\documentclass[
 10pt, aps, prx, reprint,
 amsfonts, amsmath, amssymb, 
 superscriptaddress, nofootinbib,
 floatfix
]{revtex4-2}

\usepackage[T1]{fontenc}
\usepackage[utf8]{inputenc}
\usepackage{amsthm}
\usepackage{braket}
\usepackage{enumerate}
\usepackage{xcolor}
\usepackage{graphicx}
\usepackage{enumitem}
\usepackage{thmtools}
\usepackage{thm-restate}
\usepackage{booktabs}
\usepackage{makecell}
\usepackage{microtype}
\usepackage[ruled,vlined,linesnumbered]{algorithm2e}
\usepackage{mathtools}
\usepackage{multirow}
\usepackage[hypertexnames=false]{hyperref}
\usepackage[capitalise,nameinlink]{cleveref}

\newcommand{\prep}{\textup{\textsc{Prep}}}
\newcommand{\be}{\textup{\textsc{Be}}}
\newcommand{\select}{\textup{\textsc{Sel}}}
\newcommand{\ravg}{\bar\rho}
\renewcommand{\vec}[1]{\boldsymbol{#1}}

\DeclareMathOperator{\polylog}{polylog}
\DeclareMathOperator{\tr}{tr}

\newcommand{\ketbra}[2]{\lvert#1\rangle\!\langle#2\rvert}
\newcommand{\bketbra}[2]{\big\lvert#1\big\rangle\!\big\langle#2\big\rvert}

\renewcommand{\vec}[1]{\boldsymbol{#1}}

\makeatletter
\newcommand{\apptocfile}{atoc}
\let\apptoc@orig@appendix\appendix
\renewcommand{\appendix}{%
  \apptoc@orig@appendix
  \let\apptoc@orig@addtocontents\addtocontents
  \long\def\addtocontents##1##2{%
    \def\apptoc@ext{##1}%
    \def\apptoc@toc{toc}%
    \ifx\apptoc@ext\apptoc@toc
      \apptoc@orig@addtocontents{\apptocfile}{##2}%
    \else
      \apptoc@orig@addtocontents{##1}{##2}%
    \fi
  }%
}
\newcommand{\appendixtableofcontents}{%
  \begingroup
    \setcounter{tocdepth}{3}%
    \phantomsection
    \let\addcontentsline\@gobblethree
    \section*{Contents}%
    \pdfbookmark[1]{Appendices}{apxcontents}%
    \@starttoc{\apptocfile}%
  \endgroup
}
\renewcommand\paragraph{\@startsection{paragraph}{4}{\z@}
{1ex \@plus1ex \@minus.2ex}
{-1em}
{\normalfont\normalsize\bfseries}}
\newcommand{\titlename}{\textsl{\@title}}
\makeatother

\declaretheorem{theorem}
\declaretheorem[sibling=theorem]{proposition}
\declaretheorem[sibling=theorem]{lemma}

\declaretheorem[numbered=no, name={Main Result}]{result}
\declaretheorem[style=remark, parent=theorem]{remark}
\declaretheorem[style=definition]{definition}
\declaretheorem[style=definition, sibling=definition]{assumption}

\declaretheorem[numbered=no, name=Problem]{problem}
\declaretheoremstyle[headindent=-0.5em,headfont=\bfseries,bodyfont=\normalfont,notefont=\bfseries,notebraces={}{},headpunct={.}, postheadspace=0.5em]{manualdefinition}
\declaretheorem[style=manualdefinition,numbered=no,name={}]{manualdef}

\IncMargin{1.5em}

\graphicspath{{./figures/}}

\crefname{section}{Sec.}{Secs.}

\mathchardef\mhyphen="2D

\definecolor{red4}{HTML}{B82727}
\definecolor{blue4}{HTML}{4C4CD9}
\definecolor{purple4}{HTML}{B3256C}

\allowdisplaybreaks
\hypersetup{
    colorlinks,
    citecolor=blue4,
    linkcolor=red4,
    urlcolor=purple4,
    breaklinks=true,
}

\newcommand\bbDelta{
    \Delta\kern-0.8em\scalebox{0.75}{$\Delta$}\kern0.2em}

\NewDocumentEnvironment{algoequation}{b}
  {
    \noindent
    \makebox[\hsize]{
      \raisebox{-1\height}{$\displaystyle #1$}
    }
  }
  {}

\newcommand{\OxMath}{\affiliation{Mathematical Institute, University of Oxford, Woodstock Road, Oxford OX2 6GG, United Kingdom}}
\newcommand{\WignerRCP}{\affiliation{HUN-REN Wigner Research Centre for Physics, Konkoly–Thege Mikl\'os \'ut 29-33, Budapest, H-1525, Hungary\looseness=-1}}
\newcommand{\ELTE}{\affiliation{Faculty of Informatics, E\"otv\"os Lor\'and University, P\'azm\'any P\'eter s\'et\'any 1/C, Budapest, H-1117, Hungary\looseness=-1}}
\newcommand{\QMT}{\affiliation{Quantum Motion, 9 Sterling Way, London N7 9HJ, United Kingdom}}

\begin{document}
\title{Eigenstate Preparation Through Near-Optimal Eigenprobability Filtering}

\author{Po-Wei Huang}
\email{po-wei.huang@maths.ox.ac.uk}
\OxMath
\QMT
\author{Bence Bak\'o}
\email{bako.bence@wigner.hu}
\WignerRCP
\ELTE
\OxMath
\author{B\'alint Koczor}
\email{balint.koczor@maths.ox.ac.uk}
\OxMath
\QMT
\date{\today}
\begin{abstract}
Quantum simulation is expected to be a main application of quantum computers with realistic utility in quantum chemistry, materials science and beyond. However, preparing excited or general eigenstates is a central challenge, particularly when the desired eigenvalue is not known in advance, or when the overlap with the initial state is insufficient. We introduce the \emph{Dominant Eigenstate Filtering via Eigenprobability Amplification and Thresholding} (DEFEAT) algorithm that identifies and filters the eigenstate with the largest overlap with the supplied initial state. Our key observation is that we do not need prior knowledge of the target eigenvalue, as we construct efficient twirling superoperators that map initial states to eigenprobability density operators $\rho$, diagonal in the Hamiltonian eigenbasis and encoding its spectral weights, alongside its block-encoding implementation. Our crucial innovation is the quadratic amplification of probabilities via the factorisation $\rho = \rho_{\rm sqrt}^\dagger \rho_{\rm sqrt}$, analogous to the recently introduced sum-of-squares spectral amplification (SOSSA), which we use here to amplify the separation between dominant and subdominant components. Thresholding then yields the dominant-eigenstate projector. Compared with conventional phase estimation, DEFEAT improves the dependence on the overlap with the initial state while requiring substantially fewer ancillary qubits. We prove that the query complexity of the filtering step is optimal up to logarithmic factors and establish a complementary lower bound for eigenstate preparation under purified query access to $\rho$. We validate in numerical simulations that the convergence rate of DEFEAT matches our theoretical results. Our results provide a general eigenvalue-agnostic primitive for dominant eigenstate filtering and preparation, and for estimating properties of dominant eigenstates.
\end{abstract}

\maketitle

\section{Introduction}
\label{secIntro}

For some structured problems, quantum computers promise exponential speedups over the best known classical methods. However, identifying quantum algorithms that combine substantial quantum speedups with practical or commercial relevance remains one of the central challenges in the field~\cite{babbush2026grand,zimboras2025myths}. Quantum simulation---and, in particular, the accurate preparation and characterisation of eigenstates of a target Hamiltonian---is the most natural and promising application of quantum computers. These promise substantial impact both in fundamental research, for example in many-body physics~\citep{ayral2023quantum}, and in application areas including pharmaceutical and chemical engineering~\citep{reiher2017elucidating,mcardle2020quantum} and materials science~\citep{bauer2020quantum}.

\begin{figure*}
	\includegraphics[width=0.89\textwidth]{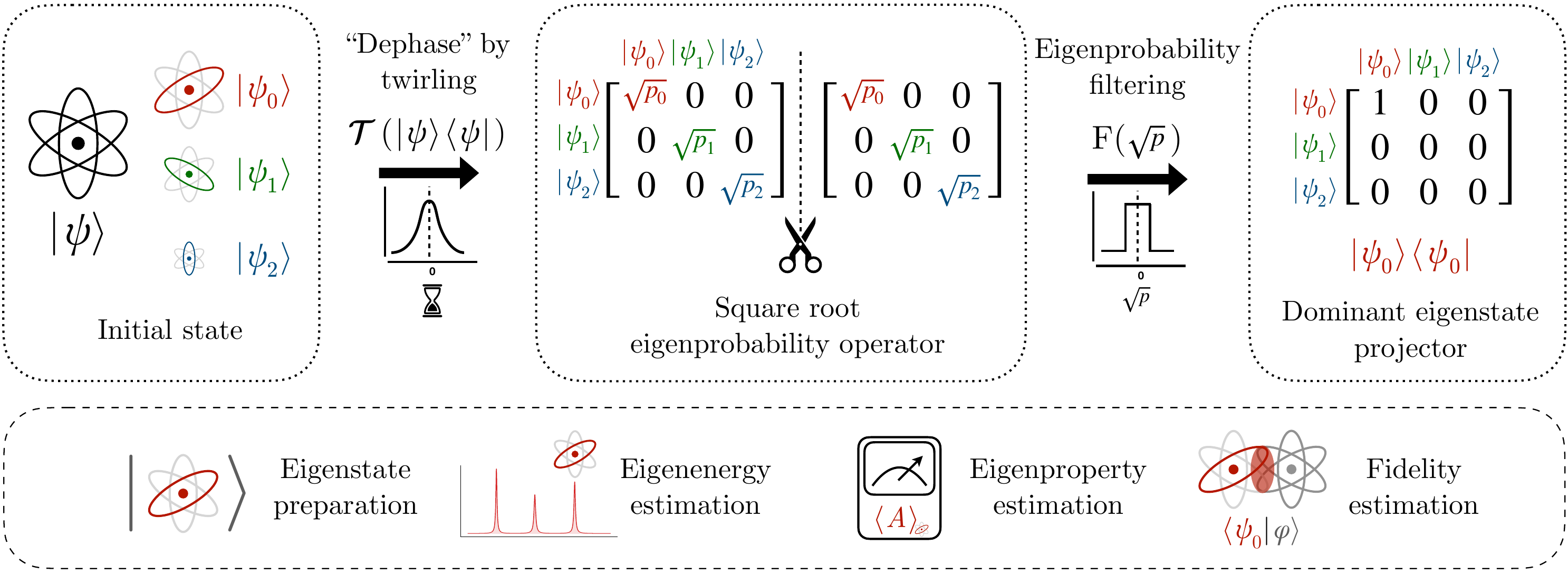}
	\caption{Overview of the \emph{Dominant Eigenstate Filtering via Eigenprobability Amplification and Thresholding} (DEFEAT) algorithm. (left) We assume access to an initial state; e.g., this is typically obtained via classical heuristics, such as DMRG~\citep{chan2011density,sharma2012spin} for quantum chemistry problems like FeMoco~\citep{reiher2017elucidating}. DEFEAT constructs a twirling superoperator to map the initial state to an eigenprobability operator $\rho$, which is diagonal in the Hamiltonian eigenbasis and encodes the spectral weights of the input state. We construct practical, efficient twirling superoperators via Gaussian-weighted averages of time-evolved states. (middle) DEFEAT exploits the factorisation $\rho = \rho_{\rm sqrt}^\dagger \rho_{\rm sqrt}$ to construct a block-encoding of the square-root of the eigenprobability operator, which amplifies the probability gap between the dominant and second-dominant eigenstate -- this is analogous to the sum-of-squares spectral amplification (SOSSA) technique of Refs.~\cite{low2025fast,king2026quantum} but with the crucial difference that SOSSA is applied to the state $\rho$ rather than to the Hamiltonian. (right) Polynomial thresholding of the amplified eigenprobability operator produces a projector onto the dominant eigenstate, enabling practical applications (bottom), such as efficient eigenstate preparation,  eigenenergy estimation, eigenproperty estimation, fidelity estimation and beyond. }
	\label{figMain}
\end{figure*}

Beyond ground states, the accurate simulation of excited states is crucial for many applications. A broad range of classical heuristics has been developed, which we briefly summarise in \cref{appLitRevClass}. These techniques underpin the interpretation and prediction of spectroscopic observables, photochemical and non-adiabatic processes, as well as charge and energy transport~\citep{dreuw2005single,krylov2008equation,casida2012progress,curchod2018ab}.
Indeed, substantial literature has emerged on quantum-computing heuristics for probing excited states, including adiabatic state-preparation schemes~\citep{aspuruguzik2005simulated,burton2025excited,lutz2026adiabatic,hwang2025preparing}, variational~\citep{higgott2019variational,jones2019variational,gocho2023excited,nakanishi2019subspace,parrish2019quantum,caditazi2024folded,boyd2022training,larose2019variational,cerezo2022variational}, or subspace and Krylov-based techniques~\citep{mcclean2017hybrid,parrish2019quantumb,motta2019determining,stair2020multireference,yoshioka2025krylov,cortes2022quantum,oumarou2025molecular}, which we briefly review in \cref{appLitRevQuant}. However, these heuristics generally do not provide rigorous performance guarantees
while the fundamental complexity of preparing excited eigenstates remains an important open problem.

Most work on quantum algorithms with provable guarantees has so far focused on ground- and low-energy-state preparation~\citep{poulin2009preparing,ge2019faster,lin2020nearoptimal,dong2022ground-state}, as well as ground-energy estimation~\citep{dong2022ground-state,lin2022heisenberg,wang2023quantum,ding2023even,babbush2018encoding}. The difficulty of preparing excited states, or more generally arbitrary eigenstates, stems from the fact that existing approaches typically require \emph{a priori} spectral information to identify and isolate eigenstates at specific energies.
This is in contrast to searching for the ground state, where the variational principle provides
a natural formulation in terms of energy minimisation. Preparing general eigenstates therefore typically requires either prior knowledge of the target eigenenergy---which can be used to resolve and isolate the desired state from the rest of the Hamiltonian spectrum~\citep{lin2020optimal,patil2026efficient}---or a guiding initial state with substantial overlap with the target eigenstate.
In the latter case, standard phase estimation~\citep{kitaev1995quantum,cleve1998quantum,nielsen2010quantum} would suffice to prepare the desired excited state with high probability~\citep{abrams1999quantum} in the first place.

However, as we show in \cref{appQPE}, textbook phase estimation is suboptimal when the guiding state's overlap with the target eigenstate is not sufficiently close to 1, or when the relevant eigenenergy is not known \emph{a priori}. Moreover, the fundamental complexity of preparing arbitrary excited states has been an open question. Efficient methods for preparing arbitrary eigenstates are therefore important both for understanding the intrinsic complexity of this task and for expanding the practical scope of quantum simulation beyond ground-state physics.

Here we develop the \emph{Dominant Eigenstate Filtering via Eigenprobability Amplification and Thresholding} (DEFEAT)
approach.
The key observation we make is that we do not
require precise knowledge of the Hamiltonian spectrum, as
we use twirling superoperators to map an initial state 
onto an eigenprobability matrix $\rho$ that is diagonal in the Hamiltonian's eigenbasis, which 
we illustrate on \cref{figMain} (left).
As $\rho$ encodes the initial state's eigenprobabilities in its diagonal,
we can use it to prepare eigenstates via a technique we term \emph{eigenprobability filtering}
that we illustrate on the \cref{figMain} (right).

A key novelty of our work is the use of sum-of-squares spectral amplification (SOSSA) techniques~\citep{low2025fast,king2026quantum}, which were introduced to quadratically amplify small \emph{energy} gaps in Hamiltonians and have been used in state-of-the-art resource estimates~\citep{somma2013spectral,low2025fast}. Here, we adapt SOSSA to a fundamentally different application: the amplification of eigenvalues of a density matrix, which we
illustrate in \cref{figMain} (middle).
We prove that the use of SOSSA improves the runtime scaling with respect to
the relevant eigenstate probability gap---SOSSA is particularly useful for probability distributions with small separation between the two largest probabilities, especially when the leading probability is not large in absolute terms.
Beyond amplification, SOSSA yields an additional constant-factor improvement as the block-encoding circuit for our amplified density matrix is shallower than the corresponding circuit without amplification.

Finally, our main result is the end-to-end algorithm DEFEAT, which takes an input state
and a Hamiltonian, and outputs
the block-encoded projector onto the eigenstate of the Hamiltonian that has the largest overlap with the initial state. 
As we state in \cref{thmFilter}, DEFEAT requires $\widetilde{\mathcal{O}} ( \Lambda^{-1} \Delta^{-1} )$ total evolution time,
$\widetilde{\mathcal{O}}(\Lambda^{-1})$
queries to the state preparation unitary $U_\psi$, and $\mathcal{O}(\log(\Delta^{-1}) + \log\log(\Lambda^{-1}))$ number of ancilla qubits.
Here, $\Lambda\le \sqrt{p_0}-\sqrt{p_1}$ denotes a lower bound for the sum-of-squares--spectral--amplified probability gap, and $\Delta$ denotes a lower bound for the minimum spectral gap between the Hamiltonian eigenstates supported over the input state.

The rest of the paper is organised as follows: \cref{secOverview} provides an overview of our methods; \cref{secEPO} describes our constructions of twirling superoperators and density operators diagonal in the Hamiltonian eigenbasis; \cref{secDefeat} introduces our filtering algorithm and details several practically important applications; \cref{secOpt} discusses optimality and query lower bounds; finally, \cref{secNum} demonstrates some algorithmic components numerically, and \cref{secDisc} concludes the work. 

\section{Overview of main results}
\label{secOverview}

An initial state $\ket{\psi}$ can be expanded into the eigenvectors
$\{\ket{\psi_k}\}_{k=0}^{N-1}$ of a problem Hamiltonian $H$
as
\begin{equation}\label{eq:initial_state}
	\ket{\psi} = \sum_{k \in \mathcal S} c_k \ket{\psi_k}.
\end{equation}
Quantum simulation algorithms typically initialise a quantum register in a
classical approximation of one of the eigenstates of $H$, and $\ket{\psi}$ is
typically only supported on a subset of all eigenvectors $\mathcal S \subseteq \{0,\dots,N-1\}$. We also define $n := \lceil \log_2 N\rceil$.

While known ground state preparation algorithms~\citep{lin2020nearoptimal,thibodeau2023nearly} prepare the block-encoded projector of 
the eigenstate $\ket{\psi_g}$ which has the least energy, 
in the present work we generalise this to the preparation of the eigenstate projector with the 
absolute largest $|c_k|$ coefficient above. Denoting spectral weights as $p_k:=|c_k|^2$, we can relabel the eigenstates in the support such that
\begin{equation}
    p_0>p_1\geq p_2\geq\cdots\geq p_{|\mathcal S|-1}>0.
    \label{eq:probs}
\end{equation}
Thus, $\ket{\psi_0}$ is the unique eigenstate having the largest overlap with the input state. Our objective is to isolate this eigenstate without assuming prior knowledge of its eigenenergy $E_0$, which we term \emph{dominant eigenstate filtering}.

\begin{problem}[Dominant eigenstate filtering]
    Assume access to a state preparation unitary $U_\psi$ that prepares the initial state $\ket{\psi} =: U_\psi \ket{0} $ defined in \cref{eq:initial_state}. The task is to construct a block-encoding of the dominant-eigenstate projector $\mathcal{P}$, or equivalently, the reflector $\mathcal{R}$, defined as
    \begin{equation*}
        \mathcal{P} := \ketbra{\psi_0}{\psi_0},
        \quad \text{and} \quad
        \mathcal{R} := \mathbb{I} - 2\ketbra{\psi_0}{\psi_0}.
    \end{equation*}
\end{problem}

The Hamiltonian, which spans the eigenspectrum through its eigenbasis and eigenenergies, is the centrepiece for modelling quantum systems, underpinning tasks such as Hamiltonian simulation~\citep{childs2012hamiltonian,low2019hamiltonian,gilyen2019quantum,low2025fast}, phase estimation~\citep{kitaev1995quantum,cleve1998quantum,nielsen2010quantum}, ground state preparation~\citep{poulin2009preparing,ge2019faster,lin2020nearoptimal,dong2022ground-state}, and ground state energy estimation~\citep{dong2022ground-state,lin2022heisenberg,wang2023quantum,ding2023even}.
However, its direct use is less suited for targeting general or excited eigenstates; unlike the ground state, which is uniquely identified by the global energy minimum, preparing excited states is substantially more involved.

For this reason, our approach is to first prepare a density matrix that we term the \emph{eigenprobability operator}: this operator
is diagonal with respect to the Hamiltonian's eigenbasis, and it encodes in its diagonal entries the probability
distribution $p_k$ that we defined above in \cref{eq:probs}. 

\begin{definition}[Eigenprobability operator]
    Given the eigenbasis $\{\ket{\psi_k}\}_{k=0}^{N-1}$ of the problem Hamiltonian $H$ and the initial state $\ket{\psi}$ in \cref{eq:initial_state} with probabilities $p_k = |c_k|^2$, we define the associated \emph{eigenprobability operator} as the dephased initial state 
      \begin{equation*}
        \rho := \mathcal{T} (\ketbra{\psi}{\psi}) = \sum_{k \in \mathcal S} p_k \ketbra{\psi_k}{\psi_k},
    \end{equation*}
    where the dephasing channel is defined as a twirling superoperator
    \begin{equation*}
    	\mathcal{T} (\,\cdot\,) :=  \sum_{k=0}^{N-1} \mathcal{P}_k  (\,\cdot\,) \mathcal{P}_k, 
    	\quad \text{with} \quad \mathcal{P}_k:= \ketbra{\psi_k}{\psi_k}.
    \end{equation*}
    \label{defEO}
\end{definition}
This operator is analogous, in the nondegenerate case, to the so-called diagonal ensemble in nonequilibrium quantum dynamics, where long-time--averaged states can be used to retain the initial probability of the eigenspectrum while eliminating cross-energy coherences~\citep{cakan2021approximating,bako2026exponential}.
In \cref{secRho}, we further detail two approximate implementations of the above twirling superoperators in the nondegenerate case: (a) a weighted average of discrete time evolution processes under the problem Hamiltonian $H$, and (b) a Chebyshev polynomial expansion. We also prove that both techniques achieve exponentially fast convergence in the approximation error.

We note that the above eigenprobability operator can be used directly as a drop-in replacement
for the role of the Hamiltonian in ground state filtering.
However, a key observation in this work is that 
the eigenprobability operator has a favourable structure that
we can exploit for developing a more efficient algorithm for dominant eigenstate preparation.

Specifically, the eigenprobability operator is both positive semi-definite and self-adjoint,
and therefore it admits a square-root factorisation in terms of a square-root eigenprobability operator
as $\rho = \rho_{\rm sqrt}^\dagger \rho_{\rm sqrt}$. We introduce a basis-dependent square-root representation as follows.
\begin{definition}[Square-root eigenprobability operator]
    Given the eigenbasis $\{\ket{\psi_k}\}_{k=0}^{N-1}$ of the problem Hamiltonian and the initial state $\ket{\psi}$ in \cref{eq:initial_state},
    the \emph{square-root eigenprobability operator} associated with $\ket{\psi}$ is defined as
    \begin{equation*}
        \rho_{\rm sqrt} := \sum_{k \in \mathcal{S}} \sqrt{p_k}\ketbra{\phi_k}{\psi_k},
    \end{equation*}
    where $\{\ket{\phi_k}\}$ is an orthonormal set and $\rho_{\rm sqrt} \in \mathbb{C}^{M\times N}$ has rank $\lvert\mathcal{S}\rvert$ with  $N = \mathrm{dim}(H)$ and $M \geq \lvert\mathcal{S}\rvert $.
    \label{defSqrt}
\end{definition}
The above factorisation was recently used for quadratically amplifying small energy gaps in Hamiltonians 
when estimating ground state properties via the SOSSA approach~\citep{somma2013spectral,low2025fast,king2026quantum,dutkiewicz2026spectral}. However, the approach requires a potentially
very expensive classical pre-computation whereby the Hamiltonian is factorised into a sum of squares of operators, due to the need to lift eigenenergies by a lower bound of the ground state energy, which is often computed via expensive SDPs.

In contrast, we apply SOSSA to 
density matrices, rather than to Hamiltonians, and our twirled eigenprobability operators $\rho$ can be naturally expressed as a sum of positive semi-definite
operators $\rho_j = B_j^\dagger B_j$, such that we obtain the  weighted sum-of-squares
factorisation as
 \begin{equation}\label{eq:sossa}
    \rho = \sum_{j} w_j \rho_j= \sum_{j} w_j B_j^\dagger B_j.
\end{equation}
Here, the operators $B_j$ are not necessarily square matrices, while $w_j > 0$ are non-negative weights of the decomposed operators.

Indeed, as we detail in \cref{secRhoSqrt}, this representation allows us to 
implement a block-encoding of the square-root eigenprobability operator in \cref{defSqrt}
as
\begin{equation}
    \rho_{\mathrm{sqrt}} = \sum_j \sqrt{w_j}\ket{j}\otimes B_j,
\end{equation}
via the operators $\{B_j\}$
 -- and we prove that this implementation only uses one query
to the initial-state preparation unitary $U_{\psi}$ or its inverse.
We also remark that \cref{defSqrt} defines $\rho_{\rm sqrt}$ in terms of its singular value decomposition (SVD), and indeed the singular values of $\rho_{\rm sqrt}$ are square roots, $\{\sqrt{p_k}\}_{k\in\mathcal{S}}$.

In \cref{secEigenFilt} we utilise quantum singular value transformation (QSVT)~\citep{gilyen2019quantum} to apply an even filtering function that filters out the non-dominant singular values, resulting in an approximation of the density matrix/projector,
$\mathcal P = \ketbra{\psi_0}{\psi_0}$
or the corresponding reflector operator
\begin{equation}
    \mathcal R = - \ketbra{\psi_0}{\psi_0} + \sum_{k\in\mathcal{S}\setminus\{0\}}\ketbra{\psi_k}{\psi_k} = \mathbb{I} - 2\ketbra{\psi_0}{\psi_0}.
\end{equation}

Following the high-level description above, we also briefly summarise assumptions required for our approach.

\begin{assumption}
We assume access to the following.

\begin{enumerate}[label=(\alph*),ref=\thedefinition(\alph*), after=\vspace{-\baselineskip}]
    \item A state preparation unitary $U_\psi$ that prepares the initial state $\ket{\psi}$
    in \cref{eq:initial_state}. 
    \label{assumptState}

    \item Controlled forward and inverse access to the time-evolution operator $U_\tau:=e^{-iH\tau}$, or forward and inverse access to the block-encoding unitary of $H$ with scaling factor $\alpha \ge \|H\|$, which we denote as $\be[\frac{H}{\alpha}]$.
    \label{assumptHam}

    \item A lower bound on the minimum spectral gap within the support $\Delta \le \min_{\substack{(j\neq k) \in\mathcal S}} |E_j-E_k|$.
    \label{assumptGap}

    \item A lower bound on the amplitude gap $\Lambda \le \sqrt{p_0}-\sqrt{p_1}.$
    \label{assumptAmpGap}

    \item Knowledge of a threshold value $\mu$ satisfying
    \begin{gather*}
        \sqrt{p_0}-\mu,\ \mu-\sqrt{p_1}\in\Omega(\Lambda).
    \end{gather*}
    \label{assumptThresh}
\end{enumerate}
\label{assumptDEFEAT}
\end{assumption}

In particular, \hyperref[assumptHam]{Assumptions~\ref*{assumptHam}} and \ref{assumptGap} are used in the construction of the twirling superoperator, and \hyperref[assumptAmpGap]{Assumptions~\ref*{assumptAmpGap}} and \ref{assumptThresh} are used in the thresholding step via QSVT. In the remaining parts of the paper, we assume that the lower bounds in \hyperref[assumptGap]{Assumptions~\ref*{assumptGap}} and \ref{assumptAmpGap} are tight, and use  $\sqrt{p_0}-\sqrt{p_1}$ to replace the use of the actual lower bound to better emphasize the complexity relations with $p_0 - p_1$, but note that in cases where knowledge of the lower bound is not tight, runtime complexities would depend on the lower bounds instead.

\begin{table*}
    \centering
    \begin{tabular}{l@{\hspace{1em}}l@{\hspace{1em}}c@{\hspace{1em}}c}
    \toprule
    Resource & Method & Eigenenergy estimation & Eigenstate preparation \\
    \midrule
    \multirow{3}{*}[-3ex]{Evolution time}
    & \makecell[l]{Textbook QPE\\(\cref{appQPE})}
    & $\displaystyle\widetilde{\mathcal{O}}\left(\frac{\|H\|}{\Gamma^3\varepsilon}\log\frac1\delta\right)$
    & $\displaystyle\widetilde{\mathcal{O}}\left(\frac{\|H\|}{\Gamma^3\sqrt{p_0}\Delta\varepsilon}\log\frac1\delta\right)$ \\
    & \makecell[l]{GaussQPE + MAE\\(\cref{appQPEMAE})}
    & $\displaystyle\widetilde{\mathcal{O}}\left(\frac{\|H\|}{\Gamma\varepsilon}\log\frac1\delta\right)$
    & $\displaystyle\widetilde{\mathcal{O}}\left(\frac{\|H\|}{\Gamma\Delta}\log\frac1{\varepsilon}\log\frac1\delta\right)$ \\
    & \makecell[l]{DEFEAT\\(\cref{secDefeat})}
    & $\displaystyle\widetilde{\mathcal{O}}\left(\left(\frac{\|H\|}{\Gamma\Delta}+\frac1\varepsilon\right)\log\frac1\delta\right)$
    & $\displaystyle\widetilde{\mathcal{O}}\left(\frac{\|H\|}{\Gamma\Delta}\log^2\frac1\varepsilon\right)$ \\
    \midrule
    \multirow{3}{*}[-3ex]{Queries to $U_\psi$}
    & Textbook QPE
    & $\displaystyle\widetilde{\mathcal{O}}\left(\frac{1}{\Gamma^2}\log\frac1\delta\right)$
    & $\displaystyle\widetilde{\mathcal{O}}\left(\frac{1}{\Gamma^2}\log\frac1\delta\right)$ \\
    & GaussQPE + MAE
    & $\displaystyle\widetilde{\mathcal{O}}\left(\frac1\Gamma\log\frac1\delta\right)$
    & $\displaystyle\widetilde{\mathcal{O}}\left(\frac1\Gamma\log\frac1\delta\right)$ \\
    &
    DEFEAT
    & $\displaystyle\widetilde{\mathcal{O}}\left(\frac1\Gamma\log\frac{1}{\delta\varepsilon}\right)$
    & $\displaystyle\widetilde{\mathcal{O}}\left(\frac1\Gamma\log\frac1\varepsilon\right)$ \\
    \midrule
    \multirow{3}{*}[-3ex]{Ancilla qubits} & Textbook QPE
    & $\displaystyle\mathcal{O}\left(\log\frac{\|H\|}{
    \Gamma\varepsilon}\right)$
    & $\displaystyle\mathcal{O}\left(\log\frac{\|H\|}{
    \Gamma\Delta\varepsilon}\right)$ \\
    & GaussQPE + MAE
    & $\displaystyle\widetilde{\mathcal{O}}\left(\frac1\Gamma+\log\frac{\|H\|}{\varepsilon}\right)$
    & $\displaystyle\widetilde{\mathcal{O}}\left(\frac1\Gamma+\log\frac{\|H\|}\Delta+\log\log\frac1{\varepsilon}\right)$ \\
    & DEFEAT
    & $\displaystyle\mathcal{O}\left(\log\frac{\|H\|}{\Delta}+\log\log\frac1\Gamma\right)$
    & $\displaystyle\mathcal{O}\left(\log\frac{\|H\|}{\Delta}+\log\log\frac1{\Gamma\varepsilon}\right)$ \\
    \bottomrule
    \end{tabular}
    \caption{\textit{Comparison to phase-estimation--based algorithms.} We compare DEFEAT for dominant eigenvalue estimation and eigenstate preparation against textbook phase estimation~\citep{abrams1999quantum}, detailed in \cref{appQPE}, and an advanced variant combining Gaussian phase estimation~\citep{rendon2024improved,rendon2023lowdepth,chen2025quantum} with multidimensional amplitude estimation~\citep{vanapeldoorn2021quantum}, detailed in \cref{appQPEMAE}. Here, we set $\Gamma \in \mathcal{O}(p_0 - p_1)$. DEFEAT improves runtime, queries to the state-preparation unitary $U_\psi$, and ancilla count over textbook phase estimation, while substantially reducing qubits relative to the advanced variant. Its additional $\log(\varepsilon^{-1})$ factor in $U_\psi$ queries for eigenenergy estimation can be traded for a $\log(\varepsilon^{-1})$ ancilla overhead by replacing statistical with QFT-based phase estimation.}
    \label{tabComp}
\end{table*}

Our main technical result is then the DEFEAT approach, which constructs the block encoding of the above projection operator and has the following algorithmic complexity.
 
\begin{result}
    Given the assumptions in \cref{assumptDEFEAT}, DEFEAT enables the preparation of the block-encoded density matrix of the dominant eigenstate $\ketbra{\psi_0}{\psi_0}$ with
    \begin{equation*}
        \widetilde{\mathcal{O}}\left(\frac{1}{\sqrt{p_0}-\sqrt{p_1}}\right).
    \end{equation*}
    queries to $U_\psi$ and its inverse.
\end{result}
As we detail in \cref{secDefeat}, DEFEAT then enables various practically important applications through the further use of amplitude amplification~\citep{brassard2002quantum}, which we also illustrate in \cref{figMain}. We further provide matching lower bounds up to logarithmic factors in \cref{secOpt}.

To showcase the performance of our algorithm, we compare it on the tasks of eigenstate preparation and eigenenergy estimation---as shown in \cref{tabComp}---to conventional phase estimation techniques, which we detail in \cref{appQPE}. Further, we propose a variant that utilises multidimensional amplitude estimation~\citep{vanapeldoorn2021quantum} to achieve a further speedup, which we derive in \cref{appQPEMAE}, and which may be of independent interest.\footnote{This algorithm also provides a quadratic speedup over the probability dependency in QMEGS~\citep{ding2024quantum} for the quantum multiple eigenvalue estimation problem~\citep{somma2019quantum}. We additionally detail these results in \cref{appQEEP}.} Notice that for these two applications, the runtime depends on $\Gamma \le p_0-p_1$ rather than on $\Lambda \le \sqrt{p_0}-\sqrt{p_1}$ due to additional amplitude amplification costs that scale with $(\sqrt{p_0})^{-1}$.

\section{Twirling superoperators and eigenprobability operators}
\label{secEPO}

\subsection{The twirling superoperator}
\label{secRho}

In this section, we detail explicit constructions of the eigenprobability operator by first developing approximations to the twirling superoperator $\mathcal{T}$ that we apply to the initial state as defined in \cref{defEO}. We consider approximations of the form of a (truncated) infinite sum of eigenvalue transforms $f(H)$ of
$H$ as
\begin{equation}
	\widetilde{\mathcal{T}} (\,\cdot\,) := \sum_{j} w_j f_j(H) [\,\cdot\,] f_j(H)^\dagger,
    \label{eq:twirl_approx}
\end{equation}
where $w_j$ are real weights. Expanding in the eigenbasis of $H$, the above equation is satisfied via the induced kernel
\begin{equation}
    \widetilde{K}(\lambda_k,\lambda_\ell) := \sum_j w_j f_j(\lambda_k) f_j^*(\lambda_\ell),
    \label{eq:kernel}
\end{equation}
where $f_j: \mathbb{R} \mapsto \mathbb{C}$ are complex functions of the eigenenergies $\lambda_k$.
Since the ideal twirling superoperator projects onto the diagonal in the eigenbasis of $H$, the corresponding ideal kernel is
\begin{equation}
    K(\lambda_k,\lambda_\ell)=\delta_{k\ell}.
\end{equation}
Thus, the problem reduces to constructing a family of spectral transformations $\{f_j\}$ and weights $\{w_j\}$ such that the induced kernel approximates the Kronecker delta above.  More realistically, we obtain approximations that evaluate to $0$, up to an exponentially small error, whenever $\vert{}\lambda_k - \lambda_l\vert{} \geq \Delta$.

Below we present two distinct twirling constructions. Although both induce dephasing of the input state, they operate under different mechanisms: the first one uses time evolution directly, while the second exploits efficient polynomial transforms in qubitisation.

We first consider \emph{time twirling}, where the primitive is the time-evolution operator $U_t = e^{-iHt}$. This approach is natural whenever Hamiltonian evolution is directly available, and also when efficient Hamiltonian simulation can be realised. In that setting, a broad range of simulation techniques are available, including product formulas~\citep{trotter1959product,suzuki1976generalized,suzuki1991general,lloyd1996universal}, randomised product formulas~\citep{campbell2019random,kiumi2025te}, and variational techniques~\citep{li2017efficient,yuan2019theory}. The resulting twirl suppresses off-diagonal terms through weighted averaging of evolution time.

We then consider \emph{Chebyshev twirling}, which instead relies on polynomial transformations of the Hamiltonian. When a block-encoding of the Hamiltonian is provided rather than the time-evolution operator, this method of twirling may be more natural, as qubitisation provides direct access to Chebyshev polynomials of the spectrum~\citep{low2019hamiltonian}. Further, compared to a time-evolution-based construction using QSVT, pure qubitisation may require less preprocessing as it avoids computation of phase angle sequences needed for simulating long time dynamics.

\subsubsection{Time twirling}
\label{secTimeTwirl}
We first consider applying time evolution which implements the function $f_t(\lambda)=e^{i\lambda t}$ in \cref{eq:twirl_approx}. We detail in \cref{appInt} that a continuous kernel can be constructed based on a continuous Gaussian filter that approximates a time averaging process~\citep{cakan2021approximating,bako2026exponential}. We also detail that this continuous time averaging process can be approximated in terms of a finite quadrature, which then enables the approximation of the twirling superoperator via finite time evolutions. The most natural finite approach uses the Gauss--Hermite quadrature; however, its scalability is constrained in practical scenarios~\cite{trefethen2022exactness}.

Alternatively, we can truncate the infinite integral over time to the finite interval $[-T, T]$, enabling Gauss--Legendre, Clenshaw--Curtis, or Riemann quadratures. In \cref{appTime}, we show that in physically relevant systems, i.e., where $\|H\|\geq1$, and $\Delta \ll 1$, the number of required quadrature points scales more efficiently for the simple Riemann sum (using a uniform step size $\tau \in \mathcal{O}(\|H\|^{-1})$) than the Gauss--Hermite quadrature. Our approach then results in the following twirling superoperator. 

\begin{definition}[Discrete time twirling]
    An approximation of the twirling superoperator in \cref{defEO} can be obtained by choosing $f_j(H)=e^{-ij H\tau}$ for some Hamiltonian that is non-degenerate over the support $\mathcal{S}$ as
    \begin{equation*}
        \widetilde{\mathcal{T}}_{\mathrm{time}}(\,\cdot\,) := \sum_{j=-M/2}^{M/2-1}w_j\, e^{-ijH\tau} [\,\cdot\,] e^{ijH\tau},
    \end{equation*}
     where the weights are given by the truncated discrete Gaussian
    $w_j = \tfrac{\tau}{\sqrt{2\pi}\sigma\mathcal Z} \exp(-\frac{(j\tau)^2}{2\sigma^2})$
    with $\sum_{j=-M/2}^{M/2-1} w_j=1$. The induced kernel is therefore
    \begin{equation*}   
        \widetilde K_{\rm time}(\lambda_k,\lambda_\ell) = \sum_{j=-M/2}^{M/2-1} w_j e^{-ij(\lambda_k - \lambda_\ell)\tau}.
    \end{equation*}
    \label{defTime}
\end{definition}

We bound the approximation error with respect to the ideal twirling superoperator in the following proposition.
\begin{restatable}[Time-twirling approximation]{proposition}{thmTimeTwirl}
Suppose the supported spectral gap satisfies $|\lambda_k-\lambda_\ell|\ge\Delta$ for all $k \ne \ell \in \mathcal{S}$.
Then achieving bounded error $\left\lVert\mathcal T(\ketbra{\psi}{\psi}) - \widetilde{\mathcal T}_{\rm time}(\ketbra{\psi}{\psi})\right\rVert \le \varepsilon$ can be achieved using the following number of quadrature points:
\begin{equation*}
    M \in \mathcal{O}\left( \frac{\|H\|}{\Delta}\log\frac1\varepsilon\right).
\end{equation*}
\label{lemma:riemann_twirl}
\end{restatable}
We defer our proof to \cref{appTime}, which proves that time twirling suppresses the off-diagonal
terms of the input state in the Hamiltonian eigenbasis.

\subsubsection{Chebyshev twirling} 
When provided with a block-encoding unitary $\be[\frac{H}{\alpha}]$ of the Hamiltonian $H$, we can naturally generate Chebyshev polynomials of $H/\alpha$ using qubitisation, where $\alpha \ge \|H\|$ is a scaling factor~\citep{low2019hamiltonian,gilyen2019quantum}. This motivates us to investigate twirling by the Chebyshev polynomial as follows.
\begin{definition}[Chebyshev twirling]
    An approximation of the twirling superoperator in \cref{defEO} is obtained by choosing  $f_j(H)=T_j(\tfrac{H}{\alpha})$ for some Hamiltonian that is non-degenerate over the support $\mathcal{S}$ as
    \begin{equation*}   \widetilde{\mathcal{T}}_{\mathrm{cheb}}(\,\cdot\,) := \sum_{j=0}^{M-1}2w_j\, T_j(\tfrac{H}{\alpha}) [\,\cdot\,] T_j(\tfrac{H}{\alpha}),
    \end{equation*}
     where the weights are given by the truncated folded discrete Gaussian
    \begin{equation*}
        w_j=\begin{cases}
        \dfrac{1}{\mathcal{Z}}, & j=0,\\[1ex]
        \dfrac{2}{\mathcal{Z}}\exp\left(-\dfrac{j^2}{2\sigma^2}\right), & 1\le j\le M-1,
        \end{cases}
    \end{equation*}
    with $\sum_{j=0}^{M-1} w_j = 1$. The induced kernel is therefore
    \begin{equation*}
        \widetilde{K}_{\mathrm{cheb}}(\lambda_k,\lambda_\ell) := \sum_{j=0}^{M-1} 2w_j T_j(\tfrac{\lambda_k}{\alpha})T_j(\tfrac{\lambda_\ell}{\alpha}).
    \end{equation*}
    \label{defCheb}
\end{definition}

We bound the approximation error in this approach via the following lemma.
\begin{restatable}[Chebyshev-twirling approximation]{lemma}{lemCheb}
Suppose the supported spectral gap satisfies $|\lambda_k-\lambda_\ell|\ge\Delta$ for all $k \ne \ell \in \mathcal{S}$ and the scaling factor $\alpha \ge \min\left\{2\|H\|, \|H\|(1 + \mathcal O(\Delta^2/\|H\|^2))\right\}$. Then achieving bounded error  $\left\lVert\mathcal T(\ketbra{\psi}{\psi}) - \widetilde{\mathcal T}_{\rm cheb}(\ketbra{\psi}{\psi})\right\rVert \le \varepsilon$
can be achieved using the following number of quadrature points:
\begin{equation*}
    M \in \mathcal{O}\left( \frac{\alpha}{\Delta}\log\frac1\varepsilon\right).
\end{equation*}

\label{lemCheb}
\end{restatable}
We defer this proof to \cref{appCheb}. This proof is similar to time-twirling, but has an additional requirement used to suppress an additional aliasing term, which appears as Chebyshev polynomials decompose purely into cosinusoids.

\subsection{Block-encoding the eigenprobability operator}

We now detail our construction of the block-encoding of the eigenprobability operator. We note that any density operator can be block-encoded naturally via its purification $\ket{\rho}$ through Lemma 25 of Ref.~\citep{gilyen2019quantum}. However, this is not necessarily compatible with the sum-of-squares decomposition of $\sum_j w_j B_j^\dagger B_j$, which is central to our approach. 

Therefore, we focus on the approximate twirling superoperators detailed in the previous subsection, which are linear combinations of positive semi-definite operators and can therefore be implemented directly via the LCU approach~\citep{childs2012hamiltonian,gilyen2019quantum}, such that the individual components in the sum-of-squares decomposition are $B_j = \bra{\psi}f_j(H)$.
For the rest of this manuscript, we focus primarily on the time-evolution--based twirling approach rather than the Chebyshev-based one, as the former applies to a wider range of settings, while we detail Chebyshev-based implementations in the Appendix. 
Specifically, using the truncated Riemann quadrature-based twirling in \cref{lemma:riemann_twirl}, the block-encoding of the eigenprobability operator can be implemented with complexity summarised in the following proposition.

\begin{restatable}[Block-encoding of $\rho$]{proposition}{propRho}
    Given a constant stepsize $\tau \in \mathcal{O}(\|H\|^{-1})$, we can construct a $(1,a+1,\varepsilon)$-block-encoding of the $n$-qubit diagonal density operator $\rho$ with 
    \begin{equation*}
        M \in\mathcal{O}\left(\frac{\|H\|}{\Delta}\log\frac{1}{\varepsilon}\right)
    \end{equation*}
   queries to $U_\tau$, where $a = \log_2 M$.
    \label{propRho}
\end{restatable}

For a proof, please refer to \cref{appRho}. A schematic depiction of the circuit implementing this block-encoding is shown in \cref{figBErho}. Here, the $\prep$ operation prepares the discretised Gaussian distribution used in the Riemann sum, i.e., 
the weights $w_j$ can be loaded into an ancillary qubit register by preparing the ancilla state $\ket{\psi_{\rm anc}} = \sum_j \sqrt{w_j} \ket{j}$. This can be achieved through standard state preparation techniques, such as  
the QSVT framework, which offers an ancilla-efficient approach, requiring merely three ancillary qubits, alongside a total gate count that scales logarithmically with respect to the target accuracy~\cite{mcardle2026quantum}. 
Furthermore, in practical scenarios, classical tensor-network methods~\cite{iaconis2024quantum} and heuristic variational techniques~\cite{zoufal2019quantum} can be used for finding an efficient Gaussian state preparation circuit.

\begin{figure}
    \centering
    \includegraphics[width=\linewidth]{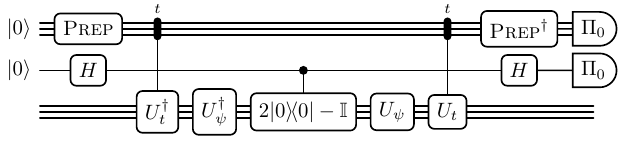}
    \caption{Block-encoding of the eigenprobability operator.}
    \label{figBErho}
\end{figure}

\begin{figure}
    \includegraphics[width=0.9\linewidth]{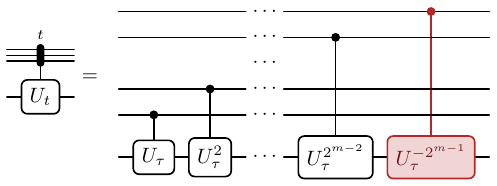}
    \caption{Implementation of multiplexed time evolution unitaries. Note that we can achieve both forward and backward time evolution by applying the highest-order controlled unitary in reverse as $U^{-2^{m-1}}$.}
    \label{figMultiplex}
\end{figure}

The second step in our implementation in \cref{figBErho} 
uses controlled time evolution operators where $U_t$ denotes the multiplexed evolution. 
While this approach has a provably favourable maximal evolution time, naively applying (controlled) time evolution operators $U_{j\tau}$ leads to an exponential number of gates. However, since the Riemann sum is defined using a fixed step-size $\tau \in \mathcal{O}(\|H\|^{-1})$, the corresponding block-encoding can be implemented using powers of $U_\tau:= e^{-iH\tau}$ with a single control on the time-register---similarly to the quantum phase estimation circuit as shown in \cref{figMultiplex}. We provide block-encoding implementations of the eigenprobability operator---generated via Chebyshev twirling for block-encoded Hamiltonians---in \cref{appRhoAlt}.

\subsection{The square-root eigenprobability operator}
\label{secRhoSqrt}
As we introduced in \cref{secOverview}, we use the SOSSA decomposition to obtain a square-root eigenprobability operator whose singular values are square roots of the singular values of $\rho$. Similar to how $\rho$ is approximated by an average of time evolved states $\widetilde{\rho}$, we can approximate $\rho_{\rm sqrt}$ directly by ``splitting'' $\widetilde{\rho}$ such that $\widetilde{\rho} = \widetilde{\rho}_{\rm sqrt}^\dagger \widetilde{\rho}_{\rm sqrt}$ as shown in \cref{figBErhoSqrt}.

Specifically, as the eigenprobability operator is approximated by a finite weighted sum of time-evolved states, we can write
\begin{equation}
    \widetilde{\rho} = \sum_j w_j e^{-iHt_j}U_\psi\ketbra{0}{0}U_\psi^\dagger e^{iHt_j},
\end{equation}
where we can set $B_j = \bra{0}U_\psi^\dagger e^{iHt_j}$. Following this, we can write the square-root eigenprobability operator as
\begin{equation}
    \widetilde{\rho}_{\rm sqrt} = \sum_j \sqrt{w_j} \ket{t_j}_{\mathcal{A}} \otimes \bra{0}_{\mathcal{B}}U_\psi^\dagger e^{iHt_j},
\end{equation}
where the circuit implementation of $\be[\widetilde{\rho}_{\rm sqrt}]$ can be found in \cref{figBErhoSqrt}. We denote the top register corresponding to the ancillary register in LCU as $\mathcal{A}$ and the main register as $\mathcal{B}$. Note that we include the tilde in the block-encoding, as the unitary block-encodes the approximate square-root eigenprobability operator $\widetilde\rho_{\rm sqrt}$ instead of the ideal $\rho_{\rm sqrt}$.
\begin{figure}
    \centering
    \includegraphics[width=\linewidth]{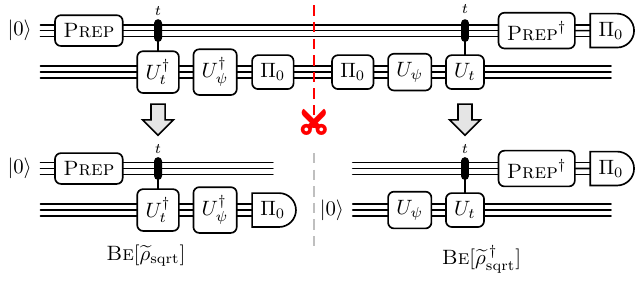}
    \caption{Block-encoding of the square-root eigenprobability operator $\rho_{\rm sqrt}$ (bottom left) and its inverse $\rho_{\rm sqrt}^\dagger$ (bottom right) obtained from splitting the block-encoding of the eigenprobability operator $\rho$ (top).}
    \label{figBErhoSqrt}
\end{figure}

\begin{figure}
    \centering
    \includegraphics[width=\linewidth]{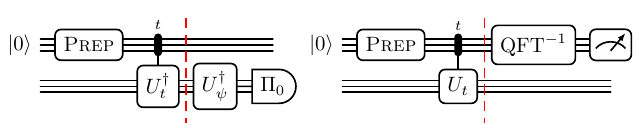}
    \caption{Comparison of the circuit implementation of the square-root eigenprobability operator $\rho_{\rm sqrt}$ (left) and quantum phase estimation (right).}
    \label{figRhoSqrtQPE}
\end{figure}

While the idea of dividing the circuit implementation of $\widetilde{\rho}$ into two halves is sufficient to prove we obtain $\widetilde{\rho}_{\rm sqrt}$,
we aim to provide additional context to facilitate conceptual understanding of how our circuit block-encodes $\widetilde{\rho}_{\rm sqrt}$. For this reason,  we illustrate how this circuit acts on an eigenstate $\ket{\psi_k}$ of the Hamiltonian and juxtapose its action with Gaussian phase estimation~\citep{rendon2024improved,rendon2023lowdepth,chen2025quantum} in \cref{figRhoSqrtQPE}. Phase estimation (in \cref{figRhoSqrtQPE}, right) applies a phase kickback~\citep{cleve1998quantum} to ``kick'' the eigenenergies encoded in the phases onto the ancillary register, and then uses an inverse quantum Fourier transform (${\rm QFT}^{-1}$)~\citep{coppersmith1994approximate} to map the eigenphases to eigenenergies. The circuit implementation of $\be[\widetilde{\rho}_{\rm sqrt}]$ in \cref{figRhoSqrtQPE} (left) is strikingly similar: it uses a Gaussian state as the initial state for the ancilla register, and applies the same phase kickback to the ancillary register, such that the joint quantum state of the main and ancillary registers is
\begin{equation}
    \bigg(\sum_j \sqrt{w_j}e^{iE_kt_j} \ket{t_j}_{\mathcal{A}}\bigg)\otimes\ket{\psi_k}_{\mathcal{B}}.
\end{equation}
Let us denote the state above on the ancilla register as $\ket{\phi_k'}_{\mathcal{A}} = \sum_j \sqrt{w_j}e^{iE_kt_j} \ket{t_j}_{\mathcal{A}}$.
The crucial difference from phase estimation is that
we do not perform readout in the ancilla register \cref{figRhoSqrtQPE} (left) and therefore the inverse
QFT is not used. Instead, the final step in our circuit in \cref{figRhoSqrtQPE} (left)
is a projection onto the initial state $\ket{\psi} = U_\psi\ket{0} = \sum_{k\in\mathcal{S}}c_k\ket{\psi_k}$ on the main register, upon which the remaining ancilla register's state becomes
\begin{equation}
    \left[\mathbb{I}_{\mathcal{A}}\otimes (\bra{0}U_\psi^\dagger)_{\mathcal{B}}\right]\ket{\phi_k'}_{\mathcal{A}}\ket{\psi_k}_{\mathcal{B}} = c_k^*\ket{\phi_k'}_{\mathcal{A}} = \sqrt{p_k}\ket{\phi_k}_{\mathcal{A}},
\end{equation}
where we absorbed the phase of $c_k^{*}$ into $\ket{\phi_k}_{\mathcal{A}}$ as the proportionality holds $c_k^{*} \propto \sqrt{p_k}$ up to a complex phase.  
The above equation therefore establishes the action as $\widetilde{\rho}_{\rm sqrt}\ket{\psi_k}=\sqrt{p_k}\ket{\phi_k}$ -- 
indeed, as  long as $\{\ket{\phi_k}\}_{k\in\mathcal{S}}$ is an orthonormal set, $\widetilde{\rho}_{\rm sqrt}$ is an appropriate implementation of $\rho_{\rm sqrt}$ as it has the identical  singular value decomposition as $\sum_{k\in\mathcal{S}}\sqrt{p_k}\ketbra{\phi_k}{\psi_k}$, as defined in \cref{defSqrt}.

It has been well established in quantum phase estimation that to estimate an eigenenergy to accuracy $\eta$, we require a maximum evolution time of $T \in\mathcal{O}(\eta^{-1})$ (considering only the dependency of the spectral gap). When $\eta \in\mathcal{O}(\Delta)$, the resulting eigenenergy estimates are sufficiently accurate that the corresponding computational basis readout states are separable and hence orthogonal. Analogously, in our setting, we can achieve the requirement of the states $\{\ket{\phi_k}\}_{k\in\mathcal{S}}$ being orthonormal with a maximum evolution time of $\mathcal{O}(\Delta^{-1})$. This recovers the $\mathcal{O}(\Delta^{-1})$ scaling in evolution time required for the implementation of the eigenprobability operator.

\begin{restatable}[Block-encoding of $\rho_{\rm sqrt}$]{proposition}{rhoSqrt}
    We can construct an explicit circuit $\be[\widetilde \rho_{\rm sqrt}]$ that provides an $\varepsilon_{\rm sqrt}$--accurate block-encoding implementation of $\rho_{\rm sqrt}$ with
    \begin{equation*}
         M \in \widetilde{\mathcal{O}}\left(\frac{\|H\|}{\Delta} \log\frac{1}{\varepsilon_{\rm sqrt}}\right)
    \end{equation*}
   queries to $U_\tau$, and requires $n + \log_2 M$ qubits to implement. 
    \label{propRhoSqrt}
\end{restatable}

We present two proofs of this in \cref{appRhoSqrt}. The first method follows the square-root decomposition and transfers the Gaussian quadrature error bounds from the previous section via the closest semi-unitary problem~\citep{higham2008functions}, and provides a general bound, but with additional explicit dependencies. The second method utilises known results on Gaussian phase estimation~\citep{rendon2023lowdepth,rendon2024improved,chen2025quantum} to provide proofs on the Riemann quadrature. 

Unlike \cref{propRho}, there is an additional logarithmic dependency that occurs in the square-root eigenprobability operator when using the Gaussian state. As the second proof reduces the construction of $\rho_{\rm sqrt}$ to Gaussian phase estimation, we can remove the additional logarithmic dependency $\mathcal{O}(\log(\Delta^{-1}))$ of the spectral gap by replacing Gaussian phase estimation with a coherent phase estimation algorithm that saturates the time complexity lower bound~\citep{mande2026tight}, such as implementations that incorporate median finding~\citep{nagaj2009fast}.

\section{Dominant eigenstate filtering and preparation}
\label{secDefeat}

\begin{algorithm*}[t]
\caption{Dominant Eigenstate Filtering via Eigenprobability Amplification and Thresholding}
\label{algoDefeat}
\Indm
\KwIn{State preparation unitary $U_\psi$ of state $\ket{\psi}$, Time evolution operator $U_\tau = e^{-iH\tau}$, Minimum support spectral gap $\Delta$, Threshold value $\mu$, Amplitude gap $\Lambda$, Accuracy $\varepsilon$ }
\KwOut{$\varepsilon$-close block-encoding of $\mathcal{R} = \mathbb{I} - 2\ketbra{\psi_0}{\psi_0}$}
\Indp
\DontPrintSemicolon
\vspace{1ex}
Prepare two registers: $\mathcal{A}$ of size $m \in \mathcal{O}(\log(\Delta^{-1})+\log\log(\varepsilon^{-1}))$, $\mathcal{B}$ of size $n = \log_2\dim(U_\tau)$, and an additional phase qubit for QSVT.\;
\For{every query to the block-encoding $\widetilde{\rho}_{\rm sqrt}$}{
    \uIf{forward access of unitary $\be[\widetilde{\rho}_{\rm sqrt}]$}{
        Prepare a truncated discrete Gaussian state in register $\mathcal{A}$ with $\sigma \in \Theta\left(\Delta^{-1}\sqrt{\log(\varepsilon^{-1})}\right)$ as follows:
        \begin{algoequation}
            \ket{0}_{\mathcal{A}} \to \sum_{j=-2^{m-1}}^{2^{m-1}-1} \sqrt{\frac{\tau}{\sqrt{2\pi}\sigma\mathcal{Z}}} e^{-\frac{(j\tau)^2}{4\sigma^2}} \ket{j}_{\mathcal{A}}.
        \end{algoequation}\;
        Apply multiplexed time evolution on $\mathcal{B}$, controlled on $\ket{j}_{\mathcal{A}}$ as follows:
        \begin{algoequation}
            \ket{j}_{\mathcal{A}} \ket{\phi}_{\mathcal{B}} \to \ket{j}_{\mathcal{A}} e^{-ijH\tau}\ket{\phi}_{\mathcal{B}}.
        \end{algoequation}\;
        Apply inverse state preparation $U_\psi^\dagger$ on $\mathcal{B}$.\;
        Apply multi-controlled Toffoli on the phase angle qubit controlled by the zero state of $\mathcal{B}$.\;
    }
    \ElseIf{inverse access of unitary  $\be[\widetilde{\rho}_{\rm sqrt}]$}{
        Execute steps 6 to 4 inversely.\;
        Apply multi-controlled Toffoli on the phase-angle qubit controlled on the zero state of $\mathcal{A}$.\;
    }
}
Apply the QSVT for the threshold function $F_{\mu}$ using a phase-angle sequence of length $\mathcal{O}(\Lambda^{-1}\log(\varepsilon^{-1}))$.\;
\Return Block-encoding unitary $\be[\mathcal{\widetilde{R}}]$ with main qubits on register $\mathcal{B}$ and ancilla qubits on register $\mathcal{A}$.
\end{algorithm*}
\begin{figure*}
\includegraphics[width=.75\linewidth]{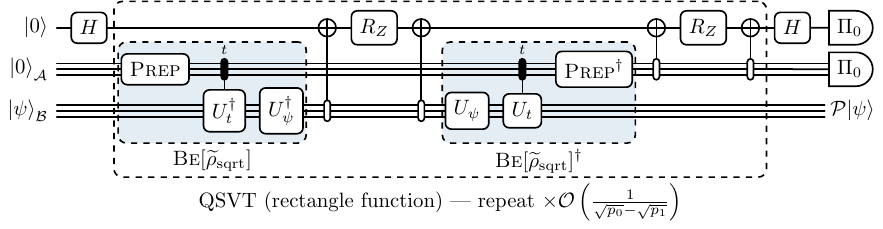}
    \caption{Full circuit implementation of DEFEAT. Note that the repeated circuit contains both $\be[\widetilde{\rho}_{\rm sqrt}]$ and $\be[\widetilde{\rho}_{\rm sqrt}]^\dagger$ as the polynomial expansion of the thresholding function is guaranteed to be even.}
    \label{figDefeat}
\end{figure*}

\subsection{The DEFEAT algorithm}
\label{secEigenFilt}

Building on the block-encoding approach of $\rho_{\rm sqrt}$ that we developed in \cref{secRhoSqrt},
we now detail our full dominant eigenstate filtering algorithm. Given a probability threshold
$\mu\in\mathbb{R}$ that separates the dominant eigenprobability from the second-largest
one as $\sqrt{p_0}\ge\mu>\sqrt{p_1}$, we define the even thresholding function as
\begin{equation}
    F_\mu(x) := \begin{cases}
        -1 & \text{if } |x| \ge \mu,\\
        1 & \text{if } |x| < \mu,
    \end{cases}
    \label{eqFilter}
\end{equation}
which we apply to the singular values of the block-encoding of $\rho_{\mathrm{sqrt}}$ using QSVT. This produces a reflector 
\begin{equation}\label{eq:reflector}
	\mathcal{R}_{\mu} =-\sum_{k: \sqrt{p_k} \ge \mu}\ketbra{\psi_k}{\psi_k} + \sum_{k: \sqrt{p_k} < \mu}\ketbra{\psi_k}{\psi_k},
\end{equation}
which equals $\mathcal{R} = \mathbb{I}-2\ketbra{\psi_0}{\psi_0}$ with a suitable $\mu$ shown above.
Alternatively, the block-encoding of $\ketbra{\psi_0}{\psi_0}$ can be prepared using similar techniques.

This serves as an analogue of the shifted sign function used in spectral thresholding techniques for ground state preparation~\citep{lin2020nearoptimal}. In that setting, the odd sign function is employed, reflecting the fact that the spectrum of the (shifted) Hamiltonian is not symmetric about zero, and one uses the sign function to distinguish between positive and negative eigenvalues relative to a threshold. In contrast, here we require a thresholding operation that performs filtering based only on the magnitude of the singular values.

Since the block-encoding maps eigenstates $\{\ket{\psi_k}\}_{k \in \mathcal{S}}$ to an orthonormal set $\{\ket{\phi_k}\}_{k \in \mathcal{S}}$, and vice versa with its inverse, the implemented polynomial must be even, i.e., composed of only even-degree terms. This ensures that the block-encoding is applied an even number of times, allowing the eigenstates to be mapped back to their original subspace.

Crucially, the construction of $\mathcal{R}_{\mu}$ requires resolving amplitudes near the threshold $\mu$. If we set $\mu \approx \frac{1}{2}(\sqrt{p_0}+\sqrt{p_1})$, the relevant distinguishing scale is set by half of the gap $\frac{1}{2}\sqrt{p_0}-\sqrt{p_1}$. The polynomial approximation of the threshold function must have a transition region narrower than this gap to correctly separate the dominant component. This leads to an overall complexity that scales with $\mathcal{O}\left((\sqrt{p_0}-\sqrt{p_1})^{-1}\right)$. We detail the quantum circuit implementation in \cref{figDefeat} and the algorithm in \cref{algoDefeat}.

We now formalise our results as follows.

\begin{restatable}[DEFEAT]{theorem}{filter}
      Given \cref{assumptDEFEAT}, DEFEAT implements an $\varepsilon$-approximate block-encoding of the projector $\mathcal{P} = \ketbra{\psi_0}{\psi_0}$ (or, equivalently, the reflection $\mathcal{R} = \mathbb{I}-2\ketbra{\psi_0}{\psi_0}$), where $\ket{\psi_0}$ denotes the dominant eigenstate. This requires
    \begin{equation*}
        \widetilde{\mathcal{O}}\left(\frac{\|H\|}{(\sqrt{p_0}-\sqrt{p_1})\Delta}\log^2\frac{1}{\varepsilon}\right),
    \end{equation*}
   queries to $U_\tau$, as well as
    \begin{equation*}
        \mathcal{O}\left(\frac{1}{\sqrt{p_0}-\sqrt{p_1}}\log\frac{1}{\varepsilon}\right)
    \end{equation*}
    queries to $U_\psi$ and 
    \begin{equation*}
        n + \mathcal{O}\left(\log\frac{\|H\|}{\Delta} + \log \log \frac{1}{(\sqrt{p_0}-\sqrt{p_1})\varepsilon}\right)
    \end{equation*}
    qubits.
    \label{thmFilter}
\end{restatable}

\begin{table*}
    \centering
    \begin{tabular}{l@{\hspace{0.8em}}c@{\hspace{0.8em}}c@{\hspace{0.8em}}c}
    \toprule
     & Evolution time & Queries to $U_\psi$ & Ancilla qubits \\
    \midrule
    \makecell[l]{Eigenenergy\\ estimation\\(\cref{propEigenenergy})}
    & $\displaystyle\widetilde{\mathcal{O}}\left(\frac{\|H\|}{(p_0-p_1)\Delta}+\frac1\varepsilon\right)$
    & $\displaystyle\widetilde{\mathcal{O}}\left(\frac{1}{p_0-p_1}\log\frac{1}{\varepsilon}\right)$
    & $\displaystyle\mathcal{O}\left(\log\frac{\|H\|}{\Delta}+\log\log\frac{1}{p_0-p_1}\right)$ \\
    \midrule
    \makecell[l]{Eigenproperty\\ estimation\\(\cref{propEigenprop})}
    & $\displaystyle\widetilde{\mathcal{O}}\left(\frac{\|H\|}{(\sqrt{p_0}-\sqrt{p_1})\Delta}\left(\frac{1}{\sqrt{p_0}}+\frac{\alpha}{\varepsilon}\right)\right)$
    & $\displaystyle\widetilde{\mathcal{O}}\left(\frac{1}{\sqrt{p_0}-\sqrt{p_1}}\left(\frac{1}{\sqrt{p_0}}+\frac{\alpha}{\varepsilon}\right)\right)$
    & $\displaystyle\mathcal{O}\left(\log \frac{\|H\|}{\Delta} + \log\log\frac{\alpha}{(p_0-p_1)\varepsilon}\right)$\\
    \midrule
    \makecell[l]{Fidelity\\ estimation\\(\cref{propFidelity})}
    & $\displaystyle\widetilde{\mathcal{O}}\left(\frac{\|H\|}{(\sqrt{p_0}-\sqrt{p_1})\Delta\varepsilon}\right)$
    & $\displaystyle\widetilde{\mathcal{O}}\left(\frac{1}{(\sqrt{p_0}-\sqrt{p_1})\varepsilon}\right)$
    & $\displaystyle\mathcal{O}\left(\log\frac{\|H\|}{\Delta} + \log\log\frac{1}{(\sqrt{p_0}-\sqrt{p_1})\varepsilon}\right)$\\
    \bottomrule
    \end{tabular}
    \caption{\emph{Resource costs of further applications of DEFEAT.} For all three tasks, we drop the dependency on the success probability $1-\delta$, which incurs an additional $\mathcal{O}(\log(\delta^{-1}))$ cost for the evolution time and query complexity of $U_\psi$. For eigenproperty estimation, we assume the observable $O$ is accessible as a block-encoding via the unitary $\be[O/\alpha]$,
    	where $\alpha$ is a scaling factor.}
    \label{tabApp}
\end{table*}

Our proof in \cref{appDefeat} is analogous to proofs for (robust) ground state filtering~\citep{lin2020nearoptimal,huang2026fullqubit}, in that we replace the ground state filter from Hamiltonian thresholding with the dominant eigenstate filter from eigenprobability thresholding. 
We also remark that using the eigenprobability operator $\rho$  in DEFEAT, rather than
our square-root operator $\rho_{\rm sqrt}$, would result in a higher overall runtime of the filtering task  that
would scale as
\begin{equation}
    \widetilde{\mathcal{O}}\left(\frac{\|H\|}{(p_0-p_1)\Delta}\log^2\frac{1}{\varepsilon}\right).
\end{equation}
Further, in the case where $\mu$ is not provided as per \hyperref[assumptThresh]{Assumption~\ref*{assumptThresh}}, we can, still obtain the block-encoded projector by searching for trial $\widetilde{\mu}$-s per binary amplitude estimation~\citep{lin2020nearoptimal}, albeit with a increased $(\sqrt{p_0})^{-1}$ runtime dependency. We detail this in \cref{appGetMu}.

\subsection{End-to-end dominant-eigenstate preparation}

Our approach for preparing the dominant eigenstate uses the reflector $\mathcal{R}$ in \cref{eq:reflector}
to reflect the phase associated with the eigenstate $\ket{\psi_0}$ in the initial state $\ket{\psi}$,
which is the basic component that enables amplitude amplification. Preparing the eigenstate
via amplitude amplification therefore requires an additional $\mathcal{O}((\sqrt{p_0})^{-1})$ repetitions of the DEFEAT circuit
as we summarise in the following statement.
\begin{restatable}[Dominant eigenstate preparation]{theorem}{eigenstate}
    Using the dominant eigenstate reflector constructed by DEFEAT (\cref{thmFilter}), we can prepare the dominant eigenstate $\ket{\psi_0}$ to square-root fidelity at least $1 - \varepsilon$ with 
    \begin{equation*}
        \widetilde{\mathcal{O}}\left(\frac{\|H\|}{(p_0-p_1)\Delta}\log^2\frac{1}{\varepsilon}\right),
    \end{equation*}
   queries to $U_\tau$, as well as
    \begin{equation*}
        \widetilde{\mathcal{O}}\left(\frac{1}{p_0-p_1}\log\frac{1}{\varepsilon}\right)
    \end{equation*}
    queries to $U_\psi$ and 
    \begin{equation*}
        n + \mathcal{O}\left(\log\frac{\|H\|}{\Delta} + \log\log\frac{1}{(p_0-p_1)\varepsilon}\right)
    \end{equation*}
    qubits.
    \label{thmEigenstate}
\end{restatable}
We defer the proof to \cref{appDesp}.
Let us again highlight that while DEFEAT depends on $(\sqrt{p_0}-\sqrt{p_1})^{-1}$, the presently considered
downstream task of eigenstate preparation picks up a factor of $(\sqrt{p_0})^{-1}$ in runtime due to amplitude amplification
which is the reason why the above complexity scales with $(p_0-p_1)^{-1}$.
However, the multiplicative factor $(\sqrt{p_0})^{-1}$ is not fundamental to eigenstate preparation, and depends on the input state:
In principle, one can apply amplitude amplification to a different initial state $\ket{\phi}$, which has higher overlap than $\ket{\psi}$, but a smaller amplitude gap such that it is not suitable for DEFEAT, but can still reduce the number of amplification steps. Further, in \cref{appDesp}, we derive an alternative construction of the algorithm such that the additional cost is $(\sqrt{\Delta})^{-1}$ instead of $(\sqrt{p_0})^{-1}$ while still using $\ket{\psi}$ as the input state.

In comparison to prior works, the substantial advantage of our approach is that it does not require \emph{a priori}
knowledge of the eigenenergy. For example, phase estimation~\citep{abrams1999quantum} can post-select on eigenstates
of a pre-specified energy, while Hamiltonian filtering~\citep{lin2020optimal} uses prior knowledge of eigenenergies
as the centre of a filtering polynomial that can isolate the relevant eigenstate.
In \cref{appQPE,appQPEMAE}, we show that our approach achieves a runtime complexity comparable to that of Gaussian phase estimation combined with multidimensional amplitude estimation, while drastically reducing the required qubit count.
Forgoing the use of multidimensional amplitude estimation
would then, however, quadratically increase the runtime in terms of the probability gap dependency.
In \cref{appHamThresh}, we also discuss an implementation of eigenstate filtering that does not require the target eigenvalue, but requires the initial overlap to be sufficiently high.

\subsection{Further applications}
We detail some further applications that build on DEFEAT and summarise their resource requirements
 in \cref{tabApp}, while deferring detailed theorem statements and proofs to \cref{appApp}.

\paragraph*{Eigenenergy estimation.} The dominant eigenenergy can be obtained by executing statistical phase estimation (QMEGS~\citep{ding2024quantum} in particular) and replacing the input state $\ket{\psi}$ with the dominant eigenstate prepared in \cref{thmEigenstate} as input. Compared to the standard QMEGS, which requires the target eigenstates to be sufficiently dominant (that is, for a single target eigenstate, to have $p_0 > 0.5$), our method only requires the eigenstate to be dominant ($p_0 > p_1$).

\paragraph*{Eigenproperty estimation.} Expected values as $\langle \psi_0 |O|\psi_0 \rangle$ can be estimated by using the Hadamard test~\citep{cleve1998quantum} and amplitude estimation~\citep{brassard2002quantum}. Apart from the initial input state, which is prepared with \cref{thmEigenstate}, subsequent reflections on the dominant eigenstate can be produced directly via the polynomial function shown in \cref{thmFilter}. Compared to prior work by \citet{bako2026exponential}, our approach provides a quadratic improvement in the dependency on $\varepsilon$, and relaxes their constraint that the initial state has to have high overlap with the dominant eigenstate.

\paragraph*{Fidelity estimation.} The fidelity between the dominant eigenstate $\ket{\psi_0}$ and another pure state $\ket{\varphi}$ can be estimated using amplitude estimation~\citep{brassard2002quantum}. Because we only require the reflector $\mathcal{R}$ in \cref{thmFilter}, in this task we avoid the extra $(\sqrt{p_0})^{-1}$ dependency present in eigenstate preparation and the other applications discussed above.

\section{Query lower bounds and optimality}
\label{secOpt}

\subsection{Optimality of filtering}
We derive quantum query lower bounds for dominant eigenstate filtering. 
Our approach 
is based on a reduction to the classical probability distribution distinction problem of \citet{belovs2019quantum},
where we can consider a scenario in which two probability distributions are encoded as quantum states whose dominant amplitudes occur at different indices.
Then an algorithm for dominant eigenstate preparation can identify these indices by preparing the corresponding dominant eigenstates. Consequently, such an algorithm could be used to distinguish the two given probability distributions. This reduction shows that dominant eigenstate preparation inherits the query complexity lower bound of the probability distribution distinction problem.
\begin{restatable}[Query lower bound for eigenstate preparation]{theorem}{filterOpt}
    Given the quantum state $\ket{\psi}$ in the form of \cref{eq:initial_state} with probabilities ordered as in \cref{eq:probs} prepared by the unitary $U_\psi$, the query complexity of
    preparing the dominant eigenstate $\ket{\psi_0}$ to square-root fidelity at least $\frac{\sqrt{3}}{2}$ without knowledge of the dominant eigenvalue satisfies the following:
    
    When $p_0 \in \Omega(1)$, the number of queries to the unitary $U_{\psi}$ is lower bounded as
    \begin{equation*}
    \Omega\left(\frac{1}{\sqrt{p_0}-\sqrt{p_1}}\right).
    \end{equation*}
    \label{thmFilterOpt}
\end{restatable}
We defer the full proof to \cref{appFilterOpt}. Note that this result is agnostic to the actual algorithm, but since our assumption $p_0 \in \Omega(1)$ implies that the number of amplification steps in the algorithm is constant, the lower bound applies directly to the filtering step of the algorithm. Comparing this lower bound with the runtime of DEFEAT given in \cref{thmFilter}, we can infer that DEFEAT achieves the optimal query complexity for filtering up to logarithmic factors. Furthermore, we detail in \cref{appFilterOpt} that a slight modification of the proof reduction shows that within the block-encoding framework itself, the filtering step is indeed near-optimal.

\subsection{Towards eigenstate preparation lower bounds}

To further elucidate the overhead of $(\sqrt{p_0})^{-1}$ associated with amplitude amplification, we consider the dominant eigenstate preparation task under a weaker model,
whereby we assume direct access to the dephased $\rho$ rather than access to its components as $\ket{\psi}$ and $U_\tau$ or $\be[\tfrac{H}{\alpha}]$. By discarding this additional structure, we can relate dominant eigenstate preparation to quantum purity amplification.

To make this connection clear, we first consider the task of quantum purity amplification~\citep{cirac1999optimal,li2024optimal}. 
In the sample-access model, given multiple copies of an unknown mixed state $\rho$, the goal is to produce a state with high fidelity to its principal eigenvector.
That is, given 
\begin{equation}
    \rho = \sum_{k \in \mathcal S} p_k \ketbra{\psi_k}{\psi_k},  
\end{equation}
we want to produce $\ket{\psi_0}$. In the following, we assume that $\rho$ refers specifically to the eigenprobability operator, rather than to a generic density matrix. \citet[Theorem 32]{grier2025streaming} has established tight sample-complexity bounds for this task, showing that any algorithm achieving fidelity $1-\varepsilon$ for the preparation of $\ket{\psi_0}$ requires
\begin{equation}
\Omega\left(\frac{1-p_0}{(p_0-p_1)^2\varepsilon}\right)
\end{equation}
copies of $\rho$ in the worst case.

To relate these bounds to our end-to-end quantum algorithm for eigenstate preparation, we consider a query-access variant of the problem, where the algorithm is given forward and inverse queries to any black-box state-preparation unitary $V$ that produces the purification of $\rho$. That is, we have
\begin{equation}
    V\ket{0}_A\ket{0}_B= \sum_{k\in \mathcal{S}} \sqrt{p_k} \ket{\psi_k}_{\mathcal{A}}\ket{\phi_k}_{\mathcal{B}},
\end{equation}
where we can recover
\begin{equation}
    \rho = \tr_{\mathcal{B}} \left(V(\ketbra{0}{0}_{\mathcal{A}}\otimes \ketbra{0}{0}_{\mathcal{B}})V^\dagger\right),
\end{equation}
by tracing out the second register. This is known in the literature as \emph{purified quantum query access}~\citep{gilyen2020distributional}. We show the following query complexity result for purity amplification based on recent work by \citet{tang2025conjugate}, which generalises sample-to-query lifting~\citep{wang2025quantum,chen2025list} techniques to state preparation oracles.

\begin{restatable}[Query lower bound for purity amplification]{theorem}{purityOpt}
    An algorithm that can prepare the principal eigenstate $\ket{\rho_0}$ to square-root fidelity at least $\frac{3\sqrt{7}}{8}$ given any state preparation unitary $V$ that provides purified quantum query access to $\rho$ has query complexity 
    \begin{equation*}
        \Omega\left(\frac{\sqrt{1-p_0}}{p_0-p_1}\right).
    \end{equation*}
    \label{thmPurity}
\end{restatable}
We show this by a simple proof by contradiction of the sample complexity simulation cost from \citet[Theorem 1.5]{tang2025conjugate}, which is deferred to \cref{appPurity}.

These results suggest that the additional overhead required for eigenstate preparation when compared to eigenstate filtering is not merely an artefact of the particular amplification procedure used in our algorithm. However, this particular lower bound does not apply to our eigenstate preparation algorithm, as our access model is more powerful than the purified quantum query access it assumes.

Indeed, our access model provides both forward and inverse access to the state preparation unitary $U_\psi$ and controlled time evolution of $H$. We can combine the two to construct a phase estimation circuit as follows, 
\begin{equation}
    {\rm QPE}_H (\textsc{Prep}\ket{0}\otimes U_\psi\ket{0} \approx \sum_{k\in \mathcal{S}} \sqrt{p_k} \ket{\widetilde E_k}\ket{\psi_k},
\end{equation}
which provides an approximate purification of the density matrix $\rho$ by tracing out the energy register.

The lower bound in \cref{thmPurity} applies to eigenstate preparation algorithms with access to any black-box purification. On the other hand, the purification above has known structure and is constructed from the separate oracles $U_\psi$ and $U_\tau$ (or $\be[\tfrac{H}{\alpha}]$), which remain available to the algorithm outside the purification circuit. Our access model is therefore strictly more structured than purified quantum query access, and the lower bound in \cref{thmPurity} cannot be transferred directly to DEFEAT.

Regardless, the lower bound for purified quantum query access provides evidence that dominant eigenstate preparation may require an additional cost beyond dominant eigenstate filtering. Establishing a matching end-to-end lower bound in the full access model of this work, including its dependence on $\sqrt{p_0}$, $\sqrt{p_0}-\sqrt{p_1}$, and the support spectral gap $\Delta$, remains an open problem.

\section{Numerics}
\label{secNum}
\begin{figure*}
    \includegraphics[width=\textwidth]{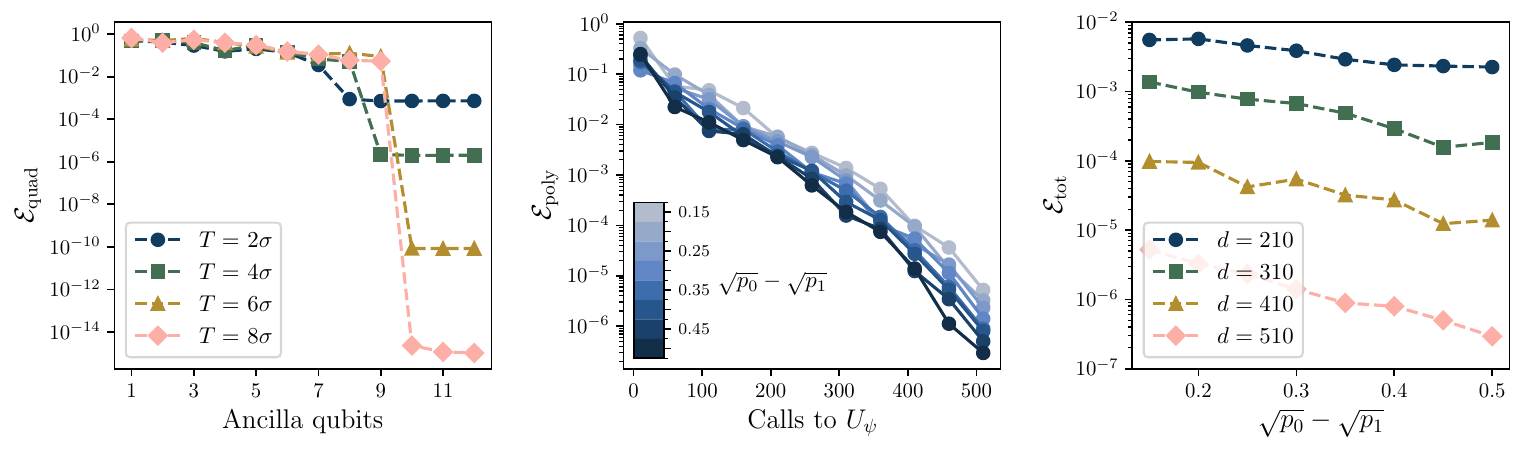}
    \caption{\emph{Simulation results.} (left) The quadrature error associated with the twirling approximation of the square-root eigenprobability operator $\mathcal{E}_{\text{quad}} = \| \rho - \widetilde{\rho}_{\rm{sqrt}}^{\dagger}\widetilde{\rho}_{\rm{sqrt}}\|$, where $\widetilde{\rho}_{\rm{sqrt}}$ is the approximation of the square-root eigenprobability operator with a $\sigma=150$ bandwidth, and varying number of ancillary qubits and cutoff $T$. (middle) The error scaling associated with the polynomial approximation of the threshold function applied to the ideal square-root eigenprobability operator $\mathcal{E}_{\text{poly}} = \| F_{\mu}(\rho_{\rm{sqrt}}) - P_{\mu}^d(\rho_{\rm{sqrt}})\|$ for different amplitude gaps obtained for a fixed $\sqrt{p_1}=0.3$ and varying dominant overlaps $\sqrt{p_0}$. Here $d$ refers to the number of calls to the state preparation $U_{\psi}$. (right) End-to-end error scaling demonstration at different $d$ calls to $U_{\psi}$, where the total error is a combination of both quadrature twirling, and threshold approximation as $\mathcal{E}_{\text{tot}} = \| F_{\mu}(\rho_{\rm{sqrt}}) - P_{\mu}^d(\widetilde{\rho}_{\rm{sqrt}})\|$. Note that here $P_{\mu}^d$ is applied to $\widetilde{\rho}_{\rm{sqrt}}$ rather than to the exact square-root eigenprobability operator.
    }
    \label{figDemos}
\end{figure*}

In this section, we numerically verify our filtering algorithm in a practical benchmarking example and illustrate the scaling of its resource requirements.
 For these experiments, we consider a $5$-qubit 1D random-field Heisenberg model with the Hamiltonian as
\begin{equation}
 H = J \sum_{j=0}^{4} \Vec{S}_j \cdot \Vec{S}_{j+1} + \sum_{j=0}^{4} h_j Z_j,
\end{equation}
where $\Vec{S}_j = \left[ X_j, Y_j, Z_j \right]$ is a Pauli vector on qubit $j$, $J=0.1$ is a coupling constant, $h_j \in [-1, 1]$ is sampled uniformly at random, and we use periodic boundary conditions. We then choose an initial state whose
dominant overlap with the ground state $\sqrt{p_0}$ we vary in the experiments, fix the 
 second most dominant overlap at $\sqrt{p_1} = 0.3$, while keeping the next $10$ energetically lowest-lying eigenstates at an equal, but non-zero overlap.
We then evaluate DEFEAT's performance while varying hyperparameters, such as the Gaussian bandwidth $\sigma$, cutoff $T$, and
the number of quadrature points $M$ (through the number of ancilla qubits).

We first demonstrate how our twirled operators approach the true eigenprobability operator as we increase the temporal cutoff $T$ and the number of ancillary qubits used for the quadrature with a fixed $\sqrt{p_0}=0.6$. The Gaussian bandwidth was chosen as $\sigma=150$ to ensure smooth convergence.
When the number of ancillary qubits used for discretisation is low,
the performance is dominated by the discretisation error, which increases with $T$. This is then suppressed, as increasing the number of ancillary qubits increases
the temporal resolution.

We then focus on the polynomial filtering approach in \cref{figDemos} (middle) and its convergence as
we increase the polynomial degree of the approximation, which requires an increased number of 
calls to $U_\psi$. Specifically, we used an approximation of the filter function in \cref{eqFilter} in terms of
the error function (${\rm erf}$) with $\mu = (\sqrt{p_0} + \sqrt{p_1})/2$. 
For this demonstration, we used the same bandwidth $\sigma=150$, $T = 8\sigma$, $\sqrt{p_0} \in\{ 0.40,\dots, 0.80\}$, and $12$ ancilla qubits.
We quantify the approximation error associated with the QSVT approach via
the distance $\mathcal{E}_{\text{poly}} = \| F_{\mu}(\rho_{\rm{sqrt}}) - P_{\mu}^d(\rho_{\rm{sqrt}})\|$, where $P_{\mu}^d$ is the polynomial approximation of
the ideal threshold function $F_{\mu}$ using $d$ calls to the state-preparation unitary. Crucially, here we are only concerned with the polynomial approximation error; therefore, we apply thresholding to the exact square-root eigenprobability operator.
\cref{figDemos} (middle) clearly demonstrates that the approximation error decreases exponentially with an increasing number of calls to the state preparation unitary. 

Finally, as an end-to-end demonstration of DEFEAT, we combine the two experiments before, and showcase the end-to-end approximation error as $\mathcal{E}_{\text{tot}} = \| F_{\mu}(\rho_{\rm{sqrt}}) - P_{\mu}^d(\widetilde{\rho}_{\rm{sqrt}})\|$. Here we use the same hyperparameters as in the previous demonstration, varying the square-root probability gap $\sqrt{p_0} - \sqrt{p_1}$. \cref{figDemos} (right) features this approximation error for an increasing number of calls to $U_{\psi}$ with $ d \in \{210, 310, 410, 510\}$ and indeed is consistent with
exponential suppression.

\section{Discussion}
\label{secDisc}

While ground-state preparation is a well-established task with optimally efficient quantum algorithms, preparing excited or more general eigenstates with optimal complexity has been an open problem.
We develop DEFEAT, which achieves near-optimal complexity for filtering the dominant eigenstate from a given input state and provides a block encoding
of this eigenstate.

Unlike conventional approaches for ground-state preparation, we do not require prior knowledge of the target eigenvalue of the Hamiltonian. Instead, we apply a twirling superoperator to prepare a density operator whose diagonal consists of the eigenprobabilities of the input state.
Another key innovation is that we apply an SOSSA decomposition to amplify the eigenprobabilities
in this density matrix, thereby mapping it to the square-root eigenprobability operator.
DEFEAT then prepares the block encoding of the relevant eigenstate by applying a threshold function. 
As we demonstrate, this block encoding can be further applied to downstream practical applications such as eigenstate preparation,
eigenenergy estimation, eigenproperty estimation, fidelity estimation, etc.

Furthermore, we solve the open problem of bounding the complexity of dominant eigenstate filtering by proving fundamental lower bounds. With this, we demonstrate that filtering via DEFEAT is optimal up to polylogarithmic factors.
However, we do note that our approach does not distinguish degenerate eigenstates, as it prepares a projector that acts
on the degenerate eigenspace rather than on a specific eigenstate. Isolating individual eigenstates within a
degenerate subspace therefore requires identifying further commuting observables, such as symmetry operators. 

Prior to our work, perhaps the most canonical method for preparing dominant eigenstates was via phase estimation methods, following the early work by \citet{abrams1999quantum}.
Here, one can first identify the target eigenenergy as the most frequent sample produced by phase estimation, and then post-select on this sample to prepare the dominant eigenstate. In \cref{appQPEMAE} we develop substantially improved variants of this approach using multidimensional amplitude estimation~\citep{vanapeldoorn2021quantum}, Gaussian phase estimation~\citep{rendon2024improved,rendon2023lowdepth,chen2025quantum}, and amplitude amplification~\citep{brassard2002quantum} that provide a speedup over na\"ive phase estimation -- however,
this improved phase estimation still requires substantially more ancillary qubits than DEFEAT.

Let us also highlight that our square-root eigenprobability operators introduce a fundamentally new application of spectral gap amplification~\citep{somma2013spectral,low2017hamiltonian,zlokapa2024hamiltonian} and sum-of-squares spectral amplification~\citep{low2025fast,king2026quantum}. Specifically, we demonstrate that these techniques extend beyond gap-amplifiable Hamiltonians and can be successfully applied to general positive semi-definite operators, such as density matrices.
We hope this observation inspires new quantum algorithms beyond Hamiltonian dynamics and provides a novel path to designing algorithms by sample-to-query lifting outside the purified quantum query access model.

Indeed, similar techniques have already been shown to be useful in algorithm design in recent years. To the best of our knowledge, an implementation of the square-root eigenprobability operator for general eigenstates has not been proposed in prior work. However, parallels can be drawn to the vector encoding framework of \citet{guo2024nonlinear} and \citet{rattew2023nonlinear}. Their approach encodes the amplitude of a quantum state in the \emph{computational basis} on the diagonal of a square block-encoding, and has been applied to applications such as continuous state preparation~\citep{rattew2023nonlinear,gonzalezconde2024efficient,ivashkov2026qkan}, the quantum acceleration of machine learning models~\citep{ivashkov2026qkan,guo2024quantum,rattew2026accelerating}, and optimal rare event sampling against the rarity threshold~\citep{guo2026quantum}.
The square-root eigenprobability operator proposed in this work can be viewed, in a sense, as a generalisation of such vector encodings to arbitrary eigenbases defined by a Hamiltonian, beyond the computational basis.

More generally, we note that this operator can also be understood as a weighted walk operator that maps eigenvectors to eigenvalues via phase kickback through multiplexed unitaries. Effectively acting as a weighted-walk analogue of phase estimation, it offers a potentially more suitable alternative for coherent applications~\citep{patel2026optimal} than using a standard phase estimation subroutine.

Finally, we remark that our framework is compatible with SOSSA-based advantages in Hamiltonian simulation. Specifically, when the target Hamiltonian used for time or Chebyshev twirling is gap-amplifiable, the resulting fast quantum simulation algorithms~\citep{low2025fast,king2026quantum} can be directly incorporated into DEFEAT. This leads to a corresponding reduction in the spectral gap dependency of the runtime complexity, complementing the improved probability-to-amplitude gap dependency established in this work.

\paragraph*{Code and data availability.}
    The source code and data for numerics can be found at \url{https://github.com/ox-quant-info/defeat}.

\begin{acknowledgments}
The authors thank Tenzan Araki and Arthur Rattew for discussions. PWH acknowledges support from the Engineering and Physical Sciences Research Council (EPSRC) Doctoral Training Partnership (EP/W524311/1) with a CASE Conversion Studentship in collaboration with Quantum Motion. PWH further acknowledges support from the Ministry of Education, Taiwan, for a Government Scholarship to Study Abroad (GSSA) and St. Catherine's College, University of Oxford, for an Alan Tayler Scholarship. BB would like to thank the support of the Hungarian National Research, Development and Innovation Office (NKFIH) through the KDP-2023 funding scheme (grant number C2245275), and the Quantum Information National Laboratory of Hungary. BK thanks UKRI for the Future Leaders Fellowship Theory to Enable Practical Quantum Advantage (MR/Y015843/1).
The numerical simulation of quantum circuits made use of the Quantum Exact Simulation Toolkit (QuEST)~\citep{jones2019quest} via the QuESTlink~\citep{jones2020questlink} frontend.
 For the purpose of Open Access, the authors have applied a CC BY public copyright license to any Author Accepted Manuscript version arising from this submission. The authors used ChatGPT and Gemini to assist in drafting and editing portions of the manuscript for language, clarity, and style. All scientific content, interpretation, and conclusions were developed and verified by the authors.
\end{acknowledgments}

\bibliography{bibliography}

\begin{thebibliography}{150}%
\makeatletter
\providecommand \@ifxundefined [1]{%
 \@ifx{#1\undefined}
}%
\providecommand \@ifnum [1]{%
 \ifnum #1\expandafter \@firstoftwo
 \else \expandafter \@secondoftwo
 \fi
}%
\providecommand \@ifx [1]{%
 \ifx #1\expandafter \@firstoftwo
 \else \expandafter \@secondoftwo
 \fi
}%
\providecommand \natexlab [1]{#1}%
\providecommand \enquote  [1]{``#1''}%
\providecommand \bibnamefont  [1]{#1}%
\providecommand \bibfnamefont [1]{#1}%
\providecommand \citenamefont [1]{#1}%
\providecommand \href@noop [0]{\@secondoftwo}%
\providecommand \href [0]{\begingroup \@sanitize@url \@href}%
\providecommand \@href[1]{\@@startlink{#1}\@@href}%
\providecommand \@@href[1]{\endgroup#1\@@endlink}%
\providecommand \@sanitize@url [0]{\catcode `\\12\catcode `\$12\catcode `\&12\catcode `\#12\catcode `\^12\catcode `\_12\catcode `\%12\relax}%
\providecommand \@@startlink[1]{}%
\providecommand \@@endlink[0]{}%
\providecommand \url  [0]{\begingroup\@sanitize@url \@url }%
\providecommand \@url [1]{\endgroup\@href {#1}{\urlprefix }}%
\providecommand \urlprefix  [0]{URL }%
\providecommand \Eprint [0]{\href }%
\providecommand \doibase [0]{https://doi.org/}%
\providecommand \selectlanguage [0]{\@gobble}%
\providecommand \bibinfo  [0]{\@secondoftwo}%
\providecommand \bibfield  [0]{\@secondoftwo}%
\providecommand \translation [1]{[#1]}%
\providecommand \BibitemOpen [0]{}%
\providecommand \bibitemStop [0]{}%
\providecommand \bibitemNoStop [0]{.\EOS\space}%
\providecommand \EOS [0]{\spacefactor3000\relax}%
\providecommand \BibitemShut  [1]{\csname bibitem#1\endcsname}%
\let\auto@bib@innerbib\@empty
%</preamble>
\bibitem [{\citenamefont {Babbush}\ \emph {et~al.}(2026)\citenamefont {Babbush}, \citenamefont {King}, \citenamefont {Boixo}, \citenamefont {Huggins}, \citenamefont {Khattar}, \citenamefont {Low}, \citenamefont {McClean}, \citenamefont {O'Brien},\ and\ \citenamefont {Rubin}}]{babbush2026grand}%
  \BibitemOpen
  \bibfield  {author} {\bibinfo {author} {\bibfnamefont {R.}~\bibnamefont {Babbush}}, \bibinfo {author} {\bibfnamefont {R.}~\bibnamefont {King}}, \bibinfo {author} {\bibfnamefont {S.}~\bibnamefont {Boixo}}, \bibinfo {author} {\bibfnamefont {W.}~\bibnamefont {Huggins}}, \bibinfo {author} {\bibfnamefont {T.}~\bibnamefont {Khattar}}, \bibinfo {author} {\bibfnamefont {G.~H.}\ \bibnamefont {Low}}, \bibinfo {author} {\bibfnamefont {J.~R.}\ \bibnamefont {McClean}}, \bibinfo {author} {\bibfnamefont {T.}~\bibnamefont {O'Brien}},\ and\ \bibinfo {author} {\bibfnamefont {N.~C.}\ \bibnamefont {Rubin}},\ }\bibfield  {title} {\bibinfo {title} {Grand challenge of quantum applications},\ }\href {https://doi.org/10.1103/6r9l-lynr} {\bibfield  {journal} {\bibinfo  {journal} {PRX Quantum}\ }\textbf {\bibinfo {volume} {7}},\ \bibinfo {pages} {020101} (\bibinfo {year} {2026})}\BibitemShut {NoStop}%
\bibitem [{\citenamefont {Zimborás}\ \emph {et~al.}(2025)\citenamefont {Zimborás}, \citenamefont {Koczor}, \citenamefont {Holmes}, \citenamefont {Borrelli}, \citenamefont {Gilyén}, \citenamefont {Huang}, \citenamefont {Cai}, \citenamefont {Acín}, \citenamefont {Aolita}, \citenamefont {Banchi}, \citenamefont {Brandão}, \citenamefont {Cavalcanti}, \citenamefont {Cubitt}, \citenamefont {Filippov}, \citenamefont {García-Pérez}, \citenamefont {Goold}, \citenamefont {Kálmán}, \citenamefont {Kyoseva}, \citenamefont {Rossi}, \citenamefont {Sokolov}, \citenamefont {Tavernelli},\ and\ \citenamefont {Maniscalco}}]{zimboras2025myths}%
  \BibitemOpen
  \bibfield  {author} {\bibinfo {author} {\bibfnamefont {Z.}~\bibnamefont {Zimborás}}, \bibinfo {author} {\bibfnamefont {B.}~\bibnamefont {Koczor}}, \bibinfo {author} {\bibfnamefont {Z.}~\bibnamefont {Holmes}}, \bibinfo {author} {\bibfnamefont {E.-M.}\ \bibnamefont {Borrelli}}, \bibinfo {author} {\bibfnamefont {A.}~\bibnamefont {Gilyén}}, \bibinfo {author} {\bibfnamefont {H.-Y.}\ \bibnamefont {Huang}}, \bibinfo {author} {\bibfnamefont {Z.}~\bibnamefont {Cai}}, \bibinfo {author} {\bibfnamefont {A.}~\bibnamefont {Acín}}, \bibinfo {author} {\bibfnamefont {L.}~\bibnamefont {Aolita}}, \bibinfo {author} {\bibfnamefont {L.}~\bibnamefont {Banchi}}, \bibinfo {author} {\bibfnamefont {F.~G. S.~L.}\ \bibnamefont {Brandão}}, \bibinfo {author} {\bibfnamefont {D.}~\bibnamefont {Cavalcanti}}, \bibinfo {author} {\bibfnamefont {T.}~\bibnamefont {Cubitt}}, \bibinfo {author} {\bibfnamefont {S.~N.}\ \bibnamefont {Filippov}}, \bibinfo {author} {\bibfnamefont {G.}~\bibnamefont {García-Pérez}}, \bibinfo {author} {\bibfnamefont {J.}~\bibnamefont {Goold}}, \bibinfo {author} {\bibfnamefont {O.}~\bibnamefont {Kálmán}}, \bibinfo {author} {\bibfnamefont {E.}~\bibnamefont {Kyoseva}}, \bibinfo {author} {\bibfnamefont {M.~A.~C.}\ \bibnamefont {Rossi}}, \bibinfo {author} {\bibfnamefont {B.}~\bibnamefont {Sokolov}}, \bibinfo {author} {\bibfnamefont {I.}~\bibnamefont {Tavernelli}},\ and\ \bibinfo {author} {\bibfnamefont {S.}~\bibnamefont {Maniscalco}},\ }\href@noop {} {\bibinfo {title} {Myths around quantum computation before full fault tolerance: What no-go theorems rule out and what they don't}} (\bibinfo {year} {2025}),\ \Eprint {https://arxiv.org/abs/2501.05694} {arXiv:2501.05694 [quant-ph]} \BibitemShut {NoStop}%
\bibitem [{\citenamefont {Ayral}\ \emph {et~al.}(2023)\citenamefont {Ayral}, \citenamefont {Besserve}, \citenamefont {Lacroix},\ and\ \citenamefont {Ruiz~Guzman}}]{ayral2023quantum}%
  \BibitemOpen
  \bibfield  {author} {\bibinfo {author} {\bibfnamefont {T.}~\bibnamefont {Ayral}}, \bibinfo {author} {\bibfnamefont {P.}~\bibnamefont {Besserve}}, \bibinfo {author} {\bibfnamefont {D.}~\bibnamefont {Lacroix}},\ and\ \bibinfo {author} {\bibfnamefont {E.~A.}\ \bibnamefont {Ruiz~Guzman}},\ }\bibfield  {title} {\bibinfo {title} {Quantum computing with and for many-body physics},\ }\href {https://doi.org/10.1140/epja/s10050-023-01141-1} {\bibfield  {journal} {\bibinfo  {journal} {Eur. Phys. J. A}\ }\textbf {\bibinfo {volume} {59}},\ \bibinfo {pages} {227} (\bibinfo {year} {2023})}\BibitemShut {NoStop}%
\bibitem [{\citenamefont {Reiher}\ \emph {et~al.}(2017)\citenamefont {Reiher}, \citenamefont {Wiebe}, \citenamefont {Svore}, \citenamefont {Wecker},\ and\ \citenamefont {Troyer}}]{reiher2017elucidating}%
  \BibitemOpen
  \bibfield  {author} {\bibinfo {author} {\bibfnamefont {M.}~\bibnamefont {Reiher}}, \bibinfo {author} {\bibfnamefont {N.}~\bibnamefont {Wiebe}}, \bibinfo {author} {\bibfnamefont {K.~M.}\ \bibnamefont {Svore}}, \bibinfo {author} {\bibfnamefont {D.}~\bibnamefont {Wecker}},\ and\ \bibinfo {author} {\bibfnamefont {M.}~\bibnamefont {Troyer}},\ }\bibfield  {title} {\bibinfo {title} {Elucidating reaction mechanisms on quantum computers},\ }\href {https://doi.org/10.1073/pnas.1619152114} {\bibfield  {journal} {\bibinfo  {journal} {Proc. Natl. Acad. Sci. U.S.A.}\ }\textbf {\bibinfo {volume} {114}},\ \bibinfo {pages} {7555} (\bibinfo {year} {2017})}\BibitemShut {NoStop}%
\bibitem [{\citenamefont {McArdle}\ \emph {et~al.}(2020)\citenamefont {McArdle}, \citenamefont {Endo}, \citenamefont {Aspuru-Guzik}, \citenamefont {Benjamin},\ and\ \citenamefont {Yuan}}]{mcardle2020quantum}%
  \BibitemOpen
  \bibfield  {author} {\bibinfo {author} {\bibfnamefont {S.}~\bibnamefont {McArdle}}, \bibinfo {author} {\bibfnamefont {S.}~\bibnamefont {Endo}}, \bibinfo {author} {\bibfnamefont {A.}~\bibnamefont {Aspuru-Guzik}}, \bibinfo {author} {\bibfnamefont {S.~C.}\ \bibnamefont {Benjamin}},\ and\ \bibinfo {author} {\bibfnamefont {X.}~\bibnamefont {Yuan}},\ }\bibfield  {title} {\bibinfo {title} {Quantum computational chemistry},\ }\href {https://doi.org/10.1103/RevModPhys.92.015003} {\bibfield  {journal} {\bibinfo  {journal} {Rev. Mod. Phys.}\ }\textbf {\bibinfo {volume} {92}},\ \bibinfo {pages} {015003} (\bibinfo {year} {2020})}\BibitemShut {NoStop}%
\bibitem [{\citenamefont {Bauer}\ \emph {et~al.}(2020)\citenamefont {Bauer}, \citenamefont {Bravyi}, \citenamefont {Motta},\ and\ \citenamefont {Chan}}]{bauer2020quantum}%
  \BibitemOpen
  \bibfield  {author} {\bibinfo {author} {\bibfnamefont {B.}~\bibnamefont {Bauer}}, \bibinfo {author} {\bibfnamefont {S.}~\bibnamefont {Bravyi}}, \bibinfo {author} {\bibfnamefont {M.}~\bibnamefont {Motta}},\ and\ \bibinfo {author} {\bibfnamefont {G.~K.-L.}\ \bibnamefont {Chan}},\ }\bibfield  {title} {\bibinfo {title} {Quantum algorithms for quantum chemistry and quantum materials science},\ }\href {https://doi.org/10.1021/acs.chemrev.9b00829} {\bibfield  {journal} {\bibinfo  {journal} {Chem. Rev.}\ }\textbf {\bibinfo {volume} {120}},\ \bibinfo {pages} {12685–12717} (\bibinfo {year} {2020})}\BibitemShut {NoStop}%
\bibitem [{\citenamefont {Chan}\ and\ \citenamefont {Sharma}(2011)}]{chan2011density}%
  \BibitemOpen
  \bibfield  {author} {\bibinfo {author} {\bibfnamefont {G.~K.-L.}\ \bibnamefont {Chan}}\ and\ \bibinfo {author} {\bibfnamefont {S.}~\bibnamefont {Sharma}},\ }\bibfield  {title} {\bibinfo {title} {The density matrix renormalization group in quantum chemistry},\ }\href {https://doi.org/10.1146/annurev-physchem-032210-103338} {\bibfield  {journal} {\bibinfo  {journal} {Annu. Rev. Phys. Chem.}\ }\textbf {\bibinfo {volume} {62}},\ \bibinfo {pages} {465–481} (\bibinfo {year} {2011})}\BibitemShut {NoStop}%
\bibitem [{\citenamefont {Sharma}\ and\ \citenamefont {Chan}(2012)}]{sharma2012spin}%
  \BibitemOpen
  \bibfield  {author} {\bibinfo {author} {\bibfnamefont {S.}~\bibnamefont {Sharma}}\ and\ \bibinfo {author} {\bibfnamefont {G.~K.-L.}\ \bibnamefont {Chan}},\ }\bibfield  {title} {\bibinfo {title} {Spin-adapted density matrix renormalization group algorithms for quantum chemistry},\ }\href {https://doi.org/10.1063/1.3695642} {\bibfield  {journal} {\bibinfo  {journal} {J. Chem. Phys.}\ }\textbf {\bibinfo {volume} {136}},\ \bibinfo {pages} {124121} (\bibinfo {year} {2012})}\BibitemShut {NoStop}%
\bibitem [{\citenamefont {Low}\ \emph {et~al.}(2025)\citenamefont {Low}, \citenamefont {King}, \citenamefont {Berry}, \citenamefont {Han}, \citenamefont {DePrince}, \citenamefont {White}, \citenamefont {Babbush}, \citenamefont {Somma},\ and\ \citenamefont {Rubin}}]{low2025fast}%
  \BibitemOpen
  \bibfield  {author} {\bibinfo {author} {\bibfnamefont {G.~H.}\ \bibnamefont {Low}}, \bibinfo {author} {\bibfnamefont {R.}~\bibnamefont {King}}, \bibinfo {author} {\bibfnamefont {D.~W.}\ \bibnamefont {Berry}}, \bibinfo {author} {\bibfnamefont {Q.}~\bibnamefont {Han}}, \bibinfo {author} {\bibfnamefont {A.~E.}\ \bibnamefont {DePrince}}, \bibinfo {author} {\bibfnamefont {A.~F.}\ \bibnamefont {White}}, \bibinfo {author} {\bibfnamefont {R.}~\bibnamefont {Babbush}}, \bibinfo {author} {\bibfnamefont {R.~D.}\ \bibnamefont {Somma}},\ and\ \bibinfo {author} {\bibfnamefont {N.~C.}\ \bibnamefont {Rubin}},\ }\bibfield  {title} {\bibinfo {title} {Fast quantum simulation of electronic structure by spectral amplification},\ }\href {https://doi.org/10.1103/pb2g-j9cw} {\bibfield  {journal} {\bibinfo  {journal} {Phys. Rev. X}\ }\textbf {\bibinfo {volume} {15}},\ \bibinfo {pages} {041016} (\bibinfo {year} {2025})}\BibitemShut {NoStop}%
\bibitem [{\citenamefont {King}\ \emph {et~al.}(2026)\citenamefont {King}, \citenamefont {Low}, \citenamefont {Babbush}, \citenamefont {Somma},\ and\ \citenamefont {Rubin}}]{king2026quantum}%
  \BibitemOpen
  \bibfield  {author} {\bibinfo {author} {\bibfnamefont {R.}~\bibnamefont {King}}, \bibinfo {author} {\bibfnamefont {G.~H.}\ \bibnamefont {Low}}, \bibinfo {author} {\bibfnamefont {R.}~\bibnamefont {Babbush}}, \bibinfo {author} {\bibfnamefont {R.~D.}\ \bibnamefont {Somma}},\ and\ \bibinfo {author} {\bibfnamefont {N.~C.}\ \bibnamefont {Rubin}},\ }\bibfield  {title} {\bibinfo {title} {Quantum simulation with sum-of-squares spectral amplification},\ }\href {https://doi.org/10.1103/m3fj-m4rm} {\bibfield  {journal} {\bibinfo  {journal} {Phys. Rev. Lett.}\ }\textbf {\bibinfo {volume} {136}},\ \bibinfo {pages} {110601} (\bibinfo {year} {2026})}\BibitemShut {NoStop}%
\bibitem [{\citenamefont {Dreuw}\ and\ \citenamefont {Head-Gordon}(2005)}]{dreuw2005single}%
  \BibitemOpen
  \bibfield  {author} {\bibinfo {author} {\bibfnamefont {A.}~\bibnamefont {Dreuw}}\ and\ \bibinfo {author} {\bibfnamefont {M.}~\bibnamefont {Head-Gordon}},\ }\bibfield  {title} {\bibinfo {title} {Single-reference ab initio methods for the calculation of excited states of large molecules},\ }\href {https://doi.org/10.1021/cr0505627} {\bibfield  {journal} {\bibinfo  {journal} {Chem. Rev.}\ }\textbf {\bibinfo {volume} {105}},\ \bibinfo {pages} {4009–4037} (\bibinfo {year} {2005})}\BibitemShut {NoStop}%
\bibitem [{\citenamefont {Krylov}(2008)}]{krylov2008equation}%
  \BibitemOpen
  \bibfield  {author} {\bibinfo {author} {\bibfnamefont {A.~I.}\ \bibnamefont {Krylov}},\ }\bibfield  {title} {\bibinfo {title} {Equation-of-motion coupled-cluster methods for open-shell and electronically excited species: The hitchhiker's guide to {Fock} space},\ }\href {https://doi.org/10.1146/annurev.physchem.59.032607.093602} {\bibfield  {journal} {\bibinfo  {journal} {Annu. Rev. Phys. Chem.}\ }\textbf {\bibinfo {volume} {59}},\ \bibinfo {pages} {433–462} (\bibinfo {year} {2008})}\BibitemShut {NoStop}%
\bibitem [{\citenamefont {Casida}\ and\ \citenamefont {Huix-Rotllant}(2012)}]{casida2012progress}%
  \BibitemOpen
  \bibfield  {author} {\bibinfo {author} {\bibfnamefont {M.}~\bibnamefont {Casida}}\ and\ \bibinfo {author} {\bibfnamefont {M.}~\bibnamefont {Huix-Rotllant}},\ }\bibfield  {title} {\bibinfo {title} {Progress in time-dependent density-functional theory},\ }\href {https://doi.org/10.1146/annurev-physchem-032511-143803} {\bibfield  {journal} {\bibinfo  {journal} {Annu. Rev. Phys. Chem.}\ }\textbf {\bibinfo {volume} {63}},\ \bibinfo {pages} {287–323} (\bibinfo {year} {2012})}\BibitemShut {NoStop}%
\bibitem [{\citenamefont {Curchod}\ and\ \citenamefont {Mart\'{\i}nez}(2018)}]{curchod2018ab}%
  \BibitemOpen
  \bibfield  {author} {\bibinfo {author} {\bibfnamefont {B.~F.~E.}\ \bibnamefont {Curchod}}\ and\ \bibinfo {author} {\bibfnamefont {T.~J.}\ \bibnamefont {Mart\'{\i}nez}},\ }\bibfield  {title} {\bibinfo {title} {Ab initio nonadiabatic quantum molecular dynamics},\ }\href {https://doi.org/10.1021/acs.chemrev.7b00423} {\bibfield  {journal} {\bibinfo  {journal} {Chem. Rev.}\ }\textbf {\bibinfo {volume} {118}},\ \bibinfo {pages} {3305–3336} (\bibinfo {year} {2018})}\BibitemShut {NoStop}%
\bibitem [{\citenamefont {Aspuru-Guzik}\ \emph {et~al.}(2005)\citenamefont {Aspuru-Guzik}, \citenamefont {Dutoi}, \citenamefont {Love},\ and\ \citenamefont {Head-Gordon}}]{aspuruguzik2005simulated}%
  \BibitemOpen
  \bibfield  {author} {\bibinfo {author} {\bibfnamefont {A.}~\bibnamefont {Aspuru-Guzik}}, \bibinfo {author} {\bibfnamefont {A.~D.}\ \bibnamefont {Dutoi}}, \bibinfo {author} {\bibfnamefont {P.~J.}\ \bibnamefont {Love}},\ and\ \bibinfo {author} {\bibfnamefont {M.}~\bibnamefont {Head-Gordon}},\ }\bibfield  {title} {\bibinfo {title} {Simulated quantum computation of molecular energies},\ }\href {https://doi.org/10.1126/science.1113479} {\bibfield  {journal} {\bibinfo  {journal} {Science}\ }\textbf {\bibinfo {volume} {309}},\ \bibinfo {pages} {1704} (\bibinfo {year} {2005})}\BibitemShut {NoStop}%
\bibitem [{\citenamefont {Burton}\ and\ \citenamefont {Filip}(2025)}]{burton2025excited}%
  \BibitemOpen
  \bibfield  {author} {\bibinfo {author} {\bibfnamefont {H.~G.~A.}\ \bibnamefont {Burton}}\ and\ \bibinfo {author} {\bibfnamefont {M.-A.}\ \bibnamefont {Filip}},\ }\href@noop {} {\bibinfo {title} {Excited state preparation on a quantum computer through adiabatic light-matter coupling}} (\bibinfo {year} {2025}),\ \Eprint {https://arxiv.org/abs/2511.22324} {arXiv:2511.22324 [quant-ph]} \BibitemShut {NoStop}%
\bibitem [{\citenamefont {Lutz}\ \emph {et~al.}(2026)\citenamefont {Lutz}, \citenamefont {Piroli}, \citenamefont {Styliaris},\ and\ \citenamefont {Cirac}}]{lutz2026adiabatic}%
  \BibitemOpen
  \bibfield  {author} {\bibinfo {author} {\bibfnamefont {M.}~\bibnamefont {Lutz}}, \bibinfo {author} {\bibfnamefont {L.}~\bibnamefont {Piroli}}, \bibinfo {author} {\bibfnamefont {G.}~\bibnamefont {Styliaris}},\ and\ \bibinfo {author} {\bibfnamefont {J.~I.}\ \bibnamefont {Cirac}},\ }\bibfield  {title} {\bibinfo {title} {Adiabatic quantum state preparation in integrable models},\ }\href {https://doi.org/10.22331/q-2026-03-18-2032} {\bibfield  {journal} {\bibinfo  {journal} {Quantum}\ }\textbf {\bibinfo {volume} {10}},\ \bibinfo {pages} {2032} (\bibinfo {year} {2026})}\BibitemShut {NoStop}%
\bibitem [{\citenamefont {Hwang}\ and\ \citenamefont {Koczor}(2025)}]{hwang2025preparing}%
  \BibitemOpen
  \bibfield  {author} {\bibinfo {author} {\bibfnamefont {W.}~\bibnamefont {Hwang}}\ and\ \bibinfo {author} {\bibfnamefont {B.}~\bibnamefont {Koczor}},\ }\bibfield  {title} {\bibinfo {title} {Preparing ground and excited states using adiabatic {CoVaR}},\ }\href {https://doi.org/10.1088/1367-2630/adb173} {\bibfield  {journal} {\bibinfo  {journal} {New J. Phys.}\ }\textbf {\bibinfo {volume} {27}},\ \bibinfo {pages} {023025} (\bibinfo {year} {2025})}\BibitemShut {NoStop}%
\bibitem [{\citenamefont {Higgott}\ \emph {et~al.}(2019)\citenamefont {Higgott}, \citenamefont {Wang},\ and\ \citenamefont {Brierley}}]{higgott2019variational}%
  \BibitemOpen
  \bibfield  {author} {\bibinfo {author} {\bibfnamefont {O.}~\bibnamefont {Higgott}}, \bibinfo {author} {\bibfnamefont {D.}~\bibnamefont {Wang}},\ and\ \bibinfo {author} {\bibfnamefont {S.}~\bibnamefont {Brierley}},\ }\bibfield  {title} {\bibinfo {title} {Variational quantum computation of excited states},\ }\href {https://doi.org/10.22331/q-2019-07-01-156} {\bibfield  {journal} {\bibinfo  {journal} {Quantum}\ }\textbf {\bibinfo {volume} {3}},\ \bibinfo {pages} {156} (\bibinfo {year} {2019})}\BibitemShut {NoStop}%
\bibitem [{\citenamefont {Jones}\ \emph {et~al.}(2019{\natexlab{a}})\citenamefont {Jones}, \citenamefont {Endo}, \citenamefont {McArdle}, \citenamefont {Yuan},\ and\ \citenamefont {Benjamin}}]{jones2019variational}%
  \BibitemOpen
  \bibfield  {author} {\bibinfo {author} {\bibfnamefont {T.}~\bibnamefont {Jones}}, \bibinfo {author} {\bibfnamefont {S.}~\bibnamefont {Endo}}, \bibinfo {author} {\bibfnamefont {S.}~\bibnamefont {McArdle}}, \bibinfo {author} {\bibfnamefont {X.}~\bibnamefont {Yuan}},\ and\ \bibinfo {author} {\bibfnamefont {S.~C.}\ \bibnamefont {Benjamin}},\ }\bibfield  {title} {\bibinfo {title} {Variational quantum algorithms for discovering {Hamiltonian} spectra},\ }\href {https://doi.org/10.1103/PhysRevA.99.062304} {\bibfield  {journal} {\bibinfo  {journal} {Phys. Rev. A}\ }\textbf {\bibinfo {volume} {99}},\ \bibinfo {pages} {062304} (\bibinfo {year} {2019}{\natexlab{a}})}\BibitemShut {NoStop}%
\bibitem [{\citenamefont {Gocho}\ \emph {et~al.}(2023)\citenamefont {Gocho}, \citenamefont {Nakamura}, \citenamefont {Kanno}, \citenamefont {Gao}, \citenamefont {Kobayashi}, \citenamefont {Inagaki},\ and\ \citenamefont {Hatanaka}}]{gocho2023excited}%
  \BibitemOpen
  \bibfield  {author} {\bibinfo {author} {\bibfnamefont {S.}~\bibnamefont {Gocho}}, \bibinfo {author} {\bibfnamefont {H.}~\bibnamefont {Nakamura}}, \bibinfo {author} {\bibfnamefont {S.}~\bibnamefont {Kanno}}, \bibinfo {author} {\bibfnamefont {Q.}~\bibnamefont {Gao}}, \bibinfo {author} {\bibfnamefont {T.}~\bibnamefont {Kobayashi}}, \bibinfo {author} {\bibfnamefont {T.}~\bibnamefont {Inagaki}},\ and\ \bibinfo {author} {\bibfnamefont {M.}~\bibnamefont {Hatanaka}},\ }\bibfield  {title} {\bibinfo {title} {Excited state calculations using variational quantum eigensolver with spin-restricted ans\"{a}tze and automatically-adjusted constraints},\ }\href {https://doi.org/10.1038/s41524-023-00965-1} {\bibfield  {journal} {\bibinfo  {journal} {npj Comput. Mater.}\ }\textbf {\bibinfo {volume} {9}},\ \bibinfo {pages} {13} (\bibinfo {year} {2023})}\BibitemShut {NoStop}%
\bibitem [{\citenamefont {Nakanishi}\ \emph {et~al.}(2019)\citenamefont {Nakanishi}, \citenamefont {Mitarai},\ and\ \citenamefont {Fujii}}]{nakanishi2019subspace}%
  \BibitemOpen
  \bibfield  {author} {\bibinfo {author} {\bibfnamefont {K.~M.}\ \bibnamefont {Nakanishi}}, \bibinfo {author} {\bibfnamefont {K.}~\bibnamefont {Mitarai}},\ and\ \bibinfo {author} {\bibfnamefont {K.}~\bibnamefont {Fujii}},\ }\bibfield  {title} {\bibinfo {title} {Subspace-search variational quantum eigensolver for excited states},\ }\href {https://doi.org/10.1103/PhysRevResearch.1.033062} {\bibfield  {journal} {\bibinfo  {journal} {Phys. Rev. Res.}\ }\textbf {\bibinfo {volume} {1}},\ \bibinfo {pages} {033062} (\bibinfo {year} {2019})}\BibitemShut {NoStop}%
\bibitem [{\citenamefont {Parrish}\ \emph {et~al.}(2019)\citenamefont {Parrish}, \citenamefont {Hohenstein}, \citenamefont {McMahon},\ and\ \citenamefont {Mart\'{\i}nez}}]{parrish2019quantum}%
  \BibitemOpen
  \bibfield  {author} {\bibinfo {author} {\bibfnamefont {R.~M.}\ \bibnamefont {Parrish}}, \bibinfo {author} {\bibfnamefont {E.~G.}\ \bibnamefont {Hohenstein}}, \bibinfo {author} {\bibfnamefont {P.~L.}\ \bibnamefont {McMahon}},\ and\ \bibinfo {author} {\bibfnamefont {T.~J.}\ \bibnamefont {Mart\'{\i}nez}},\ }\bibfield  {title} {\bibinfo {title} {Quantum computation of electronic transitions using a variational quantum eigensolver},\ }\href {https://doi.org/10.1103/PhysRevLett.122.230401} {\bibfield  {journal} {\bibinfo  {journal} {Phys. Rev. Lett.}\ }\textbf {\bibinfo {volume} {122}},\ \bibinfo {pages} {230401} (\bibinfo {year} {2019})}\BibitemShut {NoStop}%
\bibitem [{\citenamefont {Cadi~Tazi}\ and\ \citenamefont {Thom}(2024)}]{caditazi2024folded}%
  \BibitemOpen
  \bibfield  {author} {\bibinfo {author} {\bibfnamefont {L.}~\bibnamefont {Cadi~Tazi}}\ and\ \bibinfo {author} {\bibfnamefont {A.~J.~W.}\ \bibnamefont {Thom}},\ }\bibfield  {title} {\bibinfo {title} {Folded spectrum {VQE}: A quantum computing method for the calculation of molecular excited states},\ }\href {https://doi.org/10.1021/acs.jctc.3c01378} {\bibfield  {journal} {\bibinfo  {journal} {J. Chem. Theory Comput.}\ }\textbf {\bibinfo {volume} {20}},\ \bibinfo {pages} {2491–2504} (\bibinfo {year} {2024})}\BibitemShut {NoStop}%
\bibitem [{\citenamefont {Boyd}\ and\ \citenamefont {Koczor}(2022)}]{boyd2022training}%
  \BibitemOpen
  \bibfield  {author} {\bibinfo {author} {\bibfnamefont {G.}~\bibnamefont {Boyd}}\ and\ \bibinfo {author} {\bibfnamefont {B.}~\bibnamefont {Koczor}},\ }\bibfield  {title} {\bibinfo {title} {Training variational quantum circuits with {CoVaR}: Covariance root finding with classical shadows},\ }\href {https://doi.org/10.1103/PhysRevX.12.041022} {\bibfield  {journal} {\bibinfo  {journal} {Phys. Rev. X}\ }\textbf {\bibinfo {volume} {12}},\ \bibinfo {pages} {041022} (\bibinfo {year} {2022})}\BibitemShut {NoStop}%
\bibitem [{\citenamefont {LaRose}\ \emph {et~al.}(2019)\citenamefont {LaRose}, \citenamefont {Tikku}, \citenamefont {O'Neel-Judy}, \citenamefont {Cincio},\ and\ \citenamefont {Coles}}]{larose2019variational}%
  \BibitemOpen
  \bibfield  {author} {\bibinfo {author} {\bibfnamefont {R.}~\bibnamefont {LaRose}}, \bibinfo {author} {\bibfnamefont {A.}~\bibnamefont {Tikku}}, \bibinfo {author} {\bibfnamefont {E.}~\bibnamefont {O'Neel-Judy}}, \bibinfo {author} {\bibfnamefont {L.}~\bibnamefont {Cincio}},\ and\ \bibinfo {author} {\bibfnamefont {P.~J.}\ \bibnamefont {Coles}},\ }\bibfield  {title} {\bibinfo {title} {Variational quantum state diagonalization},\ }\href {https://doi.org/10.1038/s41534-019-0167-6} {\bibfield  {journal} {\bibinfo  {journal} {npj Quantum Inf.}\ }\textbf {\bibinfo {volume} {5}},\ \bibinfo {pages} {57} (\bibinfo {year} {2019})}\BibitemShut {NoStop}%
\bibitem [{\citenamefont {Cerezo}\ \emph {et~al.}(2022)\citenamefont {Cerezo}, \citenamefont {Sharma}, \citenamefont {Arrasmith},\ and\ \citenamefont {Coles}}]{cerezo2022variational}%
  \BibitemOpen
  \bibfield  {author} {\bibinfo {author} {\bibfnamefont {M.}~\bibnamefont {Cerezo}}, \bibinfo {author} {\bibfnamefont {K.}~\bibnamefont {Sharma}}, \bibinfo {author} {\bibfnamefont {A.}~\bibnamefont {Arrasmith}},\ and\ \bibinfo {author} {\bibfnamefont {P.~J.}\ \bibnamefont {Coles}},\ }\bibfield  {title} {\bibinfo {title} {Variational quantum state eigensolver},\ }\href {https://doi.org/10.1038/s41534-022-00611-6} {\bibfield  {journal} {\bibinfo  {journal} {npj Quantum Inf.}\ }\textbf {\bibinfo {volume} {8}},\ \bibinfo {pages} {113} (\bibinfo {year} {2022})}\BibitemShut {NoStop}%
\bibitem [{\citenamefont {McClean}\ \emph {et~al.}(2017)\citenamefont {McClean}, \citenamefont {Kimchi-Schwartz}, \citenamefont {Carter},\ and\ \citenamefont {de~Jong}}]{mcclean2017hybrid}%
  \BibitemOpen
  \bibfield  {author} {\bibinfo {author} {\bibfnamefont {J.~R.}\ \bibnamefont {McClean}}, \bibinfo {author} {\bibfnamefont {M.~E.}\ \bibnamefont {Kimchi-Schwartz}}, \bibinfo {author} {\bibfnamefont {J.}~\bibnamefont {Carter}},\ and\ \bibinfo {author} {\bibfnamefont {W.~A.}\ \bibnamefont {de~Jong}},\ }\bibfield  {title} {\bibinfo {title} {Hybrid quantum-classical hierarchy for mitigation of decoherence and determination of excited states},\ }\href {https://doi.org/10.1103/PhysRevA.95.042308} {\bibfield  {journal} {\bibinfo  {journal} {Phys. Rev. A}\ }\textbf {\bibinfo {volume} {95}},\ \bibinfo {pages} {042308} (\bibinfo {year} {2017})}\BibitemShut {NoStop}%
\bibitem [{\citenamefont {Parrish}\ and\ \citenamefont {McMahon}(2019)}]{parrish2019quantumb}%
  \BibitemOpen
  \bibfield  {author} {\bibinfo {author} {\bibfnamefont {R.~M.}\ \bibnamefont {Parrish}}\ and\ \bibinfo {author} {\bibfnamefont {P.~L.}\ \bibnamefont {McMahon}},\ }\href@noop {} {\bibinfo {title} {Quantum filter diagonalization: Quantum eigendecomposition without full quantum phase estimation}} (\bibinfo {year} {2019}),\ \Eprint {https://arxiv.org/abs/1909.08925} {arXiv:1909.08925 [quant-ph]} \BibitemShut {NoStop}%
\bibitem [{\citenamefont {Motta}\ \emph {et~al.}(2019)\citenamefont {Motta}, \citenamefont {Sun}, \citenamefont {Tan}, \citenamefont {O’Rourke}, \citenamefont {Ye}, \citenamefont {Minnich}, \citenamefont {Brandão},\ and\ \citenamefont {Chan}}]{motta2019determining}%
  \BibitemOpen
  \bibfield  {author} {\bibinfo {author} {\bibfnamefont {M.}~\bibnamefont {Motta}}, \bibinfo {author} {\bibfnamefont {C.}~\bibnamefont {Sun}}, \bibinfo {author} {\bibfnamefont {A.~T.~K.}\ \bibnamefont {Tan}}, \bibinfo {author} {\bibfnamefont {M.~J.}\ \bibnamefont {O’Rourke}}, \bibinfo {author} {\bibfnamefont {E.}~\bibnamefont {Ye}}, \bibinfo {author} {\bibfnamefont {A.~J.}\ \bibnamefont {Minnich}}, \bibinfo {author} {\bibfnamefont {F.~G. S.~L.}\ \bibnamefont {Brandão}},\ and\ \bibinfo {author} {\bibfnamefont {G.~K.-L.}\ \bibnamefont {Chan}},\ }\bibfield  {title} {\bibinfo {title} {Determining eigenstates and thermal states on a quantum computer using quantum imaginary time evolution},\ }\href {https://doi.org/10.1038/s41567-019-0704-4} {\bibfield  {journal} {\bibinfo  {journal} {Nat. Phys.}\ }\textbf {\bibinfo {volume} {16}},\ \bibinfo {pages} {205–210} (\bibinfo {year} {2019})}\BibitemShut {NoStop}%
\bibitem [{\citenamefont {Stair}\ \emph {et~al.}(2020)\citenamefont {Stair}, \citenamefont {Huang},\ and\ \citenamefont {Evangelista}}]{stair2020multireference}%
  \BibitemOpen
  \bibfield  {author} {\bibinfo {author} {\bibfnamefont {N.~H.}\ \bibnamefont {Stair}}, \bibinfo {author} {\bibfnamefont {R.}~\bibnamefont {Huang}},\ and\ \bibinfo {author} {\bibfnamefont {F.~A.}\ \bibnamefont {Evangelista}},\ }\bibfield  {title} {\bibinfo {title} {A multireference quantum {Krylov} algorithm for strongly correlated electrons},\ }\href {https://doi.org/10.1021/acs.jctc.9b01125} {\bibfield  {journal} {\bibinfo  {journal} {J. Chem. Theory Comput.}\ }\textbf {\bibinfo {volume} {16}},\ \bibinfo {pages} {2236–2245} (\bibinfo {year} {2020})}\BibitemShut {NoStop}%
\bibitem [{\citenamefont {Yoshioka}\ \emph {et~al.}(2025)\citenamefont {Yoshioka}, \citenamefont {Amico}, \citenamefont {Kirby}, \citenamefont {Jurcevic}, \citenamefont {Dutt}, \citenamefont {Fuller}, \citenamefont {Garion}, \citenamefont {Haas}, \citenamefont {Hamamura}, \citenamefont {Ivrii}, \citenamefont {Majumdar}, \citenamefont {Minev}, \citenamefont {Motta}, \citenamefont {Pokharel}, \citenamefont {Rivero}, \citenamefont {Sharma}, \citenamefont {Wood}, \citenamefont {Javadi-Abhari},\ and\ \citenamefont {Mezzacapo}}]{yoshioka2025krylov}%
  \BibitemOpen
  \bibfield  {author} {\bibinfo {author} {\bibfnamefont {N.}~\bibnamefont {Yoshioka}}, \bibinfo {author} {\bibfnamefont {M.}~\bibnamefont {Amico}}, \bibinfo {author} {\bibfnamefont {W.}~\bibnamefont {Kirby}}, \bibinfo {author} {\bibfnamefont {P.}~\bibnamefont {Jurcevic}}, \bibinfo {author} {\bibfnamefont {A.}~\bibnamefont {Dutt}}, \bibinfo {author} {\bibfnamefont {B.}~\bibnamefont {Fuller}}, \bibinfo {author} {\bibfnamefont {S.}~\bibnamefont {Garion}}, \bibinfo {author} {\bibfnamefont {H.}~\bibnamefont {Haas}}, \bibinfo {author} {\bibfnamefont {I.}~\bibnamefont {Hamamura}}, \bibinfo {author} {\bibfnamefont {A.}~\bibnamefont {Ivrii}}, \bibinfo {author} {\bibfnamefont {R.}~\bibnamefont {Majumdar}}, \bibinfo {author} {\bibfnamefont {Z.}~\bibnamefont {Minev}}, \bibinfo {author} {\bibfnamefont {M.}~\bibnamefont {Motta}}, \bibinfo {author} {\bibfnamefont {B.}~\bibnamefont {Pokharel}}, \bibinfo {author} {\bibfnamefont {P.}~\bibnamefont {Rivero}}, \bibinfo {author} {\bibfnamefont {K.}~\bibnamefont {Sharma}}, \bibinfo {author} {\bibfnamefont {C.~J.}\ \bibnamefont {Wood}}, \bibinfo {author} {\bibfnamefont {A.}~\bibnamefont {Javadi-Abhari}},\ and\ \bibinfo {author} {\bibfnamefont {A.}~\bibnamefont {Mezzacapo}},\ }\bibfield  {title} {\bibinfo {title} {Krylov diagonalization of large many-body {Hamiltonians} on a quantum processor},\ }\href {https://doi.org/10.1038/s41467-025-59716-z} {\bibfield  {journal} {\bibinfo  {journal} {Nat. Commun.}\ }\textbf {\bibinfo {volume} {16}},\ \bibinfo {pages} {5014} (\bibinfo {year} {2025})}\BibitemShut {NoStop}%
\bibitem [{\citenamefont {Cortes}\ and\ \citenamefont {Gray}(2022)}]{cortes2022quantum}%
  \BibitemOpen
  \bibfield  {author} {\bibinfo {author} {\bibfnamefont {C.~L.}\ \bibnamefont {Cortes}}\ and\ \bibinfo {author} {\bibfnamefont {S.~K.}\ \bibnamefont {Gray}},\ }\bibfield  {title} {\bibinfo {title} {Quantum {Krylov} subspace algorithms for ground- and excited-state energy estimation},\ }\href {https://doi.org/10.1103/PhysRevA.105.022417} {\bibfield  {journal} {\bibinfo  {journal} {Phys. Rev. A}\ }\textbf {\bibinfo {volume} {105}},\ \bibinfo {pages} {022417} (\bibinfo {year} {2022})}\BibitemShut {NoStop}%
\bibitem [{\citenamefont {Oumarou}\ \emph {et~al.}(2025)\citenamefont {Oumarou}, \citenamefont {Ollitrault}, \citenamefont {Cortes}, \citenamefont {Scheurer}, \citenamefont {Parrish},\ and\ \citenamefont {Gogolin}}]{oumarou2025molecular}%
  \BibitemOpen
  \bibfield  {author} {\bibinfo {author} {\bibfnamefont {O.}~\bibnamefont {Oumarou}}, \bibinfo {author} {\bibfnamefont {P.~J.}\ \bibnamefont {Ollitrault}}, \bibinfo {author} {\bibfnamefont {C.~L.}\ \bibnamefont {Cortes}}, \bibinfo {author} {\bibfnamefont {M.}~\bibnamefont {Scheurer}}, \bibinfo {author} {\bibfnamefont {R.~M.}\ \bibnamefont {Parrish}},\ and\ \bibinfo {author} {\bibfnamefont {C.}~\bibnamefont {Gogolin}},\ }\bibfield  {title} {\bibinfo {title} {Molecular properties from quantum {Krylov} subspace diagonalization},\ }\href {https://doi.org/10.1021/acs.jctc.5c00194} {\bibfield  {journal} {\bibinfo  {journal} {J. Chem. Theory Comput.}\ }\textbf {\bibinfo {volume} {21}},\ \bibinfo {pages} {4543–4552} (\bibinfo {year} {2025})}\BibitemShut {NoStop}%
\bibitem [{\citenamefont {Poulin}\ and\ \citenamefont {Wocjan}(2009)}]{poulin2009preparing}%
  \BibitemOpen
  \bibfield  {author} {\bibinfo {author} {\bibfnamefont {D.}~\bibnamefont {Poulin}}\ and\ \bibinfo {author} {\bibfnamefont {P.}~\bibnamefont {Wocjan}},\ }\bibfield  {title} {\bibinfo {title} {Preparing ground states of quantum many-body systems on a quantum computer},\ }\href {https://doi.org/10.1103/PhysRevLett.102.130503} {\bibfield  {journal} {\bibinfo  {journal} {Phys. Rev. Lett.}\ }\textbf {\bibinfo {volume} {102}},\ \bibinfo {pages} {130503} (\bibinfo {year} {2009})}\BibitemShut {NoStop}%
\bibitem [{\citenamefont {Ge}\ \emph {et~al.}(2019)\citenamefont {Ge}, \citenamefont {Tura},\ and\ \citenamefont {Cirac}}]{ge2019faster}%
  \BibitemOpen
  \bibfield  {author} {\bibinfo {author} {\bibfnamefont {Y.}~\bibnamefont {Ge}}, \bibinfo {author} {\bibfnamefont {J.}~\bibnamefont {Tura}},\ and\ \bibinfo {author} {\bibfnamefont {J.~I.}\ \bibnamefont {Cirac}},\ }\bibfield  {title} {\bibinfo {title} {Faster ground state preparation and high-precision ground energy estimation with fewer qubits},\ }\href {https://doi.org/10.1063/1.5027484} {\bibfield  {journal} {\bibinfo  {journal} {J. Math. Phys.}\ }\textbf {\bibinfo {volume} {60}},\ \bibinfo {pages} {022202} (\bibinfo {year} {2019})}\BibitemShut {NoStop}%
\bibitem [{\citenamefont {Lin}\ and\ \citenamefont {Tong}(2020{\natexlab{a}})}]{lin2020nearoptimal}%
  \BibitemOpen
  \bibfield  {author} {\bibinfo {author} {\bibfnamefont {L.}~\bibnamefont {Lin}}\ and\ \bibinfo {author} {\bibfnamefont {Y.}~\bibnamefont {Tong}},\ }\bibfield  {title} {\bibinfo {title} {Near-optimal ground state preparation},\ }\href {https://doi.org/10.22331/q-2020-12-14-372} {\bibfield  {journal} {\bibinfo  {journal} {Quantum}\ }\textbf {\bibinfo {volume} {4}},\ \bibinfo {pages} {372} (\bibinfo {year} {2020}{\natexlab{a}})}\BibitemShut {NoStop}%
\bibitem [{\citenamefont {Dong}\ \emph {et~al.}(2022)\citenamefont {Dong}, \citenamefont {Lin},\ and\ \citenamefont {Tong}}]{dong2022ground-state}%
  \BibitemOpen
  \bibfield  {author} {\bibinfo {author} {\bibfnamefont {Y.}~\bibnamefont {Dong}}, \bibinfo {author} {\bibfnamefont {L.}~\bibnamefont {Lin}},\ and\ \bibinfo {author} {\bibfnamefont {Y.}~\bibnamefont {Tong}},\ }\bibfield  {title} {\bibinfo {title} {Ground-state preparation and energy estimation on early fault-tolerant quantum computers via quantum eigenvalue transformation of unitary matrices},\ }\href {https://doi.org/10.1103/PRXQuantum.3.040305} {\bibfield  {journal} {\bibinfo  {journal} {PRX Quantum}\ }\textbf {\bibinfo {volume} {3}},\ \bibinfo {pages} {040305} (\bibinfo {year} {2022})}\BibitemShut {NoStop}%
\bibitem [{\citenamefont {Lin}\ and\ \citenamefont {Tong}(2022)}]{lin2022heisenberg}%
  \BibitemOpen
  \bibfield  {author} {\bibinfo {author} {\bibfnamefont {L.}~\bibnamefont {Lin}}\ and\ \bibinfo {author} {\bibfnamefont {Y.}~\bibnamefont {Tong}},\ }\bibfield  {title} {\bibinfo {title} {Heisenberg-limited ground-state energy estimation for early fault-tolerant quantum computers},\ }\href {https://doi.org/10.1103/PRXQuantum.3.010318} {\bibfield  {journal} {\bibinfo  {journal} {PRX Quantum}\ }\textbf {\bibinfo {volume} {3}},\ \bibinfo {pages} {010318} (\bibinfo {year} {2022})}\BibitemShut {NoStop}%
\bibitem [{\citenamefont {Wang}\ \emph {et~al.}(2023)\citenamefont {Wang}, \citenamefont {Fran{\c{c}}a}, \citenamefont {Zhang}, \citenamefont {Zhu},\ and\ \citenamefont {Johnson}}]{wang2023quantum}%
  \BibitemOpen
  \bibfield  {author} {\bibinfo {author} {\bibfnamefont {G.}~\bibnamefont {Wang}}, \bibinfo {author} {\bibfnamefont {D.~S.}\ \bibnamefont {Fran{\c{c}}a}}, \bibinfo {author} {\bibfnamefont {R.}~\bibnamefont {Zhang}}, \bibinfo {author} {\bibfnamefont {S.}~\bibnamefont {Zhu}},\ and\ \bibinfo {author} {\bibfnamefont {P.~D.}\ \bibnamefont {Johnson}},\ }\bibfield  {title} {\bibinfo {title} {Quantum algorithm for ground state energy estimation using circuit depth with exponentially improved dependence on precision},\ }\href {https://doi.org/10.22331/q-2023-11-06-1167} {\bibfield  {journal} {\bibinfo  {journal} {Quantum}\ }\textbf {\bibinfo {volume} {7}},\ \bibinfo {pages} {1167} (\bibinfo {year} {2023})}\BibitemShut {NoStop}%
\bibitem [{\citenamefont {Ding}\ and\ \citenamefont {Lin}(2023{\natexlab{a}})}]{ding2023even}%
  \BibitemOpen
  \bibfield  {author} {\bibinfo {author} {\bibfnamefont {Z.}~\bibnamefont {Ding}}\ and\ \bibinfo {author} {\bibfnamefont {L.}~\bibnamefont {Lin}},\ }\bibfield  {title} {\bibinfo {title} {Even shorter quantum circuit for phase estimation on early fault-tolerant quantum computers with applications to ground-state energy estimation},\ }\href {https://link.aps.org/doi/10.1103/PRXQuantum.4.020331} {\bibfield  {journal} {\bibinfo  {journal} {PRX Quantum}\ }\textbf {\bibinfo {volume} {4}},\ \bibinfo {pages} {020331} (\bibinfo {year} {2023}{\natexlab{a}})}\BibitemShut {NoStop}%
\bibitem [{\citenamefont {Babbush}\ \emph {et~al.}(2018)\citenamefont {Babbush}, \citenamefont {Gidney}, \citenamefont {Berry}, \citenamefont {Wiebe}, \citenamefont {McClean}, \citenamefont {Paler}, \citenamefont {Fowler},\ and\ \citenamefont {Neven}}]{babbush2018encoding}%
  \BibitemOpen
  \bibfield  {author} {\bibinfo {author} {\bibfnamefont {R.}~\bibnamefont {Babbush}}, \bibinfo {author} {\bibfnamefont {C.}~\bibnamefont {Gidney}}, \bibinfo {author} {\bibfnamefont {D.~W.}\ \bibnamefont {Berry}}, \bibinfo {author} {\bibfnamefont {N.}~\bibnamefont {Wiebe}}, \bibinfo {author} {\bibfnamefont {J.}~\bibnamefont {McClean}}, \bibinfo {author} {\bibfnamefont {A.}~\bibnamefont {Paler}}, \bibinfo {author} {\bibfnamefont {A.}~\bibnamefont {Fowler}},\ and\ \bibinfo {author} {\bibfnamefont {H.}~\bibnamefont {Neven}},\ }\bibfield  {title} {\bibinfo {title} {Encoding electronic spectra in quantum circuits with linear {T} complexity},\ }\href {https://doi.org/10.1103/PhysRevX.8.041015} {\bibfield  {journal} {\bibinfo  {journal} {Phys. Rev. X}\ }\textbf {\bibinfo {volume} {8}},\ \bibinfo {pages} {041015} (\bibinfo {year} {2018})}\BibitemShut {NoStop}%
\bibitem [{\citenamefont {Lin}\ and\ \citenamefont {Tong}(2020{\natexlab{b}})}]{lin2020optimal}%
  \BibitemOpen
  \bibfield  {author} {\bibinfo {author} {\bibfnamefont {L.}~\bibnamefont {Lin}}\ and\ \bibinfo {author} {\bibfnamefont {Y.}~\bibnamefont {Tong}},\ }\bibfield  {title} {\bibinfo {title} {Optimal polynomial based quantum eigenstate filtering with application to solving quantum linear systems},\ }\href {https://doi.org/10.22331/q-2020-11-11-361} {\bibfield  {journal} {\bibinfo  {journal} {Quantum}\ }\textbf {\bibinfo {volume} {4}},\ \bibinfo {pages} {361} (\bibinfo {year} {2020}{\natexlab{b}})}\BibitemShut {NoStop}%
\bibitem [{\citenamefont {Patil}\ and\ \citenamefont {Glaser}(2026)}]{patil2026efficient}%
  \BibitemOpen
  \bibfield  {author} {\bibinfo {author} {\bibfnamefont {S.}~\bibnamefont {Patil}}\ and\ \bibinfo {author} {\bibfnamefont {N.}~\bibnamefont {Glaser}},\ }\href@noop {} {\bibinfo {title} {Efficient targeting of arbitrary excited states with quantum inverse power iteration through filtering polynomials}} (\bibinfo {year} {2026}),\ \Eprint {https://arxiv.org/abs/2606.28255} {arXiv:2606.28255 [quant-ph]} \BibitemShut {NoStop}%
\bibitem [{\citenamefont {Kitaev}(1995)}]{kitaev1995quantum}%
  \BibitemOpen
  \bibfield  {author} {\bibinfo {author} {\bibfnamefont {A.~Y.}\ \bibnamefont {Kitaev}},\ }\href@noop {} {\bibinfo {title} {Quantum measurements and the abelian stabilizer problem}} (\bibinfo {year} {1995}),\ \Eprint {https://arxiv.org/abs/quant-ph/9511026} {arXiv:quant-ph/9511026 [quant-ph]} \BibitemShut {NoStop}%
\bibitem [{\citenamefont {Cleve}\ \emph {et~al.}(1998)\citenamefont {Cleve}, \citenamefont {Ekert}, \citenamefont {Macchiavello},\ and\ \citenamefont {Mosca}}]{cleve1998quantum}%
  \BibitemOpen
  \bibfield  {author} {\bibinfo {author} {\bibfnamefont {R.}~\bibnamefont {Cleve}}, \bibinfo {author} {\bibfnamefont {A.}~\bibnamefont {Ekert}}, \bibinfo {author} {\bibfnamefont {C.}~\bibnamefont {Macchiavello}},\ and\ \bibinfo {author} {\bibfnamefont {M.}~\bibnamefont {Mosca}},\ }\bibfield  {title} {\bibinfo {title} {Quantum algorithms revisited},\ }\href {https://doi.org/10.1098/rspa.1998.0164} {\bibfield  {journal} {\bibinfo  {journal} {Proc. R. Soc. Lond., A: Math. Phys. Eng. Sci.}\ }\textbf {\bibinfo {volume} {454}},\ \bibinfo {pages} {339–354} (\bibinfo {year} {1998})}\BibitemShut {NoStop}%
\bibitem [{\citenamefont {Nielsen}\ and\ \citenamefont {Chuang}(2010)}]{nielsen2010quantum}%
  \BibitemOpen
  \bibfield  {author} {\bibinfo {author} {\bibfnamefont {M.~A.}\ \bibnamefont {Nielsen}}\ and\ \bibinfo {author} {\bibfnamefont {I.~L.}\ \bibnamefont {Chuang}},\ }\href {https://doi.org/10.1017/cbo9780511976667} {\emph {\bibinfo {title} {Quantum Computation and Quantum Information: 10th Anniversary Edition}}}\ (\bibinfo  {publisher} {Cambridge University Press},\ \bibinfo {year} {2010})\BibitemShut {NoStop}%
\bibitem [{\citenamefont {Abrams}\ and\ \citenamefont {Lloyd}(1999)}]{abrams1999quantum}%
  \BibitemOpen
  \bibfield  {author} {\bibinfo {author} {\bibfnamefont {D.~S.}\ \bibnamefont {Abrams}}\ and\ \bibinfo {author} {\bibfnamefont {S.}~\bibnamefont {Lloyd}},\ }\bibfield  {title} {\bibinfo {title} {Quantum algorithm providing exponential speed increase for finding eigenvalues and eigenvectors},\ }\href {https://doi.org/10.1103/PhysRevLett.83.5162} {\bibfield  {journal} {\bibinfo  {journal} {Phys. Rev. Lett.}\ }\textbf {\bibinfo {volume} {83}},\ \bibinfo {pages} {5162} (\bibinfo {year} {1999})}\BibitemShut {NoStop}%
\bibitem [{\citenamefont {Somma}\ and\ \citenamefont {Boixo}(2013)}]{somma2013spectral}%
  \BibitemOpen
  \bibfield  {author} {\bibinfo {author} {\bibfnamefont {R.~D.}\ \bibnamefont {Somma}}\ and\ \bibinfo {author} {\bibfnamefont {S.}~\bibnamefont {Boixo}},\ }\bibfield  {title} {\bibinfo {title} {Spectral gap amplification},\ }\href {https://doi.org/10.1137/120871997} {\bibfield  {journal} {\bibinfo  {journal} {SIAM J. Comput.}\ }\textbf {\bibinfo {volume} {42}},\ \bibinfo {pages} {593} (\bibinfo {year} {2013})}\BibitemShut {NoStop}%
\bibitem [{\citenamefont {Thibodeau}\ and\ \citenamefont {Clark}(2023)}]{thibodeau2023nearly}%
  \BibitemOpen
  \bibfield  {author} {\bibinfo {author} {\bibfnamefont {M.}~\bibnamefont {Thibodeau}}\ and\ \bibinfo {author} {\bibfnamefont {B.~K.}\ \bibnamefont {Clark}},\ }\bibfield  {title} {\bibinfo {title} {Nearly-frustration-free ground state preparation},\ }\href {https://doi.org/10.22331/q-2023-08-16-1084} {\bibfield  {journal} {\bibinfo  {journal} {Quantum}\ }\textbf {\bibinfo {volume} {7}},\ \bibinfo {pages} {1084} (\bibinfo {year} {2023})}\BibitemShut {NoStop}%
\bibitem [{\citenamefont {Childs}\ and\ \citenamefont {Wiebe}(2012)}]{childs2012hamiltonian}%
  \BibitemOpen
  \bibfield  {author} {\bibinfo {author} {\bibfnamefont {A.~M.}\ \bibnamefont {Childs}}\ and\ \bibinfo {author} {\bibfnamefont {N.}~\bibnamefont {Wiebe}},\ }\bibfield  {title} {\bibinfo {title} {Hamiltonian simulation using linear combinations of unitary operations},\ }\href {https://doi.org/10.26421/QIC12.11-12-1} {\bibfield  {journal} {\bibinfo  {journal} {Quantum Info. Comput.}\ }\textbf {\bibinfo {volume} {12}},\ \bibinfo {pages} {901–924} (\bibinfo {year} {2012})}\BibitemShut {NoStop}%
\bibitem [{\citenamefont {Low}\ and\ \citenamefont {Chuang}(2019)}]{low2019hamiltonian}%
  \BibitemOpen
  \bibfield  {author} {\bibinfo {author} {\bibfnamefont {G.~H.}\ \bibnamefont {Low}}\ and\ \bibinfo {author} {\bibfnamefont {I.~L.}\ \bibnamefont {Chuang}},\ }\bibfield  {title} {\bibinfo {title} {Hamiltonian simulation by qubitization},\ }\href {https://doi.org/10.22331/q-2019-07-12-163} {\bibfield  {journal} {\bibinfo  {journal} {Quantum}\ }\textbf {\bibinfo {volume} {3}},\ \bibinfo {pages} {163} (\bibinfo {year} {2019})}\BibitemShut {NoStop}%
\bibitem [{\citenamefont {Gily\'{e}n}\ \emph {et~al.}(2019)\citenamefont {Gily\'{e}n}, \citenamefont {Su}, \citenamefont {Low},\ and\ \citenamefont {Wiebe}}]{gilyen2019quantum}%
  \BibitemOpen
  \bibfield  {author} {\bibinfo {author} {\bibfnamefont {A.}~\bibnamefont {Gily\'{e}n}}, \bibinfo {author} {\bibfnamefont {Y.}~\bibnamefont {Su}}, \bibinfo {author} {\bibfnamefont {G.~H.}\ \bibnamefont {Low}},\ and\ \bibinfo {author} {\bibfnamefont {N.}~\bibnamefont {Wiebe}},\ }\bibfield  {title} {\bibinfo {title} {Quantum singular value transformation and beyond: Exponential improvements for quantum matrix arithmetics},\ }in\ \href {https://doi.org/10.1145/3313276.3316366} {\emph {\bibinfo {booktitle} {Proceedings of the 51st Annual ACM SIGACT Symposium on Theory of Computing}}},\ \bibinfo {series and number} {STOC 2019}\ (\bibinfo  {publisher} {Association for Computing Machinery},\ \bibinfo {address} {New York, NY, USA},\ \bibinfo {year} {2019})\ pp.\ \bibinfo {pages} {193--204}\BibitemShut {NoStop}%
\bibitem [{\citenamefont {\c{C}akan}\ \emph {et~al.}(2021)\citenamefont {\c{C}akan}, \citenamefont {Cirac},\ and\ \citenamefont {Ba\~nuls}}]{cakan2021approximating}%
  \BibitemOpen
  \bibfield  {author} {\bibinfo {author} {\bibfnamefont {A.}~\bibnamefont {\c{C}akan}}, \bibinfo {author} {\bibfnamefont {J.~I.}\ \bibnamefont {Cirac}},\ and\ \bibinfo {author} {\bibfnamefont {M.~C.}\ \bibnamefont {Ba\~nuls}},\ }\bibfield  {title} {\bibinfo {title} {Approximating the long time average of the density operator: Diagonal ensemble},\ }\href {https://doi.org/10.1103/PhysRevB.103.115113} {\bibfield  {journal} {\bibinfo  {journal} {Phys. Rev. B}\ }\textbf {\bibinfo {volume} {103}},\ \bibinfo {pages} {115113} (\bibinfo {year} {2021})}\BibitemShut {NoStop}%
\bibitem [{\citenamefont {Bak\'o}\ \emph {et~al.}(2026)\citenamefont {Bak\'o}, \citenamefont {Araki},\ and\ \citenamefont {Koczor}}]{bako2026exponential}%
  \BibitemOpen
  \bibfield  {author} {\bibinfo {author} {\bibfnamefont {B.}~\bibnamefont {Bak\'o}}, \bibinfo {author} {\bibfnamefont {T.}~\bibnamefont {Araki}},\ and\ \bibinfo {author} {\bibfnamefont {B.}~\bibnamefont {Koczor}},\ }\bibfield  {title} {\bibinfo {title} {Exponential distillation of dominant eigenproperties},\ }\href {https://doi.org/10.1103/bglh-9snd} {\bibfield  {journal} {\bibinfo  {journal} {PRX Quantum}\ }\textbf {\bibinfo {volume} {7}},\ \bibinfo {pages} {010334} (\bibinfo {year} {2026})}\BibitemShut {NoStop}%
\bibitem [{\citenamefont {Dutkiewicz}\ \emph {et~al.}(2026)\citenamefont {Dutkiewicz}, \citenamefont {White}, \citenamefont {Low}, \citenamefont {III}, \citenamefont {Harrigan}, \citenamefont {Kieferova}, \citenamefont {Babbush}, \citenamefont {Berry},\ and\ \citenamefont {Rubin}}]{dutkiewicz2026spectral}%
  \BibitemOpen
  \bibfield  {author} {\bibinfo {author} {\bibfnamefont {A.}~\bibnamefont {Dutkiewicz}}, \bibinfo {author} {\bibfnamefont {A.~F.}\ \bibnamefont {White}}, \bibinfo {author} {\bibfnamefont {G.~H.}\ \bibnamefont {Low}}, \bibinfo {author} {\bibfnamefont {A.~E.~D.}\ \bibnamefont {III}}, \bibinfo {author} {\bibfnamefont {M.~P.}\ \bibnamefont {Harrigan}}, \bibinfo {author} {\bibfnamefont {M.}~\bibnamefont {Kieferova}}, \bibinfo {author} {\bibfnamefont {R.}~\bibnamefont {Babbush}}, \bibinfo {author} {\bibfnamefont {D.~W.}\ \bibnamefont {Berry}},\ and\ \bibinfo {author} {\bibfnamefont {N.~C.}\ \bibnamefont {Rubin}},\ }\href@noop {} {\bibinfo {title} {Spectral amplification for ground-state energy estimation of electronic structure in first quantization}} (\bibinfo {year} {2026}),\ \Eprint {https://arxiv.org/abs/2607.15358} {arXiv:2607.15358 [quant-ph]} \BibitemShut {NoStop}%
\bibitem [{\citenamefont {Rendon}\ \emph {et~al.}(2024)\citenamefont {Rendon}, \citenamefont {Watkins},\ and\ \citenamefont {Wiebe}}]{rendon2024improved}%
  \BibitemOpen
  \bibfield  {author} {\bibinfo {author} {\bibfnamefont {G.}~\bibnamefont {Rendon}}, \bibinfo {author} {\bibfnamefont {J.}~\bibnamefont {Watkins}},\ and\ \bibinfo {author} {\bibfnamefont {N.}~\bibnamefont {Wiebe}},\ }\bibfield  {title} {\bibinfo {title} {Improved accuracy for {Trotter} simulations using {Chebyshev} interpolation},\ }\href {https://doi.org/10.22331/q-2024-02-26-1266} {\bibfield  {journal} {\bibinfo  {journal} {{Quantum}}\ }\textbf {\bibinfo {volume} {8}},\ \bibinfo {pages} {1266} (\bibinfo {year} {2024})}\BibitemShut {NoStop}%
\bibitem [{\citenamefont {Rendon}\ and\ \citenamefont {Johnson}(2023)}]{rendon2023lowdepth}%
  \BibitemOpen
  \bibfield  {author} {\bibinfo {author} {\bibfnamefont {G.}~\bibnamefont {Rendon}}\ and\ \bibinfo {author} {\bibfnamefont {P.~D.}\ \bibnamefont {Johnson}},\ }\href@noop {} {\bibinfo {title} {Low-depth {Gaussian} state energy estimation}} (\bibinfo {year} {2023}),\ \Eprint {https://arxiv.org/abs/2309.16790} {arXiv:2309.16790 [quant-ph]} \BibitemShut {NoStop}%
\bibitem [{\citenamefont {Chen}\ \emph {et~al.}(2025{\natexlab{a}})\citenamefont {Chen}, \citenamefont {Gilyén},\ and\ \citenamefont {de~Wolf}}]{chen2025quantum}%
  \BibitemOpen
  \bibfield  {author} {\bibinfo {author} {\bibfnamefont {Y.}~\bibnamefont {Chen}}, \bibinfo {author} {\bibfnamefont {A.}~\bibnamefont {Gilyén}},\ and\ \bibinfo {author} {\bibfnamefont {R.}~\bibnamefont {de~Wolf}},\ }\bibfield  {title} {\bibinfo {title} {A quantum speed-up for approximating the top eigenvectors of a matrix},\ }in\ \href {https://doi.org/10.1137/1.9781611978322.29} {\emph {\bibinfo {booktitle} {Proceedings of the 2025 Annual ACM-SIAM Symposium on Discrete Algorithms (SODA)}}}\ (\bibinfo {year} {2025})\ pp.\ \bibinfo {pages} {994--1036}\BibitemShut {NoStop}%
\bibitem [{\citenamefont {van Apeldoorn}(2021)}]{vanapeldoorn2021quantum}%
  \BibitemOpen
  \bibfield  {author} {\bibinfo {author} {\bibfnamefont {J.}~\bibnamefont {van Apeldoorn}},\ }\bibfield  {title} {\bibinfo {title} {Quantum probability oracles \& multidimensional amplitude estimation},\ }in\ \href {https://doi.org/10.4230/LIPIcs.TQC.2021.9} {\emph {\bibinfo {booktitle} {16th Conference on the Theory of Quantum Computation, Communication and Cryptography (TQC 2021)}}},\ \bibinfo {series} {Leibniz International Proceedings in Informatics (LIPIcs)}, Vol.\ \bibinfo {volume} {197},\ \bibinfo {editor} {edited by\ \bibinfo {editor} {\bibfnamefont {M.-H.}\ \bibnamefont {Hsieh}}}\ (\bibinfo  {publisher} {Schloss Dagstuhl -- Leibniz-Zentrum f{\"u}r Informatik},\ \bibinfo {address} {Dagstuhl, Germany},\ \bibinfo {year} {2021})\ pp.\ \bibinfo {pages} {9:1--9:11}\BibitemShut {NoStop}%
\bibitem [{\citenamefont {Brassard}\ \emph {et~al.}(2002)\citenamefont {Brassard}, \citenamefont {H{\o}yer}, \citenamefont {Mosca},\ and\ \citenamefont {Tapp}}]{brassard2002quantum}%
  \BibitemOpen
  \bibfield  {author} {\bibinfo {author} {\bibfnamefont {G.}~\bibnamefont {Brassard}}, \bibinfo {author} {\bibfnamefont {P.}~\bibnamefont {H{\o}yer}}, \bibinfo {author} {\bibfnamefont {M.}~\bibnamefont {Mosca}},\ and\ \bibinfo {author} {\bibfnamefont {A.}~\bibnamefont {Tapp}},\ }\bibfield  {title} {\bibinfo {title} {Quantum amplitude amplification and estimation},\ }in\ \href {https://doi.org/10.1090/conm/305/05215} {\emph {\bibinfo {booktitle} {Quantum computation and information}}},\ \bibinfo {series} {Contemporary Mathematics}, Vol.\ \bibinfo {volume} {305}\ (\bibinfo  {publisher} {American Mathematical Society},\ \bibinfo {address} {Providence, RI, USA},\ \bibinfo {year} {2002})\ pp.\ \bibinfo {pages} {53--74}\BibitemShut {NoStop}%
\bibitem [{\citenamefont {Ding}\ \emph {et~al.}(2024)\citenamefont {Ding}, \citenamefont {Li}, \citenamefont {Lin}, \citenamefont {Ni}, \citenamefont {Ying},\ and\ \citenamefont {Zhang}}]{ding2024quantum}%
  \BibitemOpen
  \bibfield  {author} {\bibinfo {author} {\bibfnamefont {Z.}~\bibnamefont {Ding}}, \bibinfo {author} {\bibfnamefont {H.}~\bibnamefont {Li}}, \bibinfo {author} {\bibfnamefont {L.}~\bibnamefont {Lin}}, \bibinfo {author} {\bibfnamefont {H.}~\bibnamefont {Ni}}, \bibinfo {author} {\bibfnamefont {L.}~\bibnamefont {Ying}},\ and\ \bibinfo {author} {\bibfnamefont {R.}~\bibnamefont {Zhang}},\ }\bibfield  {title} {\bibinfo {title} {Quantum multiple eigenvalue {Gaussian} filtered search: an efficient and versatile quantum phase estimation method},\ }\href {https://doi.org/10.22331/q-2024-10-02-1487} {\bibfield  {journal} {\bibinfo  {journal} {Quantum}\ }\textbf {\bibinfo {volume} {8}},\ \bibinfo {pages} {1487} (\bibinfo {year} {2024})}\BibitemShut {NoStop}%
\bibitem [{\citenamefont {Somma}(2019)}]{somma2019quantum}%
  \BibitemOpen
  \bibfield  {author} {\bibinfo {author} {\bibfnamefont {R.~D.}\ \bibnamefont {Somma}},\ }\bibfield  {title} {\bibinfo {title} {Quantum eigenvalue estimation via time series analysis},\ }\href {https://doi.org/10.1088/1367-2630/ab5c60} {\bibfield  {journal} {\bibinfo  {journal} {New J. Phys.}\ }\textbf {\bibinfo {volume} {21}},\ \bibinfo {pages} {123025} (\bibinfo {year} {2019})}\BibitemShut {NoStop}%
\bibitem [{\citenamefont {Trotter}(1959)}]{trotter1959product}%
  \BibitemOpen
  \bibfield  {author} {\bibinfo {author} {\bibfnamefont {H.~F.}\ \bibnamefont {Trotter}},\ }\bibfield  {title} {\bibinfo {title} {On the product of semi-groups of operators},\ }\href {https://doi.org/10.1090/S0002-9939-1959-0108732-6} {\bibfield  {journal} {\bibinfo  {journal} {Proc. Am. Math. Soc.}\ }\textbf {\bibinfo {volume} {10}},\ \bibinfo {pages} {545–551} (\bibinfo {year} {1959})}\BibitemShut {NoStop}%
\bibitem [{\citenamefont {Suzuki}(1976)}]{suzuki1976generalized}%
  \BibitemOpen
  \bibfield  {author} {\bibinfo {author} {\bibfnamefont {M.}~\bibnamefont {Suzuki}},\ }\bibfield  {title} {\bibinfo {title} {Generalized {Trotter's} formula and systematic approximants of exponential operators and inner derivations with applications to many-body problems},\ }\href {https://doi.org/10.1007/BF01609348} {\bibfield  {journal} {\bibinfo  {journal} {Commun. Math. Phys.}\ }\textbf {\bibinfo {volume} {51}},\ \bibinfo {pages} {183–190} (\bibinfo {year} {1976})}\BibitemShut {NoStop}%
\bibitem [{\citenamefont {Suzuki}(1991)}]{suzuki1991general}%
  \BibitemOpen
  \bibfield  {author} {\bibinfo {author} {\bibfnamefont {M.}~\bibnamefont {Suzuki}},\ }\bibfield  {title} {\bibinfo {title} {General theory of fractal path integrals with applications to many-body theories and statistical physics},\ }\href {https://doi.org/10.1063/1.529425} {\bibfield  {journal} {\bibinfo  {journal} {J. Math. Phys.}\ }\textbf {\bibinfo {volume} {32}},\ \bibinfo {pages} {400–407} (\bibinfo {year} {1991})}\BibitemShut {NoStop}%
\bibitem [{\citenamefont {Lloyd}(1996)}]{lloyd1996universal}%
  \BibitemOpen
  \bibfield  {author} {\bibinfo {author} {\bibfnamefont {S.}~\bibnamefont {Lloyd}},\ }\bibfield  {title} {\bibinfo {title} {Universal quantum simulators},\ }\href {https://doi.org/10.1126/science.273.5278.1073} {\bibfield  {journal} {\bibinfo  {journal} {Science}\ }\textbf {\bibinfo {volume} {273}},\ \bibinfo {pages} {1073–1078} (\bibinfo {year} {1996})}\BibitemShut {NoStop}%
\bibitem [{\citenamefont {Campbell}(2019)}]{campbell2019random}%
  \BibitemOpen
  \bibfield  {author} {\bibinfo {author} {\bibfnamefont {E.}~\bibnamefont {Campbell}},\ }\bibfield  {title} {\bibinfo {title} {Random compiler for fast {Hamiltonian} simulation},\ }\href {https://doi.org/10.1103/PhysRevLett.123.070503} {\bibfield  {journal} {\bibinfo  {journal} {Phys. Rev. Lett.}\ }\textbf {\bibinfo {volume} {123}},\ \bibinfo {pages} {070503} (\bibinfo {year} {2019})}\BibitemShut {NoStop}%
\bibitem [{\citenamefont {Kiumi}\ and\ \citenamefont {Koczor}(2025)}]{kiumi2025te}%
  \BibitemOpen
  \bibfield  {author} {\bibinfo {author} {\bibfnamefont {C.}~\bibnamefont {Kiumi}}\ and\ \bibinfo {author} {\bibfnamefont {B.}~\bibnamefont {Koczor}},\ }\bibfield  {title} {\bibinfo {title} {{TE-PAI}: exact time evolution by sampling random circuits},\ }\href {https://doi.org/10.1088/2058-9565/ae1160} {\bibfield  {journal} {\bibinfo  {journal} {Quantum Sci. Technol.}\ }\textbf {\bibinfo {volume} {10}},\ \bibinfo {pages} {045071} (\bibinfo {year} {2025})}\BibitemShut {NoStop}%
\bibitem [{\citenamefont {Li}\ and\ \citenamefont {Benjamin}(2017)}]{li2017efficient}%
  \BibitemOpen
  \bibfield  {author} {\bibinfo {author} {\bibfnamefont {Y.}~\bibnamefont {Li}}\ and\ \bibinfo {author} {\bibfnamefont {S.~C.}\ \bibnamefont {Benjamin}},\ }\bibfield  {title} {\bibinfo {title} {Efficient variational quantum simulator incorporating active error minimization},\ }\href {https://doi.org/10.1103/PhysRevX.7.021050} {\bibfield  {journal} {\bibinfo  {journal} {Phys. Rev. X}\ }\textbf {\bibinfo {volume} {7}},\ \bibinfo {pages} {021050} (\bibinfo {year} {2017})}\BibitemShut {NoStop}%
\bibitem [{\citenamefont {Yuan}\ \emph {et~al.}(2019)\citenamefont {Yuan}, \citenamefont {Endo}, \citenamefont {Zhao}, \citenamefont {Li},\ and\ \citenamefont {Benjamin}}]{yuan2019theory}%
  \BibitemOpen
  \bibfield  {author} {\bibinfo {author} {\bibfnamefont {X.}~\bibnamefont {Yuan}}, \bibinfo {author} {\bibfnamefont {S.}~\bibnamefont {Endo}}, \bibinfo {author} {\bibfnamefont {Q.}~\bibnamefont {Zhao}}, \bibinfo {author} {\bibfnamefont {Y.}~\bibnamefont {Li}},\ and\ \bibinfo {author} {\bibfnamefont {S.~C.}\ \bibnamefont {Benjamin}},\ }\bibfield  {title} {\bibinfo {title} {Theory of variational quantum simulation},\ }\href {https://doi.org/10.22331/q-2019-10-07-191} {\bibfield  {journal} {\bibinfo  {journal} {{Quantum}}\ }\textbf {\bibinfo {volume} {3}},\ \bibinfo {pages} {191} (\bibinfo {year} {2019})}\BibitemShut {NoStop}%
\bibitem [{\citenamefont {Trefethen}(2022)}]{trefethen2022exactness}%
  \BibitemOpen
  \bibfield  {author} {\bibinfo {author} {\bibfnamefont {L.~N.}\ \bibnamefont {Trefethen}},\ }\bibfield  {title} {\bibinfo {title} {Exactness of quadrature formulas},\ }\href {https://doi.org/10.1137/20M1389522} {\bibfield  {journal} {\bibinfo  {journal} {SIAM Rev.}\ }\textbf {\bibinfo {volume} {64}},\ \bibinfo {pages} {132} (\bibinfo {year} {2022})}\BibitemShut {NoStop}%
\bibitem [{\citenamefont {McArdle}\ \emph {et~al.}(2026)\citenamefont {McArdle}, \citenamefont {Gily\'en},\ and\ \citenamefont {Berta}}]{mcardle2026quantum}%
  \BibitemOpen
  \bibfield  {author} {\bibinfo {author} {\bibfnamefont {S.}~\bibnamefont {McArdle}}, \bibinfo {author} {\bibfnamefont {A.}~\bibnamefont {Gily\'en}},\ and\ \bibinfo {author} {\bibfnamefont {M.}~\bibnamefont {Berta}},\ }\bibfield  {title} {\bibinfo {title} {Quantum state preparation without coherent arithmetic},\ }\href {https://doi.org/10.1103/ntvs-c48s} {\bibfield  {journal} {\bibinfo  {journal} {Phys. Rev. Lett.}\ }\textbf {\bibinfo {volume} {136}},\ \bibinfo {pages} {240603} (\bibinfo {year} {2026})}\BibitemShut {NoStop}%
\bibitem [{\citenamefont {Iaconis}\ \emph {et~al.}(2024)\citenamefont {Iaconis}, \citenamefont {Johri},\ and\ \citenamefont {Zhu}}]{iaconis2024quantum}%
  \BibitemOpen
  \bibfield  {author} {\bibinfo {author} {\bibfnamefont {J.}~\bibnamefont {Iaconis}}, \bibinfo {author} {\bibfnamefont {S.}~\bibnamefont {Johri}},\ and\ \bibinfo {author} {\bibfnamefont {E.~Y.}\ \bibnamefont {Zhu}},\ }\bibfield  {title} {\bibinfo {title} {Quantum state preparation of normal distributions using matrix product states},\ }\href {https://doi.org/10.1038/s41534-024-00805-0} {\bibfield  {journal} {\bibinfo  {journal} {npj Quantum Inf.}\ }\textbf {\bibinfo {volume} {10}},\ \bibinfo {pages} {15} (\bibinfo {year} {2024})}\BibitemShut {NoStop}%
\bibitem [{\citenamefont {Zoufal}\ \emph {et~al.}(2019)\citenamefont {Zoufal}, \citenamefont {Lucchi},\ and\ \citenamefont {Woerner}}]{zoufal2019quantum}%
  \BibitemOpen
  \bibfield  {author} {\bibinfo {author} {\bibfnamefont {C.}~\bibnamefont {Zoufal}}, \bibinfo {author} {\bibfnamefont {A.}~\bibnamefont {Lucchi}},\ and\ \bibinfo {author} {\bibfnamefont {S.}~\bibnamefont {Woerner}},\ }\bibfield  {title} {\bibinfo {title} {Quantum generative adversarial networks for learning and loading random distributions},\ }\href {https://doi.org/10.1038/s41534-019-0223-2} {\bibfield  {journal} {\bibinfo  {journal} {npj Quantum Inf.}\ }\textbf {\bibinfo {volume} {5}},\ \bibinfo {pages} {103} (\bibinfo {year} {2019})}\BibitemShut {NoStop}%
\bibitem [{\citenamefont {Coppersmith}(1994)}]{coppersmith1994approximate}%
  \BibitemOpen
  \bibfield  {author} {\bibinfo {author} {\bibfnamefont {D.}~\bibnamefont {Coppersmith}},\ }\href {https://arxiv.org/abs/quant-ph/0201067} {\emph {\bibinfo {title} {An approximate {Fourier} transform useful in quantum factoring}}},\ \bibinfo {type} {Tech. Rep.}\ (\bibinfo  {institution} {IBM Research Division},\ \bibinfo {year} {1994})\BibitemShut {NoStop}%
\bibitem [{\citenamefont {Higham}(2008)}]{higham2008functions}%
  \BibitemOpen
  \bibfield  {author} {\bibinfo {author} {\bibfnamefont {N.~J.}\ \bibnamefont {Higham}},\ }\href {https://doi.org/10.1137/1.9780898717778} {\emph {\bibinfo {title} {Functions of Matrices: Theory and Computation}}}\ (\bibinfo  {publisher} {Society for Industrial and Applied Mathematics},\ \bibinfo {year} {2008})\BibitemShut {NoStop}%
\bibitem [{\citenamefont {Mande}\ and\ \citenamefont {Wolf}(2026)}]{mande2026tight}%
  \BibitemOpen
  \bibfield  {author} {\bibinfo {author} {\bibfnamefont {N.~S.}\ \bibnamefont {Mande}}\ and\ \bibinfo {author} {\bibfnamefont {R.~d.}\ \bibnamefont {Wolf}},\ }\bibfield  {title} {\bibinfo {title} {Tight bounds for quantum phase estimation and related problems},\ }\href {https://doi.org/10.22331/q-2026-06-15-2140} {\bibfield  {journal} {\bibinfo  {journal} {Quantum}\ }\textbf {\bibinfo {volume} {10}},\ \bibinfo {pages} {2140} (\bibinfo {year} {2026})}\BibitemShut {NoStop}%
\bibitem [{\citenamefont {Nagaj}\ \emph {et~al.}(2009)\citenamefont {Nagaj}, \citenamefont {Wocjan},\ and\ \citenamefont {Zhang}}]{nagaj2009fast}%
  \BibitemOpen
  \bibfield  {author} {\bibinfo {author} {\bibfnamefont {D.}~\bibnamefont {Nagaj}}, \bibinfo {author} {\bibfnamefont {P.}~\bibnamefont {Wocjan}},\ and\ \bibinfo {author} {\bibfnamefont {Y.}~\bibnamefont {Zhang}},\ }\bibfield  {title} {\bibinfo {title} {Fast amplification of {QMA}},\ }\href {https://doi.org/10.26421/qic9.11-12-8} {\bibfield  {journal} {\bibinfo  {journal} {Quantum Inf. Comput.}\ }\textbf {\bibinfo {volume} {9}},\ \bibinfo {pages} {1053–1068} (\bibinfo {year} {2009})}\BibitemShut {NoStop}%
\bibitem [{\citenamefont {Huang}\ \emph {et~al.}(2026)\citenamefont {Huang}, \citenamefont {Boyd}, \citenamefont {Anselmetti}, \citenamefont {Degroote}, \citenamefont {Moll}, \citenamefont {Santagati}, \citenamefont {Streif}, \citenamefont {Ries}, \citenamefont {Marti-Dafcik}, \citenamefont {Jnane}, \citenamefont {Simon}, \citenamefont {Wiebe}, \citenamefont {Bromley},\ and\ \citenamefont {Koczor}}]{huang2026fullqubit}%
  \BibitemOpen
  \bibfield  {author} {\bibinfo {author} {\bibfnamefont {P.-W.}\ \bibnamefont {Huang}}, \bibinfo {author} {\bibfnamefont {G.}~\bibnamefont {Boyd}}, \bibinfo {author} {\bibfnamefont {G.-L.~R.}\ \bibnamefont {Anselmetti}}, \bibinfo {author} {\bibfnamefont {M.}~\bibnamefont {Degroote}}, \bibinfo {author} {\bibfnamefont {N.}~\bibnamefont {Moll}}, \bibinfo {author} {\bibfnamefont {R.}~\bibnamefont {Santagati}}, \bibinfo {author} {\bibfnamefont {M.}~\bibnamefont {Streif}}, \bibinfo {author} {\bibfnamefont {B.}~\bibnamefont {Ries}}, \bibinfo {author} {\bibfnamefont {D.}~\bibnamefont {Marti-Dafcik}}, \bibinfo {author} {\bibfnamefont {H.}~\bibnamefont {Jnane}}, \bibinfo {author} {\bibfnamefont {S.}~\bibnamefont {Simon}}, \bibinfo {author} {\bibfnamefont {N.}~\bibnamefont {Wiebe}}, \bibinfo {author} {\bibfnamefont {T.~R.}\ \bibnamefont {Bromley}},\ and\ \bibinfo {author} {\bibfnamefont {B.}~\bibnamefont {Koczor}},\ }\bibfield  {title} {\bibinfo {title} {Fullqubit alchemist: Quantum algorithm for alchemical free energy calculations},\ }\href {https://doi.org/10.1038/s41534-026-01275-2} {\bibfield  {journal} {\bibinfo  {journal} {npj Quantum Inf.}\ } (\bibinfo {year} {2026})}\BibitemShut {NoStop}%
\bibitem [{\citenamefont {Belovs}(2019)}]{belovs2019quantum}%
  \BibitemOpen
  \bibfield  {author} {\bibinfo {author} {\bibfnamefont {A.}~\bibnamefont {Belovs}},\ }\bibfield  {title} {\bibinfo {title} {Quantum algorithms for classical probability distributions},\ }in\ \href {https://doi.org/10.4230/LIPIcs.ESA.2019.16} {\emph {\bibinfo {booktitle} {27th Annual European Symposium on Algorithms (ESA 2019)}}},\ \bibinfo {series} {Leibniz International Proceedings in Informatics (LIPIcs)}, Vol.\ \bibinfo {volume} {144},\ \bibinfo {editor} {edited by\ \bibinfo {editor} {\bibfnamefont {M.~A.}\ \bibnamefont {Bender}}, \bibinfo {editor} {\bibfnamefont {O.}~\bibnamefont {Svensson}},\ and\ \bibinfo {editor} {\bibfnamefont {G.}~\bibnamefont {Herman}}}\ (\bibinfo  {publisher} {Schloss Dagstuhl -- Leibniz-Zentrum f{\"u}r Informatik},\ \bibinfo {address} {Dagstuhl, Germany},\ \bibinfo {year} {2019})\ pp.\ \bibinfo {pages} {16:1--16:11}\BibitemShut {NoStop}%
\bibitem [{\citenamefont {Cirac}\ \emph {et~al.}(1999)\citenamefont {Cirac}, \citenamefont {Ekert},\ and\ \citenamefont {Macchiavello}}]{cirac1999optimal}%
  \BibitemOpen
  \bibfield  {author} {\bibinfo {author} {\bibfnamefont {J.~I.}\ \bibnamefont {Cirac}}, \bibinfo {author} {\bibfnamefont {A.~K.}\ \bibnamefont {Ekert}},\ and\ \bibinfo {author} {\bibfnamefont {C.}~\bibnamefont {Macchiavello}},\ }\bibfield  {title} {\bibinfo {title} {Optimal purification of single qubits},\ }\href {https://doi.org/10.1103/PhysRevLett.82.4344} {\bibfield  {journal} {\bibinfo  {journal} {Phys. Rev. Lett.}\ }\textbf {\bibinfo {volume} {82}},\ \bibinfo {pages} {4344} (\bibinfo {year} {1999})}\BibitemShut {NoStop}%
\bibitem [{\citenamefont {Li}\ \emph {et~al.}(2024)\citenamefont {Li}, \citenamefont {Fu}, \citenamefont {Isogawa}, \citenamefont {Silva},\ and\ \citenamefont {Chuang}}]{li2024optimal}%
  \BibitemOpen
  \bibfield  {author} {\bibinfo {author} {\bibfnamefont {Z.}~\bibnamefont {Li}}, \bibinfo {author} {\bibfnamefont {H.}~\bibnamefont {Fu}}, \bibinfo {author} {\bibfnamefont {T.}~\bibnamefont {Isogawa}}, \bibinfo {author} {\bibfnamefont {C.}~\bibnamefont {Silva}},\ and\ \bibinfo {author} {\bibfnamefont {I.}~\bibnamefont {Chuang}},\ }\href@noop {} {\bibinfo {title} {Optimal quantum purity amplification}} (\bibinfo {year} {2024}),\ \Eprint {https://arxiv.org/abs/2409.18167} {arXiv:2409.18167 [quant-ph]} \BibitemShut {NoStop}%
\bibitem [{\citenamefont {Grier}\ \emph {et~al.}(2025)\citenamefont {Grier}, \citenamefont {Leung}, \citenamefont {Li}, \citenamefont {Pashayan},\ and\ \citenamefont {Schaeffer}}]{grier2025streaming}%
  \BibitemOpen
  \bibfield  {author} {\bibinfo {author} {\bibfnamefont {D.}~\bibnamefont {Grier}}, \bibinfo {author} {\bibfnamefont {D.}~\bibnamefont {Leung}}, \bibinfo {author} {\bibfnamefont {Z.}~\bibnamefont {Li}}, \bibinfo {author} {\bibfnamefont {H.}~\bibnamefont {Pashayan}},\ and\ \bibinfo {author} {\bibfnamefont {L.}~\bibnamefont {Schaeffer}},\ }\href@noop {} {\bibinfo {title} {Streaming quantum state purification for general mixed states}} (\bibinfo {year} {2025}),\ \Eprint {https://arxiv.org/abs/2503.22644} {arXiv:2503.22644 [quant-ph]} \BibitemShut {NoStop}%
\bibitem [{\citenamefont {Gily\'{e}n}\ and\ \citenamefont {Li}(2020)}]{gilyen2020distributional}%
  \BibitemOpen
  \bibfield  {author} {\bibinfo {author} {\bibfnamefont {A.}~\bibnamefont {Gily\'{e}n}}\ and\ \bibinfo {author} {\bibfnamefont {T.}~\bibnamefont {Li}},\ }\bibfield  {title} {\bibinfo {title} {Distributional property testing in a quantum world},\ }in\ \href {https://doi.org/10.4230/LIPIcs.ITCS.2020.25} {\emph {\bibinfo {booktitle} {11th Innovations in Theoretical Computer Science Conference (ITCS 2020)}}},\ \bibinfo {series} {Leibniz International Proceedings in Informatics (LIPIcs)}, Vol.\ \bibinfo {volume} {151},\ \bibinfo {editor} {edited by\ \bibinfo {editor} {\bibfnamefont {T.}~\bibnamefont {Vidick}}}\ (\bibinfo  {publisher} {Schloss Dagstuhl -- Leibniz-Zentrum f{\"u}r Informatik},\ \bibinfo {address} {Dagstuhl, Germany},\ \bibinfo {year} {2020})\ pp.\ \bibinfo {pages} {25:1--25:19}\BibitemShut {NoStop}%
\bibitem [{\citenamefont {Tang}\ \emph {et~al.}(2025)\citenamefont {Tang}, \citenamefont {Wright},\ and\ \citenamefont {Zhandry}}]{tang2025conjugate}%
  \BibitemOpen
  \bibfield  {author} {\bibinfo {author} {\bibfnamefont {E.}~\bibnamefont {Tang}}, \bibinfo {author} {\bibfnamefont {J.}~\bibnamefont {Wright}},\ and\ \bibinfo {author} {\bibfnamefont {M.}~\bibnamefont {Zhandry}},\ }\href@noop {} {\bibinfo {title} {Conjugate queries can help}} (\bibinfo {year} {2025}),\ \Eprint {https://arxiv.org/abs/2510.07622} {arXiv:2510.07622 [quant-ph]} \BibitemShut {NoStop}%
\bibitem [{\citenamefont {Wang}\ and\ \citenamefont {Zhang}(2025)}]{wang2025quantum}%
  \BibitemOpen
  \bibfield  {author} {\bibinfo {author} {\bibfnamefont {Q.}~\bibnamefont {Wang}}\ and\ \bibinfo {author} {\bibfnamefont {Z.}~\bibnamefont {Zhang}},\ }\bibfield  {title} {\bibinfo {title} {Quantum lower bounds by sample-to-query lifting},\ }\href {https://doi.org/10.1137/24M1638616} {\bibfield  {journal} {\bibinfo  {journal} {SIAM J. Comput.}\ }\textbf {\bibinfo {volume} {54}},\ \bibinfo {pages} {1294} (\bibinfo {year} {2025})}\BibitemShut {NoStop}%
\bibitem [{\citenamefont {Chen}\ \emph {et~al.}(2025{\natexlab{b}})\citenamefont {Chen}, \citenamefont {Wang},\ and\ \citenamefont {Zhang}}]{chen2025list}%
  \BibitemOpen
  \bibfield  {author} {\bibinfo {author} {\bibfnamefont {K.}~\bibnamefont {Chen}}, \bibinfo {author} {\bibfnamefont {Q.}~\bibnamefont {Wang}},\ and\ \bibinfo {author} {\bibfnamefont {Z.}~\bibnamefont {Zhang}},\ }\href@noop {} {\bibinfo {title} {A list of complexity bounds for property testing by quantum sample-to-query lifting}} (\bibinfo {year} {2025}{\natexlab{b}}),\ \Eprint {https://arxiv.org/abs/2512.01971} {arXiv:2512.01971 [quant-ph]} \BibitemShut {NoStop}%
\bibitem [{\citenamefont {Low}\ and\ \citenamefont {Chuang}(2017)}]{low2017hamiltonian}%
  \BibitemOpen
  \bibfield  {author} {\bibinfo {author} {\bibfnamefont {G.~H.}\ \bibnamefont {Low}}\ and\ \bibinfo {author} {\bibfnamefont {I.~L.}\ \bibnamefont {Chuang}},\ }\href@noop {} {\bibinfo {title} {Hamiltonian simulation by uniform spectral amplification}} (\bibinfo {year} {2017}),\ \Eprint {https://arxiv.org/abs/1707.05391} {arXiv:1707.05391 [quant-ph]} \BibitemShut {NoStop}%
\bibitem [{\citenamefont {Zlokapa}\ and\ \citenamefont {Somma}(2024)}]{zlokapa2024hamiltonian}%
  \BibitemOpen
  \bibfield  {author} {\bibinfo {author} {\bibfnamefont {A.}~\bibnamefont {Zlokapa}}\ and\ \bibinfo {author} {\bibfnamefont {R.~D.}\ \bibnamefont {Somma}},\ }\bibfield  {title} {\bibinfo {title} {Hamiltonian simulation for low-energy states with optimal time dependence},\ }\href {https://doi.org/10.22331/q-2024-08-27-1449} {\bibfield  {journal} {\bibinfo  {journal} {{Quantum}}\ }\textbf {\bibinfo {volume} {8}},\ \bibinfo {pages} {1449} (\bibinfo {year} {2024})}\BibitemShut {NoStop}%
\bibitem [{\citenamefont {Guo}\ \emph {et~al.}(2024{\natexlab{a}})\citenamefont {Guo}, \citenamefont {Mitarai},\ and\ \citenamefont {Fujii}}]{guo2024nonlinear}%
  \BibitemOpen
  \bibfield  {author} {\bibinfo {author} {\bibfnamefont {N.}~\bibnamefont {Guo}}, \bibinfo {author} {\bibfnamefont {K.}~\bibnamefont {Mitarai}},\ and\ \bibinfo {author} {\bibfnamefont {K.}~\bibnamefont {Fujii}},\ }\bibfield  {title} {\bibinfo {title} {Nonlinear transformation of complex amplitudes via quantum singular value transformation},\ }\href {https://link.aps.org/doi/10.1103/PhysRevResearch.6.043227} {\bibfield  {journal} {\bibinfo  {journal} {Phys. Rev. Res.}\ }\textbf {\bibinfo {volume} {6}},\ \bibinfo {pages} {043227} (\bibinfo {year} {2024}{\natexlab{a}})}\BibitemShut {NoStop}%
\bibitem [{\citenamefont {Rattew}\ and\ \citenamefont {Rebentrost}(2023)}]{rattew2023nonlinear}%
  \BibitemOpen
  \bibfield  {author} {\bibinfo {author} {\bibfnamefont {A.~G.}\ \bibnamefont {Rattew}}\ and\ \bibinfo {author} {\bibfnamefont {P.}~\bibnamefont {Rebentrost}},\ }\href@noop {} {\bibinfo {title} {Non-linear transformations of quantum amplitudes: Exponential improvement, generalization, and applications}} (\bibinfo {year} {2023}),\ \Eprint {https://arxiv.org/abs/2309.09839} {arXiv:2309.09839 [quant-ph]} \BibitemShut {NoStop}%
\bibitem [{\citenamefont {Gonzalez-Conde}\ \emph {et~al.}(2024)\citenamefont {Gonzalez-Conde}, \citenamefont {Watts}, \citenamefont {Rodriguez-Grasa},\ and\ \citenamefont {Sanz}}]{gonzalezconde2024efficient}%
  \BibitemOpen
  \bibfield  {author} {\bibinfo {author} {\bibfnamefont {J.}~\bibnamefont {Gonzalez-Conde}}, \bibinfo {author} {\bibfnamefont {T.~W.}\ \bibnamefont {Watts}}, \bibinfo {author} {\bibfnamefont {P.}~\bibnamefont {Rodriguez-Grasa}},\ and\ \bibinfo {author} {\bibfnamefont {M.}~\bibnamefont {Sanz}},\ }\bibfield  {title} {\bibinfo {title} {Efficient quantum amplitude encoding of polynomial functions},\ }\href {https://doi.org/10.22331/q-2024-03-21-1297} {\bibfield  {journal} {\bibinfo  {journal} {Quantum}\ }\textbf {\bibinfo {volume} {8}},\ \bibinfo {pages} {1297} (\bibinfo {year} {2024})}\BibitemShut {NoStop}%
\bibitem [{\citenamefont {Ivashkov}\ \emph {et~al.}(2026)\citenamefont {Ivashkov}, \citenamefont {Huang}, \citenamefont {Koor}, \citenamefont {Pira},\ and\ \citenamefont {Rebentrost}}]{ivashkov2026qkan}%
  \BibitemOpen
  \bibfield  {author} {\bibinfo {author} {\bibfnamefont {P.}~\bibnamefont {Ivashkov}}, \bibinfo {author} {\bibfnamefont {P.-W.}\ \bibnamefont {Huang}}, \bibinfo {author} {\bibfnamefont {K.}~\bibnamefont {Koor}}, \bibinfo {author} {\bibfnamefont {L.}~\bibnamefont {Pira}},\ and\ \bibinfo {author} {\bibfnamefont {P.}~\bibnamefont {Rebentrost}},\ }\bibfield  {title} {\bibinfo {title} {{QKAN}: quantum {Kolmogorov-Arnold} networks with applications in machine learning and multivariate state preparation},\ }\href {https://doi.org/10.1038/s41534-026-01202-5} {\bibfield  {journal} {\bibinfo  {journal} {npj Quantum Inf.}\ }\textbf {\bibinfo {volume} {12}},\ \bibinfo {pages} {73} (\bibinfo {year} {2026})}\BibitemShut {NoStop}%
\bibitem [{\citenamefont {Guo}\ \emph {et~al.}(2024{\natexlab{b}})\citenamefont {Guo}, \citenamefont {Yu}, \citenamefont {Choi}, \citenamefont {Han}, \citenamefont {Agrawal}, \citenamefont {Nakaji}, \citenamefont {Aspuru-Guzik},\ and\ \citenamefont {Rebentrost}}]{guo2024quantum}%
  \BibitemOpen
  \bibfield  {author} {\bibinfo {author} {\bibfnamefont {N.}~\bibnamefont {Guo}}, \bibinfo {author} {\bibfnamefont {Z.}~\bibnamefont {Yu}}, \bibinfo {author} {\bibfnamefont {M.}~\bibnamefont {Choi}}, \bibinfo {author} {\bibfnamefont {Y.}~\bibnamefont {Han}}, \bibinfo {author} {\bibfnamefont {A.}~\bibnamefont {Agrawal}}, \bibinfo {author} {\bibfnamefont {K.}~\bibnamefont {Nakaji}}, \bibinfo {author} {\bibfnamefont {A.}~\bibnamefont {Aspuru-Guzik}},\ and\ \bibinfo {author} {\bibfnamefont {P.}~\bibnamefont {Rebentrost}},\ }\href@noop {} {\bibinfo {title} {Quantum transformer: Accelerating model inference via quantum linear algebra}} (\bibinfo {year} {2024}{\natexlab{b}}),\ \Eprint {https://arxiv.org/abs/2402.16714} {arXiv:2402.16714 [quant-ph]} \BibitemShut {NoStop}%
\bibitem [{\citenamefont {Rattew}\ \emph {et~al.}(2026)\citenamefont {Rattew}, \citenamefont {Huang}, \citenamefont {Guo}, \citenamefont {Pira},\ and\ \citenamefont {Rebentrost}}]{rattew2026accelerating}%
  \BibitemOpen
  \bibfield  {author} {\bibinfo {author} {\bibfnamefont {A.~G.}\ \bibnamefont {Rattew}}, \bibinfo {author} {\bibfnamefont {P.-W.}\ \bibnamefont {Huang}}, \bibinfo {author} {\bibfnamefont {N.}~\bibnamefont {Guo}}, \bibinfo {author} {\bibfnamefont {L.}~\bibnamefont {Pira}},\ and\ \bibinfo {author} {\bibfnamefont {P.}~\bibnamefont {Rebentrost}},\ }\bibfield  {title} {\bibinfo {title} {Accelerating inference for multilayer neural networks with quantum computers},\ }in\ \href {https://proceedings.iclr.cc/paper_files/paper/2026/hash/11bcd94ac570b4e6fdca65b4b88b7a59-Abstract-Conference.html} {\emph {\bibinfo {booktitle} {International Conference on Learning Representations}}},\ Vol.\ \bibinfo {volume} {2026},\ \bibinfo {editor} {edited by\ \bibinfo {editor} {\bibfnamefont {C.}~\bibnamefont {Vondrick}}, \bibinfo {editor} {\bibfnamefont {B.}~\bibnamefont {Hariharan}}, \bibinfo {editor} {\bibfnamefont {C.}~\bibnamefont {Raffel}}, \bibinfo {editor} {\bibfnamefont {L.}~\bibnamefont {Pinto}}, \bibinfo {editor} {\bibfnamefont {D.}~\bibnamefont {Yang}},\ and\ \bibinfo {editor} {\bibfnamefont {A.}~\bibnamefont {Faust}}}\ (\bibinfo {year} {2026})\ pp.\ \bibinfo {pages} {10498--10542}\BibitemShut {NoStop}%
\bibitem [{\citenamefont {Guo}\ \emph {et~al.}(2026)\citenamefont {Guo}, \citenamefont {Huang}, \citenamefont {Wang}, \citenamefont {Thompson}, \citenamefont {Rebentrost}, \citenamefont {Gu},\ and\ \citenamefont {Yang}}]{guo2026quantum}%
  \BibitemOpen
  \bibfield  {author} {\bibinfo {author} {\bibfnamefont {N.}~\bibnamefont {Guo}}, \bibinfo {author} {\bibfnamefont {P.-W.}\ \bibnamefont {Huang}}, \bibinfo {author} {\bibfnamefont {Q.}~\bibnamefont {Wang}}, \bibinfo {author} {\bibfnamefont {J.}~\bibnamefont {Thompson}}, \bibinfo {author} {\bibfnamefont {P.}~\bibnamefont {Rebentrost}}, \bibinfo {author} {\bibfnamefont {M.}~\bibnamefont {Gu}},\ and\ \bibinfo {author} {\bibfnamefont {C.}~\bibnamefont {Yang}},\ }\href@noop {} {\bibinfo {title} {Quantum enhanced rare event discovery and sampling}} (\bibinfo {year} {2026}),\ \Eprint {https://arxiv.org/abs/2606.06316} {arXiv:2606.06316 [quant-ph]} \BibitemShut {NoStop}%
\bibitem [{\citenamefont {Patel}\ \emph {et~al.}(2026)\citenamefont {Patel}, \citenamefont {Tan}, \citenamefont {Suba\c{s}\i},\ and\ \citenamefont {Sornborger}}]{patel2026optimal}%
  \BibitemOpen
  \bibfield  {author} {\bibinfo {author} {\bibfnamefont {D.}~\bibnamefont {Patel}}, \bibinfo {author} {\bibfnamefont {S.~J.~S.}\ \bibnamefont {Tan}}, \bibinfo {author} {\bibfnamefont {Y.}~\bibnamefont {Suba\c{s}\i}},\ and\ \bibinfo {author} {\bibfnamefont {A.~T.}\ \bibnamefont {Sornborger}},\ }\bibfield  {title} {\bibinfo {title} {Optimal coherent quantum phase estimation via tapering},\ }\href {https://doi.org/10.1103/l5y6-6zxv} {\bibfield  {journal} {\bibinfo  {journal} {PRX Quantum}\ }\textbf {\bibinfo {volume} {7}},\ \bibinfo {pages} {020302} (\bibinfo {year} {2026})}\BibitemShut {NoStop}%
\bibitem [{\citenamefont {Jones}\ \emph {et~al.}(2019{\natexlab{b}})\citenamefont {Jones}, \citenamefont {Brown}, \citenamefont {Bush},\ and\ \citenamefont {Benjamin}}]{jones2019quest}%
  \BibitemOpen
  \bibfield  {author} {\bibinfo {author} {\bibfnamefont {T.}~\bibnamefont {Jones}}, \bibinfo {author} {\bibfnamefont {A.}~\bibnamefont {Brown}}, \bibinfo {author} {\bibfnamefont {I.}~\bibnamefont {Bush}},\ and\ \bibinfo {author} {\bibfnamefont {S.~C.}\ \bibnamefont {Benjamin}},\ }\bibfield  {title} {\bibinfo {title} {{QuEST} and high performance simulation of quantum computers},\ }\href {https://doi.org/10.1038/s41598-019-47174-9} {\bibfield  {journal} {\bibinfo  {journal} {Sci. Rep.}\ }\textbf {\bibinfo {volume} {9}},\ \bibinfo {pages} {10736} (\bibinfo {year} {2019}{\natexlab{b}})}\BibitemShut {NoStop}%
\bibitem [{\citenamefont {Jones}\ and\ \citenamefont {Benjamin}(2020)}]{jones2020questlink}%
  \BibitemOpen
  \bibfield  {author} {\bibinfo {author} {\bibfnamefont {T.}~\bibnamefont {Jones}}\ and\ \bibinfo {author} {\bibfnamefont {S.}~\bibnamefont {Benjamin}},\ }\bibfield  {title} {\bibinfo {title} {{QuESTlink}—{Mathematica} embiggened by a hardware-optimised quantum emulator},\ }\href {https://doi.org/10.1088/2058-9565/ab8506} {\bibfield  {journal} {\bibinfo  {journal} {Quantum Sci. Technol.}\ }\textbf {\bibinfo {volume} {5}},\ \bibinfo {pages} {034012} (\bibinfo {year} {2020})}\BibitemShut {NoStop}%
\bibitem [{\citenamefont {Schirmer}(1982)}]{schirmer1982beyond}%
  \BibitemOpen
  \bibfield  {author} {\bibinfo {author} {\bibfnamefont {J.}~\bibnamefont {Schirmer}},\ }\bibfield  {title} {\bibinfo {title} {Beyond the random-phase approximation: A new approximation scheme for the polarization propagator},\ }\href {https://doi.org/10.1103/PhysRevA.26.2395} {\bibfield  {journal} {\bibinfo  {journal} {Phys. Rev. A}\ }\textbf {\bibinfo {volume} {26}},\ \bibinfo {pages} {2395} (\bibinfo {year} {1982})}\BibitemShut {NoStop}%
\bibitem [{\citenamefont {Roos}\ \emph {et~al.}(1980)\citenamefont {Roos}, \citenamefont {Taylor},\ and\ \citenamefont {Sigbahn}}]{roos1980complete}%
  \BibitemOpen
  \bibfield  {author} {\bibinfo {author} {\bibfnamefont {B.~O.}\ \bibnamefont {Roos}}, \bibinfo {author} {\bibfnamefont {P.~R.}\ \bibnamefont {Taylor}},\ and\ \bibinfo {author} {\bibfnamefont {P.~E.}\ \bibnamefont {Sigbahn}},\ }\bibfield  {title} {\bibinfo {title} {A complete active space {SCF} method {(CASSCF)} using a density matrix formulated super-{CI} approach},\ }\href {https://doi.org/10.1016/0301-0104(80)80045-0} {\bibfield  {journal} {\bibinfo  {journal} {Chem. Phys.}\ }\textbf {\bibinfo {volume} {48}},\ \bibinfo {pages} {157} (\bibinfo {year} {1980})}\BibitemShut {NoStop}%
\bibitem [{\citenamefont {Andersson}\ \emph {et~al.}(1990)\citenamefont {Andersson}, \citenamefont {Malmqvist}, \citenamefont {Roos}, \citenamefont {Sadlej},\ and\ \citenamefont {Wolinski}}]{andersson1990second}%
  \BibitemOpen
  \bibfield  {author} {\bibinfo {author} {\bibfnamefont {K.}~\bibnamefont {Andersson}}, \bibinfo {author} {\bibfnamefont {P.~A.}\ \bibnamefont {Malmqvist}}, \bibinfo {author} {\bibfnamefont {B.~O.}\ \bibnamefont {Roos}}, \bibinfo {author} {\bibfnamefont {A.~J.}\ \bibnamefont {Sadlej}},\ and\ \bibinfo {author} {\bibfnamefont {K.}~\bibnamefont {Wolinski}},\ }\bibfield  {title} {\bibinfo {title} {Second-order perturbation theory with a casscf reference function},\ }\href {https://doi.org/10.1021/j100377a012} {\bibfield  {journal} {\bibinfo  {journal} {J. Phys. Chem.}\ }\textbf {\bibinfo {volume} {94}},\ \bibinfo {pages} {5483–5488} (\bibinfo {year} {1990})}\BibitemShut {NoStop}%
\bibitem [{\citenamefont {Angeli}\ \emph {et~al.}(2001)\citenamefont {Angeli}, \citenamefont {Cimiraglia}, \citenamefont {Evangelisti}, \citenamefont {Leininger},\ and\ \citenamefont {Malrieu}}]{angeli2001introduction}%
  \BibitemOpen
  \bibfield  {author} {\bibinfo {author} {\bibfnamefont {C.}~\bibnamefont {Angeli}}, \bibinfo {author} {\bibfnamefont {R.}~\bibnamefont {Cimiraglia}}, \bibinfo {author} {\bibfnamefont {S.}~\bibnamefont {Evangelisti}}, \bibinfo {author} {\bibfnamefont {T.}~\bibnamefont {Leininger}},\ and\ \bibinfo {author} {\bibfnamefont {J.-P.}\ \bibnamefont {Malrieu}},\ }\bibfield  {title} {\bibinfo {title} {Introduction of $n$-electron valence states for multireference perturbation theory},\ }\href {https://doi.org/10.1063/1.1361246} {\bibfield  {journal} {\bibinfo  {journal} {J. Chem. Phys.}\ }\textbf {\bibinfo {volume} {114}},\ \bibinfo {pages} {10252–10264} (\bibinfo {year} {2001})}\BibitemShut {NoStop}%
\bibitem [{\citenamefont {Cerezo}\ \emph {et~al.}(2021)\citenamefont {Cerezo}, \citenamefont {Arrasmith}, \citenamefont {Babbush}, \citenamefont {Benjamin}, \citenamefont {Endo}, \citenamefont {Fujii}, \citenamefont {McClean}, \citenamefont {Mitarai}, \citenamefont {Yuan}, \citenamefont {Cincio},\ and\ \citenamefont {Coles}}]{cerezo2021variational}%
  \BibitemOpen
  \bibfield  {author} {\bibinfo {author} {\bibfnamefont {M.}~\bibnamefont {Cerezo}}, \bibinfo {author} {\bibfnamefont {A.}~\bibnamefont {Arrasmith}}, \bibinfo {author} {\bibfnamefont {R.}~\bibnamefont {Babbush}}, \bibinfo {author} {\bibfnamefont {S.~C.}\ \bibnamefont {Benjamin}}, \bibinfo {author} {\bibfnamefont {S.}~\bibnamefont {Endo}}, \bibinfo {author} {\bibfnamefont {K.}~\bibnamefont {Fujii}}, \bibinfo {author} {\bibfnamefont {J.~R.}\ \bibnamefont {McClean}}, \bibinfo {author} {\bibfnamefont {K.}~\bibnamefont {Mitarai}}, \bibinfo {author} {\bibfnamefont {X.}~\bibnamefont {Yuan}}, \bibinfo {author} {\bibfnamefont {L.}~\bibnamefont {Cincio}},\ and\ \bibinfo {author} {\bibfnamefont {P.~J.}\ \bibnamefont {Coles}},\ }\bibfield  {title} {\bibinfo {title} {Variational quantum algorithms},\ }\href {https://doi.org/10.1038/s42254-021-00348-9} {\bibfield  {journal} {\bibinfo  {journal} {Nat. Rev. Phys.}\ }\textbf {\bibinfo {volume} {3}},\ \bibinfo {pages} {625–644} (\bibinfo {year} {2021})}\BibitemShut {NoStop}%
\bibitem [{\citenamefont {Tilly}\ \emph {et~al.}(2022)\citenamefont {Tilly}, \citenamefont {Chen}, \citenamefont {Cao}, \citenamefont {Picozzi}, \citenamefont {Setia}, \citenamefont {Li}, \citenamefont {Grant}, \citenamefont {Wossnig}, \citenamefont {Rungger}, \citenamefont {Booth},\ and\ \citenamefont {Tennyson}}]{tilly2022variational}%
  \BibitemOpen
  \bibfield  {author} {\bibinfo {author} {\bibfnamefont {J.}~\bibnamefont {Tilly}}, \bibinfo {author} {\bibfnamefont {H.}~\bibnamefont {Chen}}, \bibinfo {author} {\bibfnamefont {S.}~\bibnamefont {Cao}}, \bibinfo {author} {\bibfnamefont {D.}~\bibnamefont {Picozzi}}, \bibinfo {author} {\bibfnamefont {K.}~\bibnamefont {Setia}}, \bibinfo {author} {\bibfnamefont {Y.}~\bibnamefont {Li}}, \bibinfo {author} {\bibfnamefont {E.}~\bibnamefont {Grant}}, \bibinfo {author} {\bibfnamefont {L.}~\bibnamefont {Wossnig}}, \bibinfo {author} {\bibfnamefont {I.}~\bibnamefont {Rungger}}, \bibinfo {author} {\bibfnamefont {G.~H.}\ \bibnamefont {Booth}},\ and\ \bibinfo {author} {\bibfnamefont {J.}~\bibnamefont {Tennyson}},\ }\bibfield  {title} {\bibinfo {title} {The variational quantum eigensolver: A review of methods and best practices},\ }\href {https://doi.org/10.1016/j.physrep.2022.08.003} {\bibfield  {journal} {\bibinfo  {journal} {Phys. Rep.}\ }\textbf {\bibinfo {volume} {986}},\ \bibinfo {pages} {1–128} (\bibinfo {year} {2022})}\BibitemShut {NoStop}%
\bibitem [{\citenamefont {Peruzzo}\ \emph {et~al.}(2014)\citenamefont {Peruzzo}, \citenamefont {McClean}, \citenamefont {Shadbolt}, \citenamefont {Yung}, \citenamefont {Zhou}, \citenamefont {Love}, \citenamefont {Aspuru-Guzik},\ and\ \citenamefont {O’Brien}}]{peruzzo2014variational}%
  \BibitemOpen
  \bibfield  {author} {\bibinfo {author} {\bibfnamefont {A.}~\bibnamefont {Peruzzo}}, \bibinfo {author} {\bibfnamefont {J.}~\bibnamefont {McClean}}, \bibinfo {author} {\bibfnamefont {P.}~\bibnamefont {Shadbolt}}, \bibinfo {author} {\bibfnamefont {M.-H.}\ \bibnamefont {Yung}}, \bibinfo {author} {\bibfnamefont {X.-Q.}\ \bibnamefont {Zhou}}, \bibinfo {author} {\bibfnamefont {P.~J.}\ \bibnamefont {Love}}, \bibinfo {author} {\bibfnamefont {A.}~\bibnamefont {Aspuru-Guzik}},\ and\ \bibinfo {author} {\bibfnamefont {J.~L.}\ \bibnamefont {O’Brien}},\ }\bibfield  {title} {\bibinfo {title} {A variational eigenvalue solver on a photonic quantum processor},\ }\href {https://doi.org/10.1038/ncomms5213} {\bibfield  {journal} {\bibinfo  {journal} {Nat. Commun.}\ }\textbf {\bibinfo {volume} {5}},\ \bibinfo {pages} {4213} (\bibinfo {year} {2014})}\BibitemShut {NoStop}%
\bibitem [{\citenamefont {Epperly}\ \emph {et~al.}(2022)\citenamefont {Epperly}, \citenamefont {Lin},\ and\ \citenamefont {Nakatsukasa}}]{epperly2022theory}%
  \BibitemOpen
  \bibfield  {author} {\bibinfo {author} {\bibfnamefont {E.~N.}\ \bibnamefont {Epperly}}, \bibinfo {author} {\bibfnamefont {L.}~\bibnamefont {Lin}},\ and\ \bibinfo {author} {\bibfnamefont {Y.}~\bibnamefont {Nakatsukasa}},\ }\bibfield  {title} {\bibinfo {title} {A theory of quantum subspace diagonalization},\ }\href {https://doi.org/10.1137/21m145954x} {\bibfield  {journal} {\bibinfo  {journal} {SIAM J. Matrix Anal. Appl.}\ }\textbf {\bibinfo {volume} {43}},\ \bibinfo {pages} {1263–1290} (\bibinfo {year} {2022})}\BibitemShut {NoStop}%
\bibitem [{\citenamefont {Motta}\ \emph {et~al.}(2024)\citenamefont {Motta}, \citenamefont {Kirby}, \citenamefont {Liepuoniute}, \citenamefont {Sung}, \citenamefont {Cohn}, \citenamefont {Mezzacapo}, \citenamefont {Klymko}, \citenamefont {Nguyen}, \citenamefont {Yoshioka},\ and\ \citenamefont {Rice}}]{motta2024subspace}%
  \BibitemOpen
  \bibfield  {author} {\bibinfo {author} {\bibfnamefont {M.}~\bibnamefont {Motta}}, \bibinfo {author} {\bibfnamefont {W.}~\bibnamefont {Kirby}}, \bibinfo {author} {\bibfnamefont {I.}~\bibnamefont {Liepuoniute}}, \bibinfo {author} {\bibfnamefont {K.~J.}\ \bibnamefont {Sung}}, \bibinfo {author} {\bibfnamefont {J.}~\bibnamefont {Cohn}}, \bibinfo {author} {\bibfnamefont {A.}~\bibnamefont {Mezzacapo}}, \bibinfo {author} {\bibfnamefont {K.}~\bibnamefont {Klymko}}, \bibinfo {author} {\bibfnamefont {N.}~\bibnamefont {Nguyen}}, \bibinfo {author} {\bibfnamefont {N.}~\bibnamefont {Yoshioka}},\ and\ \bibinfo {author} {\bibfnamefont {J.~E.}\ \bibnamefont {Rice}},\ }\bibfield  {title} {\bibinfo {title} {Subspace methods for electronic structure simulations on quantum computers},\ }\href {https://doi.org/10.1088/2516-1075/ad3592} {\bibfield  {journal} {\bibinfo  {journal} {Electron. Struct.}\ }\textbf {\bibinfo {volume} {6}},\ \bibinfo {pages} {013001} (\bibinfo {year} {2024})}\BibitemShut {NoStop}%
\bibitem [{\citenamefont {Boyd}\ \emph {et~al.}(2025)\citenamefont {Boyd}, \citenamefont {Koczor},\ and\ \citenamefont {Cai}}]{boyd2025high}%
  \BibitemOpen
  \bibfield  {author} {\bibinfo {author} {\bibfnamefont {G.}~\bibnamefont {Boyd}}, \bibinfo {author} {\bibfnamefont {B.}~\bibnamefont {Koczor}},\ and\ \bibinfo {author} {\bibfnamefont {Z.}~\bibnamefont {Cai}},\ }\bibfield  {title} {\bibinfo {title} {High-dimensional subspace expansion using classical shadows},\ }\href {https://doi.org/10.1103/PhysRevA.111.022423} {\bibfield  {journal} {\bibinfo  {journal} {Phys. Rev. A}\ }\textbf {\bibinfo {volume} {111}},\ \bibinfo {pages} {022423} (\bibinfo {year} {2025})}\BibitemShut {NoStop}%
\bibitem [{\citenamefont {Ding}\ and\ \citenamefont {Lin}(2023{\natexlab{b}})}]{ding2023simultaneous}%
  \BibitemOpen
  \bibfield  {author} {\bibinfo {author} {\bibfnamefont {Z.}~\bibnamefont {Ding}}\ and\ \bibinfo {author} {\bibfnamefont {L.}~\bibnamefont {Lin}},\ }\bibfield  {title} {\bibinfo {title} {Simultaneous estimation of multiple eigenvalues with short-depth quantum circuit on early fault-tolerant quantum computers},\ }\href {https://doi.org/10.22331/q-2023-10-11-1136} {\bibfield  {journal} {\bibinfo  {journal} {Quantum}\ }\textbf {\bibinfo {volume} {7}},\ \bibinfo {pages} {1136} (\bibinfo {year} {2023}{\natexlab{b}})}\BibitemShut {NoStop}%
\bibitem [{\citenamefont {O'Brien}\ \emph {et~al.}(2019)\citenamefont {O'Brien}, \citenamefont {Tarasinski},\ and\ \citenamefont {Terhal}}]{obrien2019quantum}%
  \BibitemOpen
  \bibfield  {author} {\bibinfo {author} {\bibfnamefont {T.~E.}\ \bibnamefont {O'Brien}}, \bibinfo {author} {\bibfnamefont {B.}~\bibnamefont {Tarasinski}},\ and\ \bibinfo {author} {\bibfnamefont {B.~M.}\ \bibnamefont {Terhal}},\ }\bibfield  {title} {\bibinfo {title} {Quantum phase estimation of multiple eigenvalues for small-scale (noisy) experiments},\ }\href {https://doi.org/10.1088/1367-2630/aafb8e} {\bibfield  {journal} {\bibinfo  {journal} {New J. Phys.}\ }\textbf {\bibinfo {volume} {21}},\ \bibinfo {pages} {023022} (\bibinfo {year} {2019})}\BibitemShut {NoStop}%
\bibitem [{\citenamefont {Chan}\ \emph {et~al.}(2025)\citenamefont {Chan}, \citenamefont {Meister}, \citenamefont {Goh},\ and\ \citenamefont {Koczor}}]{chan2025algorithmic}%
  \BibitemOpen
  \bibfield  {author} {\bibinfo {author} {\bibfnamefont {H.~H.~S.}\ \bibnamefont {Chan}}, \bibinfo {author} {\bibfnamefont {R.}~\bibnamefont {Meister}}, \bibinfo {author} {\bibfnamefont {M.~L.}\ \bibnamefont {Goh}},\ and\ \bibinfo {author} {\bibfnamefont {B.}~\bibnamefont {Koczor}},\ }\bibfield  {title} {\bibinfo {title} {Algorithmic shadow spectroscopy},\ }\href {https://link.aps.org/doi/10.1103/PRXQuantum.6.010352} {\bibfield  {journal} {\bibinfo  {journal} {PRX Quantum}\ }\textbf {\bibinfo {volume} {6}},\ \bibinfo {pages} {010352} (\bibinfo {year} {2025})}\BibitemShut {NoStop}%
\bibitem [{\citenamefont {Cotler}\ \emph {et~al.}(2019)\citenamefont {Cotler}, \citenamefont {Choi}, \citenamefont {Lukin}, \citenamefont {Gharibyan}, \citenamefont {Grover}, \citenamefont {Tai}, \citenamefont {Rispoli}, \citenamefont {Schittko}, \citenamefont {Preiss}, \citenamefont {Kaufman}, \citenamefont {Greiner}, \citenamefont {Pichler},\ and\ \citenamefont {Hayden}}]{cotler2019quantum}%
  \BibitemOpen
  \bibfield  {author} {\bibinfo {author} {\bibfnamefont {J.}~\bibnamefont {Cotler}}, \bibinfo {author} {\bibfnamefont {S.}~\bibnamefont {Choi}}, \bibinfo {author} {\bibfnamefont {A.}~\bibnamefont {Lukin}}, \bibinfo {author} {\bibfnamefont {H.}~\bibnamefont {Gharibyan}}, \bibinfo {author} {\bibfnamefont {T.}~\bibnamefont {Grover}}, \bibinfo {author} {\bibfnamefont {M.~E.}\ \bibnamefont {Tai}}, \bibinfo {author} {\bibfnamefont {M.}~\bibnamefont {Rispoli}}, \bibinfo {author} {\bibfnamefont {R.}~\bibnamefont {Schittko}}, \bibinfo {author} {\bibfnamefont {P.~M.}\ \bibnamefont {Preiss}}, \bibinfo {author} {\bibfnamefont {A.~M.}\ \bibnamefont {Kaufman}}, \bibinfo {author} {\bibfnamefont {M.}~\bibnamefont {Greiner}}, \bibinfo {author} {\bibfnamefont {H.}~\bibnamefont {Pichler}},\ and\ \bibinfo {author} {\bibfnamefont {P.}~\bibnamefont {Hayden}},\ }\bibfield  {title} {\bibinfo {title} {Quantum virtual cooling},\ }\href {https://doi.org/10.1103/PhysRevX.9.031013} {\bibfield  {journal} {\bibinfo  {journal} {Phys. Rev. X}\ }\textbf {\bibinfo {volume} {9}},\ \bibinfo {pages} {031013} (\bibinfo {year} {2019})}\BibitemShut {NoStop}%
\bibitem [{\citenamefont {Huggins}\ \emph {et~al.}(2021)\citenamefont {Huggins}, \citenamefont {McArdle}, \citenamefont {O'Brien}, \citenamefont {Lee}, \citenamefont {Rubin}, \citenamefont {Boixo}, \citenamefont {Whaley}, \citenamefont {Babbush},\ and\ \citenamefont {McClean}}]{huggins2021virtual}%
  \BibitemOpen
  \bibfield  {author} {\bibinfo {author} {\bibfnamefont {W.~J.}\ \bibnamefont {Huggins}}, \bibinfo {author} {\bibfnamefont {S.}~\bibnamefont {McArdle}}, \bibinfo {author} {\bibfnamefont {T.~E.}\ \bibnamefont {O'Brien}}, \bibinfo {author} {\bibfnamefont {J.}~\bibnamefont {Lee}}, \bibinfo {author} {\bibfnamefont {N.~C.}\ \bibnamefont {Rubin}}, \bibinfo {author} {\bibfnamefont {S.}~\bibnamefont {Boixo}}, \bibinfo {author} {\bibfnamefont {K.~B.}\ \bibnamefont {Whaley}}, \bibinfo {author} {\bibfnamefont {R.}~\bibnamefont {Babbush}},\ and\ \bibinfo {author} {\bibfnamefont {J.~R.}\ \bibnamefont {McClean}},\ }\bibfield  {title} {\bibinfo {title} {Virtual distillation for quantum error mitigation},\ }\href {https://doi.org/10.1103/PhysRevX.11.041036} {\bibfield  {journal} {\bibinfo  {journal} {Phys. Rev. X}\ }\textbf {\bibinfo {volume} {11}},\ \bibinfo {pages} {041036} (\bibinfo {year} {2021})}\BibitemShut {NoStop}%
\bibitem [{\citenamefont {Koczor}(2021)}]{koczor2021exponential}%
  \BibitemOpen
  \bibfield  {author} {\bibinfo {author} {\bibfnamefont {B.}~\bibnamefont {Koczor}},\ }\bibfield  {title} {\bibinfo {title} {Exponential error suppression for near-term quantum devices},\ }\href {https://doi.org/10.1103/PhysRevX.11.031057} {\bibfield  {journal} {\bibinfo  {journal} {Phys. Rev. X}\ }\textbf {\bibinfo {volume} {11}},\ \bibinfo {pages} {031057} (\bibinfo {year} {2021})}\BibitemShut {NoStop}%
\bibitem [{\citenamefont {McArdle}\ \emph {et~al.}(2019)\citenamefont {McArdle}, \citenamefont {Jones}, \citenamefont {Endo}, \citenamefont {Li}, \citenamefont {Benjamin},\ and\ \citenamefont {Yuan}}]{mcardle2019variational}%
  \BibitemOpen
  \bibfield  {author} {\bibinfo {author} {\bibfnamefont {S.}~\bibnamefont {McArdle}}, \bibinfo {author} {\bibfnamefont {T.}~\bibnamefont {Jones}}, \bibinfo {author} {\bibfnamefont {S.}~\bibnamefont {Endo}}, \bibinfo {author} {\bibfnamefont {Y.}~\bibnamefont {Li}}, \bibinfo {author} {\bibfnamefont {S.~C.}\ \bibnamefont {Benjamin}},\ and\ \bibinfo {author} {\bibfnamefont {X.}~\bibnamefont {Yuan}},\ }\bibfield  {title} {\bibinfo {title} {Variational ansatz-based quantum simulation of imaginary time evolution},\ }\href {https://doi.org/10.1038/s41534-019-0187-2} {\bibfield  {journal} {\bibinfo  {journal} {npj Quantum Inf.}\ }\textbf {\bibinfo {volume} {5}},\ \bibinfo {pages} {75} (\bibinfo {year} {2019})}\BibitemShut {NoStop}%
\bibitem [{\citenamefont {Huo}\ and\ \citenamefont {Li}(2023)}]{huo2023error}%
  \BibitemOpen
  \bibfield  {author} {\bibinfo {author} {\bibfnamefont {M.}~\bibnamefont {Huo}}\ and\ \bibinfo {author} {\bibfnamefont {Y.}~\bibnamefont {Li}},\ }\bibfield  {title} {\bibinfo {title} {Error-resilient {Monte} {Carlo} quantum simulation of imaginary time},\ }\href {https://doi.org/10.22331/q-2023-02-09-916} {\bibfield  {journal} {\bibinfo  {journal} {Quantum}\ }\textbf {\bibinfo {volume} {7}},\ \bibinfo {pages} {916} (\bibinfo {year} {2023})}\BibitemShut {NoStop}%
\bibitem [{\citenamefont {Tsuchimochi}\ \emph {et~al.}(2023)\citenamefont {Tsuchimochi}, \citenamefont {Ryo}, \citenamefont {Ten-no},\ and\ \citenamefont {Sasasako}}]{tsuchimochi2023improved}%
  \BibitemOpen
  \bibfield  {author} {\bibinfo {author} {\bibfnamefont {T.}~\bibnamefont {Tsuchimochi}}, \bibinfo {author} {\bibfnamefont {Y.}~\bibnamefont {Ryo}}, \bibinfo {author} {\bibfnamefont {S.~L.}\ \bibnamefont {Ten-no}},\ and\ \bibinfo {author} {\bibfnamefont {K.}~\bibnamefont {Sasasako}},\ }\bibfield  {title} {\bibinfo {title} {Improved algorithms of quantum imaginary time evolution for ground and excited states of molecular systems},\ }\href {https://doi.org/10.1021/acs.jctc.2c00906} {\bibfield  {journal} {\bibinfo  {journal} {J. Chem. Theory Comput.}\ }\textbf {\bibinfo {volume} {19}},\ \bibinfo {pages} {503–513} (\bibinfo {year} {2023})}\BibitemShut {NoStop}%
\bibitem [{\citenamefont {Choi}\ \emph {et~al.}(2021)\citenamefont {Choi}, \citenamefont {Lee}, \citenamefont {Bonitati}, \citenamefont {Qian},\ and\ \citenamefont {Watkins}}]{choi2021rodeo}%
  \BibitemOpen
  \bibfield  {author} {\bibinfo {author} {\bibfnamefont {K.}~\bibnamefont {Choi}}, \bibinfo {author} {\bibfnamefont {D.}~\bibnamefont {Lee}}, \bibinfo {author} {\bibfnamefont {J.}~\bibnamefont {Bonitati}}, \bibinfo {author} {\bibfnamefont {Z.}~\bibnamefont {Qian}},\ and\ \bibinfo {author} {\bibfnamefont {J.}~\bibnamefont {Watkins}},\ }\bibfield  {title} {\bibinfo {title} {Rodeo algorithm for quantum computing},\ }\href {https://doi.org/10.1103/PhysRevLett.127.040505} {\bibfield  {journal} {\bibinfo  {journal} {Phys. Rev. Lett.}\ }\textbf {\bibinfo {volume} {127}},\ \bibinfo {pages} {040505} (\bibinfo {year} {2021})}\BibitemShut {NoStop}%
\bibitem [{\citenamefont {Kyriienko}(2020)}]{kyriienko2020quantum}%
  \BibitemOpen
  \bibfield  {author} {\bibinfo {author} {\bibfnamefont {O.}~\bibnamefont {Kyriienko}},\ }\bibfield  {title} {\bibinfo {title} {Quantum inverse iteration algorithm for programmable quantum simulators},\ }\href {https://doi.org/10.1038/s41534-019-0239-7} {\bibfield  {journal} {\bibinfo  {journal} {npj Quantum Inf.}\ }\textbf {\bibinfo {volume} {6}},\ \bibinfo {pages} {7} (\bibinfo {year} {2020})}\BibitemShut {NoStop}%
\bibitem [{\citenamefont {Yoshikura}\ \emph {et~al.}(2023)\citenamefont {Yoshikura}, \citenamefont {Ten-no},\ and\ \citenamefont {Tsuchimochi}}]{yoshikura2023quantum}%
  \BibitemOpen
  \bibfield  {author} {\bibinfo {author} {\bibfnamefont {T.}~\bibnamefont {Yoshikura}}, \bibinfo {author} {\bibfnamefont {S.~L.}\ \bibnamefont {Ten-no}},\ and\ \bibinfo {author} {\bibfnamefont {T.}~\bibnamefont {Tsuchimochi}},\ }\bibfield  {title} {\bibinfo {title} {Quantum inverse algorithm via adaptive variational quantum linear solver: Applications to general eigenstates},\ }\href {https://doi.org/10.1021/acs.jpca.3c02800} {\bibfield  {journal} {\bibinfo  {journal} {J. Phys. Chem. A}\ }\textbf {\bibinfo {volume} {127}},\ \bibinfo {pages} {6577–6592} (\bibinfo {year} {2023})}\BibitemShut {NoStop}%
\bibitem [{\citenamefont {Millar}\ \emph {et~al.}(2026)\citenamefont {Millar}, \citenamefont {Anderson}, \citenamefont {Altamura}, \citenamefont {Wallis}, \citenamefont {Sahin}, \citenamefont {Crain},\ and\ \citenamefont {Thomson}}]{millar2026imaginary}%
  \BibitemOpen
  \bibfield  {author} {\bibinfo {author} {\bibfnamefont {D.~A.}\ \bibnamefont {Millar}}, \bibinfo {author} {\bibfnamefont {L.~W.}\ \bibnamefont {Anderson}}, \bibinfo {author} {\bibfnamefont {E.}~\bibnamefont {Altamura}}, \bibinfo {author} {\bibfnamefont {O.}~\bibnamefont {Wallis}}, \bibinfo {author} {\bibfnamefont {M.~E.}\ \bibnamefont {Sahin}}, \bibinfo {author} {\bibfnamefont {J.}~\bibnamefont {Crain}},\ and\ \bibinfo {author} {\bibfnamefont {S.~J.}\ \bibnamefont {Thomson}},\ }\bibfield  {title} {\bibinfo {title} {Imaginary time spectral transforms for excited-state preparation},\ }\href {https://doi.org/10.1103/v4rc-gmcx} {\bibfield  {journal} {\bibinfo  {journal} {Phys. Rev. Res.}\ }\textbf {\bibinfo {volume} {8}},\ \bibinfo {pages} {L022042} (\bibinfo {year} {2026})}\BibitemShut {NoStop}%
\bibitem [{\citenamefont {Khinevich}\ \emph {et~al.}(2025)\citenamefont {Khinevich}, \citenamefont {Lee}, \citenamefont {Yoshioka},\ and\ \citenamefont {Mizukami}}]{khinevich2025quantum}%
  \BibitemOpen
  \bibfield  {author} {\bibinfo {author} {\bibfnamefont {V.}~\bibnamefont {Khinevich}}, \bibinfo {author} {\bibfnamefont {Y.}~\bibnamefont {Lee}}, \bibinfo {author} {\bibfnamefont {N.}~\bibnamefont {Yoshioka}},\ and\ \bibinfo {author} {\bibfnamefont {W.}~\bibnamefont {Mizukami}},\ }\href@noop {} {\bibinfo {title} {Quantum power iteration unified using generalized quantum signal processing}} (\bibinfo {year} {2025}),\ \Eprint {https://arxiv.org/abs/2507.11142} {arXiv:2507.11142 [quant-ph]} \BibitemShut {NoStop}%
\bibitem [{\citenamefont {Li}\ and\ \citenamefont {Lin}(2025)}]{li2025dissipative}%
  \BibitemOpen
  \bibfield  {author} {\bibinfo {author} {\bibfnamefont {H.-E.}\ \bibnamefont {Li}}\ and\ \bibinfo {author} {\bibfnamefont {L.}~\bibnamefont {Lin}},\ }\href@noop {} {\bibinfo {title} {Dissipative quantum algorithms for excited-state quantum chemistry}} (\bibinfo {year} {2025}),\ \Eprint {https://arxiv.org/abs/2512.19870} {arXiv:2512.19870 [quant-ph]} \BibitemShut {NoStop}%
\bibitem [{\citenamefont {Childs}\ \emph {et~al.}(2025)\citenamefont {Childs}, \citenamefont {Fu}, \citenamefont {Leung}, \citenamefont {Li}, \citenamefont {Ozols},\ and\ \citenamefont {Vyas}}]{childs2025streaming}%
  \BibitemOpen
  \bibfield  {author} {\bibinfo {author} {\bibfnamefont {A.~M.}\ \bibnamefont {Childs}}, \bibinfo {author} {\bibfnamefont {H.}~\bibnamefont {Fu}}, \bibinfo {author} {\bibfnamefont {D.}~\bibnamefont {Leung}}, \bibinfo {author} {\bibfnamefont {Z.}~\bibnamefont {Li}}, \bibinfo {author} {\bibfnamefont {M.}~\bibnamefont {Ozols}},\ and\ \bibinfo {author} {\bibfnamefont {V.}~\bibnamefont {Vyas}},\ }\bibfield  {title} {\bibinfo {title} {Streaming quantum state purification},\ }\href {https://doi.org/10.22331/q-2025-01-21-1603} {\bibfield  {journal} {\bibinfo  {journal} {{Quantum}}\ }\textbf {\bibinfo {volume} {9}},\ \bibinfo {pages} {1603} (\bibinfo {year} {2025})}\BibitemShut {NoStop}%
\bibitem [{\citenamefont {Dalzell}\ \emph {et~al.}(2025)\citenamefont {Dalzell}, \citenamefont {Gilyén}, \citenamefont {Hann}, \citenamefont {McArdle}, \citenamefont {Salton}, \citenamefont {Nguyen}, \citenamefont {Kubica},\ and\ \citenamefont {Brandão}}]{dalzell2025distillation}%
  \BibitemOpen
  \bibfield  {author} {\bibinfo {author} {\bibfnamefont {A.~M.}\ \bibnamefont {Dalzell}}, \bibinfo {author} {\bibfnamefont {A.}~\bibnamefont {Gilyén}}, \bibinfo {author} {\bibfnamefont {C.~T.}\ \bibnamefont {Hann}}, \bibinfo {author} {\bibfnamefont {S.}~\bibnamefont {McArdle}}, \bibinfo {author} {\bibfnamefont {G.}~\bibnamefont {Salton}}, \bibinfo {author} {\bibfnamefont {Q.~T.}\ \bibnamefont {Nguyen}}, \bibinfo {author} {\bibfnamefont {A.}~\bibnamefont {Kubica}},\ and\ \bibinfo {author} {\bibfnamefont {F.~G.}\ \bibnamefont {Brandão}},\ }\bibfield  {title} {\bibinfo {title} {A distillation-teleportation protocol for fault-tolerant qram},\ }in\ \href {https://doi.org/10.1109/FOCS63196.2025.00008} {\emph {\bibinfo {booktitle} {2025 IEEE 66th Annual Symposium on Foundations of Computer Science (FOCS)}}}\ (\bibinfo {year} {2025})\ pp.\ \bibinfo {pages} {38--74}\BibitemShut {NoStop}%
\bibitem [{\citenamefont {Li}\ \emph {et~al.}(2021)\citenamefont {Li}, \citenamefont {Chai}, \citenamefont {Guo}, \citenamefont {Ji}, \citenamefont {Wang}, \citenamefont {Shi}, \citenamefont {Wang}, \citenamefont {Lloyd},\ and\ \citenamefont {Du}}]{li2021resonant}%
  \BibitemOpen
  \bibfield  {author} {\bibinfo {author} {\bibfnamefont {Z.}~\bibnamefont {Li}}, \bibinfo {author} {\bibfnamefont {Z.}~\bibnamefont {Chai}}, \bibinfo {author} {\bibfnamefont {Y.}~\bibnamefont {Guo}}, \bibinfo {author} {\bibfnamefont {W.}~\bibnamefont {Ji}}, \bibinfo {author} {\bibfnamefont {M.}~\bibnamefont {Wang}}, \bibinfo {author} {\bibfnamefont {F.}~\bibnamefont {Shi}}, \bibinfo {author} {\bibfnamefont {Y.}~\bibnamefont {Wang}}, \bibinfo {author} {\bibfnamefont {S.}~\bibnamefont {Lloyd}},\ and\ \bibinfo {author} {\bibfnamefont {J.}~\bibnamefont {Du}},\ }\bibfield  {title} {\bibinfo {title} {Resonant quantum principal component analysis},\ }\href {https://doi.org/10.1126/sciadv.abg2589} {\bibfield  {journal} {\bibinfo  {journal} {Sci. Adv.}\ }\textbf {\bibinfo {volume} {7}},\ \bibinfo {pages} {eabg2589} (\bibinfo {year} {2021})}\BibitemShut {NoStop}%
\bibitem [{\citenamefont {Motwani}\ and\ \citenamefont {Raghavan}(1995)}]{motwani1995randomized}%
  \BibitemOpen
  \bibfield  {author} {\bibinfo {author} {\bibfnamefont {R.}~\bibnamefont {Motwani}}\ and\ \bibinfo {author} {\bibfnamefont {P.}~\bibnamefont {Raghavan}},\ }\href {https://doi.org/10.1017/cbo9780511814075} {\emph {\bibinfo {title} {Randomized Algorithms}}}\ (\bibinfo  {publisher} {Cambridge University Press},\ \bibinfo {year} {1995})\BibitemShut {NoStop}%
\bibitem [{\citenamefont {Mitzenmacher}\ and\ \citenamefont {Upfal}(2005)}]{mitzenmacher2005probability}%
  \BibitemOpen
  \bibfield  {author} {\bibinfo {author} {\bibfnamefont {M.}~\bibnamefont {Mitzenmacher}}\ and\ \bibinfo {author} {\bibfnamefont {E.}~\bibnamefont {Upfal}},\ }\href {https://doi.org/10.1017/cbo9780511813603} {\emph {\bibinfo {title} {Probability and Computing: Randomized Algorithms and Probabilistic Analysis}}}\ (\bibinfo  {publisher} {Cambridge University Press},\ \bibinfo {year} {2005})\BibitemShut {NoStop}%
\bibitem [{\citenamefont {Luis}\ and\ \citenamefont {Pe\v{r}ina}(1996)}]{luis1996optimum}%
  \BibitemOpen
  \bibfield  {author} {\bibinfo {author} {\bibfnamefont {A.}~\bibnamefont {Luis}}\ and\ \bibinfo {author} {\bibfnamefont {J.}~\bibnamefont {Pe\v{r}ina}},\ }\bibfield  {title} {\bibinfo {title} {Optimum phase-shift estimation and the quantum description of the phase difference},\ }\href {https://doi.org/10.1103/PhysRevA.54.4564} {\bibfield  {journal} {\bibinfo  {journal} {Phys. Rev. A}\ }\textbf {\bibinfo {volume} {54}},\ \bibinfo {pages} {4564} (\bibinfo {year} {1996})}\BibitemShut {NoStop}%
\bibitem [{\citenamefont {Bu\v{z}ek}\ \emph {et~al.}(1999)\citenamefont {Bu\v{z}ek}, \citenamefont {Derka},\ and\ \citenamefont {Massar}}]{buzek1999optimal}%
  \BibitemOpen
  \bibfield  {author} {\bibinfo {author} {\bibfnamefont {V.}~\bibnamefont {Bu\v{z}ek}}, \bibinfo {author} {\bibfnamefont {R.}~\bibnamefont {Derka}},\ and\ \bibinfo {author} {\bibfnamefont {S.}~\bibnamefont {Massar}},\ }\bibfield  {title} {\bibinfo {title} {Optimal quantum clocks},\ }\href {https://doi.org/10.1103/PhysRevLett.82.2207} {\bibfield  {journal} {\bibinfo  {journal} {Phys. Rev. Lett.}\ }\textbf {\bibinfo {volume} {82}},\ \bibinfo {pages} {2207} (\bibinfo {year} {1999})}\BibitemShut {NoStop}%
\bibitem [{\citenamefont {van Dam}\ \emph {et~al.}(2007)\citenamefont {van Dam}, \citenamefont {D'Ariano}, \citenamefont {Ekert}, \citenamefont {Macchiavello},\ and\ \citenamefont {Mosca}}]{dam2007optimal}%
  \BibitemOpen
  \bibfield  {author} {\bibinfo {author} {\bibfnamefont {W.}~\bibnamefont {van Dam}}, \bibinfo {author} {\bibfnamefont {G.~M.}\ \bibnamefont {D'Ariano}}, \bibinfo {author} {\bibfnamefont {A.}~\bibnamefont {Ekert}}, \bibinfo {author} {\bibfnamefont {C.}~\bibnamefont {Macchiavello}},\ and\ \bibinfo {author} {\bibfnamefont {M.}~\bibnamefont {Mosca}},\ }\bibfield  {title} {\bibinfo {title} {Optimal quantum circuits for general phase estimation},\ }\href {https://doi.org/10.1103/PhysRevLett.98.090501} {\bibfield  {journal} {\bibinfo  {journal} {Phys. Rev. Lett.}\ }\textbf {\bibinfo {volume} {98}},\ \bibinfo {pages} {090501} (\bibinfo {year} {2007})}\BibitemShut {NoStop}%
\bibitem [{\citenamefont {Rendon}\ \emph {et~al.}(2022)\citenamefont {Rendon}, \citenamefont {Izubuchi},\ and\ \citenamefont {Kikuchi}}]{rendon2022effects}%
  \BibitemOpen
  \bibfield  {author} {\bibinfo {author} {\bibfnamefont {G.}~\bibnamefont {Rendon}}, \bibinfo {author} {\bibfnamefont {T.}~\bibnamefont {Izubuchi}},\ and\ \bibinfo {author} {\bibfnamefont {Y.}~\bibnamefont {Kikuchi}},\ }\bibfield  {title} {\bibinfo {title} {Effects of cosine tapering window on quantum phase estimation},\ }\href {https://doi.org/10.1103/PhysRevD.106.034503} {\bibfield  {journal} {\bibinfo  {journal} {Phys. Rev. D}\ }\textbf {\bibinfo {volume} {106}},\ \bibinfo {pages} {034503} (\bibinfo {year} {2022})}\BibitemShut {NoStop}%
\bibitem [{\citenamefont {Berry}\ \emph {et~al.}(2024)\citenamefont {Berry}, \citenamefont {Su}, \citenamefont {Gyurik}, \citenamefont {King}, \citenamefont {Basso}, \citenamefont {Barba}, \citenamefont {Rajput}, \citenamefont {Wiebe}, \citenamefont {Dunjko},\ and\ \citenamefont {Babbush}}]{berry2024analyzing}%
  \BibitemOpen
  \bibfield  {author} {\bibinfo {author} {\bibfnamefont {D.~W.}\ \bibnamefont {Berry}}, \bibinfo {author} {\bibfnamefont {Y.}~\bibnamefont {Su}}, \bibinfo {author} {\bibfnamefont {C.}~\bibnamefont {Gyurik}}, \bibinfo {author} {\bibfnamefont {R.}~\bibnamefont {King}}, \bibinfo {author} {\bibfnamefont {J.}~\bibnamefont {Basso}}, \bibinfo {author} {\bibfnamefont {A.~D.~T.}\ \bibnamefont {Barba}}, \bibinfo {author} {\bibfnamefont {A.}~\bibnamefont {Rajput}}, \bibinfo {author} {\bibfnamefont {N.}~\bibnamefont {Wiebe}}, \bibinfo {author} {\bibfnamefont {V.}~\bibnamefont {Dunjko}},\ and\ \bibinfo {author} {\bibfnamefont {R.}~\bibnamefont {Babbush}},\ }\bibfield  {title} {\bibinfo {title} {Analyzing prospects for quantum advantage in topological data analysis},\ }\href {https://doi.org/10.1103/PRXQuantum.5.010319} {\bibfield  {journal} {\bibinfo  {journal} {PRX Quantum}\ }\textbf {\bibinfo {volume} {5}},\ \bibinfo {pages} {010319} (\bibinfo {year} {2024})}\BibitemShut {NoStop}%
\bibitem [{\citenamefont {D\"urr}\ and\ \citenamefont {H\o{}yer}(1996)}]{durr1996quantum}%
  \BibitemOpen
  \bibfield  {author} {\bibinfo {author} {\bibfnamefont {C.}~\bibnamefont {D\"urr}}\ and\ \bibinfo {author} {\bibfnamefont {P.}~\bibnamefont {H\o{}yer}},\ }\href@noop {} {\bibinfo {title} {A quantum algorithm for finding the minimum}} (\bibinfo {year} {1996}),\ \Eprint {https://arxiv.org/abs/quant-ph/9607014} {arXiv:quant-ph/9607014 [quant-ph]} \BibitemShut {NoStop}%
\bibitem [{\citenamefont {Pati}\ \emph {et~al.}(1993)\citenamefont {Pati}, \citenamefont {Rezaiifar},\ and\ \citenamefont {Krishnaprasad}}]{pati1993orthogonal}%
  \BibitemOpen
  \bibfield  {author} {\bibinfo {author} {\bibfnamefont {Y.}~\bibnamefont {Pati}}, \bibinfo {author} {\bibfnamefont {R.}~\bibnamefont {Rezaiifar}},\ and\ \bibinfo {author} {\bibfnamefont {P.}~\bibnamefont {Krishnaprasad}},\ }\bibfield  {title} {\bibinfo {title} {Orthogonal matching pursuit: recursive function approximation with applications to wavelet decomposition},\ }in\ \href {https://doi.org/10.1109/ACSSC.1993.342465} {\emph {\bibinfo {booktitle} {Proceedings of 27th Asilomar Conference on Signals, Systems and Computers}}}\ (\bibinfo {year} {1993})\ pp.\ \bibinfo {pages} {40--44 vol.1}\BibitemShut {NoStop}%
\bibitem [{\citenamefont {Grover}(1996)}]{grover1996fast}%
  \BibitemOpen
  \bibfield  {author} {\bibinfo {author} {\bibfnamefont {L.~K.}\ \bibnamefont {Grover}},\ }\bibfield  {title} {\bibinfo {title} {A fast quantum mechanical algorithm for database search},\ }in\ \href {https://doi.org/10.1145/237814.237866} {\emph {\bibinfo {booktitle} {Proceedings of the Twenty-Eighth Annual ACM Symposium on Theory of Computing}}},\ \bibinfo {series and number} {STOC '96}\ (\bibinfo  {publisher} {Association for Computing Machinery},\ \bibinfo {address} {New York, NY, USA},\ \bibinfo {year} {1996})\ p.\ \bibinfo {pages} {212–219}\BibitemShut {NoStop}%
\bibitem [{\citenamefont {Gordon}(1941)}]{gordon1941values}%
  \BibitemOpen
  \bibfield  {author} {\bibinfo {author} {\bibfnamefont {R.~D.}\ \bibnamefont {Gordon}},\ }\bibfield  {title} {\bibinfo {title} {Values of {Mills'} ratio of area to bounding ordinate and of the normal probability integral for large values of the argument},\ }\href {http://www.jstor.org/stable/2235868} {\bibfield  {journal} {\bibinfo  {journal} {Ann. Math. Stat.}\ }\textbf {\bibinfo {volume} {12}},\ \bibinfo {pages} {364} (\bibinfo {year} {1941})}\BibitemShut {NoStop}%
\bibitem [{\citenamefont {Rigollet}\ and\ \citenamefont {H\"{u}tter}(2023)}]{rigollet2023high}%
  \BibitemOpen
  \bibfield  {author} {\bibinfo {author} {\bibfnamefont {P.}~\bibnamefont {Rigollet}}\ and\ \bibinfo {author} {\bibfnamefont {J.-C.}\ \bibnamefont {H\"{u}tter}},\ }\href@noop {} {\bibinfo {title} {High-dimensional statistics}} (\bibinfo {year} {2023}),\ \Eprint {https://arxiv.org/abs/2310.19244} {arXiv:2310.19244 [math.ST]} \BibitemShut {NoStop}%
\bibitem [{\citenamefont {Grover}\ and\ \citenamefont {Rudolph}(2002)}]{grover2002creating}%
  \BibitemOpen
  \bibfield  {author} {\bibinfo {author} {\bibfnamefont {L.}~\bibnamefont {Grover}}\ and\ \bibinfo {author} {\bibfnamefont {T.}~\bibnamefont {Rudolph}},\ }\href@noop {} {\bibinfo {title} {Creating superpositions that correspond to efficiently integrable probability distributions}} (\bibinfo {year} {2002}),\ \Eprint {https://arxiv.org/abs/quant-ph/0208112} {arXiv:quant-ph/0208112 [quant-ph]} \BibitemShut {NoStop}%
\bibitem [{\citenamefont {Horn}\ and\ \citenamefont {Johnson}(2012)}]{horn2012matrix}%
  \BibitemOpen
  \bibfield  {author} {\bibinfo {author} {\bibfnamefont {R.~A.}\ \bibnamefont {Horn}}\ and\ \bibinfo {author} {\bibfnamefont {C.~R.}\ \bibnamefont {Johnson}},\ }\href {https://doi.org/10.1017/cbo9780511810817} {\emph {\bibinfo {title} {Matrix Analysis}}},\ \bibinfo {edition} {2nd}\ ed.\ (\bibinfo  {publisher} {Cambridge University Press},\ \bibinfo {year} {2012})\BibitemShut {NoStop}%
\bibitem [{\citenamefont {Wedin}(1972)}]{wedin1972perturbation}%
  \BibitemOpen
  \bibfield  {author} {\bibinfo {author} {\bibfnamefont {P.-A.}\ \bibnamefont {Wedin}},\ }\bibfield  {title} {\bibinfo {title} {Perturbation bounds in connection with singular value decomposition},\ }\href {https://doi.org/10.1007/bf01932678} {\bibfield  {journal} {\bibinfo  {journal} {BIT Numer. Math.}\ }\textbf {\bibinfo {volume} {12}},\ \bibinfo {pages} {99–111} (\bibinfo {year} {1972})}\BibitemShut {NoStop}%
\bibitem [{\citenamefont {Rall}(2020)}]{rall2020quantum}%
  \BibitemOpen
  \bibfield  {author} {\bibinfo {author} {\bibfnamefont {P.}~\bibnamefont {Rall}},\ }\bibfield  {title} {\bibinfo {title} {Quantum algorithms for estimating physical quantities using block encodings},\ }\href {https://doi.org/10.1103/PhysRevA.102.022408} {\bibfield  {journal} {\bibinfo  {journal} {Phys. Rev. A}\ }\textbf {\bibinfo {volume} {102}},\ \bibinfo {pages} {022408} (\bibinfo {year} {2020})}\BibitemShut {NoStop}%
\bibitem [{\citenamefont {Martyn}\ \emph {et~al.}(2021)\citenamefont {Martyn}, \citenamefont {Rossi}, \citenamefont {Tan},\ and\ \citenamefont {Chuang}}]{martyn2021grand}%
  \BibitemOpen
  \bibfield  {author} {\bibinfo {author} {\bibfnamefont {J.~M.}\ \bibnamefont {Martyn}}, \bibinfo {author} {\bibfnamefont {Z.~M.}\ \bibnamefont {Rossi}}, \bibinfo {author} {\bibfnamefont {A.~K.}\ \bibnamefont {Tan}},\ and\ \bibinfo {author} {\bibfnamefont {I.~L.}\ \bibnamefont {Chuang}},\ }\bibfield  {title} {\bibinfo {title} {Grand unification of quantum algorithms},\ }\href {https://doi.org/10.1103/PRXQuantum.2.040203} {\bibfield  {journal} {\bibinfo  {journal} {PRX Quantum}\ }\textbf {\bibinfo {volume} {2}},\ \bibinfo {pages} {040203} (\bibinfo {year} {2021})}\BibitemShut {NoStop}%
\bibitem [{\citenamefont {Huang}\ and\ \citenamefont {Koczor}(2026)}]{huang2026low}%
  \BibitemOpen
  \bibfield  {author} {\bibinfo {author} {\bibfnamefont {P.-W.}\ \bibnamefont {Huang}}\ and\ \bibinfo {author} {\bibfnamefont {B.}~\bibnamefont {Koczor}},\ }\href@noop {} {\bibinfo {title} {Low-depth amplitude estimation via statistical eigengap estimation}} (\bibinfo {year} {2026}),\ \Eprint {https://arxiv.org/abs/2603.05475} {arXiv:2603.05475 [quant-ph]} \BibitemShut {NoStop}%
\bibitem [{\citenamefont {Camps}\ and\ \citenamefont {Van~Beeumen}(2020)}]{camps2020approximate}%
  \BibitemOpen
  \bibfield  {author} {\bibinfo {author} {\bibfnamefont {D.}~\bibnamefont {Camps}}\ and\ \bibinfo {author} {\bibfnamefont {R.}~\bibnamefont {Van~Beeumen}},\ }\bibfield  {title} {\bibinfo {title} {Approximate quantum circuit synthesis using block encodings},\ }\href {https://doi.org/10.1103/PhysRevA.102.052411} {\bibfield  {journal} {\bibinfo  {journal} {Phys. Rev. A}\ }\textbf {\bibinfo {volume} {102}},\ \bibinfo {pages} {052411} (\bibinfo {year} {2020})}\BibitemShut {NoStop}%
\bibitem [{\citenamefont {Rall}\ and\ \citenamefont {Fuller}(2023)}]{rall2023amplitude}%
  \BibitemOpen
  \bibfield  {author} {\bibinfo {author} {\bibfnamefont {P.}~\bibnamefont {Rall}}\ and\ \bibinfo {author} {\bibfnamefont {B.}~\bibnamefont {Fuller}},\ }\bibfield  {title} {\bibinfo {title} {Amplitude estimation from quantum signal processing},\ }\href {https://doi.org/10.22331/q-2023-03-02-937} {\bibfield  {journal} {\bibinfo  {journal} {Quantum}\ }\textbf {\bibinfo {volume} {7}},\ \bibinfo {pages} {937} (\bibinfo {year} {2023})}\BibitemShut {NoStop}%
\bibitem [{\citenamefont {Sun}\ \emph {et~al.}(2026)\citenamefont {Sun}, \citenamefont {Zhou}, \citenamefont {Xu}, \citenamefont {Yao}, \citenamefont {Du}, \citenamefont {Zhang}, \citenamefont {Gu}, \citenamefont {Huang}, \citenamefont {Zhou}, \citenamefont {Wang}, \citenamefont {Yosifov}, \citenamefont {Dong}, \citenamefont {Huang}, \citenamefont {Serrano}, \citenamefont {Wang}, \citenamefont {Feng}, \citenamefont {Sadugol}, \citenamefont {Yu}, \citenamefont {You}, \citenamefont {Qin}, \citenamefont {Zhang}, \citenamefont {Wu}, \citenamefont {Iyer}, \citenamefont {Zhou}, \citenamefont {Li}, \citenamefont {Li}, \citenamefont {Ma}, \citenamefont {Zhao}, \citenamefont {Zeng}, \citenamefont {Zhang},\ and\ \citenamefont {Yuan}}]{sun2026quantum}%
  \BibitemOpen
  \bibfield  {author} {\bibinfo {author} {\bibfnamefont {J.}~\bibnamefont {Sun}}, \bibinfo {author} {\bibfnamefont {B.}~\bibnamefont {Zhou}}, \bibinfo {author} {\bibfnamefont {J.}~\bibnamefont {Xu}}, \bibinfo {author} {\bibfnamefont {Y.}~\bibnamefont {Yao}}, \bibinfo {author} {\bibfnamefont {Z.}~\bibnamefont {Du}}, \bibinfo {author} {\bibfnamefont {Z.}~\bibnamefont {Zhang}}, \bibinfo {author} {\bibfnamefont {Y.}~\bibnamefont {Gu}}, \bibinfo {author} {\bibfnamefont {J.}~\bibnamefont {Huang}}, \bibinfo {author} {\bibfnamefont {S.}~\bibnamefont {Zhou}}, \bibinfo {author} {\bibfnamefont {Z.}~\bibnamefont {Wang}}, \bibinfo {author} {\bibfnamefont {A.}~\bibnamefont {Yosifov}}, \bibinfo {author} {\bibfnamefont {W.}~\bibnamefont {Dong}}, \bibinfo {author} {\bibfnamefont {Y.}~\bibnamefont {Huang}}, \bibinfo {author} {\bibfnamefont {D.}~\bibnamefont {Serrano}}, \bibinfo {author} {\bibfnamefont {X.}~\bibnamefont {Wang}}, \bibinfo {author} {\bibfnamefont {T.}~\bibnamefont {Feng}}, \bibinfo {author} {\bibfnamefont {S.}~\bibnamefont {Sadugol}}, \bibinfo {author} {\bibfnamefont {W.}~\bibnamefont {Yu}}, \bibinfo {author} {\bibfnamefont {Z.}~\bibnamefont {You}}, \bibinfo {author} {\bibfnamefont {D.}~\bibnamefont {Qin}}, \bibinfo {author} {\bibfnamefont {X.-M.}\ \bibnamefont {Zhang}}, \bibinfo {author} {\bibfnamefont {Y.}~\bibnamefont {Wu}}, \bibinfo {author} {\bibfnamefont {A.}~\bibnamefont {Iyer}}, \bibinfo {author} {\bibfnamefont {Y.}~\bibnamefont {Zhou}}, \bibinfo {author} {\bibfnamefont {T.}~\bibnamefont {Li}}, \bibinfo {author} {\bibfnamefont {Y.}~\bibnamefont {Li}}, \bibinfo {author} {\bibfnamefont {X.}~\bibnamefont {Ma}}, \bibinfo {author} {\bibfnamefont {Q.}~\bibnamefont {Zhao}}, \bibinfo {author} {\bibfnamefont {P.}~\bibnamefont {Zeng}}, \bibinfo {author} {\bibfnamefont {P.}~\bibnamefont {Zhang}},\ and\ \bibinfo {author} {\bibfnamefont {X.}~\bibnamefont {Yuan}},\ }\href@noop {} {\bibinfo {title} {Quantum-classical crossover in fault-tolerant quantum dynamics simulation}} (\bibinfo {year} {2026}),\ \Eprint {https://arxiv.org/abs/2607.16116} {arXiv:2607.16116 [quant-ph]} \BibitemShut {NoStop}%
\bibitem [{\citenamefont {Grinko}\ \emph {et~al.}(2021)\citenamefont {Grinko}, \citenamefont {Gacon}, \citenamefont {Zoufal},\ and\ \citenamefont {Woerner}}]{grinko2021iterative}%
  \BibitemOpen
  \bibfield  {author} {\bibinfo {author} {\bibfnamefont {D.}~\bibnamefont {Grinko}}, \bibinfo {author} {\bibfnamefont {J.}~\bibnamefont {Gacon}}, \bibinfo {author} {\bibfnamefont {C.}~\bibnamefont {Zoufal}},\ and\ \bibinfo {author} {\bibfnamefont {S.}~\bibnamefont {Woerner}},\ }\bibfield  {title} {\bibinfo {title} {Iterative quantum amplitude estimation},\ }\href {https://doi.org/10.1038/s41534-021-00379-1} {\bibfield  {journal} {\bibinfo  {journal} {npj Quantum Inf.}\ }\textbf {\bibinfo {volume} {7}},\ \bibinfo {pages} {52} (\bibinfo {year} {2021})}\BibitemShut {NoStop}%
\end{thebibliography}%

\clearpage
\appendix

\crefalias{section}{appendix}
\crefalias{subsection}{appendix}
\counterwithin{theorem}{section}
\counterwithin{proposition}{section}
\counterwithin{lemma}{section}
\counterwithin{corollary}{theorem}
\counterwithin{assumption}{section}
\counterwithin{equation}{section}
\counterwithin{definition}{section}

\onecolumngrid
\begin{center}
    \noindent{\large\bfseries Appendices for ``\textsl{Eigenstate Preparation Through Near-Optimal Eigenprobability Filtering}''}
\end{center}
\appendixtableofcontents
\vspace{2em}
\clearpage
\twocolumngrid

\section{Related work on excited state problems}
\label{appLitRev}
\subsection{Classical algorithms in quantum chemistry}
\label{appLitRevClass}

A broad range of classical methods for molecular excited states has been developed, offering different trade-offs between accuracy and computational cost. Approaches suited primarily to states with single-reference character include configuration-interaction singles, time-dependent Hartree–Fock, time-dependent density-functional theory (TDDFT), and equation-of-motion coupled-cluster (EOM-CC) methods \citep{dreuw2005single,casida2012progress,schirmer1982beyond,krylov2008equation}. Systems involving near-degeneracies, bond rearrangement, or substantial static correlation generally require multireference methods. Complete active-space self-consistent field (CASSCF), for example, optimises both the orbitals and the configuration-interaction expansion within a selected active space, while CASPT2 and NEVPT2 recover additional dynamic correlation \citep{roos1980complete,andersson1990second,angeli2001introduction}. Density matrix renormalisation group (DMRG) methods provide an efficient active-space solver by representing the wave function as a matrix product state, enabling larger active spaces than conventional complete-active-space approaches \citep{chan2011density,sharma2012spin}. Nevertheless, classical excited-state calculations remain challenging because the many-electron configuration space grows rapidly with system size, excited states are often closely spaced or strongly correlated, and their accurate treatment may require balancing several states and different types of electron correlation.

\subsection{Quantum and hybrid approaches}
\label{appLitRevQuant}
The task of computing excited states and their properties is generally considered to be a much more difficult task than that of the ground state. Despite this, emerging recent literature has developed a plethora of heuristics that may enable preparation of excited states and general eigenstates on a quantum computer. We mainly split the existing methods into four categories: heuristic-oriented strategies that are usually more suited for near-term quantum computers, hybrid approaches that mix time evolution and classical processing, methods with theoretical guarantees that can often only be implemented on fault-tolerant machines, and algorithms that take copies of a mixed state as input and distil the dominant eigenstate.

\subsubsection{Heuristic approaches}

\paragraph*{Variational algorithms.} Variational algorithms~\citep{cerezo2021variational,tilly2022variational} constitute a major category of excited-state algorithms for near-term devices. These methods extend the Variational Quantum Eigensolver (VQE) ~\citep{peruzzo2014variational} by enforcing orthogonality to previously computed eigenstates, optimising multiple states simultaneously, or transforming the Hamiltonian such that a target excited state becomes the variational ground state. Representative approaches include Variational Quantum Deflation (VQD)~\citep{higgott2019variational,jones2019variational,gocho2023excited}, Subspace Search VQE (SSVQE)~\citep{nakanishi2019subspace,parrish2019quantum}, Folded Spectrum VQE (FS-VQE)~\citep{caditazi2024folded}, Covariance Root Finding (CoVaR)~\citep{boyd2022training,hwang2025preparing}, which reformulates variational optimisation as a covariance-based root-finding problem using classical shadows.
On the other hand, methods such as Variational Quantum State Diagonalisation (VQSD)~\citep{larose2019variational} and Variational Quantum State Eigensolver (VQSE)~\citep{cerezo2022variational} train a parameterised ansatz that diagonalises a density matrix $\rho$ into the computational basis. This allows one to obtain the largest eigenprobability of $\rho$ through sampling. Further, one can then prepare the dominant eigenstate by applying the trained ansatz inversely to the computational basis state corresponding to the largest eigenprobability.

\paragraph*{Subspace methods.}
Another strategy is to utilise subspace diagonalisation, where one can approximate eigenstates of a Hamiltonian by projecting into a low-dimensional subspace generated from provided reference states. Building on this, the quantum computer is used to compute elements of the overlap matrix between a set of basis states, after which the corresponding eigenstates can be computed on a classical computer by solving the generalised eigenvalue problem~\citep{epperly2022theory,motta2024subspace}. Notable methods include the Quantum Subspace Expansion (QSE)~\citep{mcclean2017hybrid}, Quantum Filter Diagonalisation~\citep{parrish2019quantumb} and Quantum Krylov methods~\citep{motta2019determining,stair2020multireference,yoshioka2025krylov,cortes2022quantum,oumarou2025molecular}, which can also be enhanced using classical shadows~\cite{boyd2025high}.

\paragraph*{Adiabatic methods.}
Adiabatic methods can prepare excited states by slowly evolving an initial Hamiltonian with a known eigenstate into the target Hamiltonian while remaining in the corresponding instantaneous eigenstate~\citep{aspuruguzik2005simulated,burton2025excited,lutz2026adiabatic}. While such methods can potentially provide guarantees based on the adiabatic theorem, practical applications rely on the spectral gaps along the entire adiabatic path, which can be hard to quantify and control.

\subsubsection{Time evolution-based hybrid approaches}

\paragraph*{Time-series and signal processing.}
Unlike conventional phase estimation, time-series methods treat the quantum computer as a pure time-evolution engine. By applying $e^{-iHt}$ for varying durations and measuring overlaps, the full eigenspectrum can be extracted via classical signal processing~\citep{somma2019quantum,ding2023simultaneous,ding2024quantum}. Crucially, as demonstrated by \citet{obrien2019quantum}, feeding the measurement outcomes of these short-time evolutions into a classical Bayesian or frequentist post-processor allows for the simultaneous resolution of multiple excited-state energies present in an initial superposition, bypassing the deep circuits required to collapse the system into a single state. Recently, \citet{chan2025algorithmic} drastically reduced the measurement overhead of these methods using algorithmic shadow spectroscopy. By performing randomised measurements on time-evolved states to construct classical shadows, their approach resolves excited-state energy gaps through classical post-processing using orders of magnitude fewer circuit repetitions and no ancilla qubits.

\paragraph*{Virtual distillation and cooling.}
Virtual cooling mathematically simulates physical purification to extract expectation values directly from mixed or noisy states, bypassing the need for perfect coherent state preparation~\citep{cotler2019quantum}. Recently, frameworks have shifted the burden of multi-copy quantum coherence entirely into classical post-processing. Most notably, the Exponential Distillation of Dominant Eigenproperties (DDE) protocol by \citet{bako2026exponential} applies random real-time evolution to create an average mixed state, which is then ``virtually purified'' using classical Monte Carlo integration. When paired with appropriate initial states or shifted operators, DDE achieves exponential suppression of unwanted eigencomponents without the hardware overhead of multi-copy quantum control, as in conventional virtual distillation~\cite{huggins2021virtual, koczor2021exponential}.

\paragraph*{Imaginary time evolution.}
While quantum imaginary time evolution natively cools an initial state to the ground state~\citep{mcardle2019variational,motta2019determining,huo2023error}, this non-unitary dynamics can be redirected toward excited states. Techniques like deflation and spectral transformations can be employed to target excited states~\citep{tsuchimochi2023improved}.

\subsubsection{Theoretically guaranteed methods}

\paragraph*{Phase estimation.}
Quantum phase estimation is one of the most established algorithms for estimating eigenvalues of a Hamiltonian given an input state with non-zero overlap with the corresponding eigenstate~\citep{kitaev1995quantum,cleve1998quantum,nielsen2010quantum}. Recent developments have significantly reduced the circuit overhead, including variants that estimate the eigenvalues of multiple eigenstates present in the input state while requiring only a single ancilla qubit~\citep{ding2023simultaneous,ding2024quantum}.

Apart from eigenvalue estimation, early work of \citet{abrams1999quantum} established that the textbook phase estimation using the quantum Fourier transform can be used to obtain not only the ground state, but also excited states. By post-selecting on the target eigenenergies in the phase register, one can prepare the eigenstate from the post-selected residual state. However, this requires knowledge of the target eigenvalue up to high accuracy, though this can also be acquired by using phase estimation to obtain the target eigenvalue first, as well as a sufficient overlap between the initial state and the target eigenstate.  Similarly, the Rodeo algorithm~\citep{choi2021rodeo} prepares eigenstates through repeated randomised controlled Hamiltonian evolutions and post-selection, with a circuit reminiscent of Kitaev's iterative phase estimation~\citep{kitaev1995quantum}, while retaining the same requirements of \emph{a priori} target eigenenergy knowledge and a sufficiently large initial overlap.

We provide a more detailed discussion of the phase estimation algorithm for eigenenergy estimation and eigenstate preparation in \cref{appQPE,appQPEMAE}, showing that the na\"ive textbook approach using post-selection results in a sub-optimal algorithm, and while the runtime can be improved using multidimensional amplitude estimation, this comes at a cost of highly increased qubit usage.

\paragraph*{Spectral transformations.}
Spectral transformation methods can be used for excited state preparation by transforming the Hamiltonian spectrum such that the target eigenstate becomes the dominant factor of the transformed operator. Examples of these transformations include shift-and-invert methods based on the inverse Hamiltonian $(H-\mu \mathbb I)^{-1}$~\citep{kyriienko2020quantum,yoshikura2023quantum,millar2026imaginary}, folded-spectrum methods employing $(H-\mu \mathbb{I})^2$~\citep{caditazi2024folded,khinevich2025quantum}, and spectral filtering methods that filter out other states using an eigenstate filtering polynomial~\citep{lin2020optimal,patil2026efficient}. These methods, while efficient and theoretically guaranteed, often require knowledge of the target eigenenergy to shift or filter out the desired eigenstate. In \cref{appHamThresh}, we further discuss methods that apply eigenstate filtering polynomials to Hamiltonians in the case where the eigenenergy is not known \emph{a priori}.
Recently, spectral transformations have been further incorporated into dissipative state preparation algorithms, where they are used to design Lindbladians whose steady state is the target excited state~\citep{li2025dissipative}. Such methods can eliminate the need for an initial state with significant overlap with the target eigenstate.

\subsubsection{Sample-based algorithms}
Sample-based quantum purity amplification has been studied extensively~\citep{cirac1999optimal}, with recent algorithms preparing the dominant eigenstate of a density operator from multiple copies using SWAP-based gadgets, Clebsch–Gordan transforms, or density matrix exponentiation~\citep{li2024optimal,childs2025streaming,grier2025streaming,li2024optimal,dalzell2025distillation}. Furthermore, lower bounds established by \citet{grier2025streaming} show that these sample-based purity amplification protocols are asymptotically optimal.

A related approach is resonant quantum principal component analysis~\citep{li2021resonant}, which distils the dominant eigenstate from multiple copies of a density matrix by coupling it to a frequency-tunable probe qubit and exploiting resonant transitions.

\section{Dominant eigenstate preparation with textbook phase estimation and post-selection}
\label{appQPE}

In this appendix, we describe the textbook approach to dominant eigenstate preparation based on phase estimation~\citep{kitaev1995quantum,cleve1998quantum,nielsen2010quantum} and post-selection.

After executing phase estimation, measuring the phase register and obtaining an eigenenergy $E_k$, the system register collapses to the corresponding eigenstate. Given knowledge of $E_k$, one can further enhance the success probability via amplitude amplification by reflecting the ancilla register about $\ket{E_k}$, i.e., applying $\mathbb{I} - 2\ketbra{E_k}{E_k}$. However, even for the simpler task of ground state preparation, the limited accuracy of the eigenenergy readout in the ancillary register leaves residual imperfections in the main register. This imposes a significant limitation on using phase estimation and amplitude amplification for eigenstate preparation. Specifically, the target eigenenergy $E_k$ must be known to high accuracy \emph{a priori}, and calculating this value often requires higher runtime than the actual state preparation process~\citep{ge2019faster}.

For the task of dominant eigenstate preparation, the objective is to first identify the dominant eigenstate present in the input state, and estimate its corresponding eigenenergy to high accuracy. To determine which eigenvalue is dominant, we need to sample repeatedly from the phase estimation circuit to learn the output eigenprobability distribution. Using this to identify the dominant eigenvalue, we then use this estimate in the post-selection of the eigenvalue after phase estimation to produce a residual state close enough to the dominant eigenvalue.

In particular, when phase estimation is applied to our input state $\ket{\psi} = \sum_{k \in \mathcal{S}} c_k \ket{\psi_k}$, the resulting state takes the form
\begin{align}
    {\rm QPE}(\mathbb{I}\otimes U_\psi)\ket{0}\ket{0}
     & \approx \sum_{k\in{\mathcal{S}}} c_k\ket{E_k}\ket{\psi_k}\nonumber \\
     & = \sum_{k\in{\mathcal{S}}} \sqrt{p_k}\ket{E_k}\ket{\psi_k'} ,
\end{align}
where the phase of $c_k$ is absorbed into $\ket{\psi_k'}$. Here, phase estimation provides the eigenenergies as a fixed-point encoding in the computational basis in the ancillary register. Measuring the ancillary register (and tracing out the main register) thus samples from the probability distribution over eigenenergies, where, as in the main text, the probabilities $p_k$ are given by the squared overlaps of the input state with the corresponding eigenstates.

This observation allows us to interpret phase estimation as a provider of quantum sample access to a probability distribution over computational basis states indexed by eigenenergies. Formally, we define a distribution $\vec{p'}$ over the computational basis states whereby $p'_{E_k} = p_k$ and $p'_{\ell} = 0$ for $\ell \notin \{E_k\}_{k\in \mathcal{S}}$. It follows that
\begin{equation}
    {\rm QPE}(\mathbb{I}\otimes U_\psi) \ket{0}\ket{0}\approx \sum_{k\in{\mathcal{S}}} \sqrt{p'_{E_k}}\ket{E_k}\ket{\psi_k'}
\end{equation}
acts as a quantum sample state for $\vec{p'}$.

\subsection{Estimating the dominant eigenenergy with phase estimation}
As noted previously, the first step is to generate a distribution over eigenenergies using phase estimation. In the idealised setting, one would expect an evolution time of $\mathcal{O}(\eta^{-1})$ to achieve additive accuracy $\eta$. However, as observed by \citet{ge2019faster}, when the input state is a superposition of eigenstates, an additional multiplicative factor arises that depends on the spectral weight distribution of the input. This overhead is required to ensure that the measurement outcomes concentrate sufficiently onto a single computational basis state corresponding to an eigenenergy. The measurement outcomes are then obtained from a subset of the ancillary registers, as illustrated in \cref{figQPE}. For ground state energy estimation utilising textbook phase estimation, we require $\mathcal{O}(p_G^{-1}\eta^{-1})$ evolution time for the measurement results to concentrate sufficiently and avoid obtaining energy estimates lower than the ground state energy. To facilitate further discussion, we assume the maximum evolution time, as well as the dimension size of the phase estimation ancillary register, is $L$, and the readout register is the top $m$ qubits of the ancillary register, and set $M=2^m \in \mathcal{O}(\eta^{-1})$.

\begin{figure}
    \includegraphics[width=\linewidth]{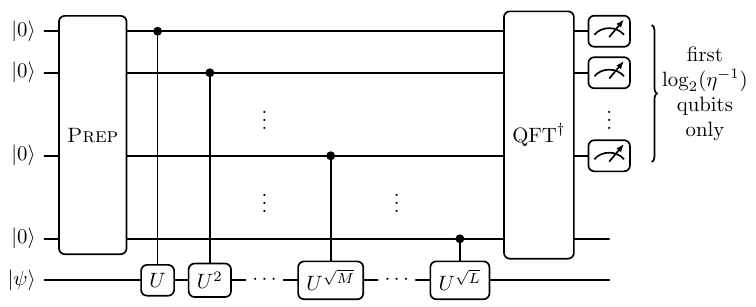}
    \caption{For input states that are not pure eigenstates, phase estimation requires additional qubits beyond the $\mathcal{O}(\log (\eta^{-1}))$ dependency.}
    \label{figQPE}
\end{figure}

\begin{lemma}[Textbook phase estimation~\citep{nielsen2010quantum}]
    For a given input eigenstate, the ``textbook'' phase estimation algorithm can estimate the eigenvalue up to additive accuracy $\eta$ with success rate of $1-\zeta$ in
   \begin{equation*}
        \widetilde{\mathcal{O}}\left(\frac{\|H\|}{\eta\zeta}\right),
    \end{equation*}
    query depth.
\end{lemma}
\begin{proof}
    Given input eigenstate $\ket{\psi_k}$ with eigenphase $\lambda_k$, phase estimation produces the output state 
    \begin{equation}
        \sum_{\ell=0}^{L-1} \alpha_{k\ell} \ket{\ell} \otimes \ket{\phi_k}
    \end{equation}
    where $\alpha_{k,\ell}$ depends on the input state.

    We first bound the probability of obtaining an estimate of $\lambda_k$ within an additive error $\eta$. Suppose we have a smaller register made up of the top $m$ qubits of the ancillary register, where we assume $M = 2^m \in\mathcal{O}(\|H\|/\eta) \le L$. We suppose a bit string approximation $b_k$ of $\lambda_k$ where we have $ 0 \le b_k < M$ and 
    \begin{equation}
        \lambda_k - \frac{1}{2M} \le \frac{b_k}{M} \le \lambda_k +\frac{1}{2M}.
    \end{equation}
    The probability of measuring $\ket{b_k}$ in the top $m$ qubits of the first register is thus the probability of obtaining an estimate of $\lambda_k$ within additive accuracy $\eta$. We define the set of bit strings of length $r$ whose first $m$ bits match $b_k$ as $\mathcal{B}_k$. 

    We can then bound the probability of not measuring $\ket{b_k}$ by some failure probability $\zeta$ such that
    \begin{equation}
        \sum_{\ell\notin \mathcal{B}_k}|\alpha_{k\ell}|^2 \le \zeta.
        \label{eqAlphaTailBound}
    \end{equation}

    For textbook phase estimation~\citep{nielsen2010quantum}, one can find
    \begin{align}
        &\sum_{\ell\notin \mathcal{B}_k}|\alpha_{k\ell}|^2 = \sum_{\frac{\ell}{L} \ge \lambda_k +\frac{1}{2M}}|\alpha_{k\ell}|^2 + \sum_{\frac{\ell}{L} \le \lambda_k -\frac{1}{2M}}|\alpha_{k\ell}|^2\nonumber\\
        &\le \sum_{\frac{\ell}{L} \ge \lambda_k +\frac{1}{2M}}\left(\frac{1}{2(\ell-L\lambda_k)}\right)^2 + \sum_{\frac{\ell}{L} \le \lambda_k -\frac{1}{2M}}\left(\frac{1}{2(L\lambda_k - \ell)}\right)^2\nonumber\\
        &\le \frac{1}{2}\sum_{L \ge \frac{L}{2M}} \frac{1}{L^2} \le \frac{1}{2} \int_{\frac{L}{2M}-1}^{\infty} \frac{1}{x^2} \mathrm dx = \frac{M}{L-2M}\le \frac{2M}{L}
    \end{align}
    Upper bounding the above with $\zeta$, we find that
    \begin{equation}
        L = \frac{2M}{\zeta} \in \mathcal{O}\left(\frac{\|H\|}{\eta\zeta}\right).
    \end{equation}
\end{proof}

We can generalise this result to cases where we have a state $\ket{\psi}$ with dominant eigenstate $\ket{\psi_0}$. We show this via an analogue of Proposition 2 in \citet{ge2019faster} for dominant eigenenergy estimation instead of ground state energy estimation.

\begin{lemma}[Dominant eigenenergy resolution via textbook phase estimation]
   Consider applying textbook phase estimation to resolve the dominant eigenenergy $E_0$ to additive accuracy $\eta \le \Delta$. The resulting phase estimation circuit requires a total controlled evolution time of
    \begin{equation*}
        \widetilde{\mathcal{O}}\left(\frac{\|H\|}{\eta(p_0-p_1)}\right),
    \end{equation*}
    and uses 
    \begin{equation*}
    \mathcal{O}\left(\log\frac{\|H\|}{\eta(p_0-p_1)}\right)
    \end{equation*}
    additional ancilla qubits.
    A single execution of the circuit produces an outcome distributed according to a perturbed spectral distribution $\{\widehat p_k\}_{k\in\mathcal{S}}$, where the probability of sampling the dominant eigenenergy is $\widehat p_0$, satisfying
    \begin{equation*}
        |\widehat p_0 - p_0| \in \mathcal{O}(p_0 - p_1).
    \end{equation*}
    \label{lemQPESingle}
\end{lemma}
\begin{proof}
    Generalizing to superposed input states $\ket{\psi} = \sum_{k\in\mathcal{S}} c_k \ket{\psi_k}$, we obtain the output state
    \begin{equation}
        \sum_{\ell=0}^{L-1} \sum_{k\in\mathcal{S}}c_k\alpha_{k\ell} \ket{\ell}\ket{\psi_k}
    \end{equation}

    The probability of measuring $\ket{b_k}$ in the top $m$ qubits can then be written as follows:
    \begin{align}
        \Pr[b_k] &= \Pr[\ell \in \mathcal{B}_k]= 1 - \sum_{j \in \mathcal{S}}\sum_{\ell\notin \mathcal{B}_k} |c_{j}|^2|\alpha_{j\ell}|^2\nonumber\\
        &= 1 - \sum_{\ell\notin \mathcal{B}_k} p_k|\alpha_{k\ell}|^2 - \sum_{j \in \mathcal{S}\setminus\{k\}}\sum_{\ell\notin \mathcal{B}_k} p_{j}|\alpha_{j\ell}|^2\nonumber\\
        &\ge 1 - \sum_{\ell\notin \mathcal{B}_k} |\alpha_{k\ell}|^2 - \sum_{j \in \mathcal{S}\setminus\{k\}}\sum_{\ell =0}^{L-1} p_{j}|\alpha_{j\ell}|^2\nonumber\\
        &= p_k - \sum_{\ell\notin \mathcal{B}_k} |\alpha_{k\ell}|^2
    \end{align}
    This then results in a final lower bound of 
    \begin{equation}
        \Pr[b_k] \ge p_k -\zeta.
    \end{equation}

    With the further assumption that $\eta \le \Delta$ such that $\mathcal{B}_{k} \cap \mathcal{B}_{j} = \emptyset$, we can also provide the upper bound of the probability as follows:
    \begin{align}
        \Pr[b_k] &= \Pr[\ell \in \mathcal{B}_k]= \sum_{j \in \mathcal{S}}\sum_{\ell\in \mathcal{B}_k} |c_{j}|^2|\alpha_{j\ell}|^2\nonumber\\
        &= \sum_{\ell\in \mathcal{B}_k} p_k|\alpha_{k\ell}|^2 + \sum_{j \in \mathcal{S}\setminus\{k\}}\sum_{\ell\in \mathcal{B}_k} p_{j}|\alpha_{j\ell}|^2\nonumber\\
        &\le p_k + \sum_{j \in \mathcal{S}\setminus\{k\}}\sum_{\ell\notin \mathcal{B}_{j}} p_{j}|\alpha_{j\ell}|^2\nonumber\\
        &\le p_k + \sum_{j \in \mathcal{S}\setminus\{k\}} p_{j}\zeta \le p_k + \zeta
    \end{align}
    where we use \cref{eqAlphaTailBound} in the second-to-last inequality.

    We also bound the probability of measuring a string $b$ that is not in the set of $\{b_k\}_{k\in \mathcal{S}}$ for all $k$ as
    \begin{align}
        \Pr\left[b \notin \{b_k\}_{k\in \mathcal{S}}\right] &= \sum_{j \in \mathcal{S}}\sum_{\ell\notin \bigcup_{k\in\mathcal{S}}\mathcal{B}_k} |c_{j}|^2|\alpha_{j\ell}|^2\nonumber\\
        &\le \sum_{j \in \mathcal{S}}\sum_{\ell\notin \mathcal{B}_j} p_{j}|\alpha_{j\ell}|^2\le \zeta,
    \end{align}
    where we used \cref{eqAlphaTailBound} in the last inequality. 

    To discern the dominant eigenenergy, we would require that 
    \begin{equation}
        \Pr[b_0] > \Pr[b_1]\text{ and } \Pr[b_0] > \Pr\left[b \notin \{b_k\}_{k\in \mathcal{S}}\right].
    \end{equation}
    Using the bounds derived above, we can see that the above can be satisfied if 
    \begin{equation}
        \Pr[b_0] \ge p_0 -\zeta > p_1 + \zeta \ge \Pr[b_1]
    \end{equation}
    and
    \begin{equation}
        \Pr[b_0] \ge p_0 -\zeta > \zeta \ge \Pr\left[b \notin \{b_k\}_{k\in \mathcal{S}}\right]
    \end{equation}
    where we obtain the bound $\zeta \in \mathcal{O} (p_0-p_1)$. Note that $\zeta$ also indicates the deviation with which the probability of an eigenenergy can be sampled.
\end{proof}
Given access to the phase estimation circuit, the most straightforward method to obtain the dominant eigenenergy is to repeatedly sample from the phase estimation circuit and identify the dominant eigenenergy. We can write the output of the phase estimation algorithm as 
\begin{multline}
    {\rm QPE}(\mathbb{I}\otimes \mathbb{I} \otimes U_\psi) \ket{0}_m\ket{0}_{r-m}\ket{0}_n\\
    = \sum_{\ell = 0}^{M-1} \sqrt{p_{\ell}'} \ket{\ell}_m \ket{\chi_\ell}_{r-m+n}
    \label{eqQPEOutput}
\end{multline}
where the $\ket{\chi_\ell}$ is the residual state after phase estimation and $\vec{p'} \in \bbDelta^M$ (where we use the symbol $\bbDelta^M$ to indicate a simplex with dimension $M$) is eigenenergy-indexed probability with the following bounds:
\begin{equation}
    \begin{cases}
        p'_{\ell} = \widehat{p}_k:\; p_k - \zeta \le \widehat{p}_k \le p_k + \zeta & \text{ if }\ell = b_k\\
        p'_{\ell} \le \zeta & \text{ if } \ell \notin \{b_k\}_{k\in\mathcal{S}}
    \end{cases}
    \label{eqQPEDistr}
\end{equation}
where $b_k \in \{0,1\}^M$ is the binary approximation of $\lambda_k$ within $\eta$. It is easy to see that sampling from the first $m$-sized register samples over $\vec{p'}$.

\begin{proposition}[Dominant eigenenergy estimation by sampling phase estimation circuits]
    With success probability $1-\delta$, the dominant eigenenergy $E_0$ can be estimated up to additive error $\eta \le \Delta$ with 
    \begin{equation*}
        \mathcal{O}\left(\frac{1}{(p_0 - p_1)^2}\log\left(\frac{1}{\delta}\min\left\{\frac{1}{p_0-p_1}, \frac{\|H\|}{\eta}\right\}\right)\right)
    \end{equation*}
    samples to the phase estimation circuit. This results in a total runtime complexity of 
    \begin{equation*}
        \widetilde{\mathcal{O}}\left(\frac{\|H\|}{(p_0 - p_1)^3\eta}\log\frac{1}{\delta}\right).
    \end{equation*}
    \label{propQPEsamp}
\end{proposition}
\begin{proof}
    The phase estimation circuit provides a classical sampling distribution $\vec p' \in \bbDelta^M$. To identify the dominant item, we must first learn the probability distribution up to $\ell_\infty$ accuracy $\xi$, before picking the maximum index. For each index, standard Chernoff bounds show that an eestimate$\widetilde{p}_{\ell}$ of $p_{\ell}'$ within additive error $\xi$ can be obtained with $\mathcal{O}(\log(1/\delta)/\xi^2 )$ samples. By a union bound, learning this for all items in the probability distribution requires 
    \begin{equation}
        \mathcal{O}\left(\frac{1}{\xi^2} \log \frac{M}{\delta}\right)
    \end{equation}
    samples. Alternatively, given that most of the entries in our probability distribution are close to 0 (upper bounded by $\zeta$), we can use $\mathcal{O}(\log(1/\delta\xi)/\xi)$ to first identify $\mathcal{O}(\xi^{-1})$ probability entries lower bounded by $p_{\ell}' > \xi$~\citep{motwani1995randomized,mitzenmacher2005probability}. Then by union bound, learning these $\mathcal{O}(\xi^{-1})$ entries in the probability distribution requires 
    \begin{equation}
        \mathcal{O}\left(\frac{1}{\xi^2} \log \frac{1}{\xi\delta}\right).
    \end{equation}

    To discern between the dominant and second dominant eigenenergies, which have the learned probabilities $\widetilde p_{b_0}$ and $\widetilde p_{b_1}$ from the sampled results, we require $\widetilde p_{b_0} \ge \widetilde p_{b_1}$. Using the bounds provided \cref{eqQPEDistr}, we obtain the following inequality:
    \begin{align}
        \widetilde p_{b_0} - \widetilde p_{b_1} &\ge (p'_{b_0}-\xi) - (p'_{b_1} +\xi) = \widehat{p}_0 - \widehat{p}_1 - 2\xi \nonumber\\
        &\ge (p_0 - \zeta) - (p_1 +\zeta) - 2\xi\nonumber\\
        &= p_0 - p_1 - 2\zeta - 2\xi
        \label{eqProbOffset}
    \end{align}
    where by $\zeta \in \mathcal{O}(p_0-p_1)$ from \cref{lemQPESingle}, we can see that $\xi \in \mathcal{O}(p_0-p_1)$.

    Thus, the total number of samples required in this 
    \begin{equation}
        \mathcal{O}\left(\frac{1}{(p_0 - p_1)^2}\log\left(\frac{1}{\delta}\min\left\{\frac{1}{p_0-p_1}, \frac{\|H\|}{\eta}\right\}\right)\right)
    \end{equation}
    where we use $M \in \mathcal{O}(\|H\|/\eta)$.

    Lastly, the total runtime complexity can be obtained by multiplying the sample complexity in this lemma with the runtime complexity in \cref{lemQPESingle} for a single-shot measurement.
\end{proof}

\subsection{Post-selecting the dominant eigenstate from phase estimation}

Given knowledge of the dominant eigenenergy obtained in the previous section, we can use this for dominant eigenstate preparation via post-selection over the known eigenenergy. Noting that the corresponding result in Proposition 3 of \citet{ge2019faster} can be used to prepare any eigenstate provided the eigenvalue. We can write the generalised version of the results as follows. Note that here we use post-selection instead of amplitude amplification as in \citet{ge2019faster} to match the usage of post-selection in eigenenergy estimation.

\begin{lemma}[Post-selecting the residual state for known eigenvalue]
	Suppose that the value of $E_k$ is known to an accuracy of 
    \begin{equation}
        \mathcal{O}\left( \sqrt{p_k}\Delta\varepsilon\right).
    \end{equation} 
    Then, post-selection on phase estimation can be used to prepare a state $\varepsilon$-close to the eigenstate with runtime complexity of
	\begin{equation*}
	\widetilde{\mathcal{O}}\left(\frac{\|H\|}{p_k^{1.5}\Delta\varepsilon}\log\frac{1}{\delta}\right),
	\end{equation*}	 
	and using $n + \mathcal{O}\left(\log (\|H\|/(p_k\Delta\varepsilon))\right)$ qubits. 
    \label{lemQPEStatePrep}
\end{lemma}
\begin{proof}
    From Proposition 3 of \citet{ge2019faster}, we know that if we run the phase estimation circuit to accuracy $\eta \in \mathcal{O}\left(\sqrt{p_k}\Delta\varepsilon\right)$, and post-select on the eigenenergy that is known \emph{a priori}, then the post-selected residual state $\ket{\psi_k^R}$ satisfies $\lVert\ket{\psi_k^R}-\ket{\psi_k}\rVert\le \varepsilon$. 
    
    From standard phase estimation results~\citep{nielsen2010quantum}, we know the sample probability of obtaining the residual state $\lVert\ket{\psi_k^R}$ is 
    \begin{equation}
        \widetilde p_k \ge \frac{4}{\pi^2} p_k.
    \end{equation}
    Then, to post-select on the correct eigenenergy, we require 
    \begin{equation}
        \mathcal{O}\left(\frac{1}{\widetilde{p}_k}\log\frac{1}{\delta}\right) \subseteq\mathcal{O}\left(\frac{1}{p_k}\log\frac{1}{\delta}\right)
    \end{equation}
    samples from standard classical sampling results. Multiplying the sample complexity by the phase estimation runtime costs that generate eigenenergies with accuracy $\eta$, we recover the runtime as shown in the lemma.
\end{proof}

We note from these results that the cost of actually preparing the eigenstate once the eigenenergy is located to high accuracy is relatively small compared to the eigenenergy estimation. This also implies that using amplitude amplification instead of post-selection does not help with the overall runtime. 

Putting the requirements of the accuracy of the eigenvalue into the eigenenergy estimation, we obtain the final runtime complexity of preparing the eigenstate through textbook phase estimation.
\begin{theorem}[Dominant eigenstate preparation by phase estimation]
    With success probability $1-\delta$, the dominant eigenstate $\ket{\psi_0}$ can be prepared up to $\ell_2$ distance $\varepsilon$ with runtime complexity
    \begin{equation*}
        \widetilde{\mathcal{O}}\left(\frac{\|H\|}{(p_0 - p_1)^3\sqrt{p_0}\Delta\varepsilon}\log\frac{1}{\delta}\right)
    \end{equation*}
    and requires 
    \begin{equation*}
        n + \mathcal{O}\left(\log \frac{\|H\|}{(p_0-p_1)\Delta\varepsilon}\right),
    \end{equation*}
    qubits. 
\end{theorem}
\begin{proof}
    We obtain the runtime complexity by plugging in
    \begin{equation}
        \frac{1}{\eta} \in \mathcal{O}\left(\frac{1}{\sqrt{p_0}\Delta\varepsilon}\right)
    \end{equation}
    required by \cref{lemQPEStatePrep} into \cref{propQPEsamp} for dominant eigenvalue estimation. Note that this runtime far exceeds that of the eigenstate preparation in \cref{lemQPEStatePrep}, and the $(\sqrt{p_0})^{-1}$ factor in the logarithmic term is subsumed by the $(p_0-p_1)^{-1}$.
\end{proof}

\section{Improved dominant eigenstate preparation with phase estimation}
\label{appQPEMAE}
In this appendix, we improve on the na\"ive phase estimation algorithm in the previous section.

First, we replace textbook phase estimation with modern tapered phase estimation~\citep{rendon2023lowdepth,rendon2024improved,chen2025quantum,luis1996optimum,buzek1999optimal,dam2007optimal,rendon2022effects,berry2024analyzing,patel2026optimal}. Instead of the uniform Hadamard state, this approach uses states that mimic the window functions used in classical signal processing for sharpening the correct output state and suppressing the other erroneous output states, thereby closing the gap between actual runtime costs and the complexity lower bounds of phase estimation~\citep{mande2026tight}.

For the next improvement, we adapt the faster phase-estimation--based ground state preparation algorithm analysed in the appendix of \citet{ge2019faster} to our setting (rather than the main algorithm presented in that work). The procedure of their algorithm is as follows:
\begin{enumerate}
    \item Apply the D\"urr--H\o{}yer minimum-finding algorithm~\citep{durr1996quantum} to the ancillary register of the phase estimation circuit to identify and read out the ground state energy $E_G$ to sufficiently high accuracy. The accuracy must be chosen such that the residual state is distilled enough to avoid a build-up of errors in the following step. The evolution time should be sufficiently long to ensure that the spread in the phase estimation distribution does not yield values below the true ground state energy.
    \item Perform amplitude amplification on the phase estimation circuit, reflecting the ancillary register about the state $\ket{E_G}$, i.e., applying $\mathbb{I} - 2\ketbra{E_G}{E_G}$.
\end{enumerate}

While the second step remains the same for us, we need to modify the first step such that the dominant eigenenergies are extracted. To replace the minimum finding, we require a procedure that, given quantum sample access to a probability distribution, identifies the index corresponding to its largest entry. This task can be accomplished using multidimensional amplitude estimation~\citep{vanapeldoorn2021quantum}, which, given quantum sample access to a probability distribution $\vec{p}$, outputs a classical estimate $\vec{\widetilde p}$ satisfying $\lVert \vec{\widetilde p} - \vec{p} \rVert_\infty \le \xi$. This allows us to reliably recover the index of the largest component of $\vec{p}$, provided the gap between the largest and second-largest probabilities is sufficiently large, and the accuracy $\xi$ of the estimate is less than this gap.

To apply multidimensional amplitude estimation, it remains to specify how we can provide this quantum sample access. In our setting, this is provided directly by the phase estimation circuit, where we can recall that
\begin{equation}
    {\rm QPE}(\mathbb{I}\otimes U_\psi) \ket{0}\ket{0}\approx \sum_{k\in{\mathcal{S}}} \sqrt{p'_{E_k}}\ket{E_k}\ket{\psi_k'}
\end{equation}
acts as a quantum sample state for $\vec{p'}$.

With this interpretation, multidimensional amplitude estimation can be applied directly to the phase estimation circuit to obtain an $\ell_\infty$-accurate estimate of $\vec{p'}$. The dominant eigenenergy $E_0$ can then be identified by identifying the index of the maximum value in $\vec{p'}$. Following this, we can now apply amplitude amplification on the phase estimation algorithm by reflecting the ancillary register about $\ket{E_0}$.

We detail the runtime costs and theorem statements in the following sections.

\subsection{Gaussian and tapered phase estimation}
To apply tapered phase estimation and to match the use of Gaussian quadratures in the DEFEAT algorithm in the main text, we discuss Gaussian phase estimation~\citep{rendon2024improved,rendon2023lowdepth,chen2025quantum}, which uses a discrete Gaussian state as input.

\begin{lemma}[Gaussian phase estimation]
   For a given input eigenstate, the Gaussian phase estimation algorithm can estimate the eigenvalue up to additive accuracy $\eta$ with a success probability of $1-\zeta$ in
   \begin{equation*}
        \mathcal{O}\left(\frac{\|H\|}{\eta}\log\frac{\|H\|}{\eta\zeta}\right),
    \end{equation*}
    query depth.
    \label{lemGaussQPE}
\end{lemma}

With the use of Gaussian phase estimation, we modify \cref{lemQPESingle} under this better runtime.

\begin{lemma}[Dominant eigenenergy resolution via Gaussian phase estimation]
   Applying Gaussian phase estimation to resolve the dominant eigenenergy $E_0$ to additive accuracy $\eta \le \Delta$ requires a total controlled evolution time of
    \begin{equation*}
        \widetilde{\mathcal{O}}\left(\frac{\|H\|}{\eta}\log \frac{1}{p_0-p_1}\right),
    \end{equation*}
    and uses 
    \begin{equation*}
    \mathcal{O}\left(\log\frac{\|H\|}{\eta}+\log \log \frac{1}{p_0-p_1}\right)
    \end{equation*}
    additional ancilla qubits.
    A single execution of the circuit produces an outcome distributed according to a perturbed spectral distribution $\{\widehat p_k\}_{k\in\mathcal{S}}$, where the probability of sampling the dominant eigenenergy is $\widehat p_0$, satisfying
    \begin{equation*}
        |\widehat p_0 - p_0| \in \mathcal{O}(p_0 - p_1).
    \end{equation*}
    \label{lemGaussQPESingle}
\end{lemma}
\begin{proof}
    Following the proof of \cref{lemQPESingle}, we again set $\zeta \in \mathcal{O}(p_0-p_1)$. Substituting this into \cref{lemGaussQPE} yields the result.
\end{proof}

\subsection{Estimating the dominant eigenenergy with multidimensional amplitude estimation}

As mentioned before, we can further improve the runtime of dominant eigenenergy estimation by replacing the na\"ive sampling with multidimensional amplitude estimation~\citep{vanapeldoorn2021quantum} as follows:

\begin{lemma}[Multidimensional amplitude estimation -- Theorem 5 and 9, \citep{vanapeldoorn2021quantum}]
    Let $\vec p \in \bbDelta^M$ and let $U_p$ provide quantum sample access to $\vec p$. Let $\xi > 0$. A probability distribution $\widetilde{\vec p} \in \bbDelta^M$ such that $\|\vec p - \widetilde{\vec p}\|_\infty \le \xi$ can be found with error probability at most $\delta  > 0$ using
    \begin{equation*}
        \mathcal{O}\left(\frac{1}{\xi}\log\frac{1}{\xi\delta}\right)
    \end{equation*}
    applications of $U_p$ and 
    \begin{equation*}
        \mathcal{O}\left(\frac{1}{\xi}\log\frac{1}{\xi}\right)
    \end{equation*}
    additional qubits.
    \label{lemMAE}
\end{lemma}

\begin{remark}
    As the phase query for multidimensional amplitude estimation only provides one probability encoded as a phase in each measurement register, we do not require extra evolution time as in phase estimation, nor does swapping the uniform Hadamards with a Gaussian state improve the performance asymptotically, similar to standard amplitude estimation.
\end{remark}

It is easy to see from \cref{eqQPEOutput} that ${\rm QPE}(\mathbb{I}\otimes U_\psi)$ provides quantum sample access to $\vec p'$ in the first register. With this, we can apply multidimensional amplitude estimation.
\begin{proposition}[Dominant eigenenergy estimation by phase estimation and multidimensional amplitude estimation]
    With success probability $1-\delta$, the dominant eigenenergy $E_0$ can be estimated up to additive error $\eta \le \Delta$ with 
    \begin{equation*}
        \mathcal{O}\left(\frac{1}{p_0 - p_1}\log\frac{1}{(p_0-p_1)\delta}\right)
    \end{equation*}
    queries to the phase estimation circuit. This results in a total runtime complexity of 
    \begin{equation*}
        \widetilde{\mathcal{O}}\left(\frac{\|H\|}{(p_0 - p_1)\eta}\right),
    \end{equation*}
    and requires $n + \widetilde{\mathcal{O}}\left((p_0 - p_1)^{-1}+\log(\|H\|/\eta)\right)$ qubits.
    \label{propQPEMAE}
\end{proposition}
\begin{proof}
    Apply \cref{lemMAE} by setting $\xi \in \mathcal{O}(p_0 - p_1)$ similarly to \cref{eqProbOffset} in \cref{propQPEsamp}. The runtime complexity is obtained by multiplying by the runtime of \cref{lemGaussQPESingle}.
\end{proof}

\subsection{Amplifying the dominant eigenstate from Gaussian phase estimation}

Given knowledge of the dominant eigenenergy obtained in the previous section, we can use this for dominant eigenstate preparation via amplitude amplification over the known eigenenergy. We make use of the following Gaussian version of Proposition 3 in \citet{ge2019faster}

\begin{lemma}[Eigenstate preparation by Gaussian phase estimation and amplitude amplification with known eigenenergy]
	Suppose that the value of $E_k$ is known to an accuracy of
    \begin{equation}
        \mathcal{O}\left(\frac{\Delta}{\log(1/p_k\varepsilon)}\right).
    \end{equation} 
    Then, phase estimation and (fixed-point) amplitude amplification can be used to prepare a state $\varepsilon$-close to the eigenstate with runtime complexity of
	\begin{equation*}
	\widetilde{\mathcal{O}}\left(\frac{\|H\|}{\sqrt{p_k}\Delta} \log\frac{1}{\varepsilon} \right),
	\end{equation*}	 
	and using $n + \mathcal{O}\left(\log (\|H\|/\Delta)+\log\log\left(p_k^{-1}\varepsilon^{-1}\right)\right)$ qubits. 
    \label{lemGaussQPEAmp}
\end{lemma}
\begin{proof}
    To perform amplitude amplification based on the correct eigenvalue, we require a bit string $b_k$ close to the known \emph{a priori} eigenvalue. Post-selecting on $\ket{b_k}$, the residual state in the main register is then 
    \begin{equation}
        \ket{\psi_k^R} = \frac{\sum_{j\in\mathcal{S}}c_{j}\alpha_{j,b_k} \ket{b_k}\ket{\psi_{j}}}{\left\lVert\sum_{j\in\mathcal{S}}c_{j}\alpha_{j,b_k} \ket{b_k}\ket{\psi_{j}}\right\rVert}.
    \end{equation}
    For simplicity, we denote $a_{j,b_k} = a_{j}^{(k)}$. Similar to \cref{eqAlphaTailBound}, we define \begin{equation}
        \sum_{\ell\ne b_j}|\alpha_{j\ell}|^2 = 1 - \left\lvert a_j^{(j)}\right\rvert^2\le \zeta.
    \end{equation}The $\ell_2$ distance between the residual state and the eigenstate is then 
    \begin{equation}
        \left\lVert\ket{\psi_k^R} -\ket{\psi_k}\right\rVert \in \Theta\left(\frac{\left\lVert\sum_{j\in\mathcal{S}\setminus\{k\}}c_{j}\alpha_j^{(k)} \ket{b_k}\ket{\psi_{j}}\right\rVert}{\left\lVert\sum_{j\in\mathcal{S}}c_j\alpha_j^{(k)} \ket{b_k}\ket{\psi_{j}}\right\rVert}\right)
    \end{equation}
    
    The numerator and denominator above can be bounded, respectively, as
    \begin{align}
        \left\lVert\sum_{j\in\mathcal{S}\setminus\{k\}}c_{j}\alpha_j^{(k)} \ket{b_k}\ket{\psi_{j}}\right\rVert^2 &= \sum_{j\in\mathcal{S}\setminus\{k\}}p_j\left|\alpha_j^{(k)}\right|^2 \nonumber\\
        &\le \sum_{j\in\mathcal{S}\setminus\{k\}}p_j \left(1-  \left|\alpha_j^{(j)}\right|^2\right)\nonumber\\
        &\le \sum_{j\in\mathcal{S}\setminus\{k\}}p_j\zeta \le \zeta
    \end{align}
    and
    \begin{align}
        \left\lVert\sum_{j\in\mathcal{S}}c_{j}\alpha_j^{(k)} \ket{b_k}\ket{\psi_{j}}\right\rVert^2  &= \sum_{j\in\mathcal{S}}p_j\left|\alpha_j^{(k)}\right|^2 \ge p_k\left|\alpha_k^{(k)}\right|^2\nonumber\\
        &\ge p_k(1-\zeta).
    \end{align}
    
    Thus, to produce a $\ket{\psi_k^R}$ that is $\varepsilon$-close to $\ket{\psi_k}$, we require
    \begin{equation}
        \sqrt{\frac{\zeta}{p_k(1-\zeta)}}\le \sqrt{\frac{2\zeta}{p_k}} \le \varepsilon,
    \end{equation}
    and we have 
    \begin{equation}
        \zeta \in \mathcal{O}(p_k\varepsilon^2).
    \end{equation}

    Therefore, to restrict the error of the residual state to $\varepsilon$, by \cref{lemQPESingle}, we need every single run of the phase estimation circuit to have runtime 
    \begin{equation}
        \widetilde{\mathcal{O}}\left(\frac{\|H\|}{\Delta}\log\frac{1}{p_k\varepsilon}\right),
    \end{equation}
    which implies that we need to know the bit string $b_k$ up to accuracy 
        \begin{equation}
        \mathcal{O}\left(\frac{\Delta}{\|H\|\log(1/p_k\varepsilon)}\right).
    \end{equation}
    Lastly, multiplying the amplitude amplification rounds of $\mathcal{O}(1/\sqrt{p_k})$, we obtain the total runtime complexity
    \begin{equation}
	\widetilde{\mathcal{O}}\left(\frac{\|H\|}{\sqrt{p_k}\Delta} \log\frac{1}{\varepsilon}\right).
	\end{equation}	
\end{proof}

Thus, using this lemma to perform dominant eigenstate preparation requires prior knowledge of the dominant eigenvalue to $\mathcal{O}(\Delta/\log(1/p_k\varepsilon))$ accuracy, which we can obtain by \cref{propQPEMAE}. We put everything together in the following result.
\begin{theorem}[Dominant eigenstate preparation by phase estimation and amplitude amplification]
    With success probability $1-\delta$, the dominant eigenstate $\ket{\psi_0}$ can be prepared up to $\ell_2$ distance $\varepsilon$ with runtime complexity
    \begin{equation*}
        \widetilde{\mathcal{O}}\left(\frac{\|H\|}{(p_0 - p_1)\Delta}\log\frac{1}{\varepsilon}\log\frac{1}{\delta}\right)
    \end{equation*}
    and requires 
    \begin{equation*}
        n + \widetilde{\mathcal{O}}\left(\frac{1}{p_0-p_1} + \log \frac{\|H\|}{\Delta}+ \log\log \frac{1}{\varepsilon}\right),
    \end{equation*}
    qubits. The qubit number dependency of $(p_0-p_1)^{-1}$ can be removed at the cost of increasing the total runtime complexity by $(p_0-p_1)^{-1}$.
    \label{propQPEprep}
\end{theorem}
\begin{proof}
    We obtain the runtime complexity by plugging in
    \begin{equation}
        \frac{1}{\eta} \in \mathcal{O}\left(\frac{1}{\Delta}\log\frac{1}{p_0\varepsilon}\right)
    \end{equation} required by \cref{lemGaussQPEAmp} into \cref{propQPEMAE} for dominant eigenvalue estimation. It is easy to see that this runtime is larger than that of \cref{lemGaussQPEAmp}. The qubit reduction can be achieved by replacing the multidimensional amplitude estimation in \cref{propQPEMAE} with sample-based learning, similar to  \cref{propQPEsamp}, but with Gaussian phase estimation instead, which removes the additional time dependency on $(p_0-p_1)^{-1}$ in \cref{propQPEsamp}.  Note that the $\log (1/p_0)$ terms are now hidden as they are strictly dominated by the $(p_0-p_1)^{-1}$ term.
\end{proof}

\subsection{Bonus: Applications to multiple eigenvalue estimation}
\label{appQEEP}
Here, we briefly depart from the main topic of dominant eigenstate preparation to discuss other uses of the eigenvalue estimation algorithm -- specifically, its application to the quantum eigenvalue estimation problem.

The quantum eigenvalue estimation problem (QEEP) is a generalisation of the phase estimation problem, where one extracts multiple (sufficiently dominant) eigenvalues from a single input instead of a sufficiently dominant eigenvalue (where the eigenstate makes up $> 50\%$ of the amplitude). 
The early work by \citet{somma2019quantum} recast this as a time-series analysis problem over expectation values obtained by Hadamard tests.
More recently, \citet{ding2023simultaneous} introduced low-depth algorithms for simultaneous estimation of multiple eigenvalues by combining Gaussian filtering with least-squares optimisation. This line of work was further advanced by \citet{ding2024quantum}, who proposed the Quantum Multiple Eigenvalue Gaussian-filtered Search (QMEGS) algorithm, integrating a search strategy analogous to orthogonal matching pursuit~\citep{pati1993orthogonal}.

The dominant eigenenergy estimation algorithm using multidimensional amplitude estimation can be viewed as a fully coherent version of QMEGS, and results in a quadratic speedup over the dependence on the probability gap. Similar to QMEGS, we assume the \emph{sufficiently dominant condition}: there exists a set of indices $\mathcal{D}\subset \{0,1,\cdots,N-1\}$ such that $p_{\min} = \min_{j\in\mathcal{D}}p_j > p_{\rm tail} = \sum_{i\notin\mathcal{D}} p_j$. We refer to $\{\lambda_m\}_{m\in \mathcal{D}}$ as the \emph{sufficiently dominant} eigenvalues. In particular, we have the following result:

\begin{theorem}[Multiple eigenvalue estimation by phase estimation and multidimensional amplitude estimation]
    Assume $p_{\rm min} > p_{\rm tail}$. With success probability $1-\delta$, the sufficiently dominant eigenvalues $\{E_k\}_{k\in\mathcal{D}}$ can be retrieved up to $\eta$ with runtime complexity
    \begin{equation*}
        \widetilde{\mathcal{O}}\left(\frac{\|H\|}{(p_{\rm min} - p_{\rm tail})\eta}\log\frac{1}{\delta}\right)
    \end{equation*}
    and requires 
    \begin{equation*}
        n + \widetilde{\mathcal{O}}\left(\frac{1}{(p_{\rm min} - p_{\rm tail})} + \log \frac{\|H\|}{\eta}\right)
    \end{equation*}
    qubits.
\end{theorem}
\begin{proof}
    We define the dominant spectral gap of eigenphases $\Delta_{\rm dom}=\min_{j,k\in \mathcal{D}, j\not=k}\left|\lambda_j-\lambda_k\right|$. We first assume $\eta/\|H\| \le \Delta_{\rm dom}$, which implies that the sufficiently dominant eigenvalues can be sufficiently resolved using phase estimation. Note that this does not include the non-dominant tail eigenvalues. 
    
    Similar to the proof in \cref{lemQPESingle}, we can obtain
    \begin{subequations}
    \begin{gather}
        \Pr[b_k] \ge p_k - \zeta,\\
        \Pr\left[b \notin \{b_k\}_{k\in \mathcal{D}}\right] \le \zeta + p_{\rm tail},
    \end{gather}
    \end{subequations}
    where an additional $p_{\rm tail}$ term in the upper bound accounts for the probability of measuring non-dominant eigenvalue bit strings, as the eigenvalue resolution does not include the tail eigenvalues.

    Thus, to discern the sufficiently dominant eigenvalues from the tail eigenvalues, we have
    \begin{equation}
        \min_{k\in\mathcal{D}} \Pr[b_k] \ge p_{\min} - \zeta > \zeta + p_{\rm tail} \ge \Pr\left[b \notin \{b_k\}_{k\in \mathcal{D}}\right],
    \end{equation}
    which gives us 
    \begin{equation}
        \zeta \in \mathcal{O}(p_{\min} - p_{\rm tail}).
    \end{equation}
    Substituting this into the phase estimation (\cref{lemQPESingle}) and multidimensional amplitude estimation (\cref{lemMAE}), where we remove the heavy hitter pre-sampling and replace the relevant bounds with a union bound over the grid resolution $\eta/\|H\|$, we obtain the bound
    \begin{equation}
        \widetilde{\mathcal{O}}\left(\frac{\|H\|}{(p_{\rm min} - p_{\rm tail})\eta}\log\frac{1}{\delta}\right)
    \end{equation}
    as shown in the theorem. The number of qubits follows from \cref{lemMAE}.

    Lastly, if we do \emph{not} assume $\eta/\|H\| \le \Delta_{\rm dom}$, then we may have multiple eigenvalues corresponding to the same bit string $b$. The upper bound for tail eigenvalue bit strings still follows,
    \begin{equation}
        \Pr\left[b \notin \{b_k\}_{k\in \mathcal{D}}\right] \le \zeta + p_{\rm tail},
    \end{equation}
    but the lower bound for sufficiently dominant eigenvalue bit strings is updated such that 
    \begin{equation}
        \Pr[b \in \{b_k\}_{k\in \mathcal{D}}] \ge 
        \sum_{j: \eta b_k \le \lambda_j \le \eta (b_k+1)} p_j - \zeta.
    \end{equation} 
    However, this can still be lower bounded by $p_{\min}$, and hence the runtime complexity remains the same as the case where $\eta/\|H\| \le \Delta_{\rm dom}$.
\end{proof}

\section{Sufficient dominant eigenstate filtering from spectral methods}
\label{appHamThresh}
In this section, we demonstrate an alternative algorithm based on spectral filtering that does not require prior knowledge of the target eigenvalue, and we discuss its shortcomings. In particular, without such knowledge, modifying the Hamiltonian-based ground-state preparation algorithm of \citet{lin2020nearoptimal} only yields a solution to the weaker task of \emph{sufficiently dominant} eigenstate preparation, where the sufficiently dominant condition is imposed on a single eigenvalue. Consequently, this Hamiltonian-based approach cannot be extended to the more general problem of dominant eigenstate preparation considered here, unless one employs our eigenprobability operator or first estimates the target eigenenergy.

To prepare the ground state without \emph{a priori} knowledge of the ground state energy, \citet{lin2020nearoptimal} first estimate it using binary amplitude estimation and binary search, locating the lowest eigenvalue by thresholding the Hamiltonian block-encoding. To adapt this to locate the dominant eigenenergy, we take inspiration from early fault-tolerant methods like \citet{lin2022heisenberg} and identify ``jumps'' in the cumulative distribution function across eigenvalues, which can be probed by thresholding the eigenvalues of the Hamiltonian. In our case, instead of identifying the first jump, which corresponds to the ground state energy, we look for the maximum jump, which corresponds to the dominant eigenenergy. 

To illustrate why this requires the eigenstate to be sufficiently dominant, we try to locate the dominant eigenenergy using binary search. We first see that in the case of perfect projectors $\Pi_\mu$ constructed from a spectral threshold $\mu$, we can measure
\begin{equation}
    a_\mu = \lVert\Pi_\mu \ket{\psi}\rVert = \sqrt{\sum_{j:E_j\le \mu} p_j}.
\end{equation}

Suppose we have the thresholds $\mu_{\rm low}$ and $\mu_{\rm high}$ such that $E_k \in [\mu_{\rm low}, \mu_{\rm high})$ with the corresponding $a_{\rm low}$ and $a_{\rm high}$. Bisecting the range to obtain $\mu_{\rm mid}$, we want to find whether $E_k < \mu_{\rm mid}$ or $E_k \ge \mu_{\rm mid}$ and using the constructed projector from the thresholds, obtain
\begin{align}
    a_{\rm mid} &= \lVert\Pi_{\mu, {\rm mid}}\ket{\psi}\rVert = \sqrt{\sum_{j:E_j< \mu_{\rm mid}} p_j}.
\end{align}
From this information, we can obtain
\begin{subequations}
\begin{align}
\sum_{j: \mu_{\rm low} \le E_j< \mu_{\rm mid}} p_j &= a_{\rm mid}^2 - a_{\rm low}^2\\
\sum_{j: \mu_{\rm mid} \le E_j< \mu_{\rm high}} p_j &= a_{\rm high}^2 - a_{\rm mid}^2.
\end{align}
\end{subequations}
However, given that we obtain a sum of eigenprobabilities, we can only determine which of the two ranges to bisect if the range containing $E_0$ is guaranteed to yield a larger sum than the other. In the worst case, this requires
\begin{equation}
    p_0 > \sum_{j \in \mathcal{S}\setminus\{0\}} p _j \Rightarrow p_0 \ge \frac{1}{2}.
\end{equation}
Hence, if binary search is used, spectral methods can only obtain the dominant eigenvalue if the eigenstate is sufficiently dominant.

One can, of course, solve this as an unstructured search problem, but given the grid size of the eigenvalues, or $\Delta$, an additional $\mathcal{O}(1/\sqrt{\Delta})$ cost is required even with the use of Grover search~\citep{grover1996fast}.

Alternatively, one can use \cref{propQPEMAE} to obtain the dominant eigenenergy first up to $\Delta$ accuracy, filter the eigenstate by applying a rectangle function around the target eigenenergy with $\mathcal{O}(\Delta)$ width, and apply amplitude amplification. This would reduce the runtime of \cref{propQPEprep} by a logarithmic factor due to the reduced accuracy for eigenenergy estimation, but it still maintains the high qubit dependency.

\section{Proofs for the twirling superoperator}
\label{appTwirl}
\subsection{Approximation via continuous time twirling}
\label{appInt}
We build on results from \citet{bako2026exponential}, which apply random time evolutions to implement an average mixed state that is nearly diagonal in the energy eigenbasis. We generalise their results to obtain a twirling superoperator and bound the spectral norm of the error $\|\mathcal{E}\|$ in the corresponding density matrix. To achieve this, we construct a continuous integral kernel rather than the summation-based kernel in \cref{eq:kernel}:
\begin{equation}
    K_{\mathrm{G}}(\lambda_k,\lambda_\ell) = \int_{-\infty}^{\infty} G(t) f_t(\lambda_k) f_t(\lambda_{\ell})^*\,\mathrm{d}t.
    \label{eq:gaussian_kernel}
\end{equation}
$K_G$ relies on the continuous set of functions $f_t(\lambda)=e^{-it\lambda}$ and the continuous set of Gaussian weights as $G(t) = (\sqrt{2\pi}\sigma)^{-1} e^{-t^2/(2\sigma^2)}$. Using this, we can construct the twirling superoperator
\begin{equation}
    \widetilde{\mathcal T}_{\rm time}^{\rm (int)}(\,\cdot\,) = \int_{-\infty}^{\infty} G(t) \,  e^{-iHt} [\,\cdot\,] e^{iHt}  \mathrm{d}t.
\end{equation}
Formally, this superoperator is equivalent to the time averaging in Lemma~1 of Ref~\cite{bako2026exponential}, which draws evolution times $t$ from a Gaussian distribution $\mathcal{N}(0,\sigma)$ and results in the time-averaged, twirled density matrix as
\begin{equation}
    \ravg =  \int_{-\infty}^{\infty} G(t) \,  e^{-i Ht} \ketbra{\psi}{\psi} e^{iHt}  \mathrm{d}t.
    \label{eq:twirl_gaussian}
\end{equation}
We derive the required Gaussian bandwidth for sufficient twirling in the following lemma.
\begin{lemma}[Continuous time twirling approximation]
    Suppose the spectral gap satisfies $|\lambda_k-\lambda_\ell|\ge\Delta$ for all $k \ne \ell \in \mathcal{S}$. To approximate the ideal twirling superoperator on $\ketbra{\psi}{\psi}$ to a fixed error $\mathcal{E}$, that is, $\left\lVert\mathcal T(\ketbra{\psi}{\psi}) - \widetilde{\mathcal T}_{\rm time}^{\rm(int)}(\ketbra{\psi}{\psi})\right\rVert \le \varepsilon$, then the required bandwidth of the Gaussian distribution scales as
    \begin{equation*}
        \sigma \in \Theta\left(\frac{1}{\Delta}\sqrt{\log\frac{1}{\mathcal{\|E\|}}}\right)
    \end{equation*}
    \label{lemma:ravg}
\end{lemma}
\begin{proof}
    Provided that we have $\rho = \mathcal{T}(\ketbra{\psi}{\psi})$, and $\bar\rho = \widetilde{\mathcal{T}}_{\rm time}^{\rm (int)}(\ketbra{\psi}{\psi})$, we can see that
    \begin{equation}
        \|\mathcal{E}\| \le \|\mathcal{E}\|_{\rm F} \le \exp\left(-\frac{\Delta^2\sigma^2}{2}\right)
    \end{equation}
    where the latter bound is obtained by evaluating off-diagonal terms of the continuous Gaussian kernel in~\cref{eq:gaussian_kernel}, as shown in the upper bound of the Hilbert-Schmidt/Frobenius norm from \citet[Lemma~1]{bako2026exponential}. This recovers the bandwidth $\sigma$ upper bound in the lemma.
    
\end{proof}

\subsection{Approximation via discrete-time twirling}
\label{appTime}
Here, we discuss the implementation of the discrete-time twirling superoperator by approximating the continuous-time twirling introduced in the previous subsection. In particular, we can approximate the Gaussian integral kernel with a Gaussian summation kernel using quadrature formulas. 

We review the cost and error dependency using two techniques: the Gauss--Hermite quadrature and the simple truncated Riemann sum. We show that the truncated Riemann sum is more efficient in physically relevant systems, and we rely on this equispaced quadrature formula for constructing the block-encoding of the eigenprobability operator.

\subsubsection{Gauss--Hermite quadrature}
The most natural quadrature formula approximating a Gaussian integral is the Gauss--Hermite quadrature, which in its standard form is defined as
\begin{equation}
    \int\limits_{-\infty}^{\infty} e^{-t^2} f(t) \mathrm dt \approx \sum_{j=0}^{M-1} w_j f(x_j),
\end{equation}
where $f(x)$ is a function, $x_j$ are roots of Hermite polynomials, and the weights $w_j$ are calculated using the images of these roots. Crucially, both the roots and weights are computable classically in $\mathcal{O}(M)$ time.

Using a change of variables, it is straightforward to extend this formulation to general Gaussian window functions $G(t) = \frac{1}{\sqrt{2\pi}\sigma} \exp(\frac{-t^2}{2\sigma^2})$ by scaling the weights and shifting the roots as
\begin{align}
    I & = \int\limits_{-\infty}^{\infty} G(t) f(t) \mathrm dt 
     \approx \sum_{j=0}^{M-1} \frac{w_j}{\sqrt{\pi}} f(\sqrt{2} \sigma x_j) \coloneqq I_M(f). 
\end{align}
The error of this generalised quadrature can be written as
\begin{equation}
    E_M = I-I_M = \frac{M! \sigma^{2M}}{(2M)!} f^{(2M)}(\xi),
\end{equation}
where $f^{(2M)}(\xi)$ is the $(2M)$-th derivative of the function at some $\xi \in (-\infty, \infty)$.

Applying the Gauss--Hermite quadrature sum instead of continuous Gaussian integration, we find that, setting $t_j = \sqrt{2}\sigma x_j$, we can construct a discrete time-twirling superoperator
\begin{equation}
    \widetilde{\mathcal T}_{\rm time}^{\rm (GH)}(\,\cdot\,) = \sum_{j=0}^{M-1} \frac{w_j}{\sqrt{\pi}}  e^{-iHt_j} [\,\cdot\,] e^{iHt_j}.
\end{equation}
Let $\ket{\psi(t)} := e^{-iHt}\ket{\psi}$. Then
\begin{equation}
\frac{\mathrm d}{\mathrm dt}\bketbra{\psi(t)}{\psi(t)} = -i\mathcal{L}\left(\bketbra{\psi(t)}{\psi(t)}\right),
\end{equation}
where $\mathcal{L}(\,\cdot\,):=[H,\,\cdot\,]$ is a Liouvillian superoperator and repeated differentiation yields
\begin{equation}
\frac{\mathrm d^{2M}}{\mathrm dt^{2M}}\left(\bketbra{\psi(t)}{\psi(t)}\right) = (-1)^M \mathcal{L}^{2M} \left(\bketbra{\psi(t)}{\psi(t)}\right).
\end{equation}

We can then obtain the difference between our time evolved density operators $\bar{\rho} = \mathcal{T}_{\rm time}^{\rm (int)}(\ketbra{\psi}{\psi})$, and $\widetilde{\rho}_M = \mathcal{T}_{\rm time}^{\rm (GH)}(\ketbra{\psi}{\psi})$, such that we have the error operator
\begin{equation}
    \mathcal{E}_M:= \bar{\rho} - \widetilde{\rho}_M = \frac{M! \sigma^{2M}}{(2M)!} (-1)^{M}\mathcal{L}^{2M} \left(\bketbra{\psi(\xi)}{\psi(\xi)}\right),
\end{equation}
for some $\xi\in(-\infty,\infty)$. We can bound the spectral norm of this error operator as
\begin{equation}
    \|\mathcal{E}_M \| \leq \frac{M! \sigma^{2M}}{(2M)!} (2\|H \|)^{2M},
\end{equation}
where we used the triangle inequality, the submultiplicativity of the spectral norm, and $\|\ketbra{\psi(t)}{\psi(t)}\| = 1$.

From this error bound, we can also derive the required number of quadrature points for a given final accuracy $\varepsilon$.

\begin{lemma}[Gauss--Hermite quadrature cost]
    To approximate the Gaussian integral 
    \begin{equation*}
        \ravg =\int_{-\infty}^{\infty} \frac{1}{\sqrt{2\pi}\sigma} e^{-t^2/2\sigma^2} e^{-iHt} \ketbra{\psi}{\psi} e^{iHt} \, \mathrm{d}t
    \end{equation*}
    with fixed accuracy $\varepsilon = \|\ravg - \widetilde{\rho}_M \|$ using the Gauss--Hermite quadrature, it suffices to use
    \begin{equation*}
    M \in \widetilde{\mathcal O}\left( (\sigma \|H \|)^2 + \log(1/\varepsilon)\right)
    \end{equation*}
    quadrature points, where $\widetilde{\rho}_M$ denotes the density operator obtained with $M$ quadrature points, and the maximal evolution time is bounded as  $T\in\widetilde{\mathcal{O}}\left(\sigma^2 \|H\| + \sigma\sqrt{\log(1/\varepsilon)}\right)$.
    \label{lemma:quadrature}
\end{lemma}

\begin{proof}
    First, let us prove the accuracy scaling of $M$. Using Robbins' bounds, we can write
    \begin{align}
        \frac{M!}{(2M)!} & < \frac{\sqrt{2\pi M}(M/e)^M \exp{(1/12M)}}
        {\sqrt{4\pi M}(2M/e)^{2M} \exp{(1/(24M+1))}} \nonumber\\
        & < \left(\frac{e}{4M}\right)^M.
    \end{align}

    Substituting this into the inequality for the accuracy, we see that it is sufficient to require $\varepsilon \ge (e \sigma^2 \|H\|^2/M)^M$, which we can rearrange as
    \begin{equation}
        \frac{\log (1/\varepsilon)}{e \sigma^2 \|H\|^2} \le \frac{M}{e \sigma^2 \|H\|^2} \log \left( \frac{M}{e \sigma^2 \|H\|^2}\right),
    \end{equation}
    and solved using the principal branch of the Lambert $W$ function $W_0$ as
    \begin{equation}
        M \ge \frac{\log(1/\varepsilon)}{W_0(\log(1/\varepsilon)/(e \sigma^2 \|H\|^2))}.
    \end{equation}
    Using the bound $W_0(x)> \log(x)-\log\log(x)$, we get the leading-order approximation

    \begin{equation}
        M \in \Omega\left( \frac{\log(1/\varepsilon)}{\log\log(1/\varepsilon) - 2\log(\sigma \|H\|)}\right).
    \end{equation}

    Now that we have shown the accuracy scaling of $M$, we need to take into account the physical constraint. Namely, for the error to be in the convergent regime (where adding quadrature points actually improves the results), we need the error sequence to monotonically decrease as
    \begin{equation}
        \frac{\|\mathcal{E}_M \|}{\|\mathcal{E}_{M+1} \|} = \frac{2M+1}{2(\sigma \|H \| )^2} > 1,
    \end{equation}
    which puts the condition on the number of quadrature points
    \begin{equation}
        M > \frac{2(\sigma \|H \|)^2 -1}{2}.
    \end{equation}
    
    Finally, since these resource bounds are additive, it suffices to use 
    \begin{equation}
        M \in \mathcal{O}\left(\left(\sigma \|H\|\right)^2 + \frac{\log(1/\varepsilon)}{\log\log(1/\varepsilon) - 2\log(\sigma \|H\|)}\right)
    \end{equation}
    quadrature points.

    From this bound, we can also obtain the maximal evolution time. Since the roots of the physicist's version of the Hermite polynomials are bounded as $|x_j|<\sqrt{2M+1}$, we can bound the maximal evolution time from the quadrature points as $T < \sigma \sqrt{4M+2}$. Consequently, the maximal evolution time is bounded as 
    \begin{equation}
        T \in \widetilde{\mathcal{O}}\left(\sigma^2 \|H\| + \sigma\sqrt{\log(1/\varepsilon)}\right).
    \end{equation}

\end{proof}

\subsubsection{Truncated Riemann sum}

Next, we discuss the simple Riemann sum, where the integral is truncated at $[-T, T]$ and discretised with a stepsize $\tau$. In this case, the required number of points $M = 2T/\tau$ is bounded in the following lemma. Using this, we can produce an approximate twirling superoperator based on Riemann quadratures
\begin{equation}
    \widetilde{\rho}_M = \widetilde{\mathcal T}_{\rm time}(\ketbra{\psi}{\psi})
\end{equation}
as defined in \cref{defTime}.

\begin{lemma}[Riemann sum cost]
    Approximating the Gaussian integral 
    \begin{equation}
        \ravg =\int_{-\infty}^{\infty} \frac{1}{\sqrt{2\pi}\sigma} e^{-t^2/2\sigma^2} e^{-iHt} \ketbra{\psi}{\psi} e^{iHt} \, \mathrm{d}t
    \end{equation}
    with fixed accuracy $\varepsilon = \|\ravg - \widetilde{\rho}_M \|$ using a simple Riemann sum on a truncated interval $[-T, T]$, requires 
    \begin{equation}
    M \in \mathcal{O} \left(\sigma \|H\| \sqrt{\log\frac{1}{\varepsilon}} + \log\frac{1}{\varepsilon}\right).
    \end{equation}
    points, where $\widetilde{\rho}_M$ denotes the density operator obtained with $M$ Riemann terms, and $T\in\mathcal{O}(\sigma \sqrt{\log(1/\varepsilon)})$
    \label{lemma:riemann}
\end{lemma}

\begin{proof}
    From the continuous time-twirling, we want to further use a discretised and truncated finite sum to approximate the Gaussian integral. In particular, our goal is to have a finite sum, for some time step $\tau$, in the form of 
    \begin{equation}
        \sum_{j = -M/2}^{M/2-1} \frac{\tau G(j\tau)}{\mathcal{Z}} e^{-t^2/2\sigma^2} e^{-ijH\tau} \ketbra{\psi}{\psi} e^{ijH\tau},
    \end{equation}
    where $G(t) = (\sqrt{2\pi\sigma})^{-1} \exp(-t^2/2\sigma^2)$ is the Gaussian function, and $\mathcal{Z}$ is the normalisation factor (such that the Gaussian weights are $\ell_1$ normalised). Note that here, the values of $j$ are taken such that when $M=2^m$ is a power of two (to enable simpler implementation on a binary representation), the value of $j$ can be implemented on $m$ bits with two's complement.

    The sources of error in this approximation are the truncation, discretisation, and normalisation. Using the triangle inequality, the total error is bounded as 
    \begin{equation}
        \mathcal{E}_{\rm total} \leq \mathcal{E}_{\rm norm} + \mathcal{E}_{\rm trunc}+ \mathcal{E}_{\rm disc}.
    \end{equation}
    Let us truncate the continuous interval to the time window $[-T, T]$. This enables us to determine the stepsize $\tau$ and the number of quadrature points $M$ afterwards. 
    
    We can write the normalisation constant as
    \begin{equation}
        \mathcal{Z} = \sum_{j = -M/2}^{M/2-1} \tau G(j\tau).
    \end{equation}
    The corresponding normalisation error is
    \begin{align}
        \mathcal{E}_{\rm norm} & = |\widetilde{\rho}_M-\mathcal{Z}\widetilde{\rho}_M| \le |1-\mathcal{Z}| \lVert\widetilde{\rho}_M\rVert = |1-\mathcal{Z}| \nonumber \\
        & = \left| \int_{-\infty}^\infty G(t) \mathrm dt -  \sum_{j = -M/2}^{M/2-1} \tau G(j\tau)\right|.
    \end{align}
    where we use the fact that $\widetilde{\rho}_M$ is normalised such that $\|\widetilde{\rho}_M\|= 1$ and  $\int_{-\infty}^{\infty}G(t)\,\mathrm{d}t=1$. This equation shows that $|1 - \mathcal{Z}|$ is bounded by its own truncation and discretisation errors. Crucially, the error for the pure Gaussian is strictly smaller than the error for the operator. Therefore, it is perfectly safe to bound it as
    \begin{equation}
        \mathcal{E}_{\rm norm} \leq \mathcal{E}_{\rm trunc}+ \mathcal{E}_{\rm disc}.
    \end{equation}
    
    Without loss of generality, we require that the truncation and discretisation errors are less than $\varepsilon/4$:
    \begin{align}
        \mathcal{E}_{\rm total} \leq 2 (\mathcal{E}_{\rm trunc}+ \mathcal{E}_{\rm disc}) \leq 2(\varepsilon/4 + \varepsilon/4) = \varepsilon.
    \end{align}

     First, let us consider the truncation error. The Gaussian tail can be bounded as
    \begin{align}
        \mathcal{E}_{\rm trunc} & = \left\| \tau \sum_{j \notin [-M/2,M/2-1]} G(j\tau)\bketbra{\psi(j\tau)}{\psi(j\tau)}\right\|\nonumber\\ 
        &\leq 2\tau\sum_{j=M/2}^{\infty} G(j\tau) \leq 2\int\limits_{T}^{\infty} G(t) \mathrm dt \nonumber\\
        &\leq \frac{2\sigma}{T\sqrt{2\pi}} e^{-\frac{T^2}{2\sigma^2}},
    \end{align}
    where we used the fact that the Gaussian is strictly decreasing on the given interval and set $T = (M/2 -1) \tau$, and, in the last line, we used Mills' inequality~\cite{gordon1941values,rigollet2023high}. By requiring these error bounds to be $\leq \varepsilon/4$, we get the sufficient condition
    \begin{align}
        T^{\star} & = \sigma \sqrt{\frac{1}{2} W_0 \left( \frac{8}{\varepsilon^2 \pi} \right)} \nonumber\\
        & \approx \sigma \sqrt{\frac{1}{2} \left( \log \frac{8}{\varepsilon^2 \pi} - \log\log \frac{8}{\varepsilon^2 \pi} \right) }. 
    \end{align}

    Next, we bound the stepsize $\tau$. By the Riemann rule
    \begin{align}
        \int\limits_{-\infty}^{\infty} g(t) \mathrm dt\approx \tau \sum_{k=-\infty}^{\infty} g(k\tau) = \sum_{m=-\infty}^{\infty} \hat{g}(2\pi m /\tau),
    \end{align}
    where for ease of notation, we used $g(t) = G(t)\bketbra{\psi(t)}{\psi(t)}$ and used the Poisson formula with the Fourier transform $\hat{g}$. Furthermore, we know that the integral of interest is the Fourier transform at frequency $0$: $\hat{g}(0) = \int_{-\infty}^{\infty} g(t) \mathrm dt$. From this, the discretisation error can be bounded as 
    \begin{align}
        \mathcal{E}_{\rm disc} & \leq \sum_{m\neq 0} \| \hat{g}(2\pi m /\tau) \| \nonumber\\
        & \leq 2 \sum_{m=1}^{\infty} \frac{1}{\sqrt{2}} \exp\left( -\frac{\sigma^2}{2} \left(\frac{2\pi m}{\tau} - 2 \| H \|\right)^2 \right)
    \end{align}
    Since the leading term is at $m=\pm1$, let us bound this as 
    \begin{equation}
        \| \hat{g}(1) \| \leq \frac{1}{\sqrt{2}} \exp\left( -\frac{\sigma^2}{2} \left(\frac{2\pi}{\tau} - 2 \| H \|\right)^2 \right) \leq \frac{\varepsilon}{4}.
    \end{equation}
    From this, we obtain a sufficient condition on the stepsize
    \begin{equation}
        \tau^{\star} \approx \frac{2\pi \sigma}{\sqrt{2\log(4/\varepsilon)} +2\sigma \| H \|}
    \end{equation}

    Finally, the required number of terms is bounded as
    \begin{align}
        M & = 2\left(\frac{T^{\star}}{\tau^{\star}} + 1\right) \nonumber    \\
        & \leq \frac{1}{\pi} \sqrt{\log\frac{8}{\varepsilon^2 \pi}\log\frac{4}{\varepsilon}} + \frac{\sqrt{2}}{\pi}\sigma \|H\| \sqrt{\log\frac{8}{\varepsilon^2 \pi}} + 2
    \end{align}

    This gives the complexity bound
    \begin{equation}
        M \in \mathcal{O} \left(\sigma \|H\| \sqrt{\log\frac{1}{\varepsilon}} + \log\frac{1}{\varepsilon}\right).
    \end{equation}
    
\end{proof}

While the dependence on the final accuracy is worse than in the case of the Gauss--Hermite quadrature, it provides better scaling in the physical parameters.

We can now prove our main result for time-twirling:
\thmTimeTwirl*
\begin{proof}
    From \cref{lemma:riemann}, we can see that the quadrature points required is $M \in \mathcal{O} \left(\sigma \|H\| \sqrt{\log(\varepsilon^{-1})} + \log(\varepsilon^{-1})\right)$. Further, from \cref{lemma:ravg}, we can see that the Gaussian bandwidth $\sigma \in \mathcal{O}\left(\Delta^{-1}\sqrt{\log(\varepsilon^{-1})}\right)$, and we can recover the scaling of the number of quadrature points in the main theorem. 
\end{proof}

Additionally, we can see that by applying \cref{lemma:ravg}, the stepsize can be found to be 
\begin{align}
    \tau^{\star} &\approx \frac{2\pi \sigma}{\sqrt{2\log(4/\varepsilon)} +2\sigma \| H \|} \nonumber\\
    &\in \mathcal{O}\left(\frac{\Delta^{-1}\sqrt{\log(\varepsilon^{-1})}}{\sqrt{\log(\varepsilon^{-1})} +\| H \|\Delta^{-1}\sqrt{\log(\varepsilon^{-1}}) }\right)\nonumber\\
    &\subseteq \mathcal{O}\left(\|H\|^{-1}\right).
\end{align}

To compare with the Gauss--Hermite quadrature results,  we can combine \cref{lemma:quadrature} and \cref{lemma:ravg} to see that the number of Gauss--Hermite quadrature points scales as
\begin{equation}
    M \in \widetilde{\mathcal{O}}\bigg(\left(\frac{\|H\|}{\Delta}\right)^2 \log\frac{1}{\varepsilon}\bigg).
\end{equation}

This shows that the truncated Riemann sum is more efficient than the Gauss--Hermite quadrature in physically relevant systems, i.e., where $\Delta\ll 1$ and $\|H\|>1$, an observation consistent with the practical applicability of the Gauss--Hermite approach~\cite{trefethen2022exactness}. Therefore, in the following, we will rely on the truncated Riemann sum, where the quadrature points are all separated by a constant stepsize $\tau$, a property that enables efficient preparation of the diagonal mixed state.

\subsection{Approximation via Chebyshev twirling}
\label{appCheb}
We now discuss the alternative twirling method based on Chebyshev polynomials instead of time evolution.
As the eigenvalues are not necessarily bounded in the range of $[-1, 1]$, to apply the Chebyshev polynomial to the eigenvalues of $H$, we need to normalise the Hamiltonian with some normalisation factor $\alpha$.

\lemCheb*
\begin{proof}
Let $\gamma_j = \arccos(\tfrac{\lambda_j}{\alpha})$. We expand the Chebyshev-twirled density matrix as follows:
\begin{align}
    \widetilde{\rho}_M &= \sum_{j=0}^{M-1} 2w_j T_j(\tfrac{H}{\alpha})\ketbra{\psi}{\psi}T_j(\tfrac{H}{\alpha})\nonumber\\
    &= \sum_{j=0}^{M-1} 2w_j \cos(j \arccos(\tfrac{H}{\alpha}))\ketbra{\psi}{\psi}\cos(j \arccos(\tfrac{H}{\alpha}))\nonumber\\
    &= \sum_{j=0}^{M-1} 2w_j \sum_{k,\ell\in\mathcal{S}}c_kc_\ell^*\cos(j\gamma_k)\cos(j\gamma_\ell)\ketbra{\psi_k}{\psi_{\ell}}\nonumber\\
    &= \begin{multlined}[t]
        \sum_{j=0}^{M-1} w_j \sum_{k,\ell\in\mathcal{S}}c_kc_\ell^*\cos(j(\gamma_k-\gamma_\ell))\ketbra{\psi_k}{\psi_{\ell}} \\
    + \sum_{j=0}^{M-1} w_j \sum_{k,\ell\in\mathcal{S}}c_kc_\ell^*\cos(j(\gamma_k+\gamma_\ell))\ketbra{\psi_k}{\psi_{\ell}}.
    \end{multlined}
\end{align}
Let $w_j$ be the truncated folded discrete Gaussian distribution such that
\begin{equation}
w_j=\begin{cases}
\dfrac{1}{\mathcal{Z}}, & j=0,\\[1ex]
\dfrac{2}{\mathcal{Z}}\exp\left(-\dfrac{j^2}{2\sigma^2}\right), & 1\le j\le M-1.
\end{cases}
\end{equation}
From the Riemann sum approximation errors derived in \cref{lemma:riemann}, we can see that approximating the discrete Gaussian sum with the Gaussian integral produces three sources of error: the normalisation error $\mathcal{E}_{\rm norm}$, the truncation error $\mathcal{E}_{\rm trunc}$, and the discretisation error $\mathcal{E}_{\rm disc}$. In particular, we can inherit the analysis of $\mathcal{E}_{\rm norm}$ and $\mathcal{E}_{\rm trunc}$ directly. For the discretisation error, the Poisson formula is used to characterise the Fourier transform of the Gaussian, but since the Gaussian is even, the cosine transform is equivalent to the Fourier transform and thus we can also use the bounds on $\mathcal{E}_{\rm disc}$ by replacing $\|H\|$ with $\lVert\arccos(H / \alpha)\rVert <\pi$. This suggests the number of quadrature points $M$ can be bounded such that 
\begin{equation}
    M \in \mathcal{O} \left(\sigma \sqrt{\log\frac{1}{\varepsilon}} + \log\frac{1}{\varepsilon}\right).
\end{equation}

Thus, approximating the (Riemann) summation by an integral, we have
\begin{equation}
    \ravg = \begin{multlined}[t]
        \int_{- \infty}^{\infty} G(t) \sum_{k,\ell\in\mathcal{S}}c_kc_\ell^*\cos((\gamma_k-\gamma_\ell)t)\ketbra{\psi_k}{\psi_{\ell}} \mathrm{d}t\\
        + \int_{- \infty}^{\infty} G(t)  \sum_{k,\ell\in\mathcal{S}}c_kc_\ell^*\cos((\gamma_k+\gamma_\ell)t)\ketbra{\psi_k}{\psi_{\ell}} \mathrm{d}t.
        \end{multlined}
\end{equation}
While the first term encodes $\gamma_k -\gamma_\ell$ and can be related to the spectral gap, the second term is an aliasing term stemming from the even parity of the Chebyshev basis, which requires an additional assumption to be suppressed.

We first show that the first term is approximately diagonal if $\sigma$ is sufficiently large. Evaluating the Gaussian integral, the first term can be written as 
\begin{align}
&\int_{- \infty}^{\infty} G(t) \sum_{k,\ell\in\mathcal{S}}c_kc_\ell^*\cos((\gamma_k-\gamma_\ell)t)\ketbra{\psi_k}{\psi_{\ell}} \,\mathrm{d}t\nonumber\\
& = \sum_{k,\ell\in\mathcal{S}}c_kc_\ell^*\exp\left(-\frac{(\gamma_k-\gamma_\ell)^2 \sigma^2}{2}\right)\ketbra{\psi_k}{\psi_{\ell}}\nonumber\\
& \begin{multlined}[b]
    =\sum_{k\in\mathcal{S}}p_k\ketbra{\psi_k}{\psi_k}\\
     + \sum_{k,\ell\in\mathcal{S}}c_kc_\ell^*\exp\left(-\frac{(\gamma_k-\gamma_\ell)^2 \sigma^2}{2}\right)\ketbra{\psi_k}{\psi_{\ell}}
\end{multlined}
\end{align}
where the off-diagonal terms can be suppressed by setting
\begin{equation}
    \sigma \in \mathcal{O}\left(\frac{1}{\Delta_{\gamma}} \sqrt{\log\frac{1}{\varepsilon}}\right),
\end{equation}
where $\Delta_{\gamma} = \min_{k\ne \ell \in \mathcal{S}} |\gamma_k - \gamma_\ell|$. As $\gamma_k = \arccos(\lambda_k/\alpha)$, from the Lipschitz bound of $\cos$, we have
\begin{equation}
    \lVert x-y\rVert \le \lVert\arccos(x)-\arccos(y)\rVert.
\end{equation}
This implies $\Delta /\alpha \le \Delta_\gamma$ and thus 
\begin{equation}
    \sigma \in \mathcal{O}\left(\frac{\alpha}{\Delta} \sqrt{\log\frac{1}{\varepsilon}}\right).
\end{equation}
Hence, the number of quadrature points is
\begin{equation}
    M \in \mathcal{O}\left(\frac{\alpha}{\Delta} \log\frac{1}{\varepsilon}\right).
\end{equation}

Lastly, we derive a sufficient condition for suppressing the aliasing term. Let us introduce the $\kappa$ such that $\forall k, \gamma_k \ge \kappa$. Then, by bounding the spectral norm of the aliasing term, we can obtain 
\begin{align}
&\left\lVert\sum_{k,\ell\in\mathcal{S}}\frac{c_kc_\ell^*}{2}\ketbra{\psi_k}{\psi_{\ell}}\int_{- \infty}^{\infty} G(j)  \cos(j(\gamma_k+\gamma_\ell)) \mathrm{d}t\right\rVert\nonumber\\
&=\left\lVert\sum_{k,\ell\in\mathcal{S}}\frac{c_kc_\ell^*}{2}\ketbra{\psi_k}{\psi_{\ell}}\exp\left(-\frac{(\gamma_k+\gamma_\ell)^2 \sigma^2}{2}\right)\right\rVert\nonumber\\
&\le\left\lVert\sum_{k,\ell\in\mathcal{S}}\frac{c_kc_\ell^*}{2}\ketbra{\psi_k}{\psi_{\ell}}\exp\left(-2\kappa^2\sigma^2\right)\right\rVert\nonumber\\
&\le\left\lVert\sum_{k,\ell\in\mathcal{S}}\frac{c_kc_\ell^*}{2}\ketbra{\psi_k}{\psi_{\ell}}\exp\left(-2\kappa^2\sigma^2\right)\right\rVert_{\rm F}\nonumber\\
&=\sqrt{\sum_{k,\ell\in\mathcal{S}}\frac{1}{4}|c_k|^2|c_\ell|^2\exp\left(-4\kappa^2\sigma^2\right)}\nonumber\\
&\le\frac{1}{2}\exp\left(-2\kappa^2\sigma^2\right)
\end{align}
Thus, if $\kappa \in \Omega(\sigma^{-1} \sqrt{\log(\varepsilon^{-1})}) \subseteq \Omega(\Delta/\alpha)$, then the aliasing term can be suppressed. Plugging in the bounds on $\sigma$, we find that 
we require
\begin{equation}
    \arccos\left(\frac{\|H\|}{\alpha}\right) \in \Omega \left(\frac{\Delta}{\alpha}\right).
\end{equation}
Suppose we have, for some constant $C$, 
\begin{equation}
\arccos\left(\frac{\|H\|}{\alpha}\right) \ge C \frac{\Delta}{\alpha}
\end{equation}
Then we can show that
\begin{equation}
\frac{\|H\|}{\alpha} \le \cos\left(\frac{C\Delta}{\alpha}\right) \le 1 - \frac{2}{\pi^2}\left(\frac{C\Delta}{\alpha}\right)^2.
\end{equation}
Solving for $\alpha$, we have
\begin{equation}
\alpha \ge \frac{\|H\|+\sqrt{\|H\|^2+8C^2\Delta^2/\pi^2}}{2},
\end{equation}
which, upon expansion, we see that 
\begin{equation}
\alpha \ge \|H\| \left(1 + \mathcal{O}\left(\frac{\Delta^2}{\|H\|^2}\right)\right)
\end{equation}
is required to suppress the aliasing term.

We can also note that if $\alpha \ge 2\|H\|$, then 
\begin{equation}
    \arccos\left(\frac{\|H\|}{\alpha}\right) \ge \frac{\pi}{3} > 1 \ge \frac{2\|H\|}{\alpha} \ge \frac{\Delta}{\alpha},
\end{equation}
which satisfies the requirement to suppress the aliasing term when we set a sufficiently large $\sigma$.
\end{proof}

\section{Construction of the eigenprobability operator}

\subsection{The eigenprobability operator from time-twirling}
\label{appRho}
In \cref{secTimeTwirl}, we have established that the nearly-diagonal density operator can be approximated using quadrature formulas as
\begin{equation}
    \ravg = \mathbb{E}_{t\sim G(t)}\left[\bketbra{\psi(t)}{\psi(t)}\right] \approx \sum_{j=-M/2}^{M/2-1} w_j \bketbra{\psi(t_j)}{\psi(t_j)},
\end{equation}
where time values $t_j$ come from a discrete set $\{t_j\}$ indexed by $j$, and the weights $w_j$ are defined by the associated quadrature formula. In the case of the Riemann sum, these weights are defined as
$ w_j = \tau G(j\tau)/\mathcal{Z}$, where $\mathcal{Z}$ is the normalizing constant accounting for the small deviation from $1$, and $t_j = j\tau$ is just the multiple of the time step $\tau$.

Importantly, since $\bketbra{\psi(t)}{\psi(t)}$ is positive semi-definite, and $\left\rVert\bketbra{\psi(t)}{\psi(t)}\right\rVert = 1$ for all $t$, we can implement its block-encoding through LCU, and by definition we have $\alpha_{\rm{LCU}} = 1$, since the quadrature weights are normalised. 
\propRho*

\begin{proof}
    From \cref{lemma:riemann_twirl}, we see that we can construct
    \begin{equation}
        \|\rho - \widetilde{\rho}_M \| \le \varepsilon
    \end{equation}
    with 
    \begin{equation}
        M \in \mathcal{O}\left(\frac{\|H\|}{\Delta}\log \frac{1}{\varepsilon}\right).
    \end{equation}

    To construct the block-encoding, we need to introduce a clock register that indexes the quadrature points and define the conditional unitary operator
    \begin{equation}
        \select = \sum_{j=-M/2}^{M/2-1} \ketbra{j}{j}_a \otimes \be\left[\bketbra{\psi(j\tau)}{\psi(j\tau)}\right],
        \label{eq:select}
    \end{equation}
    where $\ket{i}\in \mathbb{C}^M$ are basis states on $\log_2(M)$ qubits, and the latter can be obtained by an LCU between identity and
    \begin{equation}
        e^{-ijH\tau}U_\psi (2\ketbra{0}{0}-\mathbb{I}) U_\psi^\dagger e^{ijH\tau}.
    \end{equation}
    We also need to define a preparation unitary that loads the quadrature weights to the clock register
    \begin{equation}
        \prep \ket{0}_a \rightarrow \ket{G}_a \coloneqq \sum_{j=-M/2}^{M/2-1} \sqrt{w_j} \ket{j}.
    \end{equation}
    Then the block-encoding is obtained as
    \begin{align}
        \bra{G}_a \select \ket{G}_a & = \sum_{j=-M/2}^{M/2-1} w_j \be\left[\bketbra{\psi(j\tau)}{\psi(j\tau)}\right] \nonumber\\
        & = \be[\widetilde{\rho}_M]
    \end{align}

    Consequently, the state $\widetilde{\rho}_M$ can be block-encoded using $\lceil\log_2(M)\rceil$ additional qubits. 

    Finally, let us consider the cost of implementing this block-encoding. From \cref{lemma:riemann_twirl}, we see that using a time evolution operator with stepsize $\tau\in\mathcal{O}(\|H\|^{-1})$ produces a normalised time evolution operator, which we can use as the basic unit for runtime complexity. What is left to show here is that the maximal evolution time is the dominant factor in the runtime complexity. We express this directly in the proposition statement as the query complexity of $U_\tau$.
    
    Na\"ively implementing each multi-controlled term in \cref{eq:select} separately would require $M$ calls to the block-encodings. However, since the $\ket{j}$ are computational basis states, we can encode $j$ using an $m=\log_2 M$ bit two's-complement register. The corresponding power can then be implemented as
    \begin{equation}
        U_\tau^j = \prod_{k=0}^{m-2} \left(U_\tau^{2^k}\right)^{j_k} \cdot \left(U_\tau^{-2^{m-1}}\right)^{j_{m-1}},
    \end{equation}
    where the most significant bit $j_{m-1}$ encodes the sign in two's complement as shown in \cref{figMultiplex}. Consequently, the operation
    \begin{equation}
        \select = \sum_{j=-M/2}^{M/2-1} \ketbra{j}{j} \otimes U_\tau^j U_\psi \be[\ketbra{0}{0}] U_\psi^\dagger (U_\tau^\dagger)^j
    \end{equation}
    can be implemented using one ${\rm C}\mhyphen U_\tau^{2^k}$ operation for each bit of the $\prep$ register. Thus, the multiplexed implementation requires $M$ ${\rm C}\mhyphen U_\tau$ operations in total. In contrast, na\"ively implementing each of the $M$ terms separately requires $M$ number of ${\rm C}\mhyphen U_\tau^{2^k}$ operations, corresponding to $\mathcal O(M^2)$ total ${\rm C}\mhyphen U_\tau$ operations.
    
    To show that $\prep$ can be implemented efficiently, we use the fact that the quadrature weights are normalised as $\sum_{j=-M/2}^{M/2-1} w_j = 1$. Loading such a discretised Gaussian distribution can be achieved with various techniques. The Grover-Rudolph algorithm requires $\mathcal{O}(\polylog M)$ time through $\mathcal{O}(\log M)$ calls to the amplitude oracle; however, implementing coherent arithmetic requires $\mathcal{O}(\polylog M)$ ancillary qubits~\citep{grover2002creating}. Conversely, the QSVT-based approach of Ref.~\citep{mcardle2026quantum} requires only $3$ additional qubits and $\mathcal{O}(\log(1/\varepsilon)\log M)$ gates that do not change the overall asymptotic complexity of our approach. 
\end{proof}

We remark that in the case when we are given block-encoding access to $H$, we need $\widetilde{\mathcal{O}}(\alpha t)$ queries to the block-encoding to implement Hamiltonian simulation with evolution time $t$, where $\alpha\geq \| H\|$ is the scaling factor of the block-encoding. To efficiently implement different evolution times $t$, preprocessing of the phase sequences for our different times $t = j\tau$ and further fractional queries may be required. Hence, while this can be implemented in theory, a more practical option is to utilise Chebyshev-twirling as shown in the next section.

\subsection{The eigenprobability operator from Chebyshev-twirling}
\label{appRhoAlt}
Consider access to a $(\alpha, a_H, 0)$-block-encoding of the Hamiltonian $H$. After qubitisation, we obtain the walk operator
\begin{equation}
    W = \be[\tfrac{H}{\alpha}](2\ketbra{0}{0}_{a_H} - \mathbb{I})
\end{equation}
with eigenvalues $e^{\pm i\arccos(\lambda_k/\alpha)}$ and corresponding eigenstates $\ket{\psi_k}\ket{0} \pm i \ket{\perp_k}$ in the non-trivial subspace. Thus, the eigenvalues of $W^j$ can be written as 
\begin{multline}
    e^{\pm ij\arccos(\lambda_k/\alpha)} = \cos(j\arccos(\tfrac{\lambda_k}{\alpha}))\\
    \pm i\sin(j\arccos(\tfrac{\lambda_k}{\alpha}))
\end{multline}
where, upon post-selecting $\ket{0}$ on the ancilla register, we cancel out the imaginary sines, and the effective eigenvalues applied to the input eigenstate $\ket{\psi_k}$ are
\begin{equation}
    T_j(\tfrac{\lambda_k}{\alpha}) = \cos(j\arccos(\tfrac{\lambda_k}{\alpha})).
\end{equation}
This allows us to block-encode $T_j(\tfrac{H}{\alpha})$ as shown in \cref{figCheby}. Using this, we can apply Chebyshev twirling from \cref{lemCheb} to the input state, and subsequently block-encode the eigenprobability operator as follows: 
\begin{figure}
    \includegraphics{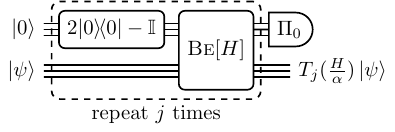}
    \caption{Block-encoding of the Chebyshev polynomial of the Hamiltonian via qubitisation.}
    \label{figCheby}
\end{figure}

\begin{figure}
    \includegraphics[width=\linewidth]{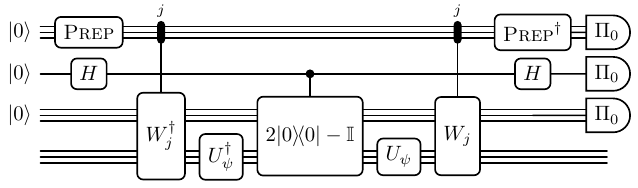}
    \caption{Block-encoding of the Chebyshev-polynomial--twirled state.}
    \label{figChebyRho}
\end{figure}

\begin{proposition}[Block-encoding of $\rho$ given Hamiltonian block-encodings]
    Given an $(\alpha, a_H, 0)$-block-encoding of $H$, we can construct $(2,a+a_H+2,\varepsilon)$-block-encoding of the $n$-qubit diagonal density operator $\rho$ with 
    \begin{equation*}
        M \in\mathcal{O}\left(\frac{\alpha}{\Delta}\log\frac{1}{\varepsilon}\right)
    \end{equation*}
    queries to $\be[\tfrac{H}{\alpha}]$, where $a = \log_2 M$.
    \label{propRhoAlt}
\end{proposition}
\begin{proof}
    From \cref{lemCheb}, we see that we can construct $\|\rho - \widetilde{\rho}_M \| \le \varepsilon$
    with 
    \begin{equation}
        M \in \mathcal{O}\left(\frac{\alpha}{\Delta}\log \frac{1}{\varepsilon}\right)
    \end{equation}
    if $\alpha \ge \min\{2\|H\|,\|H\|(1 + \mathcal O(\Delta^2/\|H\|^2))\}$. 

    We note that the latter restriction on $\alpha$ can most likely be satisfied given that producing a block-encoding with scaling factors close to $\|H\|$ is difficult in practice, and it is often the case that $\Delta \ll \|H\|$. However, in case it is not guaranteed, given that $\alpha \ge \|H\|$ is always true, we can still deliberately deflate the block-encoding of $H$ by a factor of 2 to ensure that the new scaling factor $\alpha' = 2\alpha \ge 2\|H\|$, such that the former restriction is achieved and the aliasing term is suppressed. Given the gate
    \begin{equation}
        RY\left(\frac{2\pi}{3}\right) = \begin{bmatrix} \frac{1}{2} & -\frac{\sqrt{3}}{2}\\ \frac{\sqrt{3}}{2} & \frac{1}{2}\end{bmatrix},
    \end{equation}
    we can construct a block-encoding by applying it to an extra qubit such that
    \begin{equation}
        (\bra{0}\otimes \bra{0}_{a_H}\otimes\mathbb{I})(RY(\tfrac{2\pi}{3})\otimes \be[\tfrac{H}{\alpha}])(\ket{0}\otimes \ket{0}_{a_H}\otimes\mathbb{I}) = \frac{H}{2\alpha}.
    \end{equation}
    With qubitisation of $RY\left(\frac{2\pi}{3}\right)\otimes \be[\tfrac{H}{\alpha}]= \be[\tfrac{H}{2\alpha}]$, we can produce a $(1, a_H+1, 0)$-block-encoding of $T_k(\tfrac{H}{2\alpha})$ as shown in \cref{figCheby}. Further, we can produce
    \begin{equation}
        \select = \sum_{j=0}^{M-1} \ketbra{j}{j}_{a}
        \otimes \be[T_j(\tfrac{H}{2\alpha})]\,\be[\ketbra{\psi}{\psi}]\,\be[T_j(\tfrac{H}{2\alpha})]^\dagger,
    \end{equation}
    implemented by multiplexed block-encodings of $T_j(\tfrac{H}{2\alpha})$, and again with an extra qubit for $\be[\ketbra{\psi}{\psi}]$.

    Given the $\prep$ unitary
    \begin{equation}
        \prep \ket{0} = \sum_{j=0}^{M-1} \sqrt{w_j} \ket{j}
    \end{equation}
    where $w_j$ is the folded Gaussian distribution in \cref{defCheb}, we can produce the block-encoding implementation of the twirled operator
    \begin{equation}
        \sum_{j=0}^{M-1} w_j T_j(\tfrac{H}{2\alpha})\ketbra{\psi}{\psi}T_j(\tfrac{H}{2\alpha}) = \frac{\widetilde{\rho}_M}{2}
    \end{equation}
    as shown in \cref{figChebyRho}, which implies that the scaling factor of our block-encoding is 2. Note here that the scaling factor 2 stems from the additional multiplier of 2 in the Chebyshev twirling superoperator of \cref{defCheb} and is not a result of the dilation of the block-encoding of $H$.
    
    From \cref{lemCheb}, we note that the number of quadrature points we require is
    \begin{equation}
        M \in \mathcal{O}\left(\frac{\alpha}{\Delta}\log\frac{1}{\varepsilon}\right),
    \end{equation}
    and thus $a = \log_2 M$. Given the implementation of multiplexed block-encodings of $T_j(\tfrac{H}{2\alpha})$, the total number of queries $M$ to $\be[\tfrac{H}{2\alpha}]$ is $2M$. 

    Hence, \cref{figChebyRho} produces a $(2, a+a_H+2, \varepsilon)$-block-encoding of $\rho$.
    
\end{proof}

\section{Proofs on the square-root eigenprobability operator}
\label{appRhoSqrt}

In this section, we show two versions of the runtime cost of block-encoding the square-root eigenprobability operator.

\subsection{General proof from splitting Gaussian quadratures}
First, we show a general version of the proof that works for both Gauss--Hermite quadratures and the Riemann-sum--based Gaussian quadrature, as well as Chebyshev twirling results. This proof transfers existing bounds of the quadrature approximation errors directly to the amplified results via connecting to the closest semi-unitary problem, which we detail below.

Our goal is to provide a block-encoding implementation of $\rho_{\rm sqrt} = \sum_{k \in \mathcal{S}} \sqrt{p_k}\ketbra{\phi_k}{\psi_k}$. Now, suppose that we have a unitary $\be[\rho_{\rm sqrt}]$ acting on two sets of registers that block-encodes $\rho_{\rm sqrt}$ such that
\begin{equation}
    \rho_{\rm sqrt} = (\mathbb{I} \otimes \Pi_0) \be[\rho_{\rm sqrt}] (\Pi_0 \otimes \mathbb{I})
\end{equation}
and for some input eigenstate $\ket{\psi_k}$, we have
\begin{align}
    \rho_{\rm sqrt}\ket{\psi_k} &= (\mathbb{I} \otimes \Pi_0) \be[\rho_{\rm sqrt}] (\Pi_0 \otimes \mathbb{I}) \ket{0}\ket{\psi_k}\nonumber\\
    &= \sqrt{p_k}\ket{\phi_k}\ket{0}
\end{align}

Given that the orthonormal set $\{\ket{\phi_k}\}_{k\in\mathcal{S}}$ is ill-defined here, we define a concrete set by finding the closest orthogonal set to our output states $\{\ket{\widetilde\phi_k}\}_{k\in\mathcal{S}}$ by first writing our sets in matrix form as follows:
\begin{subequations}
\begin{align}
    \widetilde{\Phi} &= \left[\ket{\widetilde\phi_0} \;\; \ket{\widetilde\phi_1}\;\; \cdots \;\; \ket{\widetilde\phi_{|\mathcal{S}|-1}}\right],\\
    \Phi &= \left[\ket{\phi_0} \;\; \ket{\phi_1}\;\; \cdots \;\; \ket{\phi_{|\mathcal{S}|-1}}\right].
\end{align}
\end{subequations}
with $\Phi^\dagger \Phi = \mathbb{I}$. The problem then reduces to finding the closest semi-unitary matrix~\citep{higham2008functions}, which we can then set as $\Phi$, using the result as follows.

\begin{lemma}[Closest semi-unitary matrix -- Theorem 8.4, \citep{higham2008functions}]
    Let $A \in \mathbb{C}^{m \times n}$ with $m \ge n$ have a polar decomposition $A = QP$, where $Q \in \mathbb{C}^{m \times n}$ is a (semi-)unitary and $P \in \mathbb{C}^{n \times n}$ is a positive semi-definite Hermitian. Then, for any unitarily invariant norm,
    \begin{equation*}
        \|Q-A\| = \min_{B \in \mathbb{C}^{m \times n}} \|B - A\|
        \quad \text{subject to } B^\dagger B = \mathbb{I}_n.
    \end{equation*}
    If $A$ has full column rank, the minimiser is unique for the Frobenius norm.
    \label{lemProcrustes}
\end{lemma}
 
A semi-unitary polar factor $Q$ can be constructed from the singular value decomposition. In particular, let
\begin{equation}
    A = W\Sigma V^\dagger
\end{equation}
be an SVD of $A$, where $W\in\mathbb{C}^{m\times m}$ and $V\in\mathbb{C}^{n\times n}$ are unitary. Then a semi-unitary polar factor is
\begin{equation}
    Q = W\,\mathbb{I}_{m,n}V^\dagger,
\end{equation}
where $\mathbb{I}_{m,n} \in \{0,1\}^{m \times n}$ is the rectangular identity matrix $\mathbb{I}_{m,n} = 
\begin{bmatrix}
    \mathbb{I}_n \\
    O_{m-n, n}
\end{bmatrix}.$ If $A$ has full column rank, then $Q$ is the unique polar factor of $A$. This, however, does not necessarily imply that $Q$ is the unique closest semi-unitary matrix for an arbitrary unitarily invariant norm other than the Frobenius norm (e.g., the spectral norm).

Using this result, we can define the closest $\Phi$ to $\widetilde{\Phi}$ as the orthonormal set that serves as the left singular vectors of $\rho_{\rm sqrt}$.
\rhoSqrt*
\begin{proof}
    We now bound the spectral error in the block-encoding. Since $\widetilde{\rho}_{\rm sqrt}$ is constructed directly from splitting $\widetilde{\rho}$, it would not necessarily satisfy \cref{defSqrt}. Moreover, the square-root eigenprobability operator defined in \cref{defSqrt} is not unique, as it depends on the choice of the orthonormal set $\{\ket{\phi_k}\}_{k\in\mathcal{S}}$. We therefore define the spectral error as the minimum spectral distance between our provided $\widetilde{\rho}_{\rm sqrt}$ and any valid square-root eigenprobability operator,
    \begin{equation}
        \varepsilon_{\rm sqrt} := \min_{\substack{\{\ket{\phi_k}\}_{k\in\mathcal{S}}\\\braket{\phi_k|\phi_\ell}=\delta_{k\ell}}}\left|\widetilde{\rho}_{\rm sqrt}-\sum_{k\in\mathcal{S}}\sqrt{p_k}\ketbra{\phi_k}{\psi_k}\right|.
    \end{equation}
    Since $\{\ket{\psi_k}\}_{k\in\mathcal{S}}$ and $\{p_k\}_{k\in\mathcal{S}}$ are fixed by the input state $\ket{\psi}$, the minimisation depends entirely over the choice of the orthonormal set $\{\ket{\phi_k}\}_{k\in\mathcal{S}}$.

    Rather than solving this weighted optimisation directly, we consider the simpler unweighted problem of finding a closest semi-unitary matrix to $\widetilde{\Phi}$ as formed from the set $\{\ket{\widetilde\phi_k}\}_{k\in\mathcal{S}}$. The resulting orthonormal set $\{\ket{\phi_k}\}_{k\in\mathcal{S}}$ defines a valid square-root eigenprobability operator and thus provides an upper bound on $\varepsilon_{\rm sqrt}$. In particular, for this choice, we have
    \begin{align}
        \varepsilon_{\rm sqrt} &\le \lVert\widetilde{\rho}_{\rm sqrt} - \rho_{\rm sqrt}\rVert  \nonumber\\
        &= \left\lVert\sum_{k\in\mathcal{S}}\sqrt{p_k} \ketbra{\widetilde{\phi}_k}{\psi_k} - \sum_{k\in\mathcal{S}}\sqrt{p_k} \ketbra{\phi_k}{\psi_k}\right\rVert \nonumber\\
        &\le \left\lVert\sum_{k\in\mathcal{S}} \ketbra{\widetilde{\phi}_k}{k} - \sum_{k\in\mathcal{S}} \ketbra{\phi_k}{k}\right\rVert\left\lVert\sum_{k\in\mathcal{S}}\sqrt{p_k}\ketbra{k}{\psi_k}\right\rVert ,
    \end{align}
    where the second inequality follows from submultiplicativity of the spectral norm. Note that the second term can be upper-bounded by 
    \begin{equation}
        \left\lVert\sum_{k\in\mathcal{S}}\sqrt{p_k}\ketbra{k}{\psi_k}\right\rVert \le \left\lVert\sum_{k\in\mathcal{S}}\sqrt{p_k}\ketbra{k}{k}\right\rVert\left\lVert\sum_{k\in\mathcal{S}}\ketbra{k}{\psi_k}\right\rVert
        \le \sqrt{p_0}
    \end{equation}
    again by submultiplicativity and noting that $\sum_{k\in\mathcal{S}}\ketbra{k}{\psi_k}$ is semi-unitary and thus has spectral norm $1$.

    We now choose $\Phi$ to minimise this upper bound. By \cref{lemProcrustes}, a closest semi-unitary matrix to $\widetilde{\Phi}$ can be constructed from its singular value decomposition such that 
    \begin{equation}
        \widetilde{\Phi} = W \widetilde\Sigma V^\dagger
    \end{equation}
    and set 
    \begin{equation}
        \Phi = W\, \mathbb{I}_{N,|\mathcal{S}|} V^\dagger.
    \end{equation}
    Consequently, we find that
    \begin{equation}
        \lVert\widetilde{\rho}_{\rm sqrt} - \rho_{\rm sqrt}\rVert \le \sqrt{p_0} \left \lVert \widetilde{\Phi} - \Phi\right\rVert = \sqrt{p_0} \left \lVert \widetilde\Sigma -  \mathbb{I}_{N,|\mathcal{S}|}\right\rVert,
    \end{equation}
    where the last equality stems from unitary invariance. To compute this term, we need to compute the singular values of $\widetilde{\Phi}$ and provide bounds on them.

    Let $\varepsilon_w$ be the upper bound of the sum of the truncation, discretisation, and normalisation errors of \cref{lemma:riemann,lemma:riemann_twirl}. Defining the Gram matrix $\widetilde G = \widetilde{\Phi}^\dagger\widetilde{\Phi}$, we find that $\braket{k|\widetilde G|k} = 1$ and when $k\ne \ell$,
    \begin{align}
        \left\lvert\braket{k|\widetilde G|\ell}\right\rvert &=  \left\lvert\braket{\widetilde\phi_k|\widetilde\phi_\ell}\right\rvert
        =  \left\lvert\braket{\widetilde\phi_k'|\widetilde\phi_\ell'}\right\rvert\nonumber\\
        &= \left\lvert\sum_{j} w_j \exp\left(-i(E_k-E_\ell)t_j\right)\right\rvert\nonumber\\
        &\le \left\lvert\int_{-\infty}^{\infty} \frac{1}{\sqrt{2\pi}\sigma}e^{-\frac{t^2}{2\sigma}}e^{-i(E_k-E_\ell)t} \mathrm dt \right\rvert+ \varepsilon_w\nonumber\\
        & = e^{-\frac{(E_k-E_\ell)^2\sigma^2}{2}} + \varepsilon_w\nonumber\\
        & \le e^{-\frac{\Delta^2\sigma^2}{2}} + \varepsilon_w \le \varepsilon_\Delta + \varepsilon_w
    \end{align}
    where we define, for the last line, $\varepsilon_\Delta = e^{-\frac{\Delta^2\sigma^2}{2}}$.

    Substituting in the singular value decomposition of $\widetilde{\Phi}$ as $\widetilde{\Phi} = W \widetilde\Sigma V^\dagger$, the Gram matrix becomes
    \begin{equation}
        \widetilde G = \widetilde{\Phi}^\dagger \widetilde{\Phi} = V \widetilde\Sigma^\dagger W^\dagger W \widetilde\Sigma V^\dagger = V \widetilde\Sigma^\dagger \widetilde\Sigma V^\dagger.
    \end{equation}
    Thus, $\widetilde G$ is unitarily diagonalizable with eigenvalues given by the diagonal entries of $\widetilde\Sigma^\dagger \widetilde\Sigma$, i.e., the squared singular values of $\widetilde{\Phi}$.

    To calculate $\left\lVert\widetilde\Sigma -  \mathbb{I}_{N,|\mathcal{S}|}\right\rVert$, let us first calculate the eigenvalues of $\widetilde\Sigma^\dagger\widetilde\Sigma -  \mathbb{I}_{|\mathcal{S}|}$. By Gerschgorin circle theorem~\citep{horn2012matrix}, we have, for all eigenvalues, 
    \begin{equation}
        \left\lvert\lambda_k(\widetilde\Sigma^\dagger \widetilde\Sigma) -1\right\rvert \le (|\mathcal{S}|-1) (\varepsilon_\Delta + \varepsilon_w),
    \end{equation}
    where we can obtain, for the bounds of the singular values,
    \begin{subequations}
    \begin{align}
        \sigma_k(\widetilde\Sigma)  &\ge \sqrt{1 - (|\mathcal{S}|-1)  (\varepsilon_\Delta + \varepsilon_w)}\\
        \sigma_k(\widetilde\Sigma) &\le \sqrt{1+(|\mathcal{S}|-1) (\varepsilon_\Delta + \varepsilon_w)}
    \end{align}
    \end{subequations}
    Therefore, we can provide the bound 
    \begin{align}
        \left\lVert\widetilde\Sigma-\mathbb{I}\right\rVert &\le 1- \sqrt{1-(|\mathcal{S}|-1) (\varepsilon_\Delta + \varepsilon_w)}\nonumber\\
        &\le (|\mathcal{S}|-1) (\varepsilon_\Delta + \varepsilon_w)
    \end{align}
    Finally, returning to the block-encoding spectral error, we have
    \begin{equation}
        \lVert\widetilde{\rho}_{\rm sqrt} - \rho_{\rm sqrt}\rVert \le \sqrt{p_0} \left \lVert \widetilde{\Phi} - \Phi\right\rVert = \sqrt{p_0} (|\mathcal{S}|-1) (\varepsilon_\Delta + \varepsilon_w).
    \end{equation}

    We can then set 
    \begin{equation}
        \varepsilon_\Delta \le \frac{\varepsilon_{\rm sqrt}}{2\sqrt{p_0}(|\mathcal{S}|-1)},\quad\varepsilon_w \le \frac{\varepsilon_{\rm sqrt}}{2\sqrt{p_0}(|\mathcal{S}|-1)}
    \end{equation}
    to bound the entire error with $\varepsilon_{\rm sqrt}$. From
    \begin{equation}
        \varepsilon_\Delta = e^{-\frac{\Delta^2\sigma^2}{2}},
    \end{equation}
    we can obtain 
    \begin{equation}
        \sigma \in \mathcal{O}\left( \frac{1}{\Delta} \sqrt{\log\frac{\sqrt{p_0}(\lvert\mathcal{S}\rvert-1)}{\varepsilon_{\rm sqrt}}}\right).
    \end{equation}
    From \cref{lemma:riemann}, we require query complexity
    \begin{equation}
        M \in \mathcal{O}\left(\sigma \|H\| \sqrt{\log\frac{\sqrt{p_0}(\lvert\mathcal{S}\rvert-1)}{\varepsilon_{\rm sqrt}}} + \log\frac{\sqrt{p_0}(\lvert\mathcal{S}\rvert-1)}{\varepsilon_{\rm sqrt}}\right).
    \end{equation}
    where we see that we need
    \begin{equation}
        M \in \mathcal{O}\left(\frac{\|H\|}{\Delta} \log\frac{\sqrt{p_0}(\lvert\mathcal{S}\rvert-1)}{\varepsilon_{\rm sqrt}}\right)
    \end{equation}
    queries to $U_\tau$.

    Finally, given that $\Delta$ is the minimum spectral gap between eigenenergies of $H$, we note that
    \begin{equation}
        (\lvert\mathcal{S}\rvert-1)\Delta \le 2\|H\|.
    \end{equation}
    With this and $\sqrt{p_0} \le 1$, the above runtime can be further simplified to be 
    \begin{equation}
        M \in \widetilde{\mathcal{O}}\left(\frac{\|H\|}{\Delta} \log\frac{1}{\varepsilon_{\rm sqrt}}\right)
    \end{equation}
\end{proof}

\begin{remark}
    We note that
    \begin{equation*}
        M \in \mathcal{O}\left(\left(\frac{\|H\|}{\Delta}\right)^2 \log\frac{\sqrt{p_0}(\lvert\mathcal{S}\rvert-1)}{\varepsilon_{\rm sqrt}}\right)
    \end{equation*}
   queries to $U_\tau$ are required for Gauss--Hermite quadratures by applying \cref{lemma:quadrature} instead.
\end{remark}

\begin{remark}
    A similar analysis can be conducted for splitting the eigenprobability operator implemented by Chebyshev twirling and accounting for the extra aliasing term in \cref{propRhoAlt} to show that $\rho_{\rm sqrt}$ can be implemented with
    \begin{equation*}
        M \in \mathcal{O}\left(\frac{\alpha}{\Delta} \log\frac{\sqrt{p_0}(\lvert\mathcal{S}\rvert-1)}{\varepsilon_{\rm sqrt}}\right)
    \end{equation*}
    queries to $\be[\tfrac{H}{\alpha}]$, but with a scaling factor of $\sqrt{2}$ instead.
\end{remark}

We note that here we have an additional dependence on the support of the input state as opposed to the block-encoding of $\rho$.

\subsection{Alternate proof for time-twirling with discretised Gaussian weights}
Here, we provide a simpler proof that shows the latter and looser asymptotic bound in \cref{propRhoSqrt} directly. We can make use of the observation that the block-encoding implementation shares some similarities with Gaussian phase estimation~\citep{rendon2023lowdepth,rendon2024improved,chen2025quantum}.

Using the Gaussian tapered input state, the phase estimation circuit produces 
\begin{equation}
    {\rm QPE}\ket{0}\ket{\psi_k} = \left(\beta_k \ket{b_k} + \sqrt{1-\beta_k^2} \ket{\perp_k}\right)\ket{\psi_k}
\end{equation}
where $1-\beta_k \le \zeta$, and $\zeta$ is the upper bound for the failure probability of phase estimation. To obtain the block-encoding of the square-root eigenprobability operator, note that the implementation by performing SOSSA on the eigenprobability operator can be found by 
\begin{equation}
    ({\rm QFT}\otimes\bra{0}U_\psi^\dagger){\rm QPE}\ket{0}\ket{\psi_k}.
\end{equation}
We can use this to show the following proof:

\begin{proof}[Alternative proof to \cref{propRhoSqrt}]
    Given $M\in \widetilde{\mathcal{O}}(\log (1/\zeta)/\Delta)$ evolution time, phase estimation generates the state $\left(\beta_k \ket{b_k}\ket{\gamma_k} + \sqrt{1-\beta_k^2} \ket{\perp_k}\right)\ket{\psi_k}$ with $\beta_k\in\mathbb{R}$ and $1-\beta_k^2 \le\zeta$ for all $k$. Note that the complex phases are absorbed into $\ket{\gamma_k}$ and $\ket{\perp_k}$. Applying QFT on the first register to cancel out the original iQFT in phase estimation, the inverse state preparation unitary, and the two zero projectors on both sides, we can write the following:
    \begin{multline}
       ({\rm QFT}\otimes\Pi_0U_\psi^\dagger){\rm QPE}(\Pi_0\otimes\mathbb{I})\\
        = \sum_{k\in\mathcal{S}} \braket{\psi|\psi_k}\bigg(\beta_k {\rm QFT}\ketbra{b_k,\gamma_k}{0} + \sqrt{1-\beta_k^2} {\rm QFT}\ketbra{\perp_k}{0}\bigg)\\
        \otimes\ketbra{0}{\psi_k}
    \end{multline}
    where $b_k \ne b_{k'}$ for all $k' \in \mathcal{S}$. This ensures that $\ket{b_k}\ket{\gamma_k}$, and subsequently, ${\rm QFT} \ket{b_k}\ket{\gamma_k}$ is orthogonal between all $k\in\mathcal{S}$.

    Absorbing the global phase and the QFT, and ignoring the zero-projected qubits, we find that we can rewrite the above in the form
    \begin{equation}
       \widetilde{\rho}_{\rm sqrt} = \sum_{k\in\mathcal{S}} \sqrt{p_k}\left(\beta_k \ketbra{\phi_k}{\psi_k} + \sqrt{1-\beta_k^2} \ketbra{\perp_k'}{\psi_k}\right),
    \end{equation}
    whereas the ideal square-root eigenprobability operator would take the form of 
    \begin{equation}
       \rho_{\rm sqrt} = \sum_{k\in\mathcal{S}} \sqrt{p_k}\ketbra{\phi_k}{\psi_k}.
    \end{equation}
    The spectral norm error can be bounded by 
    \begin{align}
       &\lVert\widetilde{\rho}_{\rm sqrt} - \rho_{\rm sqrt}\rVert\le \lVert\widetilde{\rho}_{\rm sqrt} - \rho_{\rm sqrt}\rVert_{\rm F}\nonumber\\
       &=\left\lVert\sum_{k\in\mathcal{S}} \sqrt{p_k}\left((\beta_k-1) \ketbra{\phi_k}{\psi_k} + \sqrt{1-\beta_k^2} \ketbra{\perp_k'}{\psi_k}\right)\right\rVert_{\rm F}\nonumber\\
       &\le \sqrt{\sum_{k\in\mathcal{S}} p_k (\beta_k-1) ^2 + p_k (1-\beta_k^2)}\le \sqrt{\sum_{k\in\mathcal{S}} p_k(2-2\beta_k)}\nonumber\\
       &\le \sqrt{\sum_{k\in\mathcal{S}} p_k(2-2\beta_k^2)}\le \sqrt{2\sum_{k\in\mathcal{S}} p_k\zeta} = \sqrt{2\zeta}.
    \end{align}
    Upper bounding this by $\varepsilon_{\rm sqrt}$, and replacing $\zeta$, we recover the runtime of 
    \begin{equation}
        M \in \widetilde{\mathcal{O}}\left(\frac{\|H\|}{\Delta} \log\frac{1}{\varepsilon_{\rm sqrt}}\right)
    \end{equation}
    from Gaussian phase estimation runtimes.
\end{proof}

Using this phase-estimation--based proof, we can see that a coherent phase estimation algorithm that saturates the optimal lower bound of $\mathcal{O}(\varepsilon^{-1}\log\delta^{-1})$~\citep{mande2026tight} can be used to remove the additional logarithmic dependency on the spectral gap. This could be achieved by phase estimation combined with coherent median finding~\citep{nagaj2009fast}, but actual implementation is costly due to requiring a large number of ancillary qubits and quantum sorting networks.

It may be possible to remove the logarithmic dependency using the DPSS kernel~\citep{patel2026optimal}, which is the optimal average-case taper for phase estimation, instead of the Gaussian state. However, the current analysis only bounds the average-case error rather than the worst-case error. Although empirical results suggest that the worst-case error is comparable to the average-case error, it remains unclear whether this suffices to remove the logarithmic dependency on $\Delta^{-1}$ completely.

\section{Proof for DEFEAT and dominant eigenstate preparation}

\subsection{Proof for DEFEAT}
\label{appDefeat}
Given that our prepared block-encoding of $\rho_{\rm sqrt}$ is approximate, the rectangle function applied to it must be robust to these errors and to scale after amplitude amplification.
Our robustness analysis extends prior analysis on robust ground-state preparation~\citep{huang2026fullqubit} to the non-Hermitian setting. In addition to controlling singular value perturbations, we further bound singular vector perturbations, yielding a tighter robustness guarantee.

We first recall Wedin's $\sin\Theta$ theorem, which provides bounds on singular vector perturbations as follows: 
\begin{lemma}[Wedin's $\sin\Theta$ theorem~\citep{wedin1972perturbation}]
    
    Let $A = \sum_{j=0}^{r-1} s_j \vec{u}_j \vec{v}_j^\dagger$ be the singular value decomposition of a rank-$r$ matrix, and let $\widetilde{A} = \sum_{j=0}^{r-1}\widetilde{s}_j \widetilde{\vec{u}}_j \widetilde{\vec{v}}_j^\dagger$ be its perturbed version.
    Without loss of generality, assume that $s_0$ is separated from the other perturbed singular values and 0 by the gap
    \begin{equation*}
        \Delta = \min\left\{\min_{j\ne 0}|s_0 - \widetilde{s}_j|, s_0\right\}.
    \end{equation*}
    Then 
    \begin{equation*}
        \max\{\sin\angle(\vec u_0,\widetilde{\vec u}_0), \sin\angle(\vec v_0,\widetilde{\vec v}_0)\} \le \frac{\|A-\widetilde{A}\|}{\Delta}.
    \end{equation*}
    \label{lemWedin}
\end{lemma}

Using this, we now show the results of our main theorem.
\filter*
\begin{proof}
    Suppose that we can write the singular value decomposition of the approximate square-root eigenprobability operator in \cref{propRhoSqrt} as
    \begin{equation}
        \widetilde{\rho}_{\rm sqrt} = \sum_{k\in\mathcal{S}}\widetilde{s}_k\ketbra{\widetilde\phi_k}{\widetilde\psi_k}.
    \end{equation}

    By Weyl's inequality for singular values~\citep{horn2012matrix}, we note that the difference in singular values of the square-root eigenprobability operators can be upper-bounded as follows:
    \begin{align}
        \left\lvert\sqrt{p_k}-\widetilde{s}_k\right\rvert \le \|\rho_{\rm sqrt} - \widetilde{\rho}_{\rm sqrt}\| \le \varepsilon_{\rm sqrt}.
    \end{align}

    Let $\varepsilon_{\rm sqrt} \le \frac{\sqrt{p_0} - \sqrt{p_1}}{8}$. Then, for
    \begin{equation}
        \mu \in \left[\frac{\sqrt{p_0}+3\sqrt{p_1}}{4}, \frac{3\sqrt{p_0}+\sqrt{p_1}}{4}\right],
    \end{equation}
    we can apply the polynomial approximation $\widetilde{F}_{\mu}(x)$ of the rectangle function 
    \begin{equation}
         \begin{cases}
            \widetilde{F}_{\mu}(x) \in [-1, -1 +\xi], & \forall x \in \left[-1, -\mu - \delta\right] \cup \left[\mu + \delta, 1\right]\\
            \widetilde{F}_{\mu}(x) \in [1-\xi, 1], &\forall  x \in \left[-\mu + \delta, \mu - \delta\right],
         \end{cases}
    \end{equation}
    where $\delta = \frac{\sqrt{p_0} - \sqrt{p_1}}{8}$ to $\widetilde{\rho}_{\rm sqrt}$. We can use this to produce a reflection operator $\mathcal{R}_\mu$ up to error $\xi$. By \citet[Corollary 16]{gilyen2019quantum}, we note that the polynomial approximation should be taken to degree $D \in \mathcal O (\log(1/\xi)/\delta)$.  This produces a block-encoded approximation $\widetilde{\mathcal{R}}$ of the projector $\mathbb{I} - 2\ketbra{\widetilde\psi_0}{\widetilde\psi_0}$ within spectral error $\xi$.

    Next, to bound the difference between right singular vectors of the square-root eigenprobability operator $\rho_{\rm sqrt}$ and the approximated operator $\widetilde\rho_{\rm sqrt}$ with \cref{propRhoSqrt}, we can use Wedin's $\sin\Theta$ theorem (\cref{lemWedin}) to show the following:
    \begin{equation}
        \sin\angle\left(\ket{\psi_0}, \ket{\widetilde\psi_0}\right) = \sqrt{1-\left\lvert\braket{\psi_0|\widetilde\psi_0}\right\rvert^2} \le \frac{\|\rho_{\rm sqrt} - \widetilde{\rho}_{\rm sqrt}\|}{\sqrt{p_0}-\widetilde{s}_1}.
    \end{equation}
    Since $\ketbra{\psi_0}{\psi_0}-\ketbra{\widetilde\psi_0}{\widetilde\psi_0}$ is Hermitian with rank at most 2, and $\tr(\ketbra{\psi_0}{\psi_0}-\ketbra{\widetilde\psi_0}{\widetilde\psi_0})=0$, with $\det(\ketbra{\psi_0}{\psi_0}-\ketbra{\widetilde\psi_0}{\widetilde\psi_0})=-(1-\lvert\braket{\psi_0|\widetilde\psi_0}\rvert^2)$, we can see that its non-zero eigenvalues are $\pm\sqrt{1-\lvert\braket{\psi_0|\widetilde\psi_0}\rvert^2}$, so $\|\ketbra{\psi_0}{\psi_0}-\ketbra{\widetilde\psi_0}{\widetilde\psi_0}\|_2=\sqrt{1-\lvert\braket{\psi_0|\widetilde\psi_0}\rvert^2}$. Thus, we can show that
    \begin{align}
        \left\lVert\ketbra{\psi_0}{\psi_0}-\ketbra{\widetilde\psi_0}{\widetilde\psi_0}\right\rVert&=\sqrt{1-|\braket{\psi_0|\widetilde\psi_0}|^2} \nonumber\\
        &\le\frac{\|\rho_{\rm sqrt} - \widetilde{\rho}_{\rm sqrt}\|}{\sqrt{p_0}-\widetilde{s}_1}\le\frac{\varepsilon_{\rm sqrt}}{\sqrt{p_0}-\widetilde{\sigma}_1}\nonumber\\
        &\le\frac{\varepsilon_{\rm sqrt}}{\sqrt{p_0}-\sqrt{p_1}-\varepsilon_{\rm sqrt}}\nonumber\\
        &\le\frac{8\varepsilon_{\rm sqrt}}{7(\sqrt{p_0}-\sqrt{p_1})}
    \end{align}
    while the last inequality stems from our earlier assumption that $\varepsilon_{\rm sqrt} \le \frac{\sqrt{p_0} - \sqrt{p_1}}{8}$.
    
    Thus, the total block-encoding error can be found to be
    \begin{align}
        &\left\lVert\widetilde{\mathcal{R}} - (\mathbb{I}-2\ketbra{\psi_0}{\psi_0})\right\rVert\nonumber\\
        &\le \left\lVert\widetilde{\mathcal{R}} - (\mathbb{I}-2\ketbra{\widetilde\psi_0}{\widetilde\psi_0})\right\rVert + 2\left\lVert\ketbra{\widetilde\psi_0}{\widetilde\psi_0}-\ketbra{\psi_0}{\psi_0})\right\rVert\nonumber\\
        &\le \xi +\frac{16\varepsilon_{\rm sqrt}}{7 (\sqrt{p_0} - \sqrt{p_1})}
    \end{align}
    Upper bounding the above with $\varepsilon$, we find that 
    \begin{equation}
        \xi \in \mathcal{O}(\varepsilon), \quad \varepsilon_{\rm sqrt} \in \mathcal{O}((\sqrt{p_0} - \sqrt{p_1})\varepsilon)
    \end{equation}
    
   The total runtime complexity for preparing the dominant eigenstate is obtained by combining this with \cref{propRhoSqrt} as
    \begin{equation}
        \mathcal{O}\left(\frac{\|H\|}{(\sqrt{p_0}-\sqrt{p_1})\Delta} \log \frac{1}{\varepsilon}\log\frac{\|H\|}{(\sqrt{p_0}-\sqrt{p_1})\Delta\varepsilon}\right).
    \end{equation}

    For the number of queries to $U_\psi$, we simply have to look at the degree $D$ of the polynomial expansion, that is, $D \in \mathcal{O}(\log(1/\xi)/\delta$, which is
    \begin{equation}
        \mathcal{O}\left(\frac{1}{\sqrt{p_0}-\sqrt{p_1}}\log\frac{1}{\varepsilon}\right).
    \end{equation}

    For the number of qubits required, as QSVT only adds a single qubit, the qubit dependency comes from \cref{propRhoSqrt}, which produces
    \begin{equation}
        n + \mathcal{O}\left(\log\frac{\|H\|}{\Delta} + \log \log \frac{1}{(\sqrt{p_0}-\sqrt{p_1})\varepsilon}\right)
    \end{equation}
\end{proof}

\subsection{Filtering without \texorpdfstring{\emph{a priori}}{a priori} knowledge of the threshold}
\label{appGetMu}
In this section, we discuss some relaxations of \cref{assumptDEFEAT}. 

First, the assumption on knowledge of the minimum spectral gap, as follows, can be relaxed:
\begin{manualdef}[Assumption~\ref*{assumptGap}]
    A lower bound on the minimum spectral gap within the support $\Delta \le \min_{\substack{(j\neq k) \in\mathcal S}} |E_j-E_k|$.
\end{manualdef}
In practice, we do not need explicit knowledge of the actual lower bound of the spectral gap, as we can just time-evolve for a sufficiently long time to satisfy \hyperref[assumptGap]{Assumption~\ref*{assumptGap}}.

Further, in the event that we \emph{do not} have \emph{a priori} knowledge of the probability threshold $\mu$ as shown in \hyperref[assumptThresh]{Assumption~\ref*{assumptThresh}} below, one can still estimate the value, albeit with a longer runtime.
\begin{manualdef}[Assumption~\ref*{assumptThresh}]
    Knowledge of a threshold value $\mu$ satisfying
    \begin{equation*}
        \sqrt{p_0}-\mu,\ \mu-\sqrt{p_1}\in\Omega(\Lambda).
    \end{equation*}
\end{manualdef}
We can use the same idea for finding the ground state energy from a block-encoded Hamiltonian with binary amplitude estimation~\citep{lin2020nearoptimal}. 
We use binary search to find a suitable estimate of $\mu$ by locating $\sqrt{p_0}$ and shifting downwards using the knowledge of $\sqrt{p_0}-\sqrt{p_1}$. For each candidate $\widetilde{\mu}$, we can apply the above thresholding function $F_{\widetilde\mu}$ on $\rho_{\rm sqrt}$ to approximate the reflector $\mathcal R_{\widetilde\mu} = \mathbb{I} - 2\mathcal P_{\widetilde\mu}$, where projector $\mathcal P_{\widetilde\mu} = \sum_{k:\sqrt{p_k} \ge \mu}\ketbra{\psi_k}{\psi_k}$. Using this reflector for amplitude estimation yields $\lVert \mathcal P_{\widetilde\mu}\ket{\psi}\rVert$, which discerns whether there exists a nonzero weight above the threshold, and thus distinguishes between the cases $\sqrt{p_0}>\widetilde{\mu}$ and $\sqrt{p_0}<\widetilde{\mu}$. This procedure then provides us with a suitable $\mu$ which we can use for dominant eigenstate preparation.

\begin{lemma}[Estimating $\mu$]
    Given \cref{assumptDEFEAT} without \hyperref[assumptThresh]{Assumption~\ref*{assumptThresh}}, with success probability $1-\delta$, we can estimate $\mu$ satisfying \hyperref[assumptThresh]{Assumption~\ref*{assumptThresh}}, with 
    \begin{equation*}
        \widetilde{\mathcal{O}}\left(\frac{\|H\|}{(p_0-p_1)\Delta}\log\frac{1}{\delta}\right),
    \end{equation*}
   queries to $U_\tau$, as well as
    \begin{equation*}
        \widetilde{\mathcal{O}}\left(\frac{1}{p_0-p_1}\log\frac{1}{\delta}\right)
    \end{equation*}
    queries to $U_\psi$ and 
    \begin{equation*}
        n + \mathcal{O}\left(\log\frac{\|H\|}{\Delta}+ \log \log \frac{1}{\sqrt{p_0}-\sqrt{p_1}}\right)
    \end{equation*}
    qubits.
\end{lemma}
\begin{proof}
    With binary amplitude estimation~\citep[Lemma 7 and Algorithm 1]{lin2020nearoptimal}, we can estimate $\sqrt{p_0}$ up to $\xi$ accuracy with success probability $1-\delta$ using the projector provided by \cref{thmFilter} under trial $\widetilde{\mu}$-s to replace $\mu$. This produces a total runtime complexity
    \begin{equation}
        \mathcal{O}\left(\frac{1}{\sqrt{p_0}\xi\Delta}\log\frac{1}{p_0}\log\frac{1}{\xi}\log\left(\frac{\log(\xi^{-1})}{\delta}\right)\right).
    \end{equation}

    Assuming first that we have access to the value of $\sqrt{p_0}-\sqrt{p_1}$. Setting $\xi = \frac{\sqrt{p_0} - \sqrt{p_1}}{4}$, to estimate $\sqrt{p_0}$, we can then shift the estimate down by a further $\frac{\sqrt{p_0} - \sqrt{p_1}}{2}$ to find
    \begin{equation}
        \widetilde{\mu} \in \left[\frac{\sqrt{p_0} + 3\sqrt{p_1}}{4}, \frac{3\sqrt{p_0} + \sqrt{p_1}}{4} \right].
    \end{equation}
    
    Thus, the total runtime complexity is 
    \begin{equation}
        \widetilde{\mathcal{O}}\left(\frac{1}{(p_0-p_1)\Delta}\log\frac{1}{\delta}\right).
    \end{equation}
    The qubit dependency follows from \cref{thmFilter} by having access to $\widetilde{\rho}_{\rm sqrt}$ to accuracy $\varepsilon \le \frac{\sqrt{p_0}-\sqrt{p_1}}{8}$.

    In the more general case of only having the lower bound $\Lambda$ of $\sqrt{p_0}-\sqrt{p_1}$ per \hyperref[assumptAmpGap]{Assumption~\ref*{assumptAmpGap}}, we can obtain similar results, just by replacing the appearances of $\sqrt{p_0}-\sqrt{p_1}$ with the actual lower bound. Similarly, when using \cref{thmFilter} with the trial $\widetilde{\mu}$, the amplitude-gap dependency in the runtime is replaced by $\Lambda$.
\end{proof}

\subsection{Proof for dominant eigenstate preparation}
\label{appDesp}
For eigenstate preparation, we can use the reflector in \cref{thmFilter} to amplify and prepare the dominant eigenstate.

\eigenstate*
\begin{proof}
    To prepare the dominant eigenstate, we use the reflector from \cref{thmFilter} to perform amplitude amplification on the input $\ket{\psi}$ akin to (robust) ground state preparation~\citep{lin2020nearoptimal,huang2026fullqubit}.

    Given that the overlap between the initial state and the dominant eigenstate is $\lvert\braket{\psi_0|\psi}\rvert = \sqrt{p_0}$, we can use this reflector and initial state preparation unitary $U_\psi$ to produce the dominant eigenstate via $\mathcal O ((\sqrt{p_0})^{-1})$ rounds of amplitude amplification~\citep{brassard2002quantum}. To produce the eigenstate with at least square-root fidelity $1-\varepsilon$, we require the error $\xi$ of the block-encoding $\widetilde{\mathcal R}$ to be within $\sqrt{p_0}\varepsilon$. 
    
    Substituting the block-encoding error in \cref{thmFilter} and multiplying by the amplitude amplification cost, we find that the runtime complexity for eigenstate preparation is 
    \begin{equation}
        \mathcal{O}\left(\frac{\|H\|}{(p_0-p_1)\Delta} \log \frac{1}{p_0\varepsilon}\log\frac{\|H\|}{(p_0-p_1)\Delta\varepsilon}\right).
    \end{equation}

    We can obtain the query complexity of $U_{\psi}$ by adding the number of queries to $\be[\widetilde{\rho}_{\rm sqrt}]$ and to $\mathbb{I} - 2\ketbra{\psi}{\psi}$. This is dominated by the queries to $\widetilde{\rho}_{\rm sqrt}$, where we can see that  
    \begin{equation}
        \mathcal{O}\left(\frac{\|H\|}{p_0-p_1} \log \frac{1}{p_0\varepsilon}\right).
    \end{equation}
    queries to $U_\psi$ and its inverse are needed. This recovers the asymptotics in the main statement.

    Similarly, the number of qubits can be obtained from scaling the error term dependencies.
\end{proof}

Finally, recall from the main text that the additional $(\sqrt{p_0})^{-1}$ cost from amplitude amplification is not fundamental. Here, we construct an algorithm that replaces this additional cost with $(\sqrt{\Delta})^{-1}$ by applying the shifted sign function instead of the rectangle function of the square-root eigenprobability operator to produce $\ketbra{\psi_0}{\phi_0}$. By applying QFT to the right-hand side, we can obtain
\begin{equation}
A =\ketbra{\psi_0}{E_0}
\end{equation}
where $\ket{E_0}$ is a fixed-point encoding of the dominant eigenenergy, identical to the value one would obtain via phase estimation. Given $M \in \mathcal{O}(\Delta^{-1})$, we can prepare a uniform superposition on the eigenenergy space, and apply $\be[A]$ such that
\begin{equation}
    \be[A]\left(\ket{0}\otimes\frac{1}{\sqrt{M}}\sum_{j=0}^{M-1}\ket{j}\right) = \frac{1}{\sqrt{M}}\ket{0}\ket{\psi_0} +\ket{\perp}.
\end{equation}
Then, using the reflector 
\begin{equation}
    2(\ketbra{0}{0}\otimes\ketbra{\psi_0}{\psi_0}) - \mathbb{I},
\end{equation}
which can be obtained by taking the tensor product between $\widetilde{\rho}_{\rm sqrt}$ and $\ketbra{0}{0}$ before the polynomial transformation; we can construct an algorithm that requires total evolution time 
\begin{equation}
    \widetilde{\mathcal{O}}\left(\frac{\|H\|^{3/2}}{(\sqrt{p_0}-\sqrt{p_1})\Delta^{3/2}}\log^2\frac{1}{\varepsilon}\right),
    \label{eqAltPrep}
\end{equation}
as well as
\begin{equation}
    \widetilde{\mathcal{O}}\left(\frac{1}{(\sqrt{p_0}-\sqrt{p_1})\sqrt{\Delta}}\log\frac{1}{\varepsilon}\right)
\end{equation}
queries to $U_\psi$.
This indicates that while $(\sqrt{p_0})^{-1}$ can be removed, and the optimal runtime dependency on the probability distribution is only on $(\sqrt{p_0}-\sqrt{p_1})^{-1}$, it comes at a cost of an additional dependency of $(\sqrt{\Delta})^{-1}$, which in practice is often more costly.

\section{Proof for further applications}
\label{appApp}
\subsection{Dominant eigenenergy estimation}
First, we prove the runtime complexity for eigenenergy estimation. We accomplish this by using the eigenstate prepared by \cref{thmEigenstate} as input to the QMEGS algorithm.

\begin{proposition}[Dominant eigenenergy estimation]
    Using the dominant eigenstate prepared from \cref{thmEigenstate}, we can estimate the dominant eigenenergy within an additive error $\varepsilon$ with success probability $1-\delta$ using an additional phase estimation with
    \begin{equation*}
        \widetilde{\mathcal{O}}\left(\left(\frac{\lVert H\rVert}{(p_0-p_1)\Delta}+\frac{1}{\varepsilon}\right)\log\frac{1}{\delta}\right),
    \end{equation*}
   queries to $U_\tau$, as well as
    \begin{equation*}
    \widetilde{\mathcal{O}}\left(\frac{1}{p_0-p_1}\log\frac{1}{\delta\varepsilon}\right),
    \end{equation*}
    queries to $U_\psi$ and
    \begin{equation*}
    n + \mathcal{O}\left(\log\frac{\|H\|}{\Delta} + \log\log\frac{1}{p_0-p_1}\right)
    \end{equation*}
    qubits.
    \label{propEigenenergy}
\end{proposition}
\begin{proof}
    From \cref{thmEigenstate}, we can prepare the dominant eigenstate $\ket{\psi_0}$ with square-root fidelity at least $1 - \xi$. We denote the prepared state as $\ket{\phi}$ such that
    \begin{equation}
        q_k:=\lvert\braket{\phi|\psi_k}\rvert^2.
    \end{equation}
    We can then note that 
    \begin{equation}
        q_0 \ge (1-\xi)^2.
    \end{equation}
    We can thus write (up to a phase absorbed into $\ket{\phi}$ and $\ket{\perp}$)
    \begin{equation}
        \ket{\phi} = \sqrt{q_0} \ket{\psi_0} + \sqrt{1-q_0} \ket{\perp},
    \end{equation}
    where $\ket{\perp} \in \operatorname{span} \{\ket{\psi_1}, \ket{\psi_2}, \cdots, \ket{\psi_{N-1}}\}$. Thus, we can see that the prepared approximate eigenstate has 
    \begin{equation}
        q_{\rm tail} = \sum_{j=1}^{N-1} q_j = 1- q_0 \le 2\xi-\xi^2.
    \end{equation}

    To satisfy the sufficiently dominant assumption that $q_0 > q_{\rm tail}$, we can set $\xi$ such that 
    \begin{equation}
        q_0 \ge (1- \xi)^2 > 2\xi - \xi^2 \ge q_{\rm tail}
    \end{equation}
    We then apply a statistical phase-estimation procedure, such as QMEGS~\citep{ding2024quantum}, using $\ket{\phi}$ as the input state and evolving it under the Hamiltonian $H$. By \citet[Theorem~3.1]{ding2024quantum}, with probability at least $1-\delta$, the normalised dominant eigenenergy $\lambda_0 = \frac{E_0}{\lVert H\rVert}$ can be estimated to additive error $\varepsilon$ using circuits of maximum depth
    \begin{equation}
        \mathcal{O}\left(\frac{1}{\varepsilon}\log\frac{1}{1-4\xi}\right)
    \end{equation}
    and
    \begin{align}
        &\widetilde{\mathcal{O}}\left(\frac{1}{(q_0-q_{\rm tail})^2}\log\frac{1}{\delta\varepsilon}\right)\nonumber\\
        &\subseteq\widetilde{\mathcal{O}}\left(\frac{1}{\bigl((1-\xi)^2-(2\xi-\xi^2)\bigr)^2}\log\frac{1}{\delta\varepsilon}\right)\nonumber\\
        &\subseteq\widetilde{\mathcal{O}}\left(\frac{1}{1-8\xi}\log\frac{1}{\delta\varepsilon}\right)
    \end{align}
    circuits. In particular, choosing $\xi\leq 1/16$ makes the $\xi$-dependent factors constant. The maximum circuit depth is therefore $\mathcal{O}(\varepsilon^{-1})$, while the number of circuits is
    \begin{equation}
        \widetilde{\mathcal{O}}\left(\log\frac{1}{\delta\varepsilon}\right).
    \end{equation}
    Consequently, measuring the runtime by the total Hamiltonian-evolution
    time across all circuits gives
    \begin{equation}
        \widetilde{\mathcal{O}}\left(\frac{1}{\varepsilon}\log\frac{1}{\delta}
        \right).
    \end{equation}

    Each of these parallel circuit executions requires a fresh preparation of the input state. Thus, for each of the $\widetilde{\mathcal{O}}\left(\log(\delta^{-1}\varepsilon^{-1})\right)$ circuits, we prepare $\ket{\psi_0}$ to square-root fidelity at least $7/8$. By \cref{thmEigenstate}, the total runtime complexity is 
    \begin{equation}
        \mathcal{O}\left(\frac{1}{(p_0-p_1)\Delta}\log^2\frac{1}{\varepsilon}\log\frac{1}{\delta\varepsilon}\right),
    \end{equation}
    after taking the circuit repetition cost into account. Combining the state preparation runtime complexity and time evolution cost for QMEGS, which is an additional 
    \begin{equation}
        \mathcal{O}\left(\frac{1}{\varepsilon}\log\frac{1}{\delta\varepsilon}\right),
    \end{equation} 
    runtime cost (or queries to $U_\tau= \exp(-iH\tau)$) from the evolution time for the controlled $U_\tau$ used in the actual single-ancilla phase estimation routine; we obtain the cost stated in the main result. 
    
    The qubit count follows similarly to \cref{thmEigenstate}. Recall that the square-root fidelity is set to $15/16$, so there are no longer asymptotic dependencies of accuracy $\varepsilon$ in the ancilla count, and that QMEGS adds only one additional ancilla qubit.

    The number of queries of $U_{\psi}$ similarly follows from \cref{thmEigenstate}, multiplied by $\mathcal{O}(\log(\delta^{-1}\varepsilon^{-1}))$ repetition circuits from QMEGS.
\end{proof}

As noted in the main text, we can trade off the logarithmic $\varepsilon$-dependency in the query complexity to $U_\psi$ with an increased number of qubits such that we require
\begin{equation*}
    \widetilde{\mathcal{O}}\left(\frac{1}{p_0-p_1}\log\frac{1}{\delta}\right)
\end{equation*}
queries to $U_\psi$ and
\begin{equation*}
    n + \mathcal{O}\left(\max\left\{\log\frac{\|H\|}{\Delta} + \log\log\frac{1}{p_0-p_1}, \log\frac{1}{\varepsilon}\right\}\right)
\end{equation*}
qubits by using the dominant eigenstate as input to Gaussian phase estimation as shown in \cref{lemGaussQPESingle} instead.

\subsection{Dominant eigenproperty estimation}
Next, we present the results for eigenproperty estimation. To do this, we perform the Hadamard test on the block-encoding of observable $O$, and improve the runtime complexity result with amplitude estimation. We replace the reflection around the extended state $\ket{\psi_0}\ket{0}$ with the modified reflector from \cref{thmFilter}. 

\begin{proposition}[Dominant eigenproperty estimation]
    Using the dominant eigenstate reflector from \cref{thmFilter} and an initial dominant eigenstate prepared from \cref{thmEigenstate},  given block-encoding access to an observable $O$ with scaling factor $\alpha$, we can estimate the dominant eigenproperty $\braket{\psi_0|O|\psi_0}$ within additive error $\varepsilon$ with success probability $1-\delta$ with
    \begin{equation*}
        \widetilde{\mathcal{O}}\left(\frac{\|H\|}{(\sqrt{p_0}-\sqrt{p_1})\Delta}\left(\frac{1}{\sqrt{p_0}}+\frac{\alpha}{\varepsilon}\right)\log\frac{1}{\delta}\right),
    \end{equation*}
   queries to $U_\tau$, as well as
    \begin{equation*}
        \widetilde{\mathcal{O}}\left(\frac{1}{\sqrt{p_0}-\sqrt{p_1}}\left(\frac{1}{\sqrt{p_0}}+\frac{\alpha}{\varepsilon}\right)\log\frac{1}{\delta}\right)
    \end{equation*}
    queries to $U_\psi$ and \begin{equation*}
    n + \mathcal{O}\left(\log \frac{\|H\|}{\Delta} + \log\log\frac{\alpha}{(p_0-p_1)\varepsilon}\right)
    \end{equation*}
    qubits.
    \label{propEigenprop}
\end{proposition}
\begin{proof}
    Given access to $O$ via a block-encoding unitary $\be[\tfrac{O}{\alpha}]$, we can compute the dominant eigenproperty with 
    \begin{equation}
        \frac{1}{\alpha}\braket{\psi_0|O|\psi_0} = \tr\left(\be[\tfrac{O}{\alpha}] (\ketbra{\psi_0}{\psi_0}\otimes \ketbra{0}{0})\right).
    \end{equation}
    We can estimate this using either the observable estimation algorithm from \citet{rall2020quantum} or the Hadamard test~\citep{cleve1998quantum} enhanced by amplitude estimation. To simplify the readouts, we will continue the proof using the Hadamard test. For the Hadamard test, we prepare the following state and measure Pauli-$X$ on the first qubit:
    \begin{equation}
        {\rm C}\mhyphen \be[\tfrac{O}{\alpha}] \ket{+}_{\mathcal{A}}\ket{0}_{\mathcal{B}}\ket{\psi_0}_{\mathcal{C}},
    \end{equation}
    where register $\mathcal{A}$ acts as the control of ${\rm C}\mhyphen \be[\tfrac{O}{\alpha}]$, register $\mathcal{B}$ contains the ancilla qubits for the block-encoding, and the target operator acts on register $\mathcal{C}$.
    
    The Grover iterate $\mathcal{G}$ would then be
    \begin{align}
    \mathcal{G}&=\begin{multlined}[t]({\rm C}\mhyphen \be[\tfrac{O}{\alpha}]) (\mathbb{I} - 2\ketbra{+}{+}_{\mathcal{A}}\otimes\ketbra{0}{0}_{\mathcal{B}}\otimes\ketbra{\psi_0}{\psi_0}_{\mathcal{C}})\\
    ({\rm C}\mhyphen \be[\tfrac{O}{\alpha}]^\dagger)(X_{\mathcal{A}}\otimes \mathbb{I}_{\mathcal{BC}})
    \end{multlined}\nonumber\\
    &=\begin{multlined}[t]({\rm C}\mhyphen \be[\tfrac{O}{\alpha}]) ({\rm Had}_{\mathcal{A}}\otimes \mathbb{I}_{\mathcal{BC}}) \mathcal{R}_{\rm ext}({\rm Had}_{\mathcal{A}}\otimes \mathbb{I}_{\mathcal{BC}})\\
    ({\rm C}\mhyphen \be[\tfrac{O}{\alpha}]^\dagger)(X_{\mathcal{A}}\otimes \mathbb{I}_{\mathcal{BC}}),
    \end{multlined}
    \end{align}
    where
    \begin{equation}
        \mathcal{R}_{\rm ext} = \mathbb{I} - 2\ketbra{0}{0}_{\mathcal{A}}\otimes \ketbra{0}{0}_{\mathcal{B}} \otimes \ketbra{\psi_0}{\psi_0}_{\mathcal{C}}.
    \end{equation}
    From the qubitisation formulation of amplitude amplification~\citep{martyn2021grand}, the Grover iterate $\mathcal G$ has eigenvalues $e^{\pm i\arccos(\gamma)}$ in the non-trivial subspace~\citep[Proposition 1]{huang2026low}, where $\gamma$ is the output of the Hadamard test as follows:
    \begin{align}
        \gamma &=\begin{multlined}[t]\bra{+}_{\mathcal{A}}\bra{0}_{\mathcal{B}}\bra{\psi_0}_{\mathcal{C}} {\rm C}\mhyphen \be[\tfrac{O}{\alpha}]^\dagger (X_{\mathcal{A}}\otimes \mathbb{I}_{\mathcal{BC}})\\
        {\rm C}\mhyphen \be[\tfrac{O}{\alpha}]\ket{+}_{\mathcal{A}}\ket{0}_{\mathcal{B}}\ket{\psi_0}_{\mathcal{C}}
        \end{multlined}\nonumber\\
        &= \frac{1}{\alpha}\braket{\psi_0|O|\psi_0}.
    \end{align}

    To produce $\mathcal{R}_{\rm ext}$, we can modify \cref{thmFilter} by first taking the tensor product of the block-encodings of $\widetilde{\rho}_{\rm sqrt}$ and $\ketbra{0}{0}$~\citep[Lemma 1]{camps2020approximate}, the latter of which can be obtained by LCU of the zero state reflector $\mathbb{I}-2\ketbra{0}{0}$ and identity $\mathbb{I}$. Note that the tensor product is written as follows:
    \begin{equation}
        \sum_{k} \sqrt{p_k} \ketbra{0}{0}_{\mathcal{AB}} \otimes \ketbra{0}{\psi_k}_{\mathcal{C}}\otimes \ketbra{\phi_k}{0}_{\mathcal{D}'},
    \end{equation}
    where the original $\be[\widetilde{\rho}_{\rm sqrt}]$ is applied on the input register $\mathcal{C}$ and output register $\mathcal{D}'$. Thus, applying the same polynomial function $F$ used in \cref{thmFilter} produces a block-encoding of $\mathcal{R}_{\rm ext}$, where we denote the actual block-encoded matrix as $\widetilde{\mathcal{R}}_{\rm ext}$. 
    We join the additional qubits used for the QSVT phase angles and for block-encoding $\ketbra{0}{0}$ with register $\mathcal{D}'$ to form the ancilla register $\mathcal{D}$.
    
    Suppose we prepare a block-encoding of $\widetilde{\mathcal{R}}_{\rm ext}$ such that
    \begin{equation}
        \left\lVert\widetilde{\mathcal{R}}_{\rm ext} - \mathcal{R}_{\rm ext}\right\rVert \le \xi,
    \end{equation}
    then, per unitary invariance of spectral norms, we can also prepare $\be[\widetilde{\mathcal{G}}]$, a $(1, a, \xi)$-block-encoding of $\mathcal{G}$, where $a$ is the number of ancillae required to implement $\mathcal{R}_{\rm ext}$, that is, register $\mathcal{D}$.
    
    By \citet[Corollary 55, arXiv version]{gilyen2019quantum}, we can implement a $(1, a, 4m^2\xi)$-block-encoding of $\mathcal{G}^m$, where $\mathcal{G}^m$ can be used for amplitude estimation.

    To perform amplitude estimation, we require an input state 
    \begin{equation}
        ({\rm C}\mhyphen \be[\tfrac{O}{\alpha}]_{\mathcal{ABC}} \otimes \mathbb{I}_{\mathcal{D}})\ket{+}_{\mathcal{A}}\ket{0}_{\mathcal{B}}\ket{\psi_0}_{\mathcal{C}}\ket{0}_{\mathcal{D}}.
    \end{equation}
    In particular, we can prepare $\ket{\phi}$ such that $\lvert\braket{\phi|\psi_0}\rvert \ge 1 - \eta$ via \cref{thmEigenstate}. For amplitude estimation methods without ancillae, we obtain signals of $\cos((2m+1)\cos^{-1}(\gamma))$, which is the Chebyshev polynomial of the first kind, and use these signals to recover $\gamma$~\citep{rall2023amplitude,sun2026quantum}. We implement this by preparing the state 
    \begin{equation}
        \be[\widetilde{\mathcal G}^m] ({\rm C}\mhyphen \be[\tfrac{O}{\alpha}]_{\mathcal{ABC}} \otimes \mathbb{I}_{\mathcal{D}})\ket{+}_{\mathcal{A}}\ket{0}_{\mathcal{B}}\ket{\psi_0}_{\mathcal{C}}\ket{0}_{\mathcal{D}}
    \end{equation}
    and measuring Pauli-$X$ on the second register to obtain the following:
    \begin{multline}
        \bra{+}_{\mathcal{A}}\bra{0}_{\mathcal{B}}\bra{\psi_0}_{\mathcal{C}}\bra{0}_{\mathcal{D}}({\rm C}\mhyphen \be[\tfrac{O}{\alpha}]_{\mathcal{ABC}}^\dagger\otimes \mathbb{I}_{\mathcal{D}})\\\be[\widetilde{\mathcal G}^m]^\dagger( X_{\mathcal{A}}\otimes \mathbb{I}_{\mathcal{BCD}})\be[\widetilde{\mathcal G}^m]\\ ( {\rm C}\mhyphen \be[\tfrac{O}{\alpha}]_{\mathcal{ABC}}\otimes \mathbb{I}_{\mathcal{D}})\ket{+}_{\mathcal{A}}\ket{0}_{\mathcal{B}}\ket{\psi_0}_{\mathcal{C}}\ket{0}_{\mathcal{D}},
    \end{multline}
    as shown in \cref{figEigenprop}. 
    
    We now analyse the error between the expected measurement result, obtained by this procedure, and $\gamma$.
    Let us quantify the 3 sources of error in our implementation. The first is the error propagated from the imperfect state $\ket{\phi}$. Using $W_m = ({\rm C}\mhyphen \be[\tfrac{O}{\alpha}]_{\mathcal{ABC}}^\dagger\otimes \mathbb{I}_{\mathcal{D}})\be[\widetilde{\mathcal G}^m]^\dagger( X_{\mathcal{A}}\otimes \mathbb{I}_{\mathcal{BCD}})\be[\widetilde{\mathcal G}^m] ( {\rm C}\mhyphen \be[\tfrac{O}{\alpha}]_{\mathcal{ABC}}\otimes \mathbb{I}_{\mathcal{D}})$, we obtain
    \begin{align}
        &\left\lvert\braket{+,0,\psi_0,0|W|+,0,\psi_0,0} - \braket{+,0,\phi,0|W|+,0,\phi,0}\right\rvert\nonumber\\
        &\le \lVert\ketbra{\psi_0}{\psi_0} - \ketbra{\phi}{\phi}\rVert_1 \le 2\sqrt{2\eta}.
    \end{align}
    Next is the error from the implementation of $\widetilde{\mathcal{G}}$. Let $\ket{\chi} = ({\rm C}\mhyphen \be[\tfrac{O}{\alpha}])\ket{+}_{\mathcal{A}}\ket{0}_{\mathcal{B}}\ket{\psi_0}_{\mathcal{C}}$ and 
    \begin{equation}
        \be[\widetilde{\mathcal G}^m]\ket{\chi}_{\mathcal{ABC}}\ket{0}_{\mathcal{D}} = \ket{\varphi_m}_{\mathcal{ABC}}\ket{0}_{\mathcal{D}}+\ket{\perp_m},
    \end{equation}
    where $(\mathbb I_{\mathcal{ABC}} \otimes \bra{0}_{\mathcal{D}})\ket{\perp_m}=0$.
    
    Then, by the block-encoding error of $\mathcal G^m$, there exists $\varepsilon_m \le 4m^2\xi$ such that $\left\lVert\ket{\varphi_m}-\mathcal G^m\ket{\chi}\right\rVert\le \varepsilon_m.$
    Moreover, since both $\be[\widetilde{\mathcal G}^m]$ and $\mathcal G^m$ are unitary,
    we have 
    \begin{equation}
        \lVert\ket{\varphi_m}\rVert^2+\lVert\ket{\perp_m}\rVert^2=1
        \quad\text{and}\quad
        \lVert\mathcal G^m\ket{\chi}\rVert=1.
    \end{equation}
    Hence $\lVert\ket{\varphi_m}\rVert \ge 1-\varepsilon_m$,
    so
    \begin{equation}
        \lVert\ket{\perp_m}\rVert^2 = 1-\lVert\ket{\varphi_m}\rVert^2 \le 2\varepsilon_m-\varepsilon_m^2 \le 2\varepsilon_m.
    \end{equation}
    Therefore,
    \begin{align}
        &\left\lVert \be[\widetilde{\mathcal G}^m]\ket{\chi}_{\mathcal{ABC}}\ket{0}_{\mathcal{D}} -\mathcal G^m\ket{\chi}_{\mathcal{ABC}}\ket{0}_{\mathcal{D}}\right\rVert\nonumber\\
        &\le \left\lVert\ket{\varphi_m}-\mathcal G^m\ket{\chi}\right\rVert + \lVert\ket{\perp_m}\rVert \le \varepsilon_m+\sqrt{2\varepsilon_m}\nonumber\\
        &\in \mathcal{O}(\sqrt{\varepsilon_m}) \subseteq\mathcal{O}(m\sqrt{\xi}).
    \end{align}
    Adding both errors together results in 
    \begin{equation}
        \left\lvert\braket{+,0,\phi,0|W_m|+,0,\phi,0}-\gamma\right\rvert \in \mathcal{O}(\sqrt{\eta} + m \sqrt{\xi}).
    \end{equation}
    For simplicity, we assume iterative amplitude estimation methods~\citep{grinko2021iterative,rall2023amplitude} to recover the amplitude, though we expect the following results to hold for other protocols. 

    Given the error $\mathcal{O}(\sqrt{\eta}+m\sqrt{\xi})$ arising from imperfect eigenstate preparation and the eigenstate projector, we enlarge the confidence interval at an iterative amplitude estimation round using $m$ Grover iterations by an additional $\mathcal{O}(\sqrt{\eta}+m\sqrt{\xi})$ to account for the implementation error. Let $m_{\max}$ be the maximum number of Grover iterations in a single round in the iterative process, that is, the number of iterations in the final round. Then the total error at the final iteration satisfies
    \begin{equation}
        |\widetilde{\gamma} - \gamma| \in \mathcal{O}\left(\sqrt{\eta} + m_{\max} \sqrt{\xi} + \frac{1}{m_{\max}}\right) \le \frac{\varepsilon}{\alpha}
    \end{equation}
    where the $\mathcal{O}(m_{\max}^{-1})$ term is the statistical error associated with the Heisenberg-limited amplitude-estimation procedure. It therefore suffices to choose
    \begin{equation}
        m_{\max} \in \mathcal{O}\left(\frac{\alpha}{\varepsilon}\right),\quad \eta  \in \mathcal{O}\left(\frac{\varepsilon^2}{\alpha^2}\right), \quad \xi \in \mathcal{O}\left(\frac{\varepsilon^4}{\alpha^4}\right).
    \end{equation}
    Here, $m_{\max}$ denotes the maximum number of Grover iterations used in any amplitude-estimation round, $\eta$ is the error in the input eigenstate $\ket{\psi}$, and $\xi$ is the error of the block-encoding $\be[\widetilde{\mathcal G}]$.
     
    Substituting the error bounds back into \cref{thmFilter} and \cref{thmEigenstate} for the preparation cost of $\ket{\phi}$ and $\mathcal{R}_{\rm ext}$, we recover the runtime cost of 
    \begin{equation}
        \widetilde{\mathcal{O}}\left(\frac{\|H\|}{(\sqrt{p_0}-\sqrt{p_1})\Delta}\left(\frac{1}{\sqrt{p_0}}+\frac{\alpha}{\varepsilon}\right)\log\frac{1}{\delta}\right),
    \end{equation}
    where the additional $\mathcal{O}(\log(\delta^{-1}))$ stems from median boosting the success rate.
    
    The qubit count and query complexity to $U_{\psi}$ follow similarly from \cref{thmFilter} and \cref{thmEigenstate} by plugging in the requirements for $\eta$ and $\xi$.
\end{proof}

\begin{figure}
    \includegraphics[width=\linewidth]{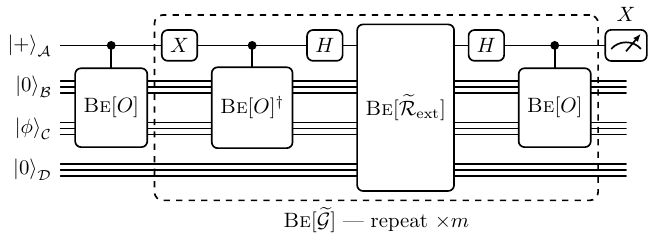}
    \caption{Circuit implementation for eigenproperty estimation. Register $\mathcal{A}$ is the Hadamard test readout qubit, register $\mathcal{B}$ is the ancilla register for the block-encoding unitary $\be[\tfrac{O}{\alpha}]$, register $\mathcal{C}$ is the main register of $\be[\tfrac{O}{\alpha}]$, and register $\mathcal{D}$ is the additional ancillae register for the block-encoding $\be[\widetilde{\mathcal R}_{\rm ext}]$.} 
    \label{figEigenprop}
\end{figure}

\subsection{Fidelity estimation with the dominant eigenstate}
Lastly, we prove the runtime complexity for fidelity estimation between some pure state $\ket{\varphi}$ and the dominant eigenstate, which can be obtained by amplitude estimation with reflectors $\mathbb{I}-2\ketbra{\varphi}{\varphi}$ and $\mathcal{R}$ from \cref{thmFilter}.

\begin{proposition}[Fidelity estimation]
    Using the dominant eigenstate reflector from \cref{thmFilter}, and given a state preparation unitary $U_{\varphi}$ that prepares $\ket{\varphi}$, we can estimate the fidelity $\lvert\braket{\varphi|\psi_0}\rvert^2$ within additive error $\varepsilon$ using
    \begin{equation*}
        \widetilde{\mathcal{O}}\left(\frac{\|H\|}{(\sqrt{p_0}-\sqrt{p_1})\Delta\varepsilon}\log\frac{1}{\delta}\right),
    \end{equation*}
   queries to $U_\tau$, as well as
    \begin{equation*}
        \widetilde{\mathcal{O}}\left(\frac{1}{(\sqrt{p_0}-\sqrt{p_1})\varepsilon}\log\frac{1}{\delta}\right)
    \end{equation*}
    queries to $U_\psi$, $\mathcal{O}(\varepsilon^{-1}\log(\delta^{-1}))$ queries to $U_\varphi$, and 
    \begin{equation*}
    n + \mathcal{O}\left(\log\frac{\|H\|}{\Delta} + \log\log\frac{1}{(\sqrt{p_0}-\sqrt{p_1})\varepsilon}\right)
    \end{equation*}
    qubits.
    \label{propFidelity}
\end{proposition}
\begin{proof}
    Given access to a pure state $\ket{\varphi}$ through its state preparation unitary $U_\varphi$, we can estimate the fidelity using amplitude estimation~\citep{brassard2002quantum}. There are two sources of error: the error $\eta$ from the imperfect block-encoding of $\mathcal{R}$ produced by \cref{thmFilter}, and the estimation error $\xi$ from amplitude estimation.

    We denote the imperfect block-encoding as $\widetilde{\mathcal{R}}$ and restructure the fidelity in terms of the reflector $\mathcal{R}$ as 
    \begin{equation}
    \lvert\braket{\varphi|\psi_0}\rvert^2 = \braket{\varphi|\mathcal{P}|\varphi} = \frac{1-\braket{\varphi|\mathcal{R}|\varphi}}{2}.
    \end{equation}
    The error can be written as
    \begin{equation}
        \frac{1-\braket{\varphi|\mathcal{R}|\varphi}}{2} - \frac{1-\braket{\varphi|\widetilde{\mathcal{R}}|\varphi}}{2} \le \frac{\eta}{2}.
    \end{equation}
    Setting $\eta = \varepsilon$ and $\xi = \varepsilon/2$, we can obtain the runtime cost of the $\mathcal{O}(\xi^{-1})\subseteq\mathcal{O}(\varepsilon^{-1})$ queries to $\widetilde{\mathcal{R}}$ for amplitude estimation as 
    \begin{equation}
        \mathcal{O}\left(\frac{\|H\|}{(\sqrt{p_0}-\sqrt{p_1})\Delta\varepsilon}\log\frac{1}{\delta}\right).
    \end{equation}
    The query complexity follows from the same result, and the qubit count follows from \cref{thmFilter} for preparing $\mathcal{R}$.
\end{proof}

\section{Proofs for lower bounds}
\label{appOptim}

\subsection{Proof for filtering lower bounds}
\label{appFilterOpt}
To provide a lower bound for dominant eigenstate filtering, we consider the easier setting of the computational basis, instead of a Hamiltonian eigenbasis. We obtain the following bounds by a reduction to the probability distinguishing problem from \citet{belovs2019quantum}.
\filterOpt*
\begin{proof}
    We provide this lower bound by reducing a specific instance of the probability distribution distinction problem to the dominant eigenstate preparation problem.

    In particular, consider the problem of distinguishing between two discrete probability distributions over a finite set $\mathcal{S}$:
    \begin{align}
        \vec{p}^{(A)} &= (p_0, p_1, p_2, \ldots, p_{|\mathcal{S}|-1}), \\
        \vec{p}^{(B)} &= (p_1, p_0, p_2, \ldots, p_{|\mathcal{S}|-1}),
    \end{align}
    where $p_0 > p_1 \ge p_2 \ge \cdots \ge p_{|\mathcal{S}|-1}$. The two distributions differ only by a swap of the first two probabilities.

    We assume quantum sample access via the state-preparation unitaries $U_{A}$ and $U_{B}$ such that
    \begin{subequations}
    \begin{align}
        U_{A}\ket{0} &= \sum_{k \in \mathcal{S}} \sqrt{p_k}\,\ket{k}, \\
        U_{B}\ket{0} &= \sqrt{p_1}\ket{0} + \sqrt{p_0}\ket{1} + \sum_{k \in \mathcal{S}\setminus\{0,1\}} \sqrt{p_k}\,\ket{k}.
    \end{align}
    \end{subequations}

    By \citet[Theorem 4]{belovs2019quantum}, given query access to such state-preparation unitaries, the quantum query complexity of distinguishing $\vec{p}^{(A)}$ from $\vec{p}^{(B)}$ is
    \begin{equation}
    \Theta\left(\frac{1}{\mathcal{D}_H\left(\vec{p}^{(A)}, \vec{p}^{(B)}\right)}\right),
    \end{equation}
    where $\mathcal{D}_H$ is the Hellinger distance, defined as
    \begin{equation}
    \mathcal{D}_H(\vec{p}, \vec{q}) = \frac{1}{\sqrt{2}} \left( \sum_{s \in \mathcal{S}} \left(\sqrt{p_s} - \sqrt{q_s}\right)^2 \right)^{1/2}.
    \end{equation}

    In the present setting, only the first two probabilities differ, so
    \begin{align}
    \mathcal{D}_H^2\left(\vec{p}^{(A)}, \vec{p}^{(B)}\right)
    &= \frac{1}{2} \left( \left(\sqrt{p_0} - \sqrt{p_1}\right)^2 + \left(\sqrt{p_1} - \sqrt{p_0}\right)^2 \right) \nonumber\\
    &= \left(\sqrt{p_0} - \sqrt{p_1}\right)^2,
    \end{align}
    and hence
    \begin{equation}
    \mathcal{D}_H\left(\vec{p}^{(A)}, \vec{p}^{(B)}\right)
    = \left|\sqrt{p_0} - \sqrt{p_1}\right|.
    \end{equation}

    Since $p_0 > p_1$, this simplifies to
    \begin{equation}
    \mathcal{D}_H\left(\vec{p}^{(A)}, \vec{p}^{(B)}\right)
    = \sqrt{p_0} - \sqrt{p_1}.
    \end{equation}

    Next, we show that by applying the dominant eigenstate preparation algorithm, we can reliably discern between the two probability distributions. We assume that we can construct a block-encoding of $\rho_{\rm sqrt}$ with $\mathcal{O}(1)$ queries to $U_{\psi}$. For general eigenstates of a Hamiltonian, we can use the QPE circuit as shown in the main text, while for computational basis states, we can use the vector encoding~\citep{guo2024nonlinear,rattew2023nonlinear} to encode the probability distribution into the singular values of $\rho_{\rm sqrt}$.
    
    Given the dominant eigenvector (in the computational basis) of $U_{A}\ket{0}$ is $\ket{0}$ and $U_{B}\ket{0}$ is $\ket{1}$, we can prepare these two states respectively up to square-root fidelity $\frac{\sqrt{3}}{{2}}$ with our algorithm given block-encoding access to $\rho_{\rm sqrt}$ and query access to $U_{\psi}$. With the prepared eigenstate, we simply measure the first qubit. If, with probability $3/4$, we have output 0, then the input distribution is $\vec{p}^{(A)}$, and vice versa. This gives us the ability to distinguish between $\vec{p}^{(A)}$ and $\vec{p}^{(B)}$, completing the reduction.

    If $p_0 \in \Omega(1)$, then we apply only constant rounds for amplitude amplification, and the cost of dominant eigenstate preparation saturates the lower bound of $\Omega((\sqrt{p_0}-\sqrt{p_1})^{-1})$ provided by probability distribution distinction.
\end{proof}

As noted, given that $\rho_{\rm sqrt}$ can be obtained by vector encoding methods for computational basis states, this proof can also establish a matching lower bound for the maximum probability finding procedure of \citet{rattew2023nonlinear}, subject to slight modifications to match the input oracle and output conventions. We again prepare the two distributions $\vec{p}^{(A)}$ and $\vec{p}^{(B)}$ and, via maximum probability finding, obtain block-encodings of $\ketbra{0}{0}$ and $\ketbra{1}{1}$, respectively. We then apply the block-encoding unitary to the all-zero state $\ket{0}\ket{0}_a$ and perform a computational-basis measurement. In particular,
\begin{align}
    \Pr(0,0_a) &= \left\lvert\bra{0}\braket{0|\be[A]|0}\ket{0}_a\right\rvert^2 \nonumber\\
    &= \left\lvert\braket{0|A|0}\right\rvert^2 = \begin{cases}
    1, & A=\ketbra{0}{0},\\
    0, & A=\ketbra{1}{1}.
    \end{cases}
\end{align}
Thus, if the outcome $\ket{0}\ket{0}_a$ is observed with high probability, we conclude that the distribution is $\vec{p}^{(A)}$; otherwise, it is $\vec{p}^{(B)}$. Hence the $\Omega((\sqrt{p_0}-\sqrt{p_1})^{-1})$ lower bound would also apply. Further, this version of the proof also directly provides a lower-bound proof for the filtering step of the eigenprobability filtering protocol as well, though due to the output of the filtering step being a block-encoding, this modified version of the proof is tailored to this access model.

\subsection{Proof for lifted purity amplification}
\label{appPurity}
We consider the task of purity amplification, where given samples of the mixed state
\begin{equation}
\rho = \sum_{k\in\mathcal{S}} p_k \ketbra{\psi_k}{\psi_k},
\end{equation}
where we retain our assumptions on the eigenprobability operator and assume
\begin{equation}
    p_0 > p_1 \ge \ldots \ge p_{|\mathcal S|-1} > 0 .
\end{equation}

Results in purity amplification~\citep{cirac1999optimal,li2024optimal,grier2025streaming} show that
\begin{equation}
    \Theta\left(\frac{1-p_0}{(p_0-p_1)^2\varepsilon}\right)
\end{equation}
samples of $\rho$ are required to obtain the principal eigenstate $\ket{\psi_0}$ up to fidelity $1-\varepsilon$.

We can use such optimality results to provide a query lower bound for the case where we have purified quantum query access instead. Recall from the main text that this is the case where we have query access to $V$ where
\begin{equation}
    V\ket{0}_A\ket{0}_B= \sum_{k\in \mathcal{S}} \sqrt{p_k} \ket{\psi_k}_A\ket{\phi_k}_B.
\end{equation}
To obtain the query complexity lower bounds, we can apply the techniques of sample-to-query lifting~\citep{wang2025quantum,chen2025list} to ``lift'' the sample complexity bounds to query complexity bounds. While the original framework is restricted to promise problems and block-encoded representations, this was generalised by \citet{tang2025conjugate} for the purified quantum query access model and for simulation tasks in general. We recount the result as follows:

\begin{lemma}[Simulating queries to a state preparation unitary given copies of the state -- Theorem 1.5~\citep{tang2025conjugate}]
    Consider a circuit that uses $Q$ calls to $U$ and $U^\dagger$, where $U$ provides purified quantum query access to the mixed state $\rho$. Let $\omega_U$ denote the output of this circuit. Then, using $n \in \mathcal{O}(Q^2 / \xi)$ copies of $\rho$, we can simulate the output of this circuit on a distribution over state preparation unitaries $V$ of $\rho$, that is, $\mathbb{E}_{V}[\omega_V]$, up to an $\xi$ error in trace distance.
    \label{propLift}
\end{lemma}

Using this result, we can show the following theorem via a proof by contradiction:
\purityOpt*
\begin{proof}
    Suppose we have a quantum query algorithm that, for every state preparation unitary $V$ of $\rho$, outputs the state $\omega_V$, where 
    \begin{equation}
        \braket{\psi_0|\omega_V|\psi_0} \ge \frac{63}{64},
    \end{equation}
    with $Q$ calls to $V$. Equivalently, the square-root fidelity of $\omega_V$ with $\ket{\psi_0}$ is at least $3\sqrt7/8$.

    The trace distance of $\omega$ and $\ketbra{\psi_0}{\psi_0}$ can be bounded using Fuchs-van de Graaf inequalities~\citep{nielsen2010quantum} such that
    \begin{equation}
        \mathcal{D}_{\tr}(\omega_V, \ketbra{\psi_0}{\psi_0}) \le \sqrt{1-\braket{\psi_0|\omega_V|\psi_0}} \le \frac{1}{8}.
    \end{equation}
    
    Let $\bar{\omega}:=\mathbb E_{V}[\omega_V]$, where the expectation is over the distribution of state-preparation unitaries appearing in \cref{propLift}. By convexity of the trace distance,
    \begin{equation}
        \mathcal{D}_{\tr}\left(\bar{\omega},\ketbra{\psi_0}{\psi_0}\right)\leq\mathbb E_V [\mathcal D_{\rm tr}\left(\omega_V,\ketbra{\psi_0}{\psi_0}\right)]\leq\frac18.
    \end{equation}
    Then, by \cref{propLift}, we have a sampling algorithm that outputs $\omega'$, where $\omega'$ is within $1/8$ trace distance from the averaged output state $\bar\omega$, using $\mathcal{O}(Q^2)$ samples to $\rho$. By the triangle inequality
    \begin{equation}
        \mathcal{D}_{\rm tr}(\omega', \ketbra{\psi_0}{\psi_0}) \le \frac{1}{4}.
    \end{equation}
    Again, by the Fuchs-van de Graaf inequalities, we can find that
    \begin{equation}
        \braket{\psi_0|\omega'|\psi_0} \ge \left(1-\mathcal{D}_{\rm tr}(\omega', \ketbra{\psi_0}{\psi_0})\right)^2 \ge \frac{9}{16}.
    \end{equation} 

    Now suppose that the query complexity
    \begin{equation}
        Q \in o\left(\frac{\sqrt{1-p_0}}{p_0-p_1}\right).
    \end{equation}
    Then, by the above result, we have a sampling algorithm that has only
    \begin{equation}
        o\left(\frac{1-p_0}{(p_0-p_1)^2}\right)
    \end{equation}
    sample complexity, and produces $\omega'$. This contradicts the lower bound provided by \citet[Theorem 32]{grier2025streaming}, which states that any sampling algorithm for purity amplification that produces a state $\omega'$ achieving fidelity of at least $9/16$ of $\ket{\psi_0}$ requires 
    \begin{equation}
        \Omega\left(\frac{1-p_0}{(p_0-p_1)^2}\right)
    \end{equation}
    samples.
    
    Thus, by proof by contradiction, we must have
    \begin{equation}
        Q \in \Omega\left(\frac{\sqrt{1-p_0}}{p_0-p_1}\right).
    \end{equation}
\end{proof}

\end{document}